\documentclass[12pt,letterpaper]{article}

\usepackage{etoolbox}
\usepackage{amsmath,amssymb,amsthm,amsfonts,amstext}
\usepackage[margin=1in,letterpaper]{geometry}
\usepackage[longnamesfirst]{natbib}
\usepackage{enumerate}
\usepackage{bm}
\usepackage{tabularx,longtable,threeparttable,booktabs}
\usepackage{siunitx}
\usepackage{color}
\usepackage{graphicx}
\usepackage{caption}
\usepackage{setspace}
\usepackage{alltt}
\usepackage{hypernat}
\usepackage{mathtools,dsfont}
\usepackage{hyperref}
\usepackage[capitalize,noabbrev]{cleveref}
\usepackage{pdflscape}
\usepackage[paper=portrait,pagesize]{typearea}

\newtoggle{SUPPLEMENTAL}\toggletrue{SUPPLEMENTAL}
\togglefalse{SUPPLEMENTAL} 

\newtoggle{BLINDED}\toggletrue{BLINDED}
\togglefalse{BLINDED} 

\numberwithin{equation}{section}

\theoremstyle{plain}

\newtheorem{theorem}{Theorem}[section]

\newtheorem{proposition}[theorem]{Proposition}

\theoremstyle{definition}

\appto\TPTnoteSettings{\linespread{1}\footnotesize}
\crefformat{footnote}{#2\footnotemark[#1]#3} 
\crefname{conjecture}{Conjecture}{Conjectures}
\crefname{section}{Section}{Sections}
\crefname{subsection}{Section}{Sections}
\crefname{subsubsection}{Section}{Sections}
\Crefname{conjecture}{Conjecture}{Conjectures}
\Crefname{section}{Section}{Sections}
\Crefname{subsection}{Section}{Sections}
\Crefname{subsubsection}{Section}{Sections}
\crefname{appendix}{Appendix}{Appendices}
\crefname{subappendix}{Appendix}{Appendices}
\crefname{subsubappendix}{Appendix}{Appendices}
\Crefname{appendix}{Appendix}{Appendices}
\Crefname{subappendix}{Appendix}{Appendices}
\Crefname{subsubappendix}{Appendix}{Appendices}
\crefname{equation}{}{}
\Crefname{equation}{Equation}{Equations}
\crefformat{enumi}{(#2#1#3)}
\crefrangeformat{enumi}{(#3#1#4)\crefrangeconjunction(#5#2#6)}
\crefmultiformat{enumi}{(#2#1#3)}{ and~(#2#1#3)}{, (#2#1#3)}{ and~(#2#1#3)}
\Crefformat{enumi}{Part (#2#1#3)}
\Crefrangeformat{enumi}{Parts (#3#1#4)\crefrangeconjunction(#5#2#6)}
\Crefmultiformat{enumi}{Parts (#2#1#3)}{ and~(#2#1#3)}{, (#2#1#3)}{ and~(#2#1#3)}
\crefname{assumption}{}{}
\Crefname{assumption}{Assumption}{Assumptions}
\newcommand{\crefrangeconjunction}{--}

\DeclareGraphicsExtensions{.pdf,.png}
\graphicspath{ {./Figures/} }

\newcommand{\citeposs}[1]{\citeauthor{#1}'s (\citeyear{#1})}

\newcommand{\numnornd}[1]{\num[round-mode=off,group-digits=integer]{#1}} 
\newcommand{\matf}[1]{\underline{\boldsymbol{\mathbf{#1}}}} 
\newcommand{\vecf}[1]{\boldsymbol{\mathbf{#1}}} 

\DeclareMathOperator{\Cov}{Cov}

\newcommand{\R}{{\mathbb R}}

\DeclareMathOperator{\E}{E}
\DeclareMathOperator{\Q}{Q}
\let\Pr\relax \DeclareMathOperator{\Pr}{P} 

\DeclareMathOperator*{\argmin}{arg\,min}

\DeclareMathOperator{\1}{\mathds{1}}
\newcommand{\Ind}[1]{\1\left\{#1\right\}}
\newcommand{\NormDist}{\mathrm{N}}

\newcommand{\UnifDist}{\textrm{Unif}}

\newcommand{\RMSE}{\textrm{RMSE}}
\newcommand{\rRMSE}{\textrm{rRMSE}}

\newcommand{\pD}[2]{\frac{\partial #1}{\partial #2}}

\newcommand{\independenT}[2]{\mathrel{\rlap{$#1#2$}\mkern2mu{#1#2}}}
\newcommand\independent{\protect\mathpalette{\protect\independenT}{\perp}}

\let\originalleft\left
\let\originalright\right
\renewcommand{\left}{\mathopen{}\mathclose\bgroup\originalleft}
\renewcommand{\right}{\aftergroup\egroup\originalright}

  \renewenvironment{thebibliography}[1]{%
    \begin{oldthebibliography}{#1}%
      \setlength{\parskip}{0ex}%
      \setlength{\itemsep}{0ex}%
  }%
  {%
    \end{oldthebibliography}%
  }

\allowdisplaybreaks[1]

\usepackage{array}

\hypersetup{
  colorlinks = false,
  urlcolor = black,
  pdfauthor = {Xin Liu},
  pdfkeywords = {economics, econometrics, statistics},
  pdftitle = {Liu: IVQR Averaging},
  pdfsubject = {econometrics},
  pdfpagemode = UseNone
}

\defcitealias{ChengLiaoShi2019}{CLS}

\title{Averaging estimation for instrumental variables quantile regression}
\author{Xin Liu\thanks{Department of Economics, University of Missouri; email: \texttt{xl6f6@mail.missouri.edu}}}

\begin{document}

\maketitle

\begin{abstract}
This paper proposes averaging estimation methods to improve the finite-sample efficiency of the instrumental variables quantile regression (IVQR) estimation. 
First, I apply \citeposs{ChengLiaoShi2019} averaging GMM framework to the IVQR model.
I propose using the usual quantile regression moments for averaging to take advantage of cases when endogeneity is not too strong.
I also propose using two-stage least squares slope moments to take advantage of cases when heterogeneity is not too strong.
The empirical optimal weight formula of \citet{ChengLiaoShi2019} helps optimize the bias--variance tradeoff, ensuring uniformly better (asymptotic) risk of the averaging estimator over the standard IVQR estimator under certain conditions.
My implementation involves many computational considerations and builds on recent developments in the quantile literature.
Second, I propose a bootstrap method 
that directly averages among IVQR, quantile regression, and two-stage least squares estimators. 
More specifically, I find the optimal weights in the bootstrap world and then apply the bootstrap-optimal weights to the original sample.
The bootstrap method is simpler to compute and generally performs better in simulations, but it lacks the formal uniform dominance results of \citet{ChengLiaoShi2019}.
Simulation results demonstrate that 
in the multiple-regressors/instruments case, 
both the GMM averaging and bootstrap estimators have uniformly smaller risk than the IVQR estimator
across data-generating processes (DGPs) with all kinds of combinations of different endogeneity levels and heterogeneity levels. 
In DGPs with a single endogenous regressor and instrument,
where averaging estimation is known to have least opportunity for improvement, 
the proposed averaging estimators outperform the IVQR estimator in some cases but not others.

\vspace*{0.5\baselineskip}


\end{abstract}

\newpage

\onehalfspacing

\section{Introduction}


Since \citeposs{KoenkerBassett1978} seminal work, quantile regression (QR) has become a useful tool to capture unobserved heterogeneous effects in policy analysis or program evaluation.
However, in practice endogeneity commonly results in inconsistent estimates of the conventional quantile regression.
To address this endogeneity issue in quantile regression, 
\citet{ChernozhukovHansen2005} propose an instrumental variables method to identify the structural quantile function or (conditional) quantile treatment effects.
Subsequently, \citet{ChernozhukovHansen2006} and then many others have provided estimation methods for the instrumental variables quantile regression (IVQR) model.
Although there are important distinctions, for now I simply refer to ``the'' IVQR estimator.

Although the IVQR estimator has desirable large-sample properties like consistency and asymptotic normality,
it can have imprecise estimates due to its large finite-sample variance,\footnote{It is possible that the IVQR variance is infinite, similar to IV \citep{Kinal1980}.
More technically, the word ``variance'' in this paper means a different measure of dispersion/spread that is never infinite, like trimmed/truncated variance or interquartile range.}
just as the two-stage least squares (2SLS) estimator may have a substantial finite-sample dispersion.

The goal of this paper is to propose new estimation methods to improve the finite-sample estimation efficiency of IVQR.
I propose two such methods. 
The first applies the averaging generalized method of moments (GMM) framework of  \citet*{ChengLiaoShi2019} (hereafter \citetalias{ChengLiaoShi2019}) to IVQR. 
The second is a bootstrap averaging method,
which uses the bootstrap world's optimal averaging weights on IVQR, QR, and 2SLS estimators.

This paper has four main contributions.
First, beyond extending the 2SLS/OLS averaging to the analogous IVQR/QR averaging, this paper shows it is also helpful to include 2SLS in averaging to improve on IVQR.
Second, the implementation of \citetalias{ChengLiaoShi2019}/GMM averaging is not trivial for IVQR.
In particular, it needs two-step GMM estimation and nonparametric Jacobian matrix estimation, which are discussed or extended in this paper (and my code).
Third, the bootstrap IVQR/QR/2SLS averaging method is new and outperforms \citetalias{ChengLiaoShi2019}/GMM averaging in simulations, 
although it lacks theoretical results like in \citetalias{ChengLiaoShi2019}.
Fourth, the simulations compare various methods in a wide range of data-generating process (DGP) types (varying endogeneity, heterogeneity, distributional shapes, etc.),
although they are still limited.

\paragraph{Averaging estimation in this paper}
%
My first method follows the framework of \citetalias{ChengLiaoShi2019}.
\citetalias{ChengLiaoShi2019} first define a benchmark conservative GMM estimator based on a set of valid moments.
Then they add an additional set of possibly misspecified moments to the conservative moments to obtain the ``aggressive'' GMM estimator.
On one hand, the aggressive GMM estimator might be more biased than the conservative GMM estimator since additional moments might not be valid.
On the other hand, by adding this additional information, the aggressive GMM estimator might significantly reduce variance and overall mean squared error (MSE) compared to the conservative GMM estimator.
\citetalias{ChengLiaoShi2019} propose an optimal weight formula to average the conservative and aggressive GMM estimators.
They show that under certain conditions the averaging estimator uniformly dominates the conservative GMM estimator in asymptotic risk.

This paper applies \citetalias{ChengLiaoShi2019} to IVQR as follows.
The conservative GMM estimator uses only the IVQR moments.
Two types of additional moments are proposed: 
the conventional QR moments and the 2SLS slope moments (excluding the intercept term). 
The motivation for proposing these two additional moments comes from the bias--variance tradeoff. 
When there is not much endogeneity in the model, 
the conventional QR estimator is little biased; 
meanwhile, the QR estimates usually have smaller variance than IVQR estimates. 
Therefore, introducing QR moments as the additional moments can reduce variance and maybe achieve an overall reduction in MSE, although it might increase the bias. 
When there is not much heterogeneity across quantiles, 
2SLS and IVQR at any quantile will usually have similar slope estimates, thus similar bias; 
at the same time, 2SLS estimates usually have a smaller variance than IVQR estimates, except in cases like a fat-tailed error term.
Therefore, using 2SLS slope moments as the additional moments can also improve efficiency by reducing the overall MSE. 
From another perspective, 2SLS can be viewed as a limiting case of smoothed IVQR estimation as the smoothing bandwidth goes to infinity \citep[\S2.2]{KaplanSun2017}.
These two types of additional moments yield two types of aggressive estimators to average with the IVQR estimator.
I apply \citetalias{ChengLiaoShi2019}'s empirical averaging weight formula to obtain the averaging estimator. 
Simulation results demonstrate the averaging estimator has MSE uniformly below or equal to that of the IVQR estimator at all quantiles under certain uniform dominance conditions.

Besides the \citetalias{ChengLiaoShi2019} GMM averaging method, 
I propose a new bootstrap averaging method. 
The bootstrap averaging method averages the IVQR, QR, and 2SLS estimators in the bootstrap world with a grid of fixed weights, 
and picks the weight that minimizes the robust root mean squared error (robust RMSE, or rRMSE) as the bootstrap optimal weight. 
This optimal weight is then used to average the IVQR, QR, and 2SLS estimators in the original sample to obtain the bootstrap averaging estimator.

The motivation for the bootstrap averaging method is the same as for the additional QR and 2SLS slope moments in the GMM averaging method.
The QR and 2SLS estimators might have smaller MSE than the IVQR estimator in some DGPs with little endogeneity or little heterogeneity, respectively.
The hope is that the bootstrap world is similar enough to the real world that the bootstrap method places more weight on the 2SLS and/or QR estimator 
when they perform better than the IVQR estimator, 
and put more weight on the IVQR estimator when 2SLS and QR have larger MSE than IVQR.

The bootstrap averaging method, although lacking theoretical results, has potential advantages over the \citetalias{ChengLiaoShi2019} averaging GMM method when applied to IVQR.
Bootstrap averaging is easier for computation since it avoids highly over-identified quantile GMM, which has a difficult criterion function to minimize.
Moreover, the bootstrap averaging estimator has better performance in simulation results. 
This is partly because the bootstrap can average among three different estimators, whereas \citetalias{ChengLiaoShi2019} only averages between two estimators.
It might also be because often in finite samples QR outperforms the aggressive IVQR-QR GMM estimator 
and 2SLS outperforms the aggressive IVQR-2SLS GMM estimator 
in situations when the additional moments provide significant variance reduction.

\paragraph{Literature}~

\textit{Averaging estimation }
Averaging estimation originates from Stein-like shrinkage estimation and has recently been reinvestigated and extended by many authors to improve estimation efficiency.
\Citet{JamesStein1961} propose an estimator that shrinks the least squares estimator toward zero.
The James--Stein estimator can be viewed as in the class of averaging estimation, 
in that it averages the least squares estimator and zero.
Under certain conditions, the James--Stein estimator dominates the least squares estimator in terms of a strict reduction of MSE.
A limitation of the James--Stein estimator is it restricts to the normal distribution. 
\Citet{Maasoumi1978} applies the idea of Stein-like shrinkage estimation to simultaneous equations by averaging the three-stage least squares (3SLS) estimator with the least squares (LS) estimator.
He shows while the 3SLS and 2SLS estimators have no finite moments (therefore, unbounded risk) in some cases, the averaging estimator has finite moments (therefore, bounded risk).
Recently, \citet{Hansen2017} applies averaging estimation to a single equation instrumental variables model. 
He averages 2SLS and OLS estimators with weight depending on the statistic for testing exogeneity in the model.
He shows that the averaging estimator uniformly dominates 2SLS in asymptotic risk
when the number of endogenous regressors is greater than two.
It extends the James--Stein framework to a general error term distribution but still limits to homoskedasticity.
The averaging GMM method of \citetalias{ChengLiaoShi2019} works in a general framework with no normality or homoskedasticity restriction.
\citetalias{ChengLiaoShi2019} provide supporting simulations to show the uniform dominance results in their averaging GMM framework can hold with both Gaussian and non-Gaussian errors.

\textit{IVQR estimation and computation }
Since \citeposs{ChernozhukovHansen2005} seminal work on IVQR identification, many researches focus on IVQR estimation and computational efficiency.
The challenging computation of IVQR estimation comes from the non-differentiable IVQR moments and the non-convex GMM objective function. 
\Citet{ChernozhukovHansen2006} first propose a two-step inverse quantile regression method with a grid search on the endogenous regressors' coefficients to compute the IVQR estimator.
This IVQR estimator is asymptotically equivalent to a GMM estimator.
However, 
its computation time scales poorly with the number of endogenous regressors.
\Citet{ChenLee2018} propose an exact GMM estimator using mixed-integer quadratic programming, but it also has long computation time.
\Citet{Zhu2019} proposes a $k$-step correction approach using mixed integer linear programming.
This estimator is asymptotically equivalent to the GMM estimator and has computational efficiency in models with multiple endogenous regressors.
\Citet{KaidoWuthrich2019} decompose IVQR estimation into conventional QR sub-problems.
\Citet{KaplanSun2017} propose smoothing the IVQR moments, which helps both estimation and computational efficiency. 
However, their results only apply to an exactly-identified linear model and iid sampling. 
\Citet{deCastroGalvaoKaplanLiu2019} propose a GMM estimator using the same smoothed IVQR moments, 
extending results to over-identified nonlinear models and dependent data. 
%
I use the \citet{deCastroGalvaoKaplanLiu2019} in this paper because \citetalias{ChengLiaoShi2019} require a two-step GMM estimator.


\textit{Other identification methods in quantile regression with endogeneity }
There are three main approaches to address endogeneity in quantile regression.
They are based on different sets of assumptions appropriate for different empirical settings; none is strictly ``better'' or ``worse.''
First, as noted, \citet{ChernozhukovHansen2005} use an instrumental variables approach to identify the structural quantile function and (conditional) quantile treatment effects. 
Second, the local quantile treatment effect (LQTE) model \citep{AbadieEtAl2002} identifies the (conditional) quantile treatment effect for the sub-population of ``compliers'' in the binary treatment variable case, 
parallel to the local average treatment effect (LATE) model.
Third, triangular models and control functions have been used by \citet{Chesher2003}, \citet{Lee2007}, and others.
%
See \citet[\S9.2.5]{ChernozhukovHansenWuthrich2017} and \citet[\S10.5]{MellyWuthrich2017} for more detailed comparisons with additional references.
Different from these listed studies that focus on identification, 
I focus on improving estimation efficiency within the IVQR framework of \citet{ChernozhukovHansen2005}.

\paragraph{Outline}
The organization of this paper is as follows.
\Cref{sec:setup} presents the model setup. 
\Cref{sec:est:avg} presents the GMM averaging estimation method. 
\Cref{sec:bs} presents the bootstrap averaging estimation method.
\Cref{sec:sim} presents simulation results. 
\Cref{sec:conclusion} concludes.
Proofs and additional computational details are collected in the appendix.
Code is provided for all methods and simulations.

\paragraph{Notation}
For scalar/vector/matrix variable formatting, $\vecf{X}$ is a random vector with elements $X_j$, $\vecf{x}$ is a non-random vector with elements $x_j$, $Y$ and $y$ are random and non-random scalars, respectively, and $\matf{M}$ and $\matf{m}$ are random and non-random matrices with row $i$, column $j$ elements $M_{ij}$ and $m_{ij}$.
For vector/matrix multiplication, all vectors are treated as column vectors.
Also, $\Ind{\cdot}$ is the indicator function, 
$\E(\cdot)$ expectation, 
$\Q_\tau(\cdot)$ the $\tau$-quantile, 
$\Pr(\cdot)$ probability, 
and $\NormDist(\mu,\sigma^2)$ the normal distribution. 
Acronyms used include those for 
instrumental variables (IV), 
two-stage least squares (2SLS), 
generalized method of moment (GMM),
[smoothed] instrumental variables quantile regression ([S]IVQR), 
probability density function (PDF), 
cumulative distribution function (CDF), 
quantile regression (QR), 
conditional quantile function (CQF),
data-generating process (DGP), 
[root] mean squared error ([R]MSE), 
asymptotic mean squared error (AMSE),
and interquantile range (IQR).

\section{Model Setup}
\label{sec:setup}

We are interested in estimating the parameter  $\vecf{\beta}_{0\tau}\in\mathcal{B}\subseteq\R^{d_\beta}$ in a linear quantile model
that uniquely satisfies the conditional probability
\begin{equation}\label{eqn:CH05-cond-quantile}
\tau 
= \Pr\left(   Y_{i} \leq \vecf{X}_{i}' \vecf{\beta}_{0\tau}  \mid \vecf{Z}_{i} \right)
,
\end{equation}
where 
$\tau \in (0,1) $ is a given quantile level;  
$ Y_{i}$ is the outcome variable; 
$\vecf{X}_i= \big( \vecf{X}_{exog, i}'  , \vecf{D}_{i}'  \big)'  \in\mathcal{X}\subseteq\R^{d_X}$ is the vector of regressors;
$ \vecf{D}_{i}$ is the vector of potentially endogenous explanatory variables; 
$\vecf{X}_{exog, i}$ is the vector of exogenous explanatory variables;
and $\vecf{Z}_i=\big( \vecf{X}_{exog, i}'  , \vecf{Z}_{excl, i}'  \big)'  \in\mathcal{Z}\subseteq\R^{d_Z}$ is the full vector of instruments, which contains both the exogenous explanatory variables $\vecf{X}_{exog, i}$ and excluded instruments $\vecf{Z}_{excl, i}$. 
The conditional probability in \cref{eqn:CH05-cond-quantile} comes from \citeposs{ChernozhukovHansen2005} identification result in their Theorem 1,
which states conditions under which the $\vecf{\beta}_{0\tau}$ satisfying \cref{eqn:CH05-cond-quantile} is a structural parameter or includes a (conditional) quantile treatment effect parameter.

For intuition about identification, suppose a structural random coefficient model
\begin{align}  \label{eqn:random-coeffi}
Y & = \vecf{X}' \vecf{\beta}(U)
,
\end{align}
with unobserved scalar $U \sim \UnifDist(0,1)$, and $\vecf{\beta}(\cdot)$ is the vector-valued function of $U$ satisfying the monotonicity condition that 
$\vecf{X}' \vecf{\beta}(u)$ is increasing in $u$ 
for any $\vecf{X}=\vecf{x}$ in its support $\mathcal{X}$.
Monotonicity implies that $Y\le\vecf{X}'\vecf{\beta}(\tau)$ is equivalent to $U\le\tau$.
If additionally $\vecf{X}$ is exogenous with $\vecf{X} \independent U$, then 
$\Pr( Y \le \vecf{X}' \vecf{\beta}(\tau) \mid \vecf{X} ) = 
\Pr( U \le \tau \mid \vecf{X})$ by monotonicity, $\Pr( U \le \tau \mid \vecf{X})=\Pr(U\le\tau)$ by exogeneity, and $\Pr(U\le\tau)=\tau$ by the normalization $U\sim\UnifDist(0,1)$.
If $\vecf{X}$ and $U$ are dependent, 
then we can instead condition on instrument $\vecf{Z}$ satisfying $\vecf{Z}\independent U$, yielding \cref{eqn:CH05-cond-quantile} where $\vecf{\beta}_{0\tau} \equiv \vecf{\beta}(\tau)$.

The conditional probability in \cref{eqn:CH05-cond-quantile} can be written as a conditional expectation,
\begin{equation}\label{eqn:condi-moments} 
0 =
\E\left[  \Ind{ Y_{i} - \vecf{X}_{i}' \vecf{\beta}_{0\tau} \le 0} - \tau \mid \vecf{Z}_i \right] .
\end{equation}
By the law of iterated expectations, \cref{eqn:condi-moments} implies the unconditional moments
\begin{equation}\label{eqn:moments} 
\vecf{0}_{d_Z\times 1} = 
\E\left\{ \vecf{Z}_i \left[ \Ind{Y_{i} - \vecf{X}_{i}' \vecf{\beta}_{0\tau} \le 0} - \tau \right] \right\} .
\end{equation}

Most IVQR estimators use \cref{eqn:moments}.
In principle, the $\vecf{Z}_i$ in \cref{eqn:moments} could be replaced by the (nonparametrically estimated) optimal instruments based on \cref{eqn:condi-moments}.
However, this is not the focus of this paper.
This paper estimates structural parameter $\vecf{\beta}_{0\tau} $ 
based on the unconditional moments in \cref{eqn:moments}.

We introduce a few notations. 
Define the population map $\vecf{M}_{1} \colon \mathcal{B} \times \mathcal{T} \mapsto \R^{d_Z}$ as
\begin{align}\label{eqn:def-M}
\vecf{M}_{1}(\vecf{\beta},\tau) 
&\equiv \E\left[ \vecf{g}_{1i}(\vecf{\beta},\tau) \right] 
,\\ \label{eqn:def-gu}
\vecf{g}_{1i}(\vecf{\beta},\tau)
&\equiv \vecf{g}_{1}( Y_i,\vecf{X}_{i},\vecf{Z}_i,\vecf{\beta},\tau)
 \equiv \vecf{Z}_i\left[ \Ind{ Y_{i} - \vecf{X}_{i}' \vecf{\beta} \le 0} - \tau \right] 
,
\end{align}
where 
the subscript $1$ denotes the original IVQR moments, used as the ``conservative'' moments in GMM averaging later. 
Population moments \cref{eqn:moments} imply
\begin{equation}\label{eqn:def-beta0}
 \vecf{M}_{1}(\vecf{\beta}_{0\tau},\tau) =\vecf{0}. 
\end{equation}
Denote the sample moments as
\begin{equation}\label{eqn:def-M-hat}
    \hat{\vecf{M}}_{1n}\left(\vecf{\beta},\tau\right)
\equiv \frac{1}{n}\sum_{i=1}^{n} \vecf{g}_{1i}(\vecf{\beta},\tau) .
\end{equation}



\section{Averaging {GMM} Estimation for {IVQR}}
\label{sec:est:avg}

This section describes the first averaging estimation method. 
It follows the averaging GMM framework in \citetalias{ChengLiaoShi2019}.
It contributes to provide two types of additional moments 
and apply the averaging estimation method to the instrument variables quantile regression.

\citetalias{ChengLiaoShi2019} define a ``conservative'' GMM estimator based on a set of valid moments
and an ``aggressive'' GMM estimator based on the set of moments that combines the valid moments and additional, possibly misspecified moments.
\citetalias{ChengLiaoShi2019} average between these two GMM estimators
and show that under certain conditions the averaging GMM estimator uniformly dominates the conservative GMM estimator.

\subsection*{Step 1: conservative moments and conservative estimator}

To apply the \citetalias{ChengLiaoShi2019} averaging GMM framework to the IVQR model, 
I use \citeposs{deCastroGalvaoKaplanLiu2019} smoothed two-step GMM method for computation to obtain the conservative GMM estimator.
It replaces the IVQR moment functions
\begin{equation}
\vecf{g}_{1i}(\vecf{\beta},\tau)
= \vecf{Z}_i\left[ \Ind{ Y_{i} - \vecf{X}_{i}' \vecf{\beta} \le 0} - \tau \right]
\end{equation}
with smoothed moment functions
\begin{equation}\label{Smoothed-moments}
\vecf{g}_{1i}(\vecf{\beta},\tau)
=  \vecf{Z}_i \bigl[ \tilde{I}\bigl( \big(-Y_{i} + \vecf{X}_{i}' \vecf{\beta}_{0\tau} \big) / h_n \bigr) - \tau \bigr] ,
\end{equation}
where $h_n$ is the sequence of smoothing bandwidth and $\tilde{I}(\cdot)$ is the same smoothed indicator function as defined in \citet{deCastroGalvaoKaplanLiu2019}. 
The conservative GMM estimator is thus
\begin{equation}\label{eqn:def-est-beta1}
\hat{\vecf{\beta}}_{\mathrm{1}}
= \argmin_{\vecf{\beta} \in \mathcal{B} } 
    \hat{\vecf{M}}_{1n}(\vecf{\beta},\tau)'  \check{\matf{\Sigma}}^{-1}_{1}
    \hat{\vecf{M}}_{1n}(\vecf{\beta},\tau),
\end{equation}
where $\check{\matf{\Sigma}}_{1}$ is some consistent estimator of the (long-run) variance of the IVQR sample moments from the first step.
With iid sampling, 
\begin{equation}\label{eqn:Omega-bar}
\check{\matf{\Sigma}}_{1}=\frac{1}{n}\sum_{i=1}^{n} \vecf{g}_{1i}(\check{\vecf{\beta}},\tau) \vecf{g}_{1i}(\check{\vecf{ \beta}},\tau)' 
- \hat{\vecf{M}}_{1n}( \check{\vecf{\beta}}, \tau )  \hat{\vecf{M}}_{1n}( \check{\vecf{\beta}}, \tau )',
\end{equation}
where $\check{\vecf{\beta}}$ is some initial consistent estimator of $\vecf{\beta}_{0\tau}$. 
Specifically, my $\check{\vecf{\beta}}$ is \citeposs{deCastroGalvaoKaplanLiu2019} method of moments estimator with instrument vector equal to the linear projection of $\vecf{X}$ onto $\vecf{Z}$.

\subsection*{Step 2: additional moments}

In addition to the conservative moments $\E[\vecf{g}_{1i}(\vecf{\beta},\tau)]=\vecf{0}_{d_Z\times1}$, we have ``additional moments'' based on $\vecf{g}^{*}( \vecf{\beta}, \tau)$ that might or might not be valid. 
If the additional moments are valid 
(i.e., $\vecf{0}_{r*} = \E[ \vecf{g}^{*}(\vecf{\beta}_{0\tau}, \tau) ]$),
then adding additional valid information to estimation will reduce variance and improve efficiency of the IVQR estimator. 
If the additional moments are misspecified 
(i.e., $\vecf{0}_{r*} \ne \E[ \vecf{g}^{*}(\vecf{\beta}_{0\tau}, \tau) ]$), 
then combining these invalid additional moments with the original valid IVQR moments will result in a biased aggressive GMM estimator.
However, the misspecified moments could still be helpful, 
if as a result the aggressive GMM estimator
has a large reduction in variance, and an overall reduction in MSE.

In principle, 
the averaging estimator can always be at least as good as the conservative IVQR estimator, 
if it puts zero weight on the aggressive estimator 
when the additional moments are misspecified severely enough. 
In practice, 
this desirable result may not exactly hold due to the estimation error of the empirical weight in finite samples.

I propose two different types of additional moments. 
The first type is the conventional QR moments.
If the structural model is a linear conditional $\tau$-quantile function (CQF), 
then
\begin{equation} 
\Pr( Y_{i} - \vecf{X}_{i}' \vecf{\beta}_{0\tau} \le 0 \mid \vecf{X}_{i}) = \tau . 
\end{equation}
This conditional quantile restriction can be rewritten as a conditional expectation, 
\begin{equation}\label{eqn:condi-QR}
0 =
\E\left\{  \Ind{ Y_{i} - \vecf{X}_{i}' \vecf{\beta}_{0\tau} \le 0} - \tau \mid \vecf{X}_{i} \right\} .
\end{equation}
Using the law of iterated expectations, \cref{eqn:condi-QR} implies certain unconditional QR moments
that are also the first-order condition of the population minimization problem of the expectation of the check function. 
Writing these unconditional QR moments in two separate parts,
\begin{align}\label{eqn:excludX-moments} 
\vecf{0}_{d_{X exog}\times 1} &= 
\E\left\{ \vecf{X}_{exog, i} \left[ \Ind{Y_{i} - \vecf{X}_{i}' \vecf{\beta}_{0\tau} \le 0} - \tau \right] \right\} ,
\\
\label{eqn:QR-moments} 
\vecf{0}_{d_D\times 1} &= 
\E\left\{ \vecf{D}_{i} \left[ \Ind{Y_{i} - \vecf{X}_{i}' \vecf{\beta}_{0\tau} \le 0} - \tau \right] \right\} .
\end{align}

I use \cref{eqn:QR-moments} but not \cref{eqn:excludX-moments} for the additional moments.
The first part \cref{eqn:excludX-moments} is already contained in the IVQR population moments \cref{eqn:moments}, therefore the ``conservative moments,'' as the exogenous regressors $\vecf{X}_{exog}$ are contained in the full instruments $\vecf{Z}$.
I use the second part \cref{eqn:QR-moments}, 
the QR moments with potentially endogenous regressors $\vecf{D}$, 
as the ``additional moments'' to compute the ``aggressive moments,'' the ``aggressive estimator,'' and the ``averaging estimator.''
(For computation, I again use a smoothed version of the moment function, like \cref{Smoothed-moments} but with $\vecf{D}_i$ replacing $\vecf{Z}_i$.)
I call this the IVQR-QR type of averaging.




The motivation for using the conventional QR moments as the additional moments is that when there is little endogeneity 
(i.e., the additional QR moments are only slightly misspecified),
the QR estimator is only a little biased.
Meanwhile, the QR estimator usually has lower variance than the IVQR estimator.
Therefore the aggressive estimator has a lower variance than the IVQR estimator.
When the DGP has severe endogeneity 
(i.e., the additional QR moments are severely misspecified),
the QR estimator and aggressive estimator have larger bias than the IVQR estimator.
Ideally, more weight is put on the aggressive estimator when there is little endogeneity, 
and more weight is put on the conservative IVQR estimator when there is much endogeneity.

As an alternative to the QR type of additional moments, I propose using the 2SLS slope moment functions
\begin{equation}\label{eqn:def-g-star-2SLS}
\vecf{g}_{i}^{*}( \vecf{\beta}, \tau) 
\equiv (\vecf{Z}_{-1,i}-\bar{\vecf{Z}}_{-1}) \left( Y_{i} - \vecf{X}_{i}' \vecf{\beta} \right) ,
\end{equation}
where $\vecf{Z}_{-1}$ is the instruments without the intercept term, i.e., $\vecf{Z}=(1,\vecf{Z}_{-1})$.
We use the demeaned instruments $(\vecf{Z}_{-1,i} - \bar{\vecf{Z}}_{-1})$ in these additional moments, 
where $\bar{\vecf{Z}}_{-1} \equiv 1/n \sum_{i=1}^{n} \vecf{Z}_{-1,i}$, to represent the 2SLS slope condition $\Cov(\vecf{Z}_{-1}, Y-\vecf{X}'\vecf{\beta}_0)=\vecf{0}$.
With the 2SLS slope moments as the additional moments, 
we obtain the corresponding aggressive and averaging estimators.
I call this the IVQR-2SLS type of averaging.


The motivation for using the 2SLS slope moments as the additional moments is that when there is not much heterogeneity across quantiles, 
2SLS and IVQR at any quantile have similar slope (but not intercept) estimates, thus similar bias.
Moreover, the smoothed IVQR estimator of \citet{KaplanSun2017} has slope estimates approach the 2SLS slope estimates as the smoothing bandwidth goes to infinity (see their \S2.2).
Meanwhile, the 2SLS estimator usually has smaller variance than IVQR estimator,
especially at quantile levels that are away from the median and closer to the tails.%
\footnote{One exception to 2SLS having smaller variance is with fat-tailed error terms, which I investigate in the simulations in \cref{sec:sim}.}
For example, even in an empirically-based simulation with substantial heterogeneity (that causes 2SLS to be biased), Table 3 of \citet{KaplanSun2017} shows 2SLS to be more efficient than IVQR at four out of five quantile levels.
Incorporating the additional 2SLS slope moments, the aggressive estimator usually has smaller variance than the conservative IVQR estimator.
The averaging estimator improves efficiency over the IVQR estimator 
by putting more weight on the aggressive estimator 
when it has large enough reduction in variance, 
and putting more weight on the IVQR estimator 
when the 2SLS slope moments are severely misspecified 
and the bias increase overwhelms the variance reduction.

Besides the two types of additional moments proposed in this paper, 
the same idea suggests using IVQR slope moments with other quantile levels to be the additional moments. 
The intuition is that 
when there is no or little heterogeneity, there will not be much difference in estimates across quantiles. 
This is related to the $L$-estimation method in conventional quantile regression by \citet{KoenkerPortnoy1986}.

\subsection*{Step 3: aggressive moments and aggressive estimator}

Incorporating the additional moments 
(either the potentially endogenous QR moments or the 2SLS slope moments) 
to the IVQR moments, 
we obtain the aggressive moments. Define the aggressive GMM estimator as 
\begin{equation}\label{eqn:def-est-beta2}
\hat{\vecf{\beta}}_{\mathrm{2}}
= \argmin_{\vecf{\beta} \in \mathcal{B} } 
    \hat{\vecf{M}}_{2n}(\vecf{\beta},\tau)'  \check{\matf{\Sigma}}^{-1}_{2}
    \hat{\vecf{M}}_{2n}(\vecf{\beta},\tau),
\end{equation}
where
\begin{equation}
\label{eqn:m2}
\hat{\vecf{M}}_{2n}(\vecf{\beta},\tau) 
\equiv \frac{1}{n}\sum_{i=1}^{n} \vecf{g}_{2i}( \vecf{\beta}, \tau) , 
\quad
\vecf{g}_{2i}( \vecf{\beta}, \tau) 
\equiv (\vecf{g}_{1i}( \vecf{\beta}, \tau)',\vecf{g}^{*}_{i}( \vecf{\beta}, \tau)' )' .
\end{equation}
The aggressive moments contain both the conservative moments and additional moments.
The GMM weighting matrix $\check{\matf{\Sigma}}_{2}^{-1}$ is constructed in the same way as $\check{\matf{\Sigma}}_{1}^{-1}$,
except with aggressive moments (using $\vecf{g}_2$) instead of conservative moments (using $\vecf{g}_1$).
The subscript ``2'' denotes ``aggressive'' here.

\subsection*{Step 4: averaging estimator}

Following \citetalias{ChengLiaoShi2019}, define the averaging GMM estimator as
\begin{equation}\label{eqn:def-g-star-AVG}
\hat{\vecf{\beta}}_{\mathrm{AVG}}=(1-\hat{w})\hat{\vecf{\beta}}_{\mathrm{1}} + \hat{w} \hat{\vecf{\beta}}_{\mathrm{2}} .
\end{equation}
The empirical averaging weight $\hat{w}$ in \citetalias{ChengLiaoShi2019} is the sample analog of the optimal weight:
\begin{equation}\label{weight}
    \hat{w}=\frac{\mathrm{tr}(\matf{\Upsilon} (\hat{\matf{\Sigma}}_1-\hat{\matf{\Sigma}}_2))}{n(\hat{\vecf{\beta}}_1-\hat{\vecf{\beta}}_2)' \matf{\Upsilon} (\hat{\vecf{\beta}}_1-\hat{\vecf{\beta}}_2)+\mathrm{tr}(\matf{\Upsilon}(\hat{\matf{\Sigma}}_1-\hat{\matf{\Sigma}}_2))}
\end{equation} 
where 
\begin{equation}
    \hat{\matf{\Sigma}}_{k} = (\hat{\matf{G}}'_{k} \check{\matf{\Sigma}}^{-1}_{k}  \hat{\matf{G}}_{k})^{-1} , \ \text{for} \ k=1, 2.
\end{equation}
The middle part $\check{\matf{\Sigma}}_{k}$
is the estimator of the covariance matrix of the conservative sample moments and aggressive sample moments for $k=1$ and $k=2$, respectively; 
its inverse $\check{\matf{\Sigma}}^{-1}_{k}$ is the efficient two-step GMM weighting matrix.
Define the estimator of the Jacobian matrix of the conservative moments and aggressive moments ($k=1$ and $k=2$, respectively) as
\begin{equation}\label{eqn:est-G}
\hat{\matf{G}}_{k}
=
 \frac{1}{n} \sum_{i=1}^{n}
\pD{}{\vecf{\beta}' }  \vecf{g}_{ki}(\hat{\vecf{\beta}}_{1},\tau) .
\end{equation}
Both $ \check{\matf{\Sigma}}_{k}$ and $\hat{\matf{G}}_{k}$, for $k=1$ and $k=2$, are evaluated at the conservative GMM estimator $\hat{\vecf{\beta}}_1$.
Therefore, they are consistent regardless of misspecification of the additional moments.
The diagonal matrix $\matf{\Upsilon} $ measures how much we weight each element in the parameter vector in the (scaled) loss function $(\hat{\vecf{\beta}}-\vecf{\beta}_{0\tau})'\matf{\Upsilon}(\hat{\vecf{\beta}}-\vecf{\beta}_{0\tau})$, as in (3.9) of \citetalias{ChengLiaoShi2019}. 
When $\matf{\Upsilon}$ is the identity matrix, 
the expected loss (i.e., risk) becomes the sum of MSEs of each component of the estimator vector, 
$\sum_{j=1}^{d_X} \E[(\hat\beta_j-\beta_{0\tau j})^2]$.

Averaging estimation can be considered as a bias--variance tradeoff.
The averaging weights are crucial for good performance of the averaging estimator.
Too much weight on the aggressive estimator can result in large squared bias and large MSE.
Too little weight on the aggressive estimator can result in large variance and again large MSE.
The optimal averaging weight ideally balances the variance and the squared bias, minimizing MSE.
This is like the bandwidth choice for kernel regression and other nonparametric estimators.

\subsection{Bandwidth choice}
\label{sec:bandwidth}

%
The \citetalias{ChengLiaoShi2019} empirical weight formula in \cref{weight} requires estimation of the conservative and aggressive parameters, covariance matrix, and (population) Jacobian matrix, which involves a conditional density that must be nonparametrically estimated.
The performance of the averaging GMM method heavily depends on whether the empirical weight is estimated accurately or not.


To estimate the parameters and the covariance matrix with the smoothed GMM approach of \citet{deCastroGalvaoKaplanLiu2019},
I use the smallest possible smoothing bandwidth, 
for two reasons.
First, \citet{KaplanSun2017} note smoothing can reduce MSE.%
\footnote{\Citet{KaplanSun2017} derive an MSE-optimal bandwidth for estimating the smoothed estimating equation IVQR estimator.
This bandwidth is typically much larger than the smallest possible smoothing that makes computation feasible.}
However, the goal of this paper is to demonstrate that 
it is the averaging method, instead of smoothing, that can improve estimation efficiency.
Second,
the \citetalias{ChengLiaoShi2019} averaging GMM framework assumes the conservative moments are valid and that the conservative estimator is not biased, but smoothing introduces some bias.
Using the smallest possible smoothing bandwidth makes the bias of the conservative IVQR estimator small enough to be negligible.


For IVQR and even QR, 
the population Jacobian matrix involves a conditional PDF, 
which is commonly estimated by a nonparametric kernel estimator.
The usual kernel estimator is actually the same as the standard sample Jacobian when $\hat{\vecf{\beta}}$ is based on smoothed moments.

To precisely estimate the Jacobian matrix of the IVQR moments and of the aggressive moments, 
I modify \citeposs{Kato2012} bandwidth.
He provides the optimal bandwidth for conventional QR based on asymptotic mean squared error (AMSE).
He also provides a simplified version assuming independent, standard normal regression errors.
I extend his Gaussian plug-in bandwidth to allow for any error variance, and I adapt his formulas for QR to IVQR.
With QR, the vector $\vecf{X}$ acts as both the regressors and the instruments; for IVQR, the parts of the bandwidth formulas where $\vecf{X}$ acts as instruments are replaced by $\vecf{Z}$.
For the QR Gaussian plug-in, I build on \citeposs{Kato2012} results to prove in \cref{prop:bandwidth} that in the model 
$Y=\vecf{X}' \vecf{\beta}_0 + U$
with
$U \mid \vecf{X} \sim \NormDist( \mu , \sigma^2 )$
and
$\Q_{\tau}(U \mid \vecf{X})=0$,
the sample analog of the AMSE-optimal bandwidth is
\begin{equation*}
\begin{split}
\hat{h}_{\mathrm{opt}}
&= n^{-1/5} \left( \frac{4.5\sum_{j,k=1}^{d_X} \left( n^{-1} \sum_{i=1}^{n} X_{ij}^2 X_{ik}^2 \right) }{\hat\alpha(\tau)\sum_{j,k=1}^{d_X} \left( n^{-1} \sum_{i=1}^{n} X_{ij} X_{ik}  \right)^2 } \right)^{1/5},
\\
\hat\alpha(\tau)
&\equiv \frac{1}{\hat\sigma^5} \left[ 1-\left(\Phi^{-1} \left( \tau \right) \right)^2     \right]^2 \phi \left(\Phi^{-1} \left( \tau \right) \right),
\end{split}
\end{equation*}
where 
$\Phi(\cdot)$ and $\phi(\cdot)$ are the standard normal CDF and PDF.
Details are in \cref{sec:app-bandwidth}.

\section{Bootstrap method}
\label{sec:bs}

In addition to the theoretically-based averaging GMM estimation in \cref{sec:est:avg}, 
I also propose a bootstrap averaging estimator for IVQR.
Simulation performance of this bootstrap averaging estimator is in \cref{sec:sim}.

The bootstrap averaging estimator comes from the same motivating idea in \cref{sec:est:avg}.
That is, compared to the IVQR estimator, the QR or 2SLS estimator might have smaller variance, and overall smaller MSE, though larger bias.
This is especially true with only mild endogeneity or heterogeneity.

There are some differences with \cref{sec:est:avg} that may enable the bootstrap's better performance in simulations.
\Cref{sec:est:avg} considers averaging between two estimators (conservative and aggressive), whereas 
the bootstrap method averages among three estimators: IVQR, 2SLS, and QR.
Computationally, bootstrap averaging is simpler and easier, not requiring two-step GMM with a large degree of overidentification.

The bootstrap method algorithm is as follows.
\begin{enumerate}[1.]
\item\label{step:orig} Compute the IVQR, 2SLS, and QR estimators 
$(\hat{\vecf{\beta}}_{\mathrm{IVQR}}, \hat{\vecf{\beta}}_{\mathrm{2SLS}}, \hat{\vecf{\beta}}_{\mathrm{QR}})$ 
using the original sample $(Y_i,  \vecf{X}_{endo,i}, \vecf{X}_{exog,i}, \vecf{Z}_{excl,i})$ for $i=1, \ldots, n$.
\item\label{step:bs-sample} 
Draw bootstrap sample $b$: 
$(Y_i^{(b)},  \vecf{X}_{endo,i}^{(b)}, \vecf{X}_{exog,i}^{(b)}, \vecf{Z}_{excl,i}^{(b)})$ 
for $i=1, \ldots, n$. 
\item\label{step:bs-est} Use the bootstrap sample 
to compute the bootstrapped IVQR, 2SLS, and QR estimators 
$(\hat{\vecf{\beta}}_{\mathrm{IVQR}}^{(b)} , \hat{\vecf{\beta}}_{\mathrm{2SLS}}^{(b)} , \hat{\vecf{\beta}}_{\mathrm{QR}}^{(b)})$. 
\item\label{step:bs-avg} Compute averaging estimators in the bootstrap sample for a fixed grid of weights $(w_{1s}, w_{2s}, w_{3s})$ for $s=1, \ldots, \numnornd{13701}$, with $0 \le w_{1s}, w_{2s}, w_{3s} \le 1$ and $w_{1s}+w_{2s}+w_{3s}=1$; specifically,  $\hat{\vecf{\beta}}_{\mathrm{AVG},s}^{(b)} 
= w_{1s}\hat{\vecf{\beta}}_{\mathrm{IVQR}}^{(b)} 
+ w_{2s}\hat{\vecf{\beta}}_{\mathrm{2SLS}}^{(b)} 
+ w_{3s}\hat{\vecf{\beta}}_{\mathrm{QR}}^{(b)}$.
\item Repeat steps \ref{step:bs-sample}--\ref{step:bs-avg} for $b=1, \ldots, B$.
\item\label{step:bs-rmse} Treating $\hat{\vecf{\beta}}_{\mathrm{IVQR}}$ as the true population parameter in the bootstrap world, compute the RMSE of the $\numnornd{13071}$ different averaging estimators:
\begin{equation*}
    \RMSE(\hat{\vecf{\beta}}_{\mathrm{AVG,s}}) = \sqrt{\frac{1}{B} \sum_{b=1}^{B} (\hat{\vecf{\beta}}_{\mathrm{AVG,s}}^{(b)} - \hat{\vecf{\beta}}_{\mathrm{IVQR}} )^2 }.
\end{equation*}
\item\label{step:w-star} Define the ``bootstrapped optimal weight'' as minimizing the RMSE in step \ref{step:bs-rmse}:
\begin{equation*}
(w_1^*, w_2^*, w_3^*)
\equiv (w_{1s^*}, w_{2s^*}, w_{3s^*})
,\quad
s^* \equiv \argmin_{s \in \{1,\ldots, 13701\} } \mathrm{RMSE}(\hat{\vecf{\beta}}_{\mathrm{AVG,s}}) .
\end{equation*}
\item Using the bootstrapped optimal weight in step \ref{step:w-star}, define the ``bootstrapped averaging estimator'' as
$\hat{\theta}_{\mathrm{AVG.BS}} 
\equiv w_1^*\hat{\vecf{\beta}}_{\mathrm{IVQR}} 
+ w_2^*\hat{\vecf{\beta}}_{\mathrm{2SLS}} 
+ w_3^*\hat{\vecf{\beta}}_{\mathrm{QR}}$.
\end{enumerate}

The computation time for the bootstrap averaging estimator is reasonably fast even with a large weight grid size like $\numnornd{13701}$. 
The actual IVQR, 2SLS, and QR estimators only need to be computed once per bootstrap draw; the $\numnornd{13701}$ is just arithmetic.
In my implementation, $\hat{\vecf{\beta}}_{\mathrm{IVQR}}$ and $\hat{\vecf{\beta}}_{\mathrm{IVQR}}^{(b)}$ are obtained by first projecting the regressors onto the instruments, 
and then solving smoothed versions of the exactly-identified equations, using standard numerical methods.
So the additional computation time of bootstrap averaging over GMM averaging is not too big.%
\footnote{For example, in simulation model 1, 
running 400 replications with 250 bootstrap draws per replication
to compute both the GMM averaging estimator and the bootstrap averaging estimator 
takes 2.75 times longer than running 400 replications with only GMM averaging.}


\section{Simulations}\label{sec:sim}

This section reports simulation results showing the finite-sample performance of the two types of GMM averaging estimator and the bootstrap averaging estimator, relative to the IVQR estimator.

I consider three different simulation models, with many DGPs within each model.
Simulation model 1 presents a case where the uniform dominance condition does not hold.
Simulation models 2 and 3 present cases where all the averaging estimators uniformly dominate the IVQR estimator.
Simulation model 2 closely follows the \citetalias{ChengLiaoShi2019} simulation model (S2), but with modification to IVQR.
Since one important feature of quantile regression is to capture unobserved heterogeneity, simulation model 3 includes slope heterogeneity.

The performance of the averaging estimators are measured by their robust root mean squared error (robust RMSE) relative to that of the IVQR estimator.
The robust RMSE is computed by replacing the bias with median bias and replacing the standard deviation with interquartile range (IQR) divided by 1.349. 
It equals RMSE for normal distributions but is more robust to outliers. 
More specifically, robust RMSE (rRMSE) is computed as
\begin{equation}\label{eqn:rRMSE}
\rRMSE = \sqrt{
  \sum_{j=1}^{d_\theta} [ (median\  bias_j)^2 + (IQR_j/1.349)^2 ]
} ,
\end{equation}
where the median bias is the median of estimators among the $M$ replications minus the true parameter value,
\begin{equation}
     median\  bias_j=median_{m \in \{1,\ldots,M\} }(\hat{\theta}_{jm})-\theta_j,  
     \ \textrm{for} \ j=1, \ldots,d_\theta,
\end{equation}
the IQR is the difference between the $0.75$-quantile and $0.25$-quantile of the estimators among the $M$ replications,
\begin{equation}
IQR_j=0.75\text{quantile}_{m \in \{1 \ldots,M\} }(\hat{\theta}_{jm})-
     0.25\text{quantile}_{m \in \{1 \ldots,M\} }(\hat{\theta}_{jm}),  
     \ \textrm{for} \ j=1, \ldots,d_\theta,
\end{equation}
$M$ is the number of replications in total,
and $d_\theta$ is the number of parameters.

I normalize the rRMSE of the IVQR estimator to 1 and use the relative rRMSE to see how other estimators perform relative to the IVQR estimator.
That is, I divide all rRMSEs by the IVQR rRMSE to get the relative rRMSE.
If the relative rRMSE is above (below) 1, that means the estimator has larger (smaller) rRMSE than IVQR, i.e., it performs worse (better) than IVQR.
Ideally, an estimator 
has relative rRMSE (weakly) below 1 in all different DGPs. 
In this case, 
we say the proposed estimator uniformly dominates the IVQR estimator: regardless of the true DGP, it performs as good or better than IVQR.
Such an estimator is unambiguously preferred to IVQR.
An estimator may still be preferred even without uniform dominance, but it would depend on the user's preferences (e.g., the estimator's Bayes risk may be larger or smaller than IVQR's depending on the user's prior over DGPs).

\subsection{Simulation Model 1}

\subsubsection{Simulation DGP}
\label{subsec:simDGP1}

This Job Training Partnership Act-based simulation DGP generalizes DGP 1 in \citet{deCastroGalvaoKaplanLiu2019} to a class of DGPs that allow different combinations of endogeneity, heterogeneity, and fat-tail levels.
Consider a structural random coefficient model that describes the impact of a job training program $D_i$ on individual $i$'s earnings $Y_i$,
\begin{equation}
Y_i=\beta(U_i)+\gamma(U_i) \times D_i .
\end{equation}
The unobserved 
scalar $U_i \sim \UnifDist(0,1)$ incorporates other earnings determinants like ability.
The individual-specific intercept and slope depend on $U_i$, through functions $\beta(\cdot)$ and $\gamma(\cdot)$.

The slope function is set as 
$\gamma(U_i)=100c_2 U_i^4$.
The nonnegative constant $c_2$ indicates the degree of treatment effect heterogeneity.

The intercept function is set as
$\beta(U_i)=60+Q(U_i)$,
with two possible $Q(\cdot)$.
First, $Q(U_i)$ follows a $\chi^2_3$ distribution.
Second, $Q(U_i)$ follows a $t$-distribution with $c_3$ degrees of freedom.
Each $c_3$ value represents a different fat-tail degree of earnings.
As $c_3$ increases, the earnings distribution becomes less fat-tailed, approaching a normal distribution.

The job training offer (eligibility) $Z_i$ is completely randomized with $\Pr(Z_i=1)=\Pr(Z_i=0)=1/2$.
This $Z_i$ is a valid instrument for $D_i$.

The relationship between the randomized offer and the self-selection of participation in the program is described as a conditional probability,
\begin{equation}\label{SimModel1:ProbFn}
\Pr(D_i=1 \mid Z_i=1, U_i) = 0.5+c_1(U_i-0.5) ,
\quad
\Pr(D_i=1 \mid Z_i=0,U_i)=0 ,
\end{equation}
where $D_i=1$ if the individual actually takes the training and $D_i=0$ if not.
The constant $c_1\in[0,1]$ indicates the endogeneity level of $D_i$, with larger $c_1$ meaning more endogeneity.

Altogether I ran 242 DGPs (combinations of $c_1$, $c_2$, and $c_3$) in this simulation model 1.
Since many results are very similar, I selected 14 DGPs representative of different combinations of endogeneity, heterogeneity, and fat-tail level.
More details are in \cref{app:DGP-1}.

\subsubsection{Simulation results}

\Cref{tab:est:M1:tau=0.5} presents the finite-sample rRMSE of the estimators proposed in this paper relative to that of the IVQR estimator at the median, $\tau=0.5$.

\begin{table}[htbp]
    \centering\caption{\label{tab:est:M1:tau=0.5} Finite-sample relative rRMSE in JTPA-based simulation model 1, $\tau=0.5$.}
            \sisetup{round-precision=2,round-mode=places}
    \begin{threeparttable}
    \begin{tabular}{clllccclcclccc}
    \toprule
    &          &          & & \multicolumn{2}{c}{$\mathrm{IVQR.2SLS}$} &               & \multicolumn{2}{c}{$\mathrm{IVQR.QR} $} &               &  &     \\
    \cmidrule{5-6} \cmidrule{8-9}
    $\mathrm{DGP }$  &  $\mathrm{Endog}$ &  $\mathrm{Hetero}$ &  $\mathrm{Fattail}$ &  $\mathrm{AVG} $ & $\mathrm{AGG}$       & $\mathrm{2SLS}$ & $\mathrm{AVG}$ & $\mathrm{AGG}$         & $\mathrm{QR}$ & $\mathrm{BS} $   & $\mathrm{CON}$ \\
    \midrule
1          &     No &     No &     NA &  $\num[math-rm=\mathbf]{0.951569}$  & $\num{0.868586}$ & $\num{0.945289}$  & $\num[math-rm=\mathbf]{0.881229}$  &  $\num{0.764277}$ & $\num{0.602878}$ & $\num[math-rm=\mathbf]{0.724973}$ & $\num{0.313962}$   \\
2          &   Some &     No &     NA &  $\num[math-rm=\mathbf]{0.974909}$  & $\num{1.020234}$ & $\num{1.050889}$  & $\num[math-rm=\mathbf]{0.961879}$  &  $\num{1.050789}$ & $\num{1.059315}$ & $\num[math-rm=\mathbf]{0.969991}$ & $\num{0.319387}$   \\
3          &   Some &     No &     NA &  $\num[math-rm=\mathbf]{0.968523}$  & $\num{0.927275}$ & $\num{1.014554}$  & $\num{1.127433}$  &  $\num{1.739923}$ & $\num{1.766773}$ & $\num[math-rm=\mathbf]{0.992771}$ & $\num{0.321078}$   \\
4          &   Much &     No &     NA &  $\num[math-rm=\mathbf]{0.974521}$  & $\num{0.938316}$ & $\num{0.918147}$  & $\num{1.044576}$  &  $\num{2.563305}$ & $\num{2.711448}$ & $\num[math-rm=\mathbf]{0.986249}$ & $\num{0.337619}$   \\
5          &     No &   Some &     NA &  $\num{1.132561}$  & $\num{4.617217}$ & $\num{4.100392}$  & $\num[math-rm=\mathbf]{0.871133}$  &  $\num{0.722243}$ & $\num{0.569145}$ & $\num[math-rm=\mathbf]{0.917096}$ & $\num{1.392060}$   \\
6          &   Some &   Some &     NA &  $\num{1.129272}$  & $\num{5.039539}$ & $\num{4.449811}$  & $\num[math-rm=\mathbf]{0.915169}$  &  $\num{0.883725}$ & $\num{0.779233}$ & $\num[math-rm=\mathbf]{0.977591}$ & $\num{1.347838}$   \\
7          &   Much &   Some &     NA &  $\num{1.084652}$  & $\num{5.800216}$ & $\num{5.098121}$  & $\num{1.096158}$  &  $\num{2.811002}$ & $\num{2.947386}$ & $\num{1.143536}$ & $\num{1.202500}$   \\
8          &   Much &   Some &     NA &  $\num{1.082160}$  & $\num{7.635008}$ & $\num{6.867073}$  & $\num{1.041133}$  &  $\num{3.677603}$ & $\num{3.527626}$ & $\num{1.204501}$ & $\num{1.495871}$   \\
9          &   Some &   Much &     NA &  $\num{1.057804}$  & $\num{5.663275}$ & $\num{5.180694}$  & $\num[math-rm=\mathbf]{0.946151}$  &  $\num{0.860991}$ & $\num{0.830940}$ & $\num[math-rm=\mathbf]{0.949558}$ & $\num{2.627995}$   \\
10         &     No &     No &   Some &  $\num[math-rm=\mathbf]{0.949378}$  & $\num{0.998784}$ & $\num{1.174431}$  & $\num[math-rm=\mathbf]{0.891689}$  &  $\num{0.782840}$ & $\num{0.535217}$ & $\num[math-rm=\mathbf]{0.697091}$ & $\num{0.176753}$   \\
11         &     No &     No &   Much &  $\num{1.007083}$  & $\num{1.004834}$ & $\num{28.514470}$  & $\num[math-rm=\mathbf]{0.772184}$  &  $\num{0.694622}$ & $\num{0.526553}$ & $\num[math-rm=\mathbf]{0.776365}$ & $\num{0.194113}$   \\
12         &   Much &     No & Little &  $\num[math-rm=\mathbf]{0.968218}$  & $\num{0.885921}$ & $\num{0.929345}$  & $\num{1.087361}$  &  $\num{2.562201}$ & $\num{2.624060}$ & $\num{1.031583}$ & $\num{0.169276}$   \\
13         &   Much &     No &   Some &  $\num[math-rm=\mathbf]{0.968126}$  & $\num{1.082454}$ & $\num{1.265369}$  & $\num{1.029175}$  &  $\num{2.745966}$ & $\num{2.726339}$ & $\num{1.083136}$ & $\num{0.161290}$   \\
14         &   Much &     No &   Much &  $\num{1.009193}$  & $\num{1.034325}$ & $\num{29.684063}$  & $\num{1.030598}$  &  $\num{2.750409}$ & $\num{2.557961}$ & $\num{1.063432}$ & $\num{0.203994}$   \\
\bottomrule
            \end{tabular}
\begin{tablenotes}
 \item $\num{400}$ replications. $\num{50}$ bootstraps. Sample size is 1000.\\
 Columns 1--4 describe the DGP.
 Columns 5--11 report the relative rRMSE (i.e., rRMSE divided by the IVQR rRMSE) of the IVQR-2SLS averaging estimator, IVQR-2SLS aggressive estimator, 2SLS estimator, IVQR-QR averaging estimator, IVQR-QR aggressive estimator, QR estimator, and bootstrap averaging estimator, respectively. 
 Column 12 reports the absolute rRMSE of the IVQR estimator (whose relative rRMSE is 1 by definition).
 Columns 5, 8, and 11 are the three averaging estimators of primary interest.
 Columns 6, 7, 9, and 10 are reported for reference to better see the mechanisms of the averaging methods.
\end{tablenotes}
\end{threeparttable}
\end{table}

\Cref{tab:est:M1:tau=0.5} shows no averaging estimator uniformly dominates IVQR in simulation model 1.
That is, no column has relative rRMSE below $1$ for each DGP.
But, they do not suffer much in the least favorable cases, even when 2SLS and/or QR (or their aggressive GMM counterparts) have relative rRMSE well above 1.

As expected, the IVQR-2SLS averaging estimator performs better than (or almost the same as) the IVQR estimator in the DGPs with no slope heterogeneity.
The efficiency gain is not much (up to 5\%). 
In the DGPs with slope heterogeneity, the relative rRMSE of IVQR-2SLS averaging is around $1.06$ to $1.13$,
while that of IVQR-2SLS aggressive estimator and 2SLS estimator is around 5 to 7.
This indicates that although IVQR-2SLS averaging estimator is worse than IVQR,
its rRMSE is much closer to 1 compared with the 2SLS estimator and aggressive estimator, due to putting most of the weight on the conservative estimator in these cases.

The IVQR-QR averaging estimator performs around 11--23\% better than IVQR estimator in the DGPs with no endogeneity and no heterogeneity.
As the endogeneity level of treatment variable increase, 
the QR moments become more misspecified, 
and the performance of the IVQR-QR averaging estimator becomes less favorable.
In the DGPs with much endogeneity, 
the IVQR-QR averaging estimator performs worse than the IVQR estimator by around 3--13\%.
Compared with the relative rRMSE of the QR estimator and that of the IVQR-QR aggressive estimator, 
the averaging estimator is much closer to 1 and saved from even worse performance by putting most of the weight on the conservative estimator.

The bootstrap estimator that averages among IVQR, 2SLS, and QR performs better than IVQR in all but one DGP in which either 2SLS or QR (or both) has relative rRMSE below 1.
In the DGPs in which both 2SLS and QR are worse than IVQR, 
the bootstrap averaging estimator performs worse than IVQR, but not by as much.
For example, in DGPs 7, 8, 13, and 14,
the bootstrap averaging estimator has relative rRMSE from 1.06 to 1.20, compared to 1.27 to 29.68 for 2SLS and 2.56 to 3.53 for QR.

\subsubsection{Results at other quantiles}

Results at other quantiles are shown in \cref{appdx:M1}.
It includes results at $\tau=0.2$ up to $\tau=0.8$.

For IVQR-2SLS averaging, there is a ``magic quantile'' at which 2SLS performs well across all DGPs, regardless of heterogeneity.
In simulation model 1, the slope is $\gamma(U)=100 c_2 U^4$ with $U \sim \UnifDist(0,1)$.
The 2SLS population slope is $\E[100 c_2 U^4]=100c_2\E(U^4)=20c_2$, where $\E(U^4)=0.2$ is the fourth moment of a standard uniform distribution.
The $\tau$-IVQR slope is $100c_2\tau^4$; since $0.7^4=0.24$, the slope is very close to the 2SLS slope when $\tau=0.7$ ($24c_2$ vs.\ $20c_2$).
At a slightly smaller $\tau$ (not run in simulations), the IVQR and 2SLS population slopes are identical.
So among $\tau=0.2,0.3,\ldots,0.8$, the ``magic quantile'' is $\tau=0.7$ in simulation model 1.
This is true regardless of $c_2$; it depends only on how close $\tau^4$ is to $\E(U^4)$.
Simulation results confirm that at the $0.7$-quantile, 
2SLS has relative rRMSE much below 1 even in the DGPs 5--9
with some or much slope heterogeneity.
The IVQR-2SLS averaging estimator also has relative rRMSE below 1 in all the DGPs.

As $\tau$ moves away from the ``magic quantile,'' the 2SLS estimator begins to show the pattern that 
its relative rRMSE is below 1 in DGPs with no slope heterogeneity, 
but above 1 in the DGPs with some or much slope heterogeneity.
In DGPs 5--9 with slope heterogeneity,
the IVQR-2SLS aggressive estimator's relative rRMSE can be as large as 12,
but that of averaging estimator is still close to 1 (bounded by 1.19).
The averaging estimator's rRMSE is less than 20\% worse than IVQR in the least favorable situation.

For the IVQR-QR averaging estimator, the patterns at other quantiles are almost the same as at the median.


\Cref{tab:S1:IVQR-BS:comparetau} presents the bootstrap averaging results at different quantiles.
Although bootstrap averaging does not uniformly dominate the conservative IVQR estimator, it offers significant improvement in a variety of DGPs with relatively small downside.
Of the $98$ relative rRMSEs reported in the BS columns in \cref{tab:est:M1:tau=0.5,tab:S1:IVQR-BS:comparetau}, only five are $1.15$ or above, with the largest being $1.29$, compared to $30$ relative rRMSEs below $0.85$.

\begin{table}[htbp]
    \centering\caption{\label{tab:S1:IVQR-BS:comparetau} Relative rRMSE of bootstrap averaging in simulation model 1.}
            \sisetup{round-precision=2,round-mode=places}
    \begin{threeparttable}
    \begin{tabular}{ccccScccrcccccScccc}
    \toprule
    &                     \multicolumn{2}{c}{$\tau=0.2$}&   \multicolumn{2}{c}{$\tau=0.3$}            &&  \multicolumn{2}{c}{$\tau=0.4$} &   \multicolumn{2}{c}{$\tau=0.6$}             &&  
    \multicolumn{2}{c}{$\tau=0.7$} &    \multicolumn{2}{c}{$\tau=0.8$}           
      \\
    \cmidrule{2-3} \cmidrule{4-5} \cmidrule{7-8} \cmidrule{9-10} \cmidrule{12-13}  \cmidrule{14-15}  
    $\mathrm{DGP }$   
    & $\mathrm{BS}$ & $\mathrm{CON}$ & $\mathrm{BS}$ & $\mathrm{CON}$ 
    && $\mathrm{BS}$ & $\mathrm{CON}$ & $\mathrm{BS}$ & $\mathrm{CON}$ 
    && $\mathrm{BS}$ & $\mathrm{CON}$ & $\mathrm{BS}$ & $\mathrm{CON}$ 
 \\
    \midrule
1            &   $\num[math-rm=\mathbf]{0.717590}$  & $\num{0.200699}$
& $\num[math-rm=\mathbf]{0.777669}$  & $\num{0.222795}$ 
&& $\num[math-rm=\mathbf]{0.746059}$  & $\num{0.262160}$ 
& $\num[math-rm=\mathbf]{0.728276}$  & $\num{0.374756}$ 
&&  $\num[math-rm=\mathbf]{0.667956}$  & $\num{0.441311}$ 
& $\num[math-rm=\mathbf]{0.688927}$ & $\num{0.586038}$
\\
2            &   $\num[math-rm=\mathbf]{0.805050}$  & $\num{0.238507}$ 
& $\num[math-rm=\mathbf]{0.910770}$  & $\num{0.242750}$ 
&& $\num[math-rm=\mathbf]{0.942130}$  & $\num{0.275171}$
& $\num[math-rm=\mathbf]{0.902930}$  & $\num{0.386376}$ 
&&  $\num[math-rm=\mathbf]{0.950457}$  & $\num{0.397837}$ 
& $\num[math-rm=\mathbf]{0.778211}$ & $\num{0.531690}$ 
\\
3            &    $\num[math-rm=\mathbf]{0.973480}$  & $\num{0.268979}$   
& $\num{1.025453}$  & $\num{0.289591}$
&& $\num{1.026058}$  & $\num{0.339472}$   
& $\num{1.052818}$  & $\num{0.383936}$ 
&&   $\num[math-rm=\mathbf]{0.951108}$  & $\num{0.406742}$ 
& $\num[math-rm=\mathbf]{0.863831}$ & $\num{0.590445}$
\\
4            &    $\num[math-rm=\mathbf]{0.997812}$  & $\num{0.286971}$ 
& $\num{1.038450}$  & $\num{0.280525}$ 
&& $\num[math-rm=\mathbf]{0.985335}$  & $\num{0.323231}$
& $\num{1.015765}$  & $\num{0.327001}$  
&&   $\num[math-rm=\mathbf]{0.918032}$  & $\num{0.409502}$  
& $\num[math-rm=\mathbf]{0.870871}$ & $\num{0.473768}$
\\
5            &   $\num[math-rm=\mathbf]{0.931778}$  & $\num{0.262225}$  
& $\num[math-rm=\mathbf]{0.881007}$  & $\num{0.435454}$  
&& $\num[math-rm=\mathbf]{0.900719}$  & $\num{0.845524}$ 
& $\num[math-rm=\mathbf]{0.961853}$  & $\num{2.284178}$ 
&&   $\num[math-rm=\mathbf]{0.810318}$  & $\num{3.002037}$ 
& $\num[math-rm=\mathbf]{0.927856}$ & $\num{3.823902}$
\\
6            &   $\num[math-rm=\mathbf]{0.880377}$  & $\num{0.292543}$   
&  $\num[math-rm=\mathbf]{0.907597}$  & $\num{0.471132}$ 
&& $\num[math-rm=\mathbf]{0.896864}$  & $\num{0.753200}$ 
& $\num[math-rm=\mathbf]{0.932082}$  & $\num{2.005065}$ 
&&  $\num[math-rm=\mathbf]{0.747725}$  & $\num{2.849327}$
& $\num[math-rm=\mathbf]{0.920452}$ & $\num{3.823448}$
\\
7            &  $\num{1.182218}$  & $\num{0.334383}$ 
& $\num{1.143677}$  & $\num{0.521376}$ 
&&  $\num{1.107408}$  & $\num{0.773788}$ 
& $\num{1.099726}$  & $\num{1.626560}$
&&  $\num[math-rm=\mathbf]{0.807219}$  & $\num{1.970702}$ 
& $\num[math-rm=\mathbf]{0.988952}$ & $\num{2.182250}$
\\
8            &   $\num{1.235545}$  & $\num{0.392785}$ 
&  $\num{1.192488}$  & $\num{0.664187}$  
&&  $\num{1.105882}$  & $\num{0.949278}$ 
& $\num{1.292962}$  & $\num{2.084858}$ 
&&  $\num[math-rm=\mathbf]{0.803794}$  & $\num{2.818585}$ 
& $\num{1.015731}$ & $\num{3.089752}$  
\\
9            &     $\num[math-rm=\mathbf]{0.957760}$  & $\num{0.336555}$  
& $\num{1.009758}$  & $\num{0.705640}$  
&& $\num{1.001349}$  & $\num{1.341854}$ 
& $\num[math-rm=\mathbf]{0.955567}$  & $\num{4.101009}$ 
&&  $\num[math-rm=\mathbf]{0.720695}$  & $\num{5.971537}$ 
& $\num[math-rm=\mathbf]{0.941050}$ & $\num{7.506477}$
\\
10            &   $\num[math-rm=\mathbf]{0.720202}$  & $\num{0.262022}$  
& $\num[math-rm=\mathbf]{0.800783}$  & $\num{0.190439}$ 
&& $\num[math-rm=\mathbf]{0.747600}$  & $\num{0.186220}$ 
& $\num[math-rm=\mathbf]{0.738735}$  & $\num{0.192385}$ 
&&  $\num[math-rm=\mathbf]{0.804744}$  & $\num{0.189647}$ 
& $\num[math-rm=\mathbf]{0.747379}$ & $\num{0.229356}$ 
\\
11            &  $\num[math-rm=\mathbf]{0.841286}$  & $\num{0.421388}$  
& $\num[math-rm=\mathbf]{0.825398}$  & $\num{0.276590}$
&& $\num[math-rm=\mathbf]{0.767596}$  & $\num{0.213634}$ 
& $\num[math-rm=\mathbf]{0.742150}$  & $\num{0.210796}$ 
&&  $\num[math-rm=\mathbf]{0.748903}$  & $\num{0.319371}$
& $\num[math-rm=\mathbf]{0.781523}$ & $\num{0.467432}$
\\
12            &   $\num[math-rm=\mathbf]{0.808155}$  & $\num{0.296785}$  
& $\num[math-rm=\mathbf]{0.911510}$  & $\num{0.228425}$
&& $\num[math-rm=\mathbf]{0.970161}$  & $\num{0.181721}$
& $\num{1.080718}$  & $\num{0.152095}$ 
&&  $\num{1.060233}$  & $\num{0.155356}$ 
& $\num{1.070966}$ & $\num{0.153497}$  
\\
13            &    $\num[math-rm=\mathbf]{0.898756}$  & $\num{0.331279}$ 
& $\num[math-rm=\mathbf]{0.971984}$  & $\num{0.254676}$ 
&& $\num{1.057774}$  & $\num{0.194289}$  
& $\num{1.041635}$  & $\num{0.158934}$ 
&&  $\num{1.023785}$  & $\num{0.167077}$ 
& $\num{1.107483}$ & $\num{0.175413}$  
\\
14            &   $\num[math-rm=\mathbf]{0.833138}$  & $\num{0.793780}$ 
& $\num[math-rm=\mathbf]{0.967510}$  & $\num{0.406347}$ 
&& $\num[math-rm=\mathbf]{0.908508}$  & $\num{0.253328}$
& $\num{1.075077}$  & $\num{0.202276}$ 
&&  $\num{1.076085}$  & $\num{0.231732}$  
& $\num{1.111387}$ & $\num{0.325663}$  
\\
\bottomrule
\end{tabular}
\begin{tablenotes}
 \item $\num{400}$ replications, $\num{50}$ bootstraps, sample size $1000$. ``CON'' is the rRMSE of the conservative IVQR estimator; ``BS'' is the bootstrap averaging relative rRMSE (bootstrap rRMSE divided by conservative rRMSE).
\end{tablenotes}
            \end{threeparttable}
            \end{table}

\subsection{Simulation Model 2}
\label{sim:model2}

This simulation model is close to simulation DGP 2 in \citetalias{ChengLiaoShi2019}.
The \citetalias{ChengLiaoShi2019} model considers endogenous regressors at a fixed endogeneity level with a set of valid instruments, while introducing potentially invalid instruments in the additional moments with varying degrees of endogeneity. 
Since my paper proposes using the IVQR-2SLS and IVQR-QR types of additional moments, no additional instruments are involved.
Instead of varying the endogeneity of additional instruments, I let the regressors' endogeneity level vary across DGPs.
Otherwise, the DGPs are the same as in \citetalias{ChengLiaoShi2019}.
This simulation model 2 has enough continuous endogenous regressors and instruments that uniform dominance holds.

\subsubsection{Simulation DGP}
\label{subsec:simDGP2}

Consider the linear model
\begin{equation}
Y=\theta_0 + \sum_{j=1}^{6} \theta_j X_j + u,
\end{equation}
where $\theta_1=\theta_2=\cdots=\theta_6=2.5$ and $\theta_0=1 $; 
$Y$ is the dependent variable; 
and there are six endogenous regressors $(X_{1}, \ldots, X_{6})$ 
and twelve excluded valid instruments $(Z_{1}, \ldots, Z_{12})$.
The data $(Y_i, X_{1,i}, \ldots, X_{6,i}, Z_{1,i}, \ldots, Z_{12,i})$ are sampled iid for $i=1,\ldots,n$.
The regressors are generated by 
\begin{equation}
X_j=(Z_j+Z_{j+6})/2+\epsilon_j, 
\textrm{ for }
j=1, \ldots, 6 .
\end{equation}
The $(Z_{1}, \ldots, Z_{12}, \epsilon_1, \ldots, \epsilon_6, u) $ are generated from a multivariate normal distribution with mean zero and covariance matrix $diag(\matf{I}_{12\times12}, \matf{\Sigma}_{7\times7})$, where 
\begin{equation}\label{data-generate}
    \matf{\Sigma}_{7\times7} =
\left[ {
\begin{array}{cc}
\matf{I}_{6\times6} &  c_0 \times \vecf{1}_{6\times1}     \\
c_0 \times \vecf{1}_{1\times6}  & 1    
\end{array}
}
\right] .
\end{equation}
The $\matf{I}_{6\times6}$ and $\matf{I}_{12\times12}$ are the identity matrix. 
The $\vecf{1}_{6\times1}$ and $\vecf{1}_{1\times6}$ are the column and row vector of ones, respectively. 
In each DGP, $c_0$ is a fixed constant. 
It measures how much the regressors and the structural error are correlated. 
For $\matf{\Sigma}_{7\times7}$ to be positive definite it requires $c_0\le 0.4$.
Thus, let $c_0$ take value in $\{0, 0.05, 0.1, \ldots, 0.4 \} $.
The higher the $c_0$ value, the more endogenous the regressors.
In the special case $c_0=0$, the regressors are generated independently from the structural error.
The simulation results demonstrate that 
in the case $c_0=0.4$,
the regressors are endogenous enough that the QR estimator performs much worse than IVQR estimator.

In this data-generating process, 
the error term $u$ follows a Gaussian distribution. 
Then, I make a location shift to make its $\tau$-quantile to be zero. 
This shift is to make the model to satisfy the conditional IVQR restriction that 
conditional on instruments, the $\tau$-quantile of error term is zero.

The conservative IVQR moment functions use the twelve instruments, 
\begin{equation}\label{moments:cons}
g_{1,j}(\beta)= Z_j (\Ind{Y-X_1\beta_1-\cdots-X_6\beta_6 \le 0}-\tau) 
\textrm{ for }
j=1, \ldots, 12.
\end{equation}
The additional moment functions for IVQR-QR averaging are
\begin{equation}\label{moments:add:IVQR-QR}
g_{j}^{*}(\beta)= X_j (\Ind{Y-X_1\beta_1-\cdots-X_6\beta_6 \le 0}-\tau) 
\textrm{ for }
j=1, \ldots, 6.
\end{equation}
The additional moment functions for IVQR-2SLS averaging are
\begin{equation}\label{moments:add:IVQR-2SLS}
g_{j}^{*}(\beta)= (Z_{j}-\bar{Z}_{j})u
\textrm{ for }
j=1, \ldots, 12.
\end{equation}

The IVQR estimator is the GMM estimator from the 12 moments using \cref{moments:cons}.
The IVQR-QR aggressive estimator is the GMM estimator from the 18 moments using \cref{moments:cons,moments:add:IVQR-QR}.
The IVQR-2SLS aggressive estimator is the GMM estimator from the 24 moments using \cref{moments:cons,moments:add:IVQR-2SLS}.

\subsubsection{Simulation results}
\label{subsec:simDGP2-avg}

\begin{table}[htbp!]
    \centering\caption{\label{tab:M2:S1S22:tau=0.5} Relative rRMSE in simulation model 2, $\tau=0.5$.}
            \sisetup{round-precision=2,round-mode=places}
    \begin{threeparttable}
    \begin{tabular}{cSScccScccScccc}
    \toprule
    &                    && \multicolumn{2}{c}{$\mathrm{IVQR.2SLS}$}&               && \multicolumn{2}{c}{$\mathrm{IVQR.QR} $} &               &&  &    \\
    \cmidrule{4-5} \cmidrule{8-9}
    $\mathrm{DGP }$  &  $\mathrm{Endog}$  &&  $\mathrm{AVG} $ & $\mathrm{AGG}$       & $\mathrm{2SLS}$ && $\mathrm{AVG}$ & $\mathrm{AGG}$         & $\mathrm{QR}$ && $\mathrm{BS} $    & $\mathrm{SEE}$ \\
    \midrule
1          &    0  &&   $\num[math-rm=\mathbf]{0.876621}$  & $\num{0.816984}$ & $\num{0.769302}$  && $\num[math-rm=\mathbf]{0.627007}$  &  $\num{0.464573}$ & $\num{0.457321}$ && $\num[math-rm=\mathbf]{0.554361}$ &  $\num{0.815396}$  \\
2          & 0.05  &&   $\num[math-rm=\mathbf]{0.919416}$  & $\num{0.923405}$ & $\num{0.783947}$  && $\num[math-rm=\mathbf]{0.685586}$  &  $\num{0.566005}$ & $\num{0.548458}$ && $\num[math-rm=\mathbf]{0.623790}$ &  $\num{0.847613}$  \\
3          & 0.1   &&   $\num[math-rm=\mathbf]{0.929650}$  & $\num{0.945819}$ & $\num{0.799861}$  && $\num[math-rm=\mathbf]{0.775141}$  &  $\num{0.808968}$ & $\num{0.800620}$ && $\num[math-rm=\mathbf]{0.704582}$ &  $\num{0.840211}$  \\
4          & 0.15  &&   $\num[math-rm=\mathbf]{0.920974}$  & $\num{0.912700}$ & $\num{0.755202}$  && $\num[math-rm=\mathbf]{0.894906}$  &  $\num{1.163727}$ & $\num{1.163325}$ && $\num[math-rm=\mathbf]{0.789633}$ &  $\num{0.829286}$  \\
5          &  0.2  &&   $\num[math-rm=\mathbf]{0.917827}$  & $\num{0.927735}$ & $\num{0.789565}$  && $\num[math-rm=\mathbf]{0.910303}$  &  $\num{1.433761}$ & $\num{1.406581}$ && $\num[math-rm=\mathbf]{0.828754}$ &  $\num{0.829167}$  \\
6          & 0.25  &&   $\num[math-rm=\mathbf]{0.925927}$  & $\num{0.914773}$ & $\num{0.789525}$  && $\num[math-rm=\mathbf]{0.952613}$  &  $\num{1.798988}$ & $\num{1.768014}$ && $\num[math-rm=\mathbf]{0.843464}$ &  $\num{0.869715}$  \\
7          &  0.3  &&   $\num[math-rm=\mathbf]{0.881145}$  & $\num{0.867900}$ & $\num{0.756437}$  && $\num[math-rm=\mathbf]{0.963672}$  &  $\num{2.075631}$ & $\num{2.051260}$ && $\num[math-rm=\mathbf]{0.840108}$ &  $\num{0.785066}$  \\
8          & 0.35  &&   $\num[math-rm=\mathbf]{0.901901}$  & $\num{0.911520}$ & $\num{0.708683}$  && $\num[math-rm=\mathbf]{0.974877}$  &  $\num{2.550676}$ & $\num{2.505178}$ && $\num[math-rm=\mathbf]{0.812182}$ &  $\num{0.801656}$  \\
9          &  0.4  &&   $\num[math-rm=\mathbf]{0.914493}$  & $\num{0.945680}$ & $\num{0.774773}$  && $\num[math-rm=\mathbf]{0.964646}$  &  $\num{2.958814}$ & $\num{2.886673}$ && $\num[math-rm=\mathbf]{0.843013}$ &  $\num{0.853452}$  \\
\bottomrule
            \end{tabular}
\begin{tablenotes}
 \item $\num{200}$ replications. $\num{50}$ bootstraps. Sample size is 1000.
\end{tablenotes}
\end{threeparttable}
\end{table}

\begin{figure}[htbp!]
 \centering
\includegraphics[width=0.5\textwidth]{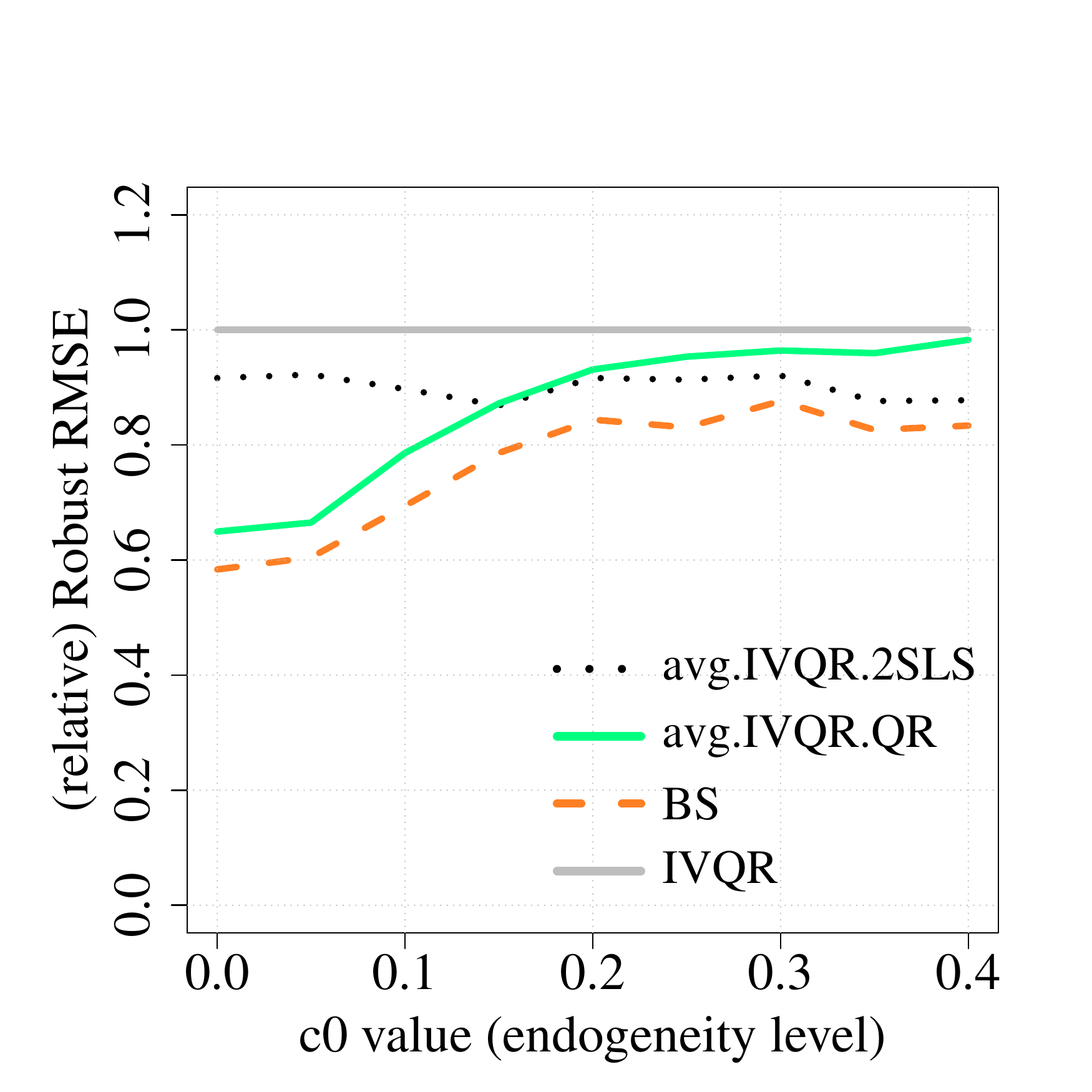}
\caption{\label{fig:sim-M2:CLS-FH-S1S22-tau0.5}%
Relative rRMSE in simulation model 2; $\tau=0.5$, $n=1000$, $\num{200}$ replications, $\num{50}$ bootstraps.} 
\end{figure}

In simulation model 2, 
all three averaging estimators 
uniformly dominate the IVQR estimator. 
\Cref{fig:sim-M2:CLS-FH-S1S22-tau0.5} shows these three averaging estimators' relative rRMSE is bounded below 1
in a class of DGPs with different endogeneity level. 
\Cref{tab:M2:S1S22:tau=0.5} reports the rRMSE of various estimators relative to that of the IVQR estimator.
The second column reports the value of $c_0$, 
which indicates how much regressor endogeneity there is in the DGP.
Columns 3, 6, and 9 are the three averaging estimators proposed in this paper. 
Columns 4 and 5 are the IVQR-2SLS aggressive estimator and the 2SLS estimator for reference.
Columns 7 and 8 are the IVQR-QR aggressive estimator and the QR estimator for reference. 
Column 10 is the smoothed estimating equations (SEE) estimator proposed by \citet{KaplanSun2017},
using their code's plug-in bandwidth.

In simulation model 2, 
the regressors' coefficients are set to be fixed.
There is no slope heterogeneity.
The 2SLS estimator performs uniformly better than the IVQR estimator, 
in the sense that its relative rRMSE is less than 1 across all DGPs.
The IVQR-2SLS aggressive estimator and averaging estimator both have relative rRMSE less than 1 and uniformly dominate the IVQR estimator.

The more endogeneity of regressors in a DGP, 
the more misspecified the QR moments, 
and the worse the performance of QR.
In DGPs 1--3 with the least endogeneity, 
the QR estimator performs better than the IVQR estimator. 
In DGPs 4--9 with more endogeneity, 
QR is increasingly worse than IVQR.
The IVQR-QR aggressive estimator follows the same pattern as the QR estimator.
However, the IVQR-QR averaging estimator has relative rRMSE less than 1 in all DGPs,
even in the DGPs where the QR moments are very misspecified 
and the QR and IVQR-QR aggressive estimators are much worse than the IVQR estimator.
This shows the empirical optimal weight formula works as desired, putting more weight on the aggressive estimator when the additional moments reduce variance more relative to the increased bias, and putting more weight on the conservative estimator when the additional moments are severely misspecified.

The bootstrap averaging estimator performs best of all.
It not only uniformly dominates the IVQR estimator, up to 45\% efficiency gain, 
but also is uniformly better than either type of GMM averaging estimator. 
Moreover, 
the bootstrap averaging estimator usually performs better than the SEE estimator, although not uniformly.
As there is no slope heterogeneity in simulation model 2, 
the IVQR-2SLS averaging estimator's relative rRMSE is always below 1, mostly around 0.9.
The IVQR-QR averaging estimator's relative rRMSE is also below 1 in all DGPs, increasing from 0.63 to 0.97 as $c_0$ increases.
The bootstrap averaging estimator has even smaller relative rRMSE, ranging from 0.55 to 0.84 as $c_0$ increases.

The simulation results at other quantiles are reported in \cref{appdx:M2:Gaussian}. 
The results share the same pattern as that at the median.
All three averaging estimators uniformly dominate the IVQR estimator at quantiles from $\tau=0.2$ to $\tau=0.8$.

\subsubsection{Results with non-Gaussian error term}
\label{non-Gaussian}

The uniform dominance in this fixed-coefficient model with multiple endogenous regressors also holds with a non-Gaussian error term.
To illustrate this, 
I set the error term to follow a chi-square distribution with $4$ degrees of freedom.
More specifically, 
I first generate the error term in the same way as in \cref{data-generate}.
Second, I transform the error term to
$u^{*}= F_{\chi_4^2 }^{-1}(\Phi(u))$
where $\Phi(\cdot)$ is the CDF of the original Gaussian error term $u$ and
$F_{\chi_4^2}^{-1}$ is the inverse CDF of a $\chi^2_4$ distribution.
Finally, I shift $u^*$ to have its $\tau$-quantile equal zero.

The simulation results have similar patterns as in the Gaussian error case; see the results tables and figures in \cref{appdx:M2:non-Gaussian}.
All three averaging estimators have relative rRMSE less than 1 in all the DGPs at all quantiles, 
expect three cases of IVQR-2SLS averaging at $\tau=0.2$ having relative rRMSE between 1.01 and 1.03.
However, these three rare cases seem due to simulation error since their relative rRMSEs are less than or equal to 1 when running more simulation replications.

\subsection{Simulation Model 3}
\label{sim:model3}

Seeing slope heterogeneity across quantiles is part of the value of quantile regression. 
Simulation model 3 extends simulation model 2 to allow for slope heterogeneity.

\subsubsection{Simulation DGP}
\label{subsec:simDGP3}

Consider a linear model 
\begin{equation}
Y=\theta_0 + \sum_{j=1}^{6} \theta_j X_j + u,
\end{equation}
with 
$\theta_0=1 $ and random coefficients for individual $i$ equal to 
\begin{equation}\label{slope:hetero}
\theta_{1i}=\theta_{2i}=\cdots=\theta_{6i}=hetero \times [F(u_i)]^4, 
\textrm{ for }
i=1,\ldots,n,
\end{equation}
where $F(\cdot)$ is the CDF of $u$.
The slope term \cref{slope:hetero} is set to be a function of the rank of the error term in its distribution. 
It represents the slope heterogeneity feature. 
The term ``$hetero$'' is a fixed constant in a DGP, the same as $c_2$ in simulation model 1.
When $hetero=0$, 
the slope is a constant zero, 
which means the model has no slope heterogeneity.
The larger the ``$hetero$'' value, 
the more slope heterogeneity in the DGP.
In simulation model 3, 
$hetero$ takes value in $\{0,0.1,\ldots,1\}$.
I consider both 
Gaussian and non-Gaussian errors.
The non-Gaussian error follows the same transformation as in \cref{non-Gaussian}.
The error term has the same location shift as in \cref{sim:model2} to make its $\tau$-quantile equal zero.

As in simulation model 2, there are six endogenous regressors $(X_{1}, \ldots, X_{6})$ and twelve valid instruments $(Z_{1}, \ldots, Z_{12})$, and sampling is iid.
The regressors are generated by 
\begin{equation}
X_j=Z_j+Z_{j+6}+\epsilon_j, 
\textrm{ for }
j=1, \ldots, 6 .
\end{equation}

The monotonicity condition in IVQR requires that 
the structural quantile function is increasing in $u$ given any $\vecf{X}=\vecf{x}\in\mathcal{X}$ \citep[Condition A1]{ChernozhukovHansen2005}. 
For the monotonicity condition to hold in this simulation model,
I restrict the regressors to have non-negative support. 
More specifically, 
I shift each regressor to the right by 3.1 times its standard deviation.
There is less than 0.001 probability 
that a regressor remains negative, in which case I set it to zero.

As in simulation model 2, the $(Z_{1}, \ldots, Z_{12}, \epsilon_1, \ldots, \epsilon_6, u) $ 
are generated from a multivariate normal distribution 
with mean zero and covariance matrix $diag(\matf{I}_{12\times12}, \matf{\Sigma}_{7\times7})$, where 
\begin{equation}
    \matf{\Sigma}_{7\times7} =
\left[ {
\begin{array}{cc}
\matf{I}_{6\times6} &  c_0 \times \vecf{1}_{6\times1}     \\
c_0 \times \vecf{1}_{1\times6}  & 1    
\end{array}
}
\right] .
\end{equation}
Since $c_0 \in \{0, 0.05, 0.1, \ldots, 0.4 \}$, there are $11\times 9$ DGP combinations of $(hetero, c_0)$.

I simulate the rRMSE as in \cref{eqn:rRMSE}, where the true population parameters are $\theta_j=hetero \times \tau^4$ for all $j=1,\ldots,6$.

\subsubsection{Simulation results}
\label{subsec:simDGP3-avg}

\begin{table}[htbp]
    \centering\caption{\label{tab:bound:S3:tau=0.5} Bounds of relative rRMSE in simulation model 3, $\tau=0.5$.}
            \sisetup{round-precision=3,round-mode=places}
    \begin{threeparttable}
    \begin{tabular}{cScccScccSc}
    \toprule
           && \multicolumn{8}{c}{$\mathrm{Fixed \ Endog }$}               \\
    \cmidrule{3-10}
           && \multicolumn{2}{c}{$\mathrm{No}\ (c_0=0)$}              && \multicolumn{2}{c}{$\mathrm{Some}\ (c_0=0.2) $}     && \multicolumn{2}{c}{$\mathrm{Much}\ (c_0=0.4) $}    \\
     \cmidrule{3-4} \cmidrule{6-7} \cmidrule{9-10}
              && $\mathrm{Lower}$   & $\mathrm{Upper}$   &&   $\mathrm{Lower}$  & $\mathrm{Upper}$    && $\mathrm{Lower}$  & $\mathrm{Upper}$                \\
    \cmidrule{1-10}
        $\hat{\theta}_{\mathrm{AVG.2SLS}}$ &&  $\num{0.921558}$  & $\num{1.006412}$  && $\num{0.960113}$  &  $\num{1.024411}$ && $\num{0.958913}$ & $\num{0.996904}$  \\
  $\hat{\theta}_{\mathrm{AVG.QR}}$ &&  $\num{0.830799}$  & $\num{0.919396}$  && $\num{0.996671}$  &  $\num{1.002288}$ && $\num{0.998377}$ & $\num{1.001428}$  \\
        $\hat{\theta}_{\mathrm{BS}}$ &&  $\num{0.769622}$  & $\num{0.929397}$  && $\num{0.850083}$  &  $\num{1.014216}$ && $\num{0.859666}$ & $\num{1.006431}$  \\
    \cmidrule{1-10}
                && \multicolumn{8}{c}{$\mathrm{Fixed \ Hetero }$}               \\
      \cmidrule{3-10}
                && \multicolumn{2}{c}{$\mathrm{No}\ (\mathrm{hetero}=0)$}              && \multicolumn{2}{c}{$\mathrm{Some}\ (\mathrm{hetero}=0.5) $}     && \multicolumn{2}{c}{$\mathrm{Much}\ (\mathrm{hetero}=1) $}    \\
      \cmidrule{3-4} \cmidrule{6-7} \cmidrule{9-10}
                && $\mathrm{Lower}$   & $\mathrm{Upper}$   &&   $\mathrm{Lower}$  & $\mathrm{Upper}$    && $\mathrm{Lower}$  & $\mathrm{Upper}$               \\
      \cmidrule{1-10}
        $\hat{\theta}_{\mathrm{AVG.2SLS}}$ &&  $\num{0.906251}$  & $\num{0.991062}$  && $\num{0.973068}$  &  $\num{1.021872}$ && $\num{0.981040}$ & $\num{1.015973}$  \\
  $\hat{\theta}_{\mathrm{AVG.QR}}$ &&  $\num{0.857462}$  & $\num{1.000541}$  && $\num{0.875519}$  &  $\num{1.001435}$ && $\num{0.888192}$ & $\num{1.001399}$  \\
        $\hat{\theta}_{\mathrm{BS}}$ &&  $\num{0.802951}$  & $\num{0.883027}$  && $\num{0.877296}$  &  $\num{1.003506}$ && $\num{0.886404}$ & $\num{1.018404}$  \\
\bottomrule
            \end{tabular}
            \begin{tablenotes}
            \item $\num{200}$ replications. $\num{50}$ bootstraps. Sample size is 1000.\\
            \end{tablenotes}
            \end{threeparttable}
            \end{table}

In \cref{tab:bound:S3:tau=0.5},
I fix the endogeneity level indicator $c_0$ at zero, 0.2, or 0.4 (highest value), and vary the slope heterogeneity across all values $hetero=0,0.1,\ldots,1$.
Similarly, I fix $hetero$ at zero, 
0.5, or 1 (highest value), and vary the endogeneity level across all values $c_0=0,0.05,\ldots,0.4$.
Altogether this includes 60 DGPs, covering the boundaries of the 99 DGPs and indicating the patterns.
\Cref{tab:bound:S3:tau=0.5} reports the upper and lower bounds of the relative rRMSE of the three averaging estimators in these six cases (i.e., 3 fixed endogeneity levels, 3 fixed heterogeneity levels).

\Cref{tab:bound:S3:tau=0.5} summarizes results at the median in these six cases.
More detailed results tables are included in the supplemental appendix.
We can see that all three averaging estimators have relative rRMSE lower bounds strictly less than 1.
The bootstrap averaging estimator has the smallest lower bound in five cases and is essentially tied for smallest (up to simulation error) in the sixth.
With no endogeneity and varying heterogeneity, 
IVQR-QR and bootstrap averaging both have upper bound strictly below 1, while the IVQR-2SLS averaging estimator has upper bound roughly 1 (up to simulation error).
With no heterogeneity and varying endogeneity, 
IVQR-2SLS and bootstrap averaging both have upper bound strictly below 1, while IVQR-QR averaging has upper bound roughly 1. 
In the other four cases, the three averaging estimators all have upper bound roughly 1. 
Although some exceed 1 (largest value 1.024), this is believed to be entirely due to simulation error: all the upper bounds reduce to $1.000\pm0.001$ with a larger number of simulation replications.

\Cref{fig:M3:S1S22:tau0.5} is similar to \cref{tab:bound:S3:tau=0.5}, but it visualizes results for all 60 DGPs (not just bounds).
It shows that in simulation model 3, 
all three averaging estimators uniformly dominate the IVQR estimator.
Using any of these three averaging estimators would provide more precise estimation than the IVQR estimator.

\begin{figure}[htbp!]
\includegraphics[width=0.45\textwidth, height=0.3\textheight, trim=35 20 20 70]{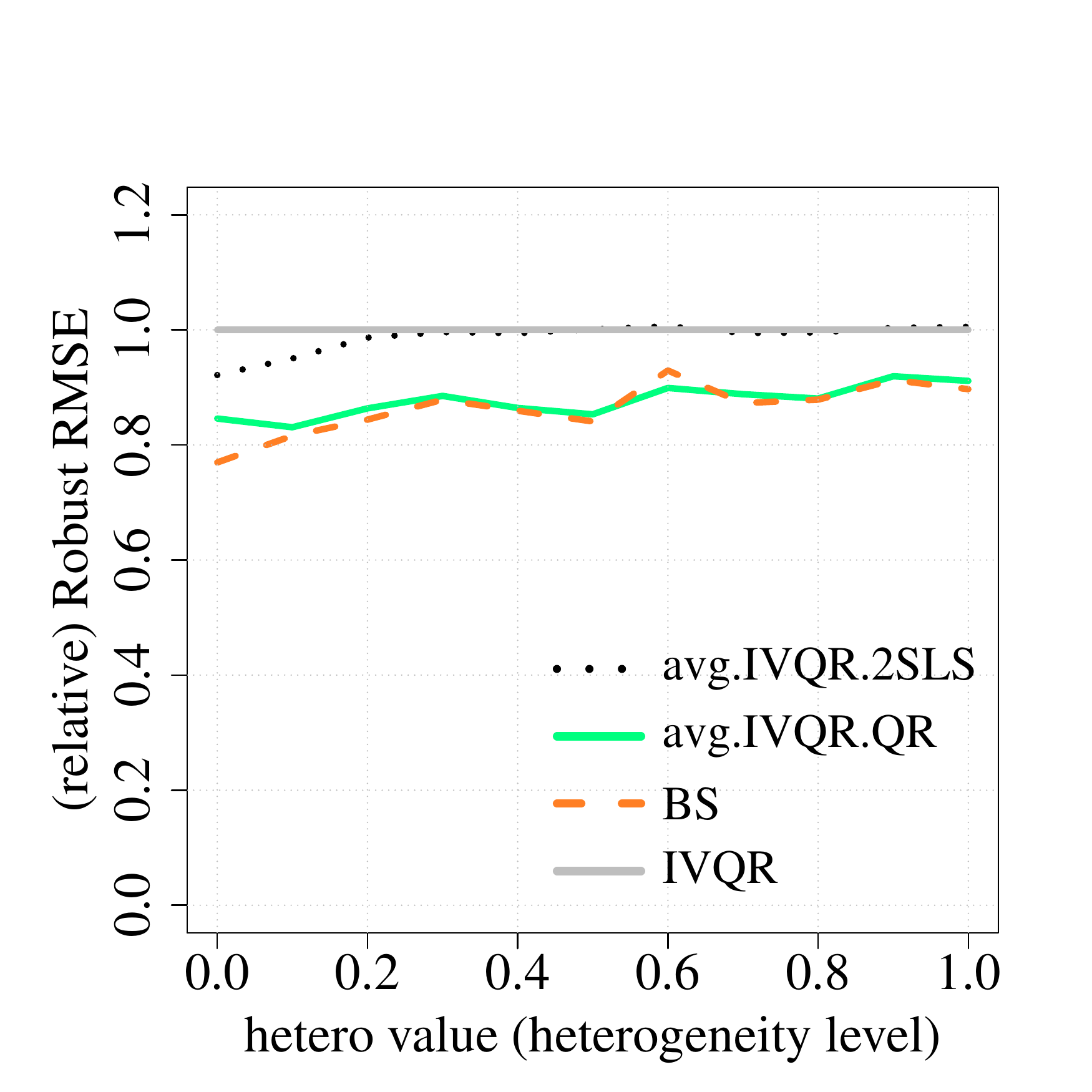}
\hfill
\includegraphics[width=0.45\textwidth, height=0.3\textheight, trim=35 20 20 70]{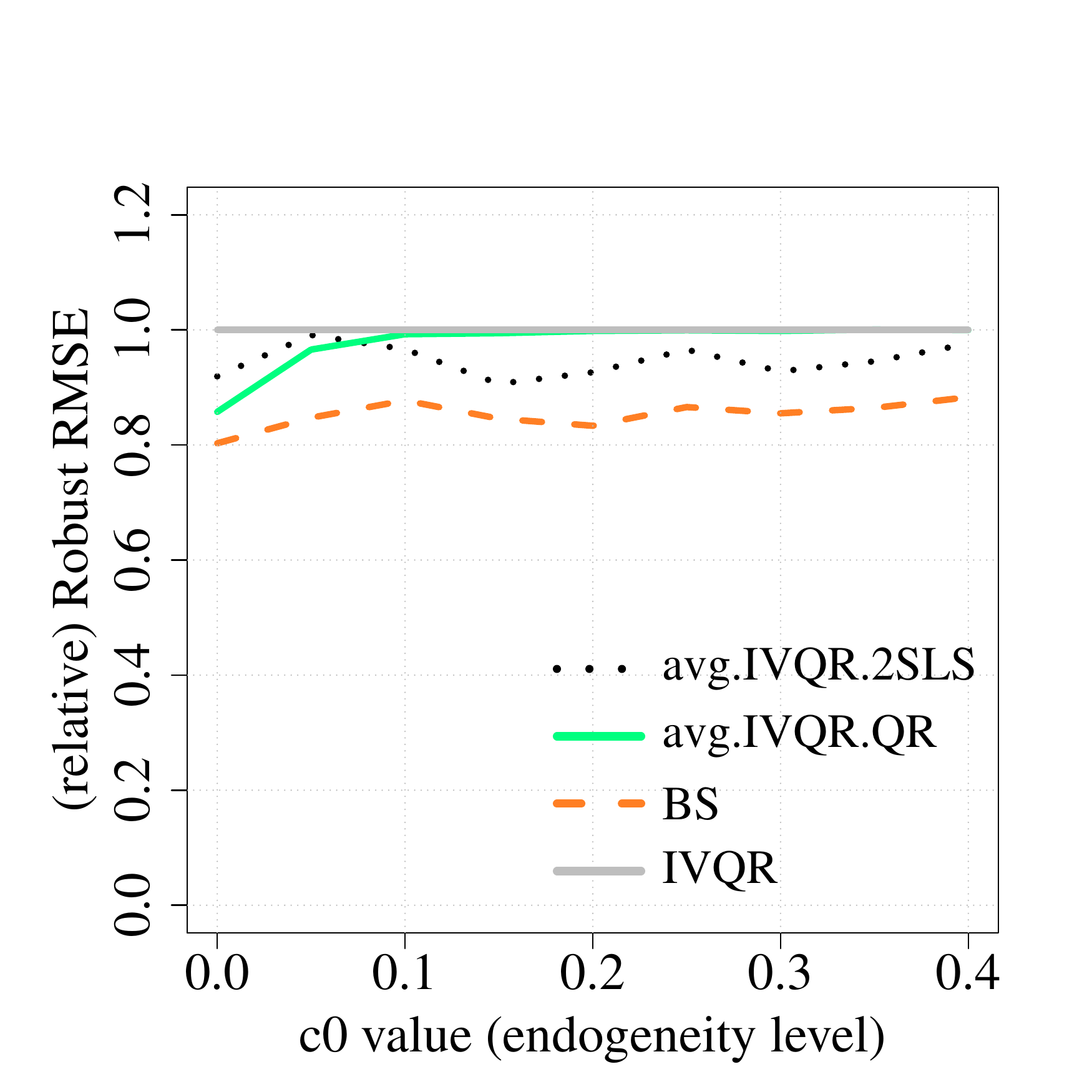}
\hfill\null
\includegraphics[width=0.45\textwidth, height=0.3\textheight, trim=35 20 20 70]{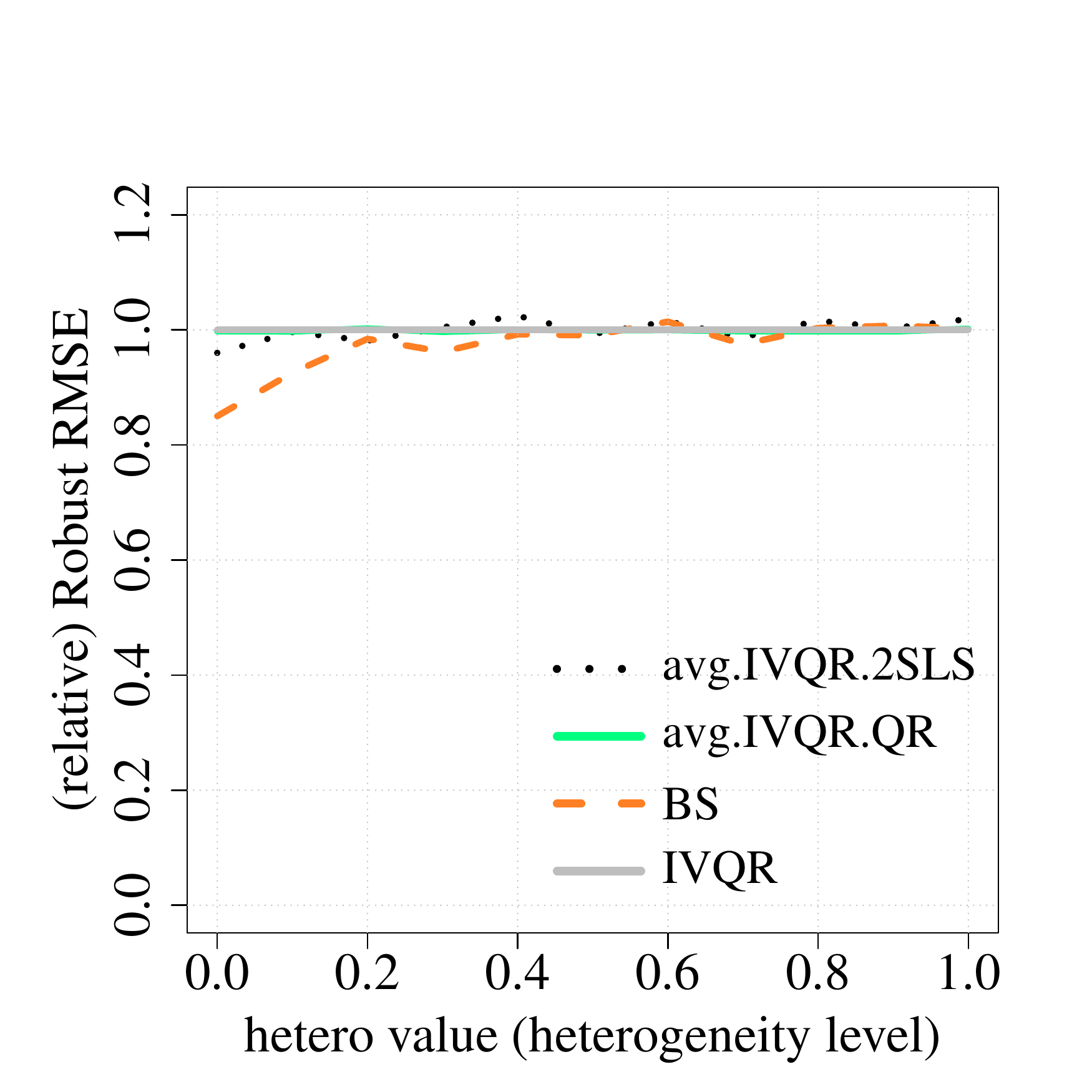}
\hfill
\includegraphics[width=0.45\textwidth, height=0.3\textheight, trim=35 20 20 70 ]{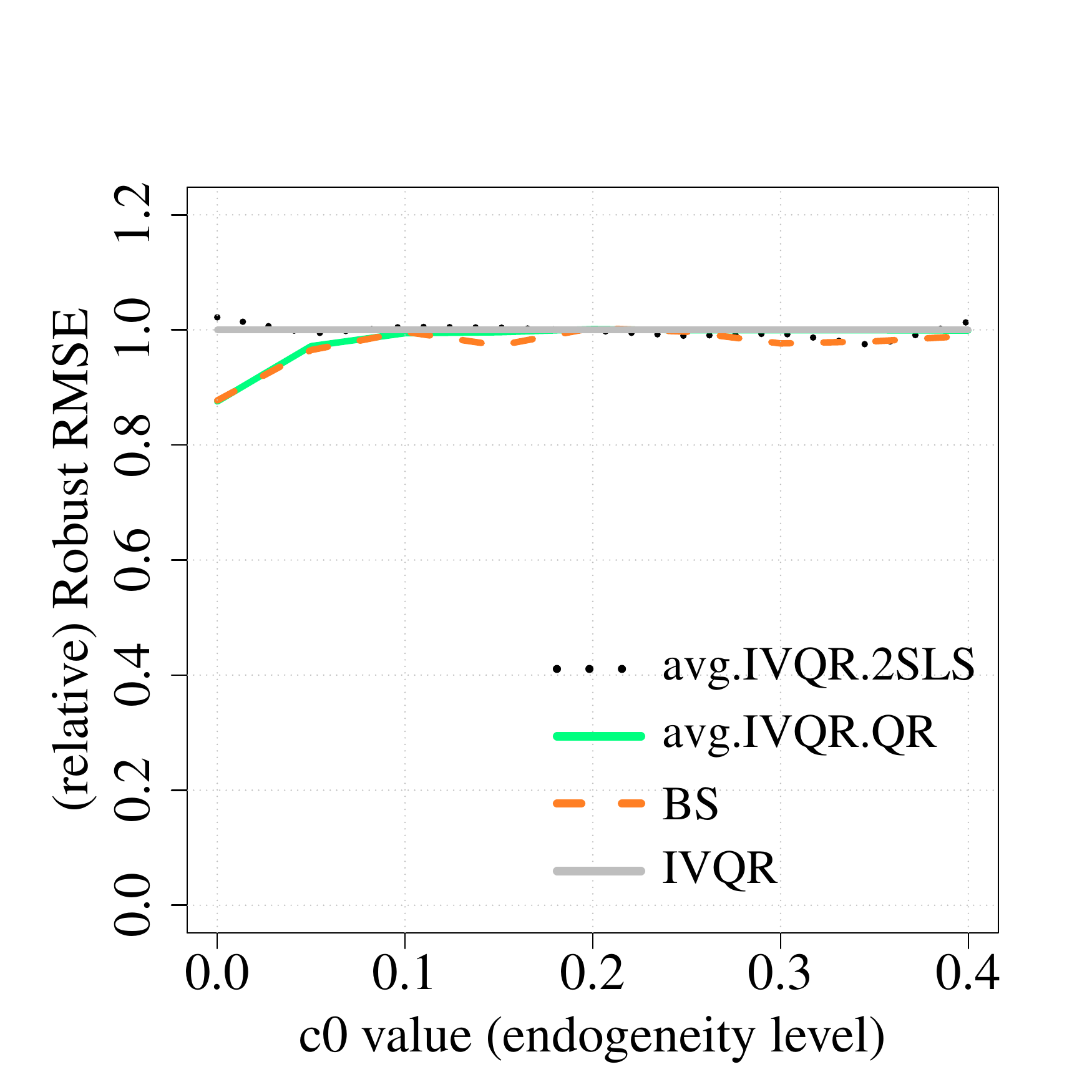}
\hfill\null
\includegraphics[width=0.45\textwidth, height=0.3\textheight, trim=35 20 20 70 ]{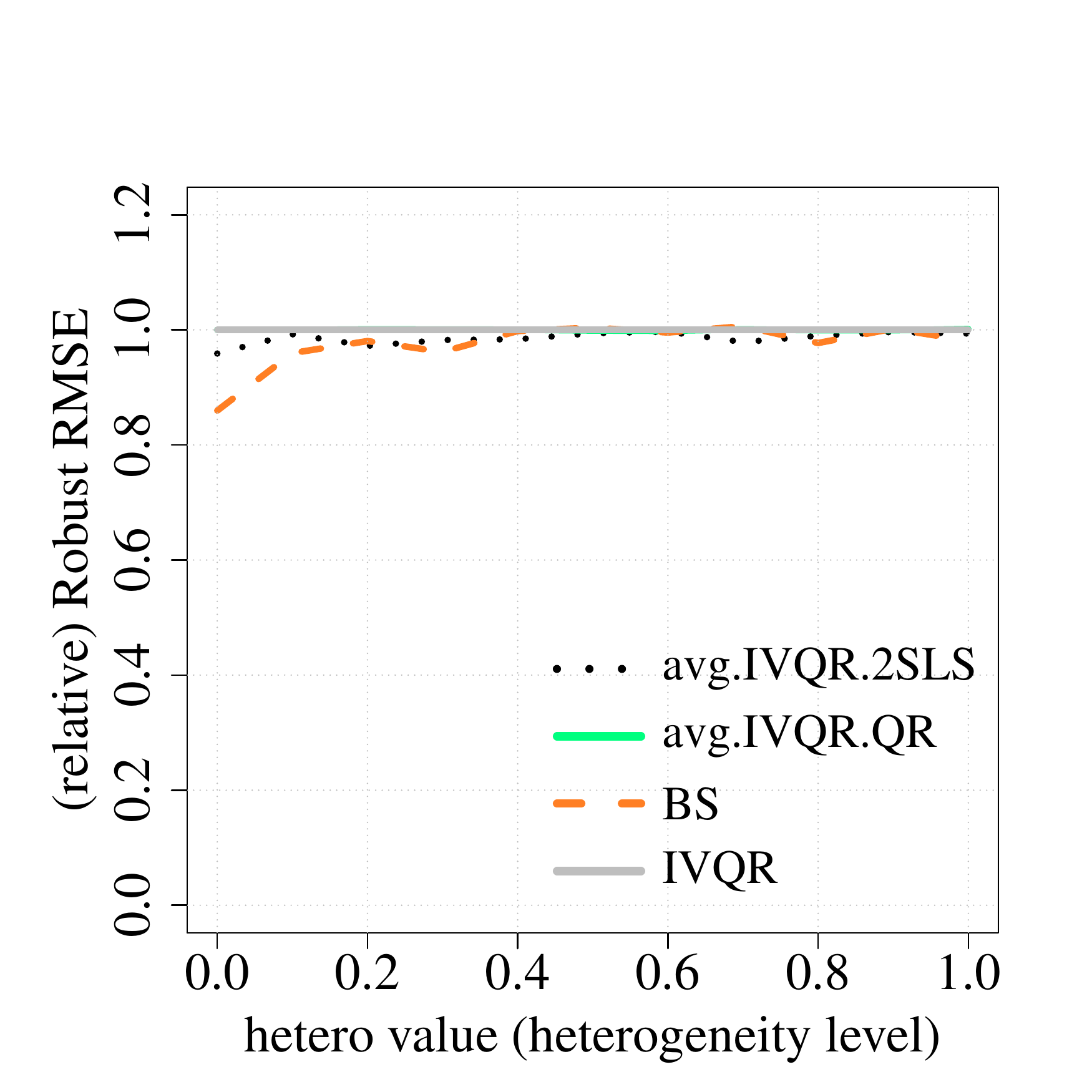}
\hfill
\includegraphics[width=0.45\textwidth, height=0.3\textheight, trim=35 20 20 70]{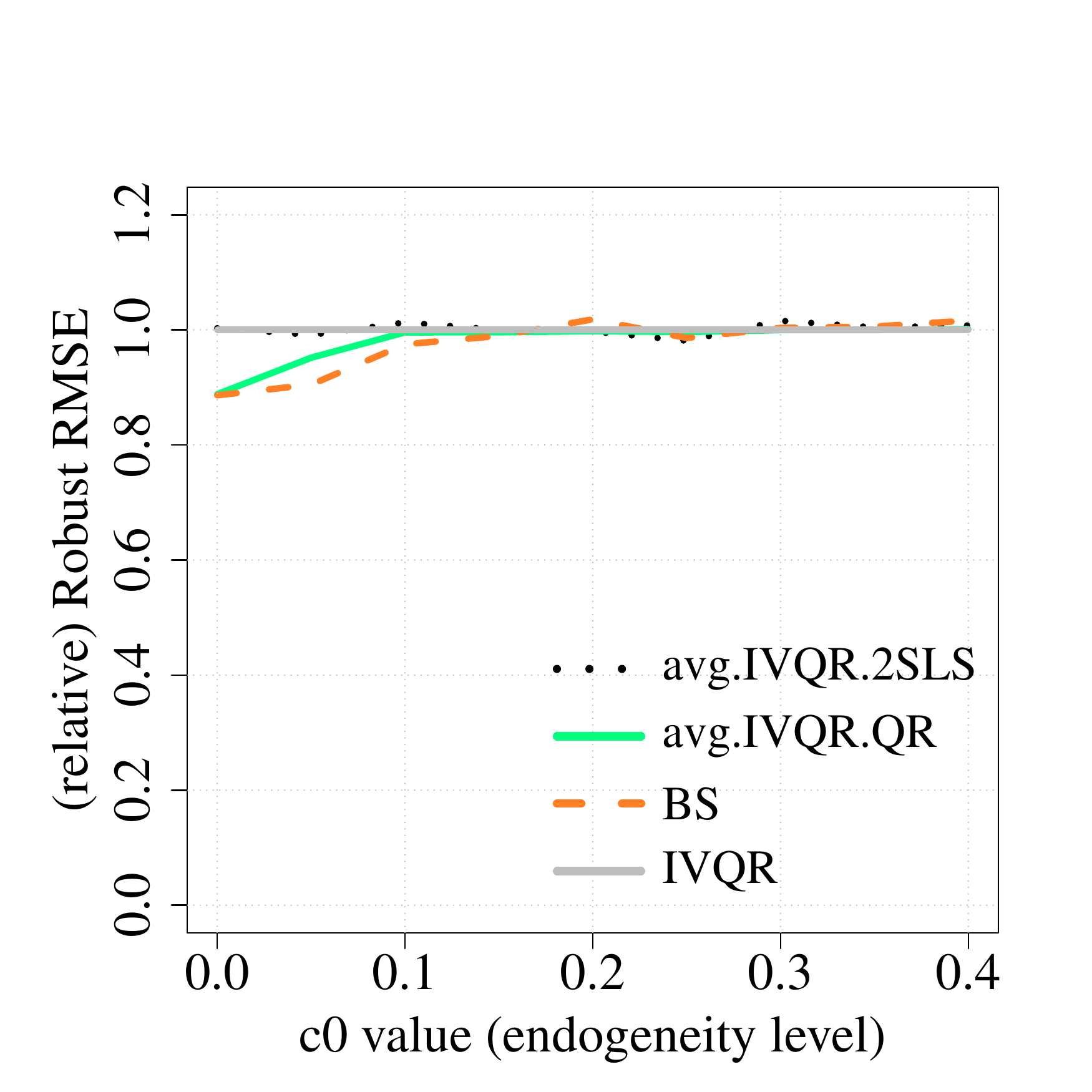}
\caption{\label{fig:M3:S1S22:tau0.5}%
 Relative rRMSE in simulation model 3 at median quantile level in 6 cases: fixed endogeneity level $c_0=0$  (left up), $c_0=0.2$ (left middle), $c_0=0.4$ (left bottom) and varying heterogeneity; and fixed heterogeneity level $hetero=0$ (right up), $hetero=0.5$ (right middle), $hetero=1$ (right bottom) and varying endogeneity,
 based on $\num{200}$ replications and $\num{50}$ bootstraps. Sample size n=1000.} 
\end{figure}

\subsubsection{Results at other quantiles}

\Cref{appdx:M3} includes results like \cref{fig:M3:S1S22:tau0.5,tab:bound:S3:tau=0.5}
at quantiles ranging from $\tau=0.2$ to $\tau=0.8$.
The supplemental appendix provides yet more detailed tables for these cases.

For some quantile levels, for some DGPs, the bootstrap averaging estimator's relative rRMSE is above 1,
but 1.075 is the maximum among 420 values (i.e., 60 DGPs with 7 values of $\tau$ each).
Again, this is partly due to simulation error.
The relative rRMSE exceeds 1.05 only 19 times out of 420.
In 96 out of 420 DGPs, the relative rRMSE is between 1 to 1.04.
In the other 305 out of 420 cases, 
the bootstrap averaging estimator has relative rRMSE less than 1.
At every quantile, 
the bootstrap averaging estimator's relative rRMSE can reach as low as 0.79, and at some quantiles even 0.574.

The QR estimator and the IVQR-QR averaging estimator each have similar patterns across quantiles in the 60 DGPs shown in the table in the supplemental appendix.
One finding is that in the much-endogeneity case 
(i.e., $c_0\ge0.35$),
the IVQR-QR aggressive estimator has a computation problem and performs very poorly.
The IVQR-QR averaging estimator, however, still has relative rRMSE of 1. 
This indicates the empirical weight puts almost all weight on the conservative IVQR estimator and 
zero weight on the IVQR-QR aggressive estimator.
The averaging method works in a desirable way.
The IVQR-QR averaging estimators all show uniform dominance over the IVQR estimator.

\begin{figure}[htbp]
\centering
\includegraphics[width=0.45\textwidth, height=0.3\textheight, trim=35 20 20 70]{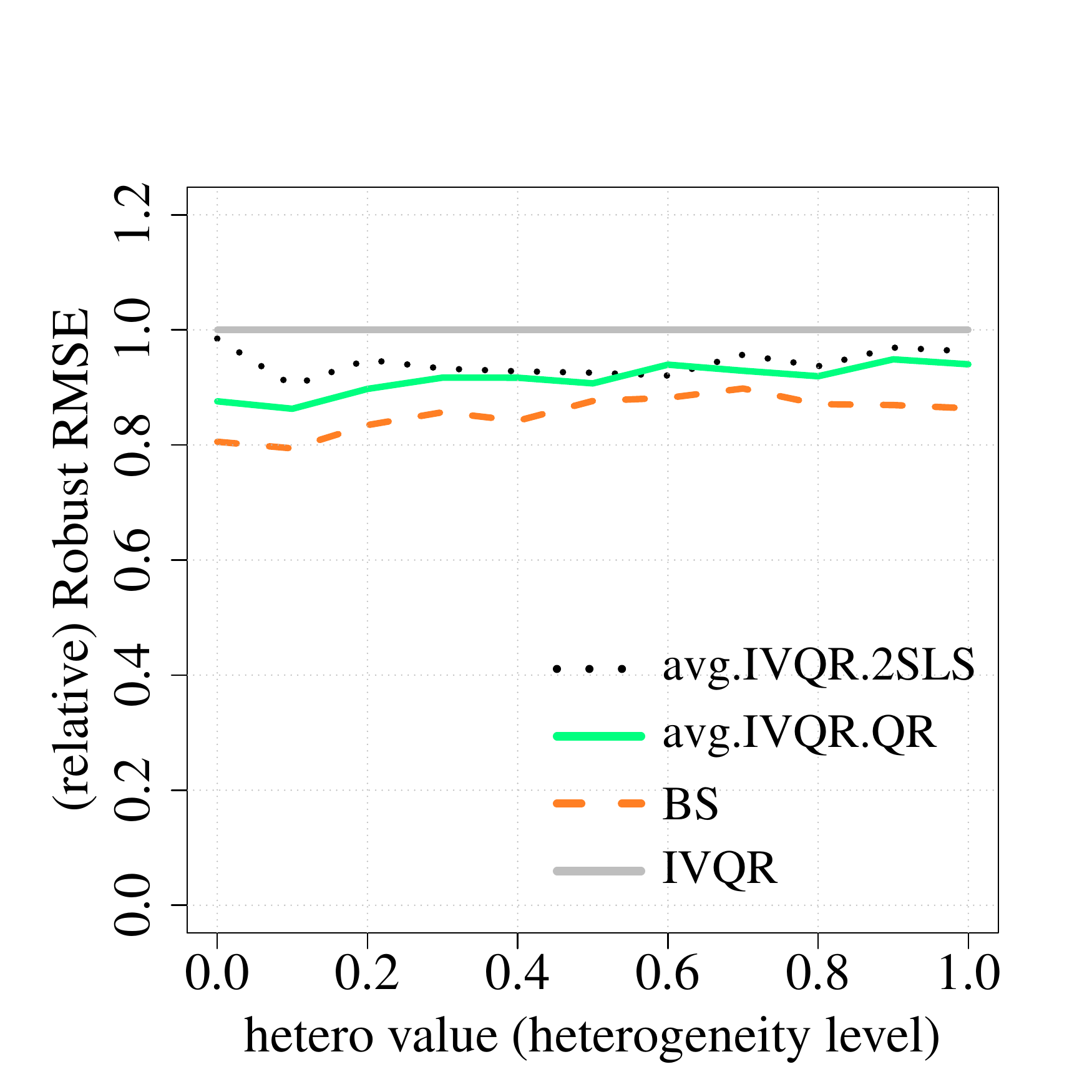}
\hfill
\includegraphics[width=0.45\textwidth, height=0.3\textheight, trim=35 20 20 70]{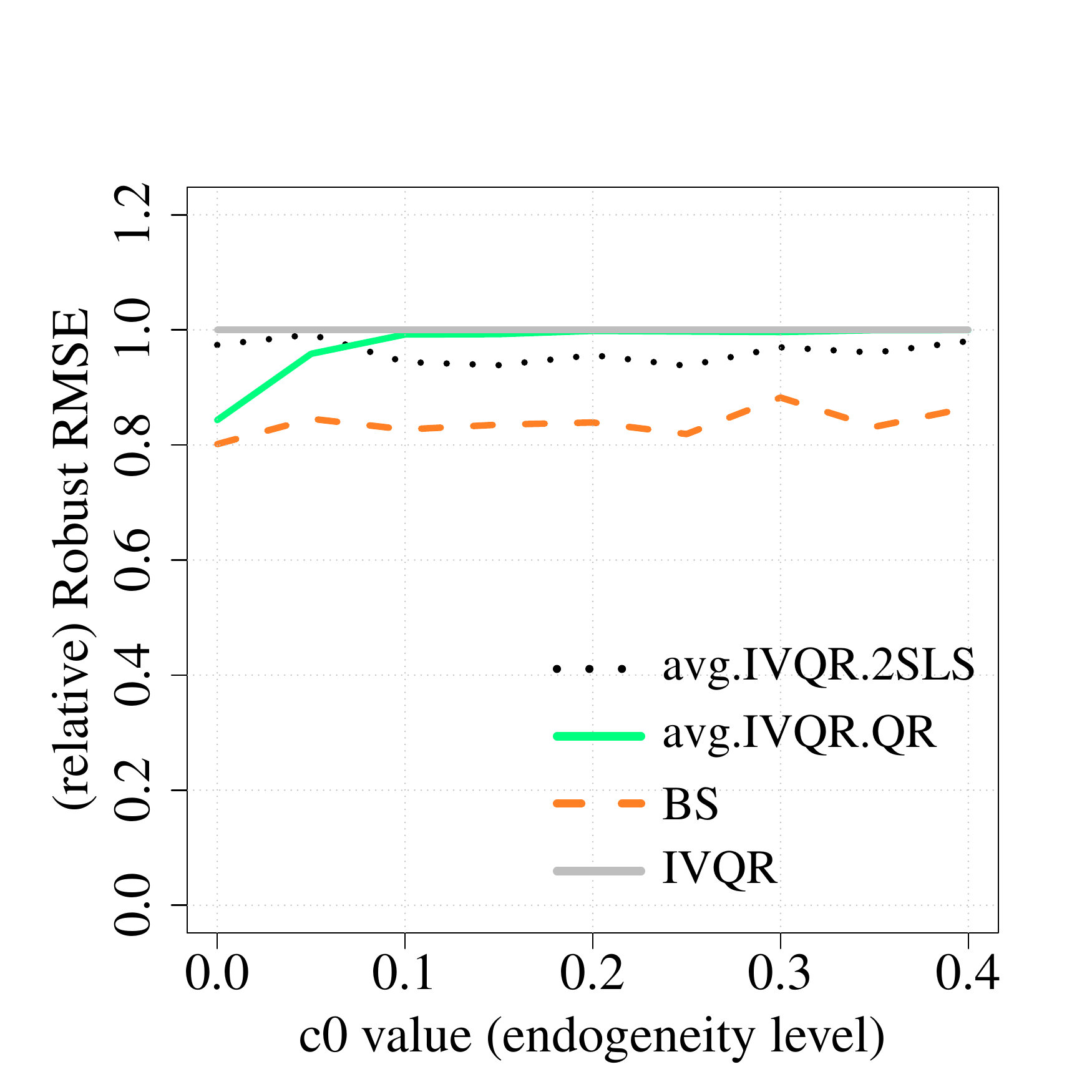}
\\
\includegraphics[width=0.45\textwidth, height=0.3\textheight, trim=35 20 20 70]{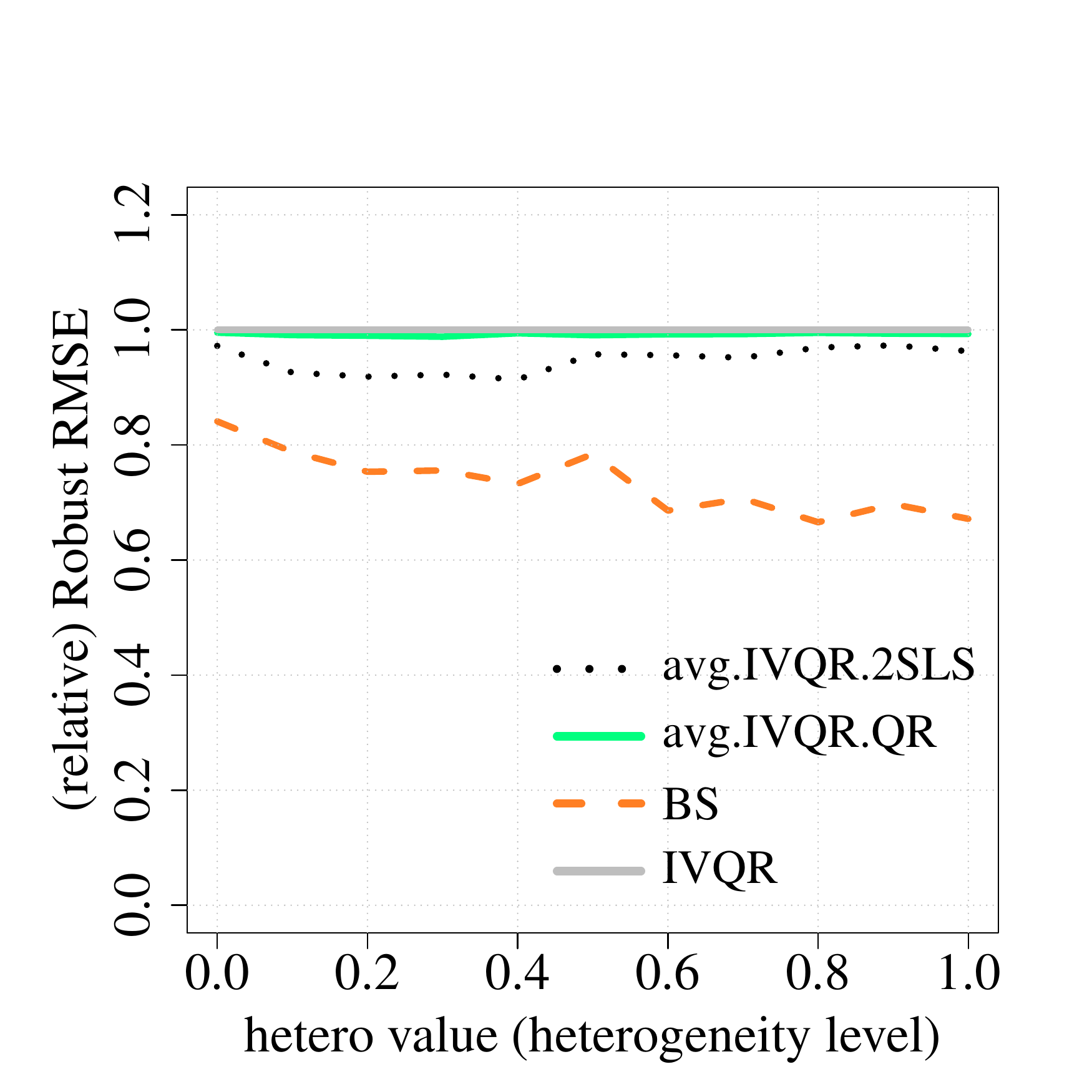}
\hfill
\includegraphics[width=0.45\textwidth, height=0.3\textheight, trim=35 20 20 70 ]{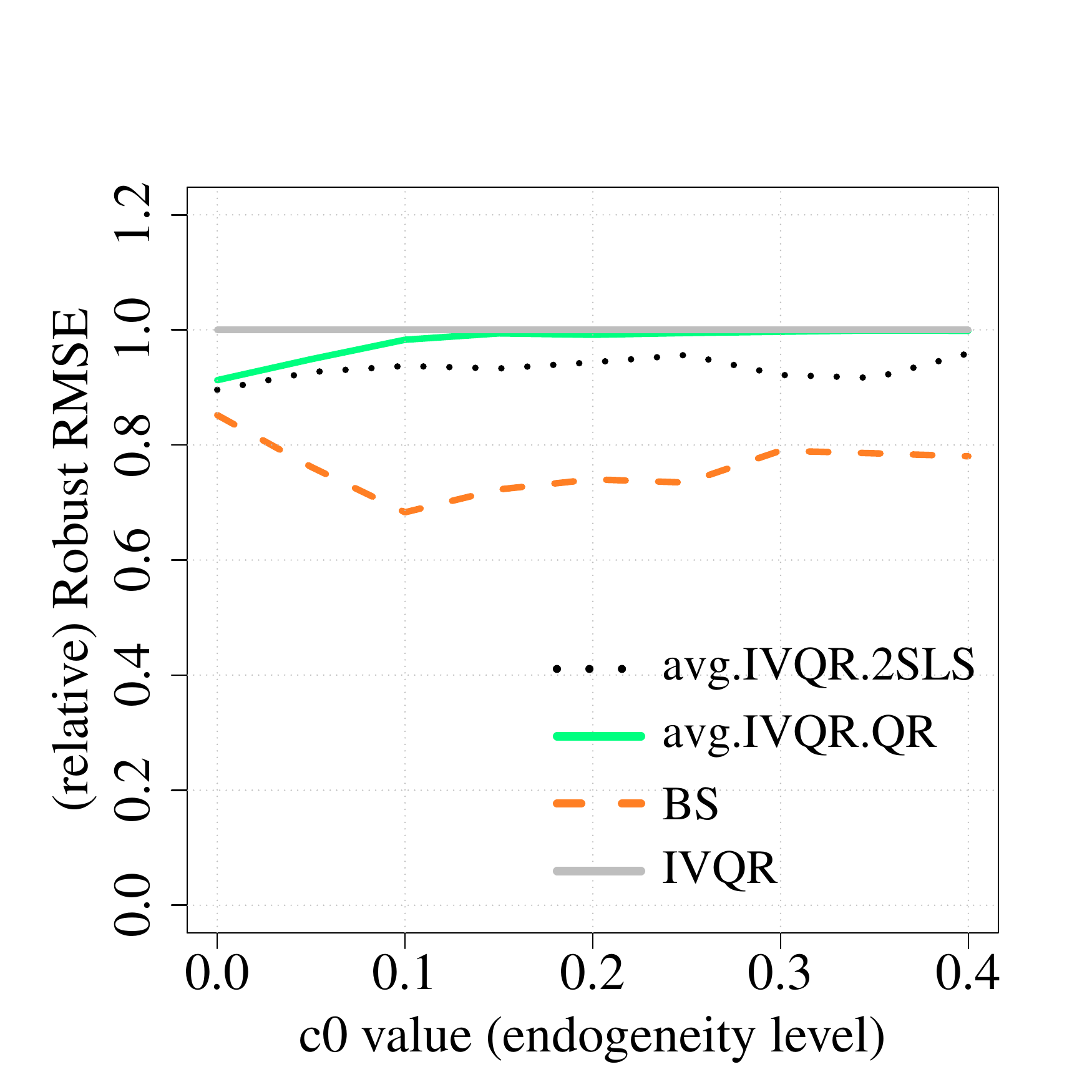}
\hfill\null
\includegraphics[width=0.45\textwidth, height=0.3\textheight, trim=35 20 20 70 ]{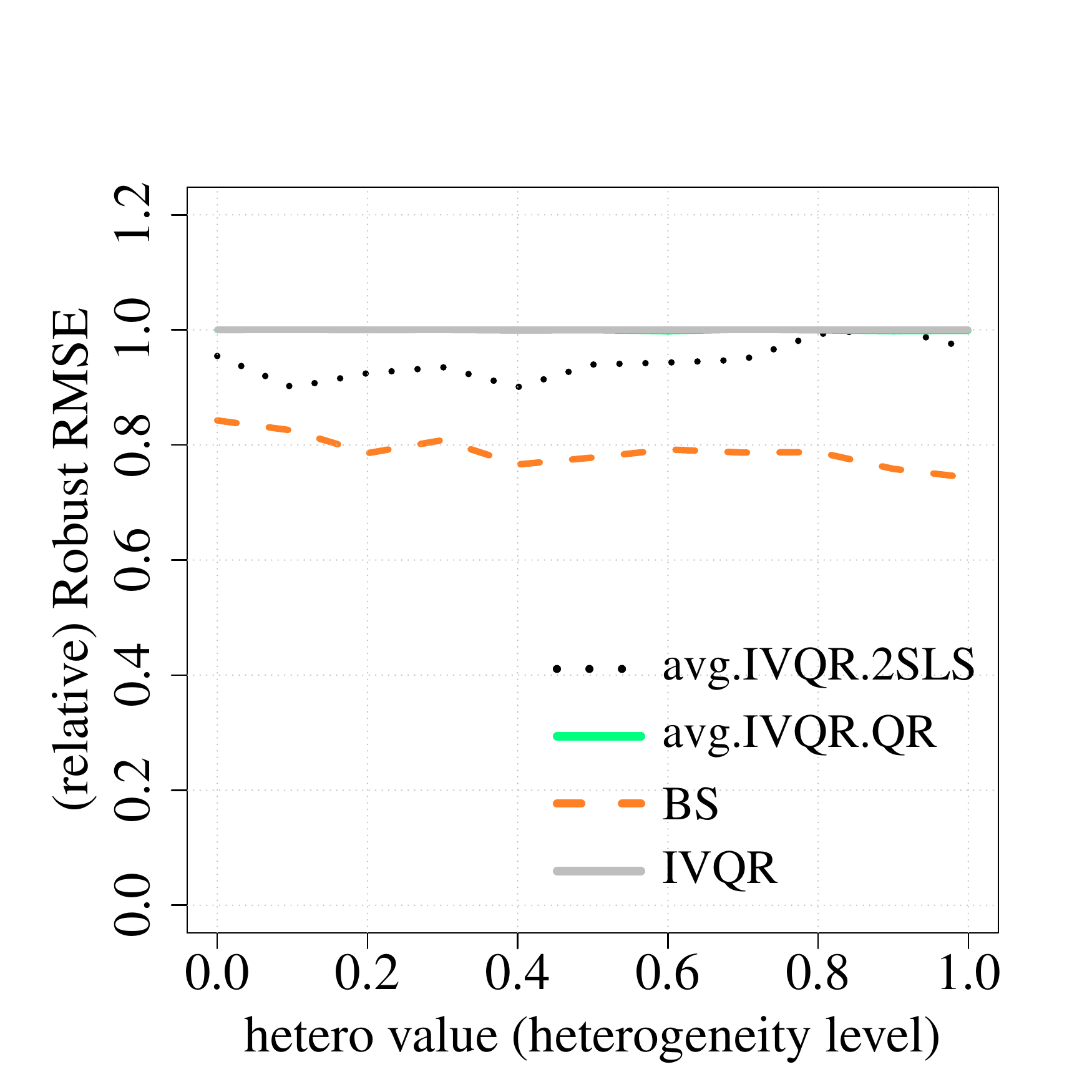}
\hfill
\includegraphics[width=0.45\textwidth, height=0.3\textheight, trim=35 20 20 70]{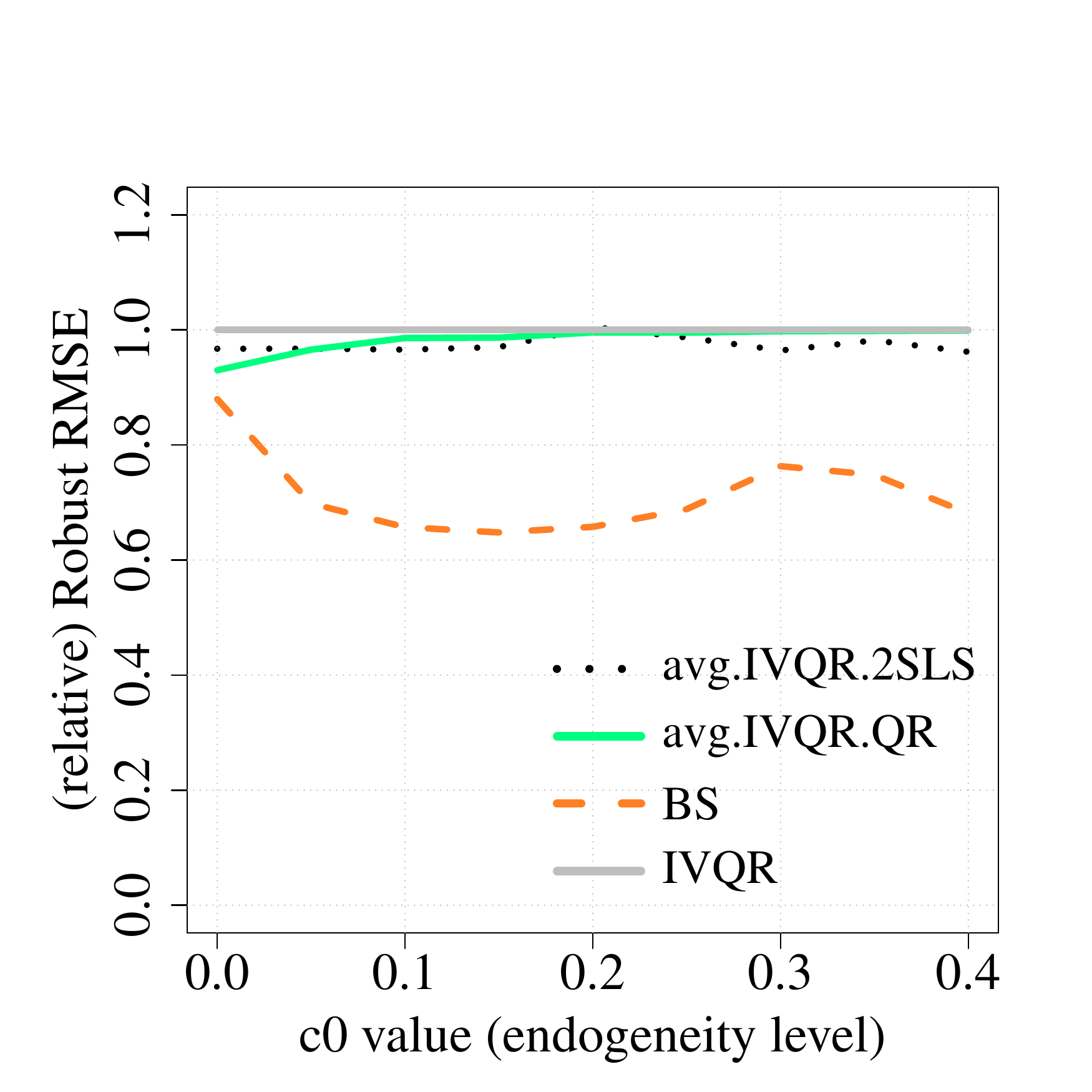}
\caption{\label{fig:M3:S1S22:tau0.7}%
Relative rRMSE in simulation model 3 at $\tau=0.7$ quantile level in 6 cases: fixed endogeneity level $c_0=0$  (left up), $c_0=0.2$ (left middle), $c_0=0.4$ (left bottom) and varying heterogeneity; and fixed heterogeneity level $hetero=0$ (right up), $hetero=0.5$ (right middle), $hetero=1$ (right bottom) and varying endogeneity,
based on $\num{200}$ replications and $\num{50}$ bootstraps. Sample size $n=1000$.} 
\end{figure}

The ``magic quantile'' phenomenon for 2SLS and IVQR-2SLS averaging again applies in simulation model 3, 
as seen when looking at the results across quantiles. 
In simulation model 3, 
the slope is $\theta(u_i)=hetero \times [F(u_i)]^4$ where $u \sim N(0,1)$
and $F(\cdot)$ is its CDF.
The 2SLS population slope is $\E[hetero \times [F(u_i)]^4 ]$.
Since $F(u_i)\sim\UnifDist(0,1)$, this equals $hetero$ times the fourth moment of a standard uniform distribution, which is 0.2.
The IVQR slope is $hetero \times \tau^4$. 
Since $0.68^4=0.20$, the 2SLS and IVQR slopes are equal when $\tau=0.68$.
For $\tau$ near $0.68$, even with much slope heterogeneity across units, the 2SLS and IVQR slopes are similar.
As $\tau$ gets farther from $0.68$, the IVQR slope increasingly differs from the 2SLS slope, with the rate of increase depending on $hetero$.

The simulation results show that at $\tau=0.6,0.7,0.8$, the 2SLS estimator has relative rRMSE less than 1 uniformly in all 60 DGPs, 
even in the DGPs with high $hetero$ values.
Correspondingly, the bootstrap averaging estimator is much lower than 1 at these three quantiles.
\Cref{fig:M3:S1S22:tau0.7} shows the performance of the three averaging estimators at $\tau=0.7$.
The bootstrap averaging estimator's relative rRMSE is significantly below 1,
and much below that of the IVQR-2SLS averaging estimator.
This is probably because the bootstrap method averages the 2SLS estimator directly, rather than through aggressive GMM that may not weight the 2SLS slope moments as heavily as is optimal.

Compared with the ``magic tau'' results in simulation model 1, 
in which the 2SLS and IVQR-2SLS averaging estimators have relative rRMSE much lower than 1 at quantile $\tau=0.7$ only,
here in simulation model 3, the good performance of the 2SLS and IVQR-2SLS averaging estimators
shows up at quantiles $\tau=0.6$, $\tau=0.7$, and $\tau=0.8$.
The ``magic tau'' is the same (0.68) in these two simulated models.
The difference in the performance at $\tau=0.6$ and $\tau=0.8$ 
is because the ``magic tau'' is only related with the (median) bias part of rRMSE.
That is, 
at $\tau=0.7$, 
the 2SLS and IVQR slopes are still very close (almost no bias of 2SLS estimates).
At $\tau=0.6$ and $\tau=0.8$, 
the 2SLS and IVQR slopes are somewhat close (little bias of 2SLS estimates).
In one simulated model,
the bias at $\tau=0.6$ or $\tau=0.8$ is still small compared to the (large) variance, 
whereas in another simulated model with smaller variance, 
the bias may start to dominate variance even at $\tau=0.6$ or $\tau=0.8$.

\section{Conclusion}
\label{sec:conclusion}

This paper contributes two averaging estimation methods to improve finite-sample efficiency in IVQR estimation.
First, I implement the averaging GMM of \citet{ChengLiaoShi2019}, proposing two types of additional moments based respectively on conventional QR and 2SLS, and considering other important implementation details such as bandwidths, extending a result of \citet{Kato2012} along the way.
Second, 
I propose a new method that uses the bootstrap to estimate optimal weights for averaging IVQR, QR, and 2SLS estimators.

This paper provides simulation evidence that these three averaging estimators outperform the IVQR estimator across all kinds of DGPs in large models with multiple endogenous regressors, as well as across quantile levels from $\tau=0.2$ to $\tau=0.8$.
Bootstrap averaging offers especially substantial efficiency gains.


Future work could involve developing theory for the practically successful bootstrap method, 
or investigating averaging across quantiles, 
or adding non-trivial smoothing into this averaging framework.
Additionally, an IVQR application of \citet{ArmstrongKolesar2019} may better use additional-but-possibly-misspecified moments when the object of interest is a scalar instead of the full parameter vector.

\singlespacing
\bibliographystyle{chicago} 

\appendix

\onehalfspacing

\section{Bandwidth for estimating {Jacobian} matrix}
\label{sec:app-bandwidth}

\Citet{Kato2012} provides an AMSE-optimal bandwidth for estimating the Jacobian matrix in conventional quantile regression.
He shows that in the linear CQF model
\begin{equation*}
    Y=\vecf{X}' \vecf{\beta}_0 + U \ \text{with } \  Q_{\tau}(U \mid \vecf{X})=0, 
\end{equation*}
the AMSE-optimal bandwidth for estimating $\E[f_0(0\mid\vecf{X})\vecf{X}\vecf{X}']$ is 
\begin{equation}\label{general:h}
    h_{\mathrm{opt}}=n^{-1/5} \left( \frac{ 4.5 \sum_{j,k=1}^{d_X} \E \left[f_0 ( 0 \mid \vecf{X}) X_j^2 X_k^2\right] }{\sum_{j,k=1}^{d_X} \left( \E \left[ f_0^{(2)} (0\mid \vecf{X}) X_j X_k \right]  \right)^2 } \right)^{1/5},
\end{equation}
where $f_0(\cdot \mid \vecf{X})$ is the conditional PDF of $U$, conditional on the regressors $\vecf{X}$; and $f_0^{(2)}(\cdot \mid \vecf{X})$ is the conditional PDF's second derivative.
If additionally $U \independent \vecf{X}$, 
then $f_0(u\mid \vecf{X})=f_0(u)$, and formula \cref{general:h} reduces to
\begin{equation}\label{eqn:h-unconditional}
h_{\mathrm{opt}}=n^{-1/5} \left( \frac{4.5\sum_{j,k=1}^{d_X} \E[f_0(0) X_j^2X_k^2] }{\sum_{j,k=1}^{d_X} \left( \E[ f_0^{(2)} (0) X_j X_k ]  \right)^2 } \right)^{1/5}.
\end{equation}
where $f_0(\cdot)$ is the unconditional PDF of $U$ 
and $f_0^{(2)}(\cdot)$ is its second derivative.

\Citet{Kato2012} further simplifies \cref{eqn:h-unconditional} when $U$ is standard normal.
That is, in the linear CQF model
\begin{equation*}
Y = \vecf{X}' \vecf{\theta}_0 + U 
\ \text{with} \  
U \mid \vecf{X} \sim \NormDist(0,1) ,
\end{equation*}
the sample analog of the optimal bandwidth \cref{general:h} reduces to 
\begin{equation}\label{bandwidth:Kato:Gaussian}
\hat{h}_{\mathrm{opt}}=n^{-1/5} \left( \frac{4.5\sum_{j,k=1}^{d_X} \left( n^{-1} \sum_{i=1}^{n} X_{ij}^2 X_{ik}^2 \right) }{\alpha(\tau)\sum_{j,k=1}^{d_X} \left( n^{-1} \sum_{i=1}^{n} X_{ij} X_{ik}  \right)^2 } \right)^{1/5},
\end{equation}
where 
\begin{equation}\label{bandwidth:Kato:Gaussian:alpha}
\alpha(\tau) \equiv \left[ 1 - \left( \Phi^{-1} \left( \tau \right) \right)^2     \right]^2 \phi \left(\Phi^{-1} \left( \tau \right) \right) ,
\end{equation}
and $\Phi(\cdot)$ and $\phi(\cdot)$ are standard normal CDF and PDF.%
\footnote{The version in \citet{Kato2012}, $\alpha(\tau)=\left[ 1-(\Phi^{-1}(\tau))     \right]^2 \phi(\Phi^{-1}(\tau) ) $, appears to have a typo; see the below proof of \cref{prop:bandwidth}.}

I extend the optimal bandwidth formula in \cref{bandwidth:Kato:Gaussian} 
to allow for general scale parameter. 

\begin{proposition}\label{prop:bandwidth}
In a linear CQF model
\begin{equation*}
Y=\vecf{X}' \vecf{\beta}_0 + U 
\ \text{with} \ 
U \mid \vecf{X} \sim \NormDist( \mu, \sigma^2 ), 
\ \text{and}\ 
\Q_{\tau}(U \mid \vecf{X})=0,
\end{equation*}
the sample analog of the AMSE-optimal bandwidth for estimating the population Jacobian matrix is
\begin{equation}\label{bandwidth:Gaussian}
\hat{h}_{\mathrm{opt}}
= n^{-1/5} \left( \frac{4.5\sum_{j,k=1}^{d_X} \left( n^{-1} \sum_{i=1}^{n} X_{ij}^2 X_{ik}^2 \right) }{\hat\alpha(\tau)\sum_{j,k=1}^{d_X} \left( n^{-1} \sum_{i=1}^{n} X_{ij} X_{ik}  \right)^2 } \right)^{1/5},
\end{equation}
where 
\begin{equation}\label{bandwidth:Gaussian:alpha}
\hat\alpha(\tau) \equiv
\frac{1}{\hat\sigma^5} \left[ 1-\left(\Phi^{-1} \left( \tau \right) \right)^2     \right]^2 \phi \left(\Phi^{-1} \left( \tau \right) \right) 
\end{equation}
and $\Phi(\cdot)$ and $\phi(\cdot)$ are standard normal CDF and PDF.
\end{proposition}
\begin{proof}
We can consider the $\NormDist(\mu, \sigma^2)$ distribution of $U$ as the $\NormDist(0, \sigma^2)$ distribution with a shift of $\mu$.
We have 
\begin{equation*}
f_U(0)=\phi_{N(0, \sigma^2)}\left(\Phi_{N(0, \sigma^2)}^{-1} \left( \tau \right) \right),
\end{equation*}
where $\phi_{N(0, \sigma^2)}$ and $\Phi_{N(0, \sigma^2)}$ are the $\NormDist(0, \sigma^2)$ PDF and CDF.

For a standard normal distribution, 
\begin{equation*}
    \phi(x)= \frac{1}{\sqrt{2\pi}} e^{-x^2/2} ,
\end{equation*}
which has the property that
\begin{equation*}
    \phi'(x)= \frac{1}{\sqrt{2\pi}} e^{-x^2/2} \left(- \frac{2x}{2} \right)=-x\phi(x) .
\end{equation*}

For a $\NormDist(0, \sigma^2)$ density function, 
\begin{equation*}
    \phi_{N(0, \sigma^2)}(x)= \frac{1}{\sqrt{2\pi \sigma^2 }} e^{- \frac{x^2}{2\sigma^2}  },
\end{equation*}
its first derivative is 
\begin{equation*}
    \phi_{N(0, \sigma^2)}'(x)= \frac{1}{\sqrt{2\pi \sigma^2}} e^{- \frac{x^2}{2\sigma^2}} \left(- \frac{2x}{2\sigma^2 } \right)
    = \left(-\frac{x}{\sigma^2} \right) \phi_{N(0, \sigma^2)}(x).
\end{equation*}
By the chain rule, its second derivative is 
\begin{align*}
    \phi_{N(0, \sigma^2)}^{(2)}(x) & = 
    \left(-\frac{1}{\sigma^2}\right) \phi_{N(0, \sigma^2)}(x)+
    \left(-\frac{x}{\sigma^2}\right) \phi_{N(0, \sigma^2)}'(x)\\
   &= \left(-\frac{1}{\sigma^2}\right) \phi_{N(0, \sigma^2)}(x)+
    \left(-\frac{x}{\sigma^2}\right) 
    \left(-\frac{x}{\sigma^2} \right) \phi_{N(0, \sigma^2)}(x)\\
   &= \frac{1}{\sigma^2} 
    \left(\frac{x^2}{\sigma^2}-1 \right) 
     \phi_{N(0, \sigma^2)}(x)    
.
\end{align*}
Therefore,
\begin{equation*}
    f_U^{(2)}(0) =
    \frac{1}{\sigma^2} \left( \frac{ \left( \Phi_{N(0, \sigma^2)}^{-1} \left( \tau \right) \right)^2 }{\sigma^2}  -1 \right) 
    \phi_{N(0, \sigma^2)}\left(\Phi_{N(0, \sigma^2)}^{-1} \left( \tau \right) \right).
\end{equation*}
Since 
\begin{equation*}
    \Phi_{N(0, \sigma^2)}^{-1}(\tau)= \sigma \Phi_{N(0, 1)}^{-1} \left( \tau \right)
\textrm{ and }
    \phi_{N(0, \sigma^2)}(x)= \frac{1}{\sigma} \phi_{N(0, 1)} \left(x/\sigma\right),
\end{equation*}
$ f_U(0)$ and $ f_U^{(2)}(0)$ further reduce to 
\begin{align*}
     f_U(0) & = \frac{1}{\sigma} \phi_{N(0, 1)}(\Phi_{N(0, \sigma^2)}^{-1}(\tau)/\sigma)  \\
     & =  \frac{1}{\sigma} \phi_{N(0, 1)}(\sigma \Phi_{N(0, 1)}^{-1}(\tau)/\sigma)\\
     & =  \frac{1}{\sigma} \phi_{N(0, 1)}( \Phi_{N(0, 1)}^{-1}(\tau)),
\end{align*}
and 
\begin{align*}
     f_U^{(2)}(0) 
     & = \frac{1}{\sigma^2} \left( \frac{ \left(\Phi_{N(0, \sigma^2)}^{-1}(\tau)\right)^2 }{\sigma^2}  -1 \right) 
     \frac{1}{\sigma} \phi_{N(0, 1)} 
     \left(\Phi_{N(0, \sigma^2)}^{-1}(\tau)/\sigma \right)  \\
     & =  \frac{1}{\sigma^2} 
     \left( \frac{ \left(\sigma \Phi_{N(0, 1)}^{-1}(\tau)\right)^2 }{\sigma^2}  -1 \right) \frac{1}{\sigma} \phi_{N(0, 1)}
     \left( \Phi_{N(0, 1)}^{-1}(\tau) \right)\\
     & =  \frac{1}{\sigma^3} 
     \left( \left(\Phi_{N(0, 1)}^{-1}(\tau)\right)^2  -1 \right) \phi_{N(0, 1)}\left( \Phi_{N(0, 1)}^{-1}(\tau)\right).
\end{align*}
Then the optimal bandwidth in \cref{general:h} reduces to 
\begin{align*}
h_{\mathrm{opt}}
&= n^{-1/5} \left( \frac{4.5 \sum_{j,k=1}^{d} \E \left[f_0 \left(0 \mid X \right) X_j^2X_k^2 \right] }{\sum_{j,k=1}^{d} \left( \E\left[ f_0^{(2)} \left(0\mid X\right) X_j X_k\right]  \right)^2 } \right)^{1/5}  \\
&=  n^{-1/5} \left( \frac{4.5\sum_{j,k=1}^{d} \E[f_0(0) X_{j}^2X_{k}^2] }{\sum_{j,k=1}^{d} \left( \E[f_0^{(2)} (0) X_{j} X_{k}]   \right)^2 } \right)^{1/5}
\\&= n^{-1/5} \left( \frac{4.5 f_0(0) \sum_{j,k=1}^{d} \E[X_{j}^2 X_{k}^2 }{\left(f_0^{(2)} (0) \right)^2\sum_{j,k=1}^{d} \left( \E[X_{j} X_{k}]  \right)^2 } \right)^{1/5}
\\&= n^{-1/5} \left( \frac{4.5 \sum_{j,k=1}^{d} \E[X_{j}^2 X_{k}^2] }{\alpha(\tau)\sum_{j,k=1}^{d} \left( \E[X_{j} X_{k}]  \right)^2 } \right)^{1/5},    
\end{align*}
where
\begin{align*}
\alpha(\tau)
&= \frac{\left[f_0^{(2)} (0) \right]^2}{f_0(0)}
= \frac{\left\{ \dfrac{1}{\sigma^3} 
     \left[ (\Phi_{N(0, 1)}^{-1}(\tau))^2  -1 \right] \phi_{N(0, 1)}\left( \Phi_{N(0, 1)}^{-1}(\tau)\right) \right\}^2 }{\dfrac{1}{\sigma} \phi_{N(0, 1)}\left( \Phi_{N(0, 1)}^{-1}(\tau)\right)}\\
     & = \frac{1}{\sigma^5} 
     \left[ (\Phi_{N(0, 1)}^{-1}(\tau))^2  -1 \right]^2 \phi_{N(0, 1)}\left( \Phi_{N(0, 1)}^{-1}(\tau)\right).
\end{align*}

In the special case $\sigma^2=1$, it reduces to $\alpha(\tau)=\left[ 1-\left(\Phi^{-1}(\tau)\right)^2     \right]^2 \phi \left(\Phi^{-1}(\tau) \right)$.
The sample analog replaces $\sigma$ with $\hat\sigma$ and replaces the expectations with sample averages.
\end{proof}

The optimal bandwidth formula \cref{bandwidth:Gaussian} is very helpful in practice, 
especially with multiple regressors since nonparametric conditional density derivative estimators converge slowly.
In these cases, using the optimal bandwidth formula \cref{general:h} 
might suffer large estimation error.

\citeposs{Kato2012} optimal bandwidth formula is under the assumption that regressors are exogenous. 
That is, $X$ represents both regressor vector and the full instrument vector.
I modify Kato's optimal bandwidth formula by replacing $X_{j}X_{k}$ term with $X_{j}Z_{k}$ for $j=1, \ldots, d_X$ and $k = 1, \ldots, d_Z$ 
to adapt to IVQR.
I also replace the conditional density functions that are conditioning on regressors to that conditioning on both the regressors and the instruments, since the IVQR object of interest is $\E[f_0(0\mid\vecf{X},\vecf{Z})\vecf{X}\vecf{Z}']$; e.g., see (16) in \citet{KaplanSun2017}.

Following this modification for IVQR, 
the general AMSE-optimal bandwidth formula \cref{general:h} for estimating the Jacobian matrix becomes
\begin{equation}\label{general:h:XZ}
    h_{\mathrm{opt}}=n^{-1/5} \left( \frac{4.5 \sum_{j=1}^{d_X} \sum_{k=1}^{d_Z} \E \left[ f_0 ( 0 \mid \vecf{X}, \vecf{Z}) X_j^2 Z_k^2 \right] }{\sum_{j=1}^{d_X} \sum_{k=1}^{d_Z} \left( \E \left[ f_0^{(2)} (0 \mid \vecf{X}, \vecf{Z} ) X_j Z_k \right]  \right)^2 } \right)^{1/5},
\end{equation}
and the AMSE-optimal plug-in bandwidth formula \cref{bandwidth:Gaussian} in the general Gaussian model becomes 
\begin{equation}\label{bandwidth:Gaussian:XZ}
\hat{h}_{\mathrm{opt}}=
n^{-1/5} \left( \frac{4.5 \sum_{j=1}^{d_X} \sum_{k=1}^{d_Z} \left( n^{-1} \sum_{i=1}^{n} X_{ij}^2 Z_{ik}^2 \right) }{\hat\alpha(\tau)\sum_{j=1}^{d_X} \sum_{k=1}^{d_Z} \left( n^{-1} \sum_{i=1}^{n} X_{ij} Z_{ik}  \right)^2 } \right)^{1/5} ,
\end{equation}
with the same $\hat\alpha(\tau)$ as in \cref{bandwidth:Gaussian:alpha}.

In \cref{sec:sim}, for simulation model 1,
I use optimal bandwidth formula \cref{general:h:XZ} to estimate the Jacobian matrix.
Model 1 is a small model with one binary treatment variable and one binary instrument, so the conditional PDF and its second derivative are relatively easy to estimate nonparametrically.
The optimal bandwidth formula \cref{general:h:XZ} could probably be improved, 
but it seems to work well in the simulations; any further improvement to the bandwidth would only improve the performance of the averaging GMM estimators.

For simulation models 2 and 3, I instead use the optimal bandwidth formula \cref{bandwidth:Gaussian:XZ}. 
These two simulation models are large models with multiple continuous endogenous regressors and instruments.
Fully nonparametric estimation of the conditional density and its second derivative is possible but would have large estimation error given my sample sizes.


\section{Simulation DGP details}

\subsection{Simulation model 1 details}
\label{app:DGP-1}

Simulation model 1 is a DGP based on the Job Training Partnership Act (JTPA) empirical application of \citet{KaplanSun2017}, which in turn is based on \citet{AbadieEtAl2002}.

Consider a structural random coefficient model that describes the impact of a job training program on an individual's earnings,\footnote{This JTPA-based simulation DGP generalizes DGP 1 in \citet{deCastroGalvaoKaplanLiu2019} 
to a class of DGPs that account for 
different combinations of endogeneity, heterogeneity, and fat-tail levels.}
\begin{equation}
Y_i = \beta(U_i)+\gamma(U_i) \times D_i ,
\end{equation}
where $Y_i$ is the outcome variable that denotes individual $i$'s earning.
Here, $U_i \sim \UnifDist(0,1)$ captures individual $i$'s unobserved innate ability.
The intercept $\beta(\cdot)$ and slope $\gamma(\cdot)$ are functions of the individual's innate ability.
That is, 
$\beta(U_i)$ and $\gamma(U_i)$ denote the individual-specific intercept and slope, respectively.

The job training program is randomly offered to individuals, 
denoted as a binary variable $Z_i$.
When $Z_i=1$, it means individual $i$ is offered the job training, and $Z_i=0$ means individual $i$ is not offered the training.
The offer probability is $\Pr(Z_i=1)=1/2$.

When the individual is offered the job training program, he can choose to take this training program or not.
If individual $i$ actually takes the training, then $D_i=1$, otherwise $D_i=0$.
There is self-selection based on the unobserved innate ability.
I model the take-up decision as
\begin{equation}\label{appd:SimModel1:ProbFn}
\Pr(D_i=1 \mid Z_i=1, U_i)=0.5+c_1(U_i-0.5) ,
\quad
\Pr(D_i=1 \mid Z_i=0,U_i)=0 .
\end{equation}
It describes the conditional probability of the individual taking the job training program, 
conditional on the fact that the individual is offered, 
and given his innate ability.

The constant $c_1 \in [0, 1] $ captures the endogeneity level of the observed choice variable $D_i$.
When $c_1=0$, 
the conditional probability \cref{appd:SimModel1:ProbFn} becomes a constant $0.5$. 
It means the individual's innate ability is irrelevant to the decision to take the job training offer:
\begin{equation}
\Pr(D_i=1 \mid Z_i=1, U_i)
= \Pr(D_i=1 \mid Z_i=1)
= 0.5.
\end{equation}
It implies 
\begin{equation}
  \Pr(D_i=1 \mid U_i) 
= \Pr(D_i=1) .
\end{equation}
The treatment variable $D_i$ is uncorrelated with the individual's unobserved innate ability. 
There is no endogeneity in the model. 
When $c_1=1$, the conditional probability function becomes 
\begin{equation}
\Pr(D_i=1 \mid Z_i=1, U_i) = U_i ,
\end{equation}
which depends more strongly on innate ability than for any other $c_1<1$.

Each $c_1$ value corresponds to a DGP with a distinct endogeneity level. 
As $c_1$ increases, 
innate ability plays an increasing role in the self-selection process, i.e., increasing endogeneity of $D_i$.
I consider $c_1=0$ and $c_1=1$ as the two endpoints, 
denoting the DGPs with no endogeneity 
and most endogeneity, respectively.

One of the main features of quantile regression is to capture slope heterogeneity. 
I use the individual-specific random coefficient model to represent the slope heterogeneity.  
It assumes 
the treatment effect of training program on each individual is heterogeneous. 
The slope term, 
which measures the treatment effect of the training program, 
is a function of the individual's specific unobserved innate ability. 
Let the slope function be
\begin{equation}
\gamma(U_i)=100c_2 U_i^4.
\end{equation}
The nonnegative constant $c_2$ indicates how much treatment effect heterogeneity there is  across individuals with different ability levels.
Each $c_2$ value corresponds to a different DGP with a specific heterogeneity level.

When $c_2=0$, the slope term becomes a constant (zero) for all individuals.
That means the treatment effect is the same for both high and low ability individuals.
There is no heterogeneity.

The larger $c_2$ is,
the larger is the difference of the treatment effect between lower and higher ability individuals, i.e., larger heterogeneity.
In the simulation model,
$c_2$ increases from 0 to 1 to indicate a class of DGPs with different (increasing) heterogeneity levels.

Two types of (random) intercept are considered.
The first is $\beta(U_i)=60+Q(U_i)$,
where $Q(\cdot)$ is the $\chi^2_3$ quantile function.
It represents that without training, 
higher-ability people have higher earnings, 
and lower-ability people have lower earnings.
The second 
replaces $Q(\cdot)$ by the quantile function of the $t$-distribution with $c_3\ge1$ degrees of freedom.
%
In this second case, the slope term is set to be a constant zero, i.e., no heterogeneous training effects.
When $c_3=1$,
the distribution of earning is very fat-tailed (Cauchy)
with respect to individuals' innate ability.
As $c_3$ increases, 
the distribution of earning becomes less fat-tailed,
approaching a normal distribution. 
Each $c_3$ value represents a specific DGP with a different fat-tail level of the earnings distribution with respect to individuals' ability levels.

Overall, 
this simulation model considers $11$ possible values of each simulation parameter: 
$c_1 \in \{ 0, 0.1,  \ldots, 1 \} $, $c_2 \in \{ 0, 0.1, \ldots, 1 \} $, and 
$c_3 \in \{ 1, 1.5,  \ldots, 4.5, 5, 6, 7 \} $.
In the first case of the $\chi^2_3$ intercept term, 
there is no $c_3$ involved. 
In total, it has $11\times 11=121$ DGPs of the $(c_1, c_2, c_3=\mathrm{NA}) $ combinations.  
In the second case of the $t$-distributed intercept term, 
different $c_3$ values are involved, but $c_2$ is set to zero.
It also has $11\times 11=121$ DGPs of the $(c_1, c_2=0, c_3) $ combinations. 
Therefore, this simulation model has $121+121=242$ DGPs altogether.

I first run simulations with all 242 DGPs. 
Based on the simulation results, 
I select 14 representative DGPs with distinct combinations of endogeneity, heterogeneity, and fat-tail levels that span the ranges of each.
I label the endogeneity level as ``No,'' ``Some,'' or ``Much,'' based on the simulation results with respect to $c_1$ values. 
Similarly, I label heterogeneity level in four catogories: ``No,'' ``Little,'' ``Some,'' and ``Much,'' with respect to $c_2$; 
and fat-tail level in four categories: ``NA,'' ``Little,'' ``Some,'' and ``Much,'' based on $c_3$.%
\footnote{Specifically, based on the simulation results, I label $c_1=0$ as ``no endogeneity''; $c_1 \in [0.1, 0.3]$ as ``some endogeneity''; and $c_1 \ge 0.5$ as ``very much endogeneity'' in a DGP.
I label $c_2=0$ as ``no heterogeneity''; $c_2=0.1$ as ``little heterogeneity''; $c_2 \in [0.3, 0.6]$ as ``some heterogeneity''; $c_2=0.9$ as ``very much heterogeneity.''
I label $c_3=1$ as ``very much fat tail''; $c_3=3$ as ``some fat tail''; and the case $c_3 \ge 5$ as ``little fat tail.''}

\subsection{Uniform dominance condition}

The fact that the averaging estimators fail to show uniform dominance over the IVQR estimator in simulation model 1 is essentially because the uniform dominance condition does not hold for model 1.
A few studies give more detailed discussion about the uniform dominance condition. 
\citetalias{ChengLiaoShi2019} Theorem 5.2 gives a sufficient condition for the uniform dominance results in the averaging GMM framework.
In some cases, this condition reduces to a requirement of at least four parameters to be considered, whereas simulation model 1 considers only one parameter.
\citetalias{ChengLiaoShi2019} footnote 9 gives a review about the uniform dominance condition in the Gaussian location model in other studies. 
In most of the studies, the uniform dominance condition requires more than four parameters to be estimated. 
\Citet{Hansen2017} 
shows the averaging estimator between the 2SLS estimator and the OLS estimator will uniformly dominate the 2SLS estimator under the condition that the number of endogenous regressors is at least 3.

A common understanding about the uniform dominance condition in averaging estimation is that 
the averaging method has uniform dominance in large models with at least 4 endogenous regressors.
In contrast, in a small model with only one endogenous regressor, 
the averaging estimator does not have uniform dominance over the conservative estimator.
Both \citetalias{ChengLiaoShi2019} simulated model 3 and this paper's simulated model 1 provide simulation findings that in the model with one endogenous regressor, the averaging estimator does not show any uniform dominance over the conservative estimator. 
The models in \cref{sim:model2,sim:model3} are larger models that have 6 endogenous regressors.
The simulation findings suggest the desirable uniform dominance of the averaging estimator over the IVQR estimator.

In practice, it is quite often that the (endogenous) treatment variables we are interested in have a few interaction terms with other variables. 
In such cases, 
multiple endogenous regressors are involved in the model, 
and the uniform dominance condition is easily satisfied.
Then, using the averaging methods proposed in this paper can improve estimation efficiency over the IVQR estimator.

\section{Additional simulation results}
\label{app:sim}

\subsection{Simulation model 1 results at different quantile levels}
\label{appdx:M1}

\begin{table}[htbp!]
    \centering\caption{\label{tab:S1:IVQR-2SLS:comparetau} Finite-sample relative rRMSE of IVQR-2SLS averaging, aggressive estimators, and IV estimator in JTPA-based simulation model 1.}
            \sisetup{round-precision=2,round-mode=places}
    \begin{threeparttable}
    \begin{tabular}{ccccScccScccccScccc}
    \toprule
    &                     \multicolumn{4}{c}{$\tau=0.2$} &&               \multicolumn{4}{c}{$\tau=0.3$} &&          \multicolumn{4}{c}{$\tau=0.4$} \\
    \cmidrule{2-5} \cmidrule{7-10}  \cmidrule{12-15}  
    &                     \multicolumn{2}{c}{$\mathrm{IVQR.2SLS}$}&               && & \multicolumn{2}{c}{$\mathrm{IVQR.2SLS}$} &               &&  &
    \multicolumn{2}{c}{$\mathrm{IVQR.2SLS}$} &               
      \\
    \cmidrule{2-3} \cmidrule{7-8} \cmidrule{12-13}  
    $\mathrm{DGP }$   
    & $\mathrm{AVG}$ & $\mathrm{AGG}$ & $\mathrm{2SLS}$ & $\mathrm{CON}$ 
    && $\mathrm{AVG}$ & $\mathrm{AGG}$ & $\mathrm{2SLS}$ & $\mathrm{CON}$ 
    && $\mathrm{AVG}$ & $\mathrm{AGG}$ & $\mathrm{2SLS}$ & $\mathrm{CON}$ 
 \\
    \midrule
1            &   $\num[math-rm=\mathbf]{0.941798}$  & $\num{0.871945}$ & $\num{0.597444}$   & $\num{0.200699}$ 
&& $\num{1.004959}$  & $\num{1.016194}$ & $\num{1.332099}$   & $\num{0.222795}$ 
&& $\num[math-rm=\mathbf]{0.968725}$  & $\num{0.967158}$ & $\num{1.132076}$   & $\num{0.262160}$
\\
2            &   $\num[math-rm=\mathbf]{0.909462}$  & $\num{0.873552}$ & $\num{0.801688}$   & $\num{0.238507}$  
&& $\num[math-rm=\mathbf]{0.991528}$  & $\num{0.989130}$ & $\num{1.382659}$   & $\num{0.242750}$ 
&&  $\num{1.007672}$  & $\num{1.030468}$ & $\num{1.219754}$   & $\num{0.275171}$ 
\\
3            &   $\num[math-rm=\mathbf]{0.993463}$  & $\num{1.260438}$ & $\num{1.187777}$   & $\num{0.268979}$
&& $\num{1.006754}$  & $\num{0.984176}$ & $\num{1.124869}$   & $\num{0.289591}$
&& $\num[math-rm=\mathbf]{0.983354}$  & $\num{0.953060}$ & $\num{0.959583}$   & $\num{0.339472}$
\\
4            &   $\num[math-rm=\mathbf]{0.981629}$  & $\num{1.679892}$ & $\num{1.975982}$   & $\num{0.286971}$
&& $\num[math-rm=\mathbf]{0.983002}$  & $\num{1.003114}$ & $\num{1.105015}$   & $\num{0.280525}$
&&  $\num[math-rm=\mathbf]{0.984543}$  & $\num{0.932994}$ & $\num{0.959015}$   & $\num{0.323231}$
\\
5            &   $\num[math-rm=\mathbf]{0.927540}$  & $\num{0.828456}$ & $\num{0.608481}$   & $\num{0.262225}$ 
&& $\num{1.007973}$  & $\num{4.365245}$ & $\num{18.063972}$   & $\num{0.435454}$
&&  $\num{1.038360}$  & $\num{8.001935}$ & $\num{8.480688}$   & $\num{0.845524}$ 
\\
6            &   $\num[math-rm=\mathbf]{0.946775}$  & $\num{0.859464}$ & $\num{0.709994}$   & $\num{0.292543}$ 
&& $\num{1.003724}$  & $\num{7.188967}$ & $\num{17.311542}$   & $\num{0.471132}$ 
&&  $\num{1.013825}$  & $\num{9.459444}$ & $\num{9.905162}$   & $\num{0.753200}$ 
\\
7            &   $\num{1.071536}$  & $\num{1.977770}$ & $\num{2.317601}$   & $\num{0.334383}$ 
&& $\num{1.034245}$  & $\num{4.796462}$ & $\num{14.867341}$   & $\num{0.521376}$
&&  $\num{1.043314}$  & $\num{8.191025}$ & $\num{9.342972}$   & $\num{0.773788}$ 
\\
8            &   $\num{1.114988}$  & $\num{2.230924}$ & $\num{2.519431}$   & $\num{0.392785}$ 
&& $\num{1.013410}$  & $\num{5.689098}$ & $\num{19.541229}$   & $\num{0.664187}$
&&  $\num{1.015020}$  & $\num{10.928502}$ & $\num{12.754568}$   & $\num{0.949278}$ 
\\
9            &    $\num[math-rm=\mathbf]{0.875496}$  & $\num{0.818780}$ & $\num{0.730609}$   & $\num{0.336555}$ 
&& $\num{1.031585}$  & $\num{12.035078}$ & $\num{26.191465}$   & $\num{0.705640}$
&&  $\num{1.032594}$  & $\num{9.957519}$ & $\num{12.605145}$   & $\num{1.341854}$ 
\\
10            &   $\num[math-rm=\mathbf]{0.881206}$  & $\num{0.812385}$ & $\num{0.515840}$   & $\num{0.262022}$  
&& $\num[math-rm=\mathbf]{0.985131}$  & $\num{1.020977}$ & $\num{1.090029}$   & $\num{0.190439}$
&&  $\num{1.014666}$  & $\num{1.010461}$ & $\num{1.114727}$   & $\num{0.186220}$ 
\\
11            &   $\num[math-rm=\mathbf]{0.929318}$  & $\num{0.864419}$ & $\num{0.642252}$   & $\num{0.421388}$ 
&& $\num[math-rm=\mathbf]{0.996671}$  & $\num{0.996941}$ & $\num{20.011687}$   & $\num{0.276590}$ 
&&  $\num[math-rm=\mathbf]{0.992762}$  & $\num{1.011250}$ & $\num{25.908973}$   & $\num{0.213634}$
\\
12            &   $\num[math-rm=\mathbf]{0.992317}$  & $\num{1.473281}$ & $\num{1.680375}$   & $\num{0.296785}$  
&& $\num[math-rm=\mathbf]{0.906470}$  & $\num{0.846264}$ & $\num{0.688699}$   & $\num{0.228425}$ 
&&  $\num[math-rm=\mathbf]{0.937887}$  & $\num{0.888482}$ & $\num{0.865700}$   & $\num{0.181721}$ 
\\
13            &   $\num[math-rm=\mathbf]{0.981839}$  & $\num{1.514421}$ & $\num{1.645674}$   & $\num{0.331279}$ 
&& $\num[math-rm=\mathbf]{0.952824}$  & $\num{0.887062}$ & $\num{0.801376}$   & $\num{0.254676}$ 
&&  $\num[math-rm=\mathbf]{0.970872}$  & $\num{0.997993}$ & $\num{1.050454}$   & $\num{0.194289}$
\\
14            &   $\num[math-rm=\mathbf]{0.850159}$  & $\num{1.156678}$ & $\num{1.171333}$   & $\num{0.793780}$ 
&& $\num{1.006167}$  & $\num{1.004608}$ & $\num{14.901967}$   & $\num{0.406347}$ 
&&  $\num[math-rm=\mathbf]{0.990812}$  & $\num{1.007877}$ & $\num{23.903279}$   & $\num{0.253328}$ 
\\
\midrule
    &                     \multicolumn{4}{c}{$\tau=0.6$} &&               \multicolumn{4}{c}{$\tau=0.7$} &&          \multicolumn{4}{c}{$\tau=0.8$} \\
    \cmidrule{2-5} \cmidrule{7-10}  \cmidrule{12-15}  
    &                     \multicolumn{2}{c}{$\mathrm{IVQR.2SLS}$}&               && & \multicolumn{2}{c}{$\mathrm{IVQR.2SLS}$} &               &&  &
    \multicolumn{2}{c}{$\mathrm{IVQR.2SLS}$} &               
      \\
    \cmidrule{2-3} \cmidrule{7-8} \cmidrule{12-13}  
    DGP   
    & $\mathrm{AVG}$ & $\mathrm{AGG}$ & $\mathrm{2SLS}$ & $\mathrm{CON}$ 
    && $\mathrm{AVG}$ & $\mathrm{AGG}$ & $\mathrm{2SLS}$ & $\mathrm{CON}$ 
    && $\mathrm{AVG}$ & $\mathrm{AGG}$ & $\mathrm{2SLS}$ & $\mathrm{CON}$ 
 \\
    \midrule
1            &    $\num[math-rm=\mathbf]{0.887270}$  & $\num{0.768824}$ & $\num{0.791940}$   & $\num{0.374756}$ 
&& $\num[math-rm=\mathbf]{0.767970}$  & $\num{0.657689}$ & $\num{0.672506}$   & $\num{0.441311}$ 
&&  $\num[math-rm=\mathbf]{0.732586}$  & $\num{0.522260}$ & $\num{0.536476}$  & $\num{0.586038}$  
\\
2            &    $\num[math-rm=\mathbf]{0.968070}$  & $\num{0.897544}$ & $\num{0.868689}$   & $\num{0.386376}$
&& $\num[math-rm=\mathbf]{0.974841}$  & $\num{0.851511}$ & $\num{0.843664}$   & $\num{0.397837}$
&& $\num[math-rm=\mathbf]{0.794872}$  & $\num{0.580228}$ & $\num{0.503935}$   & $\num{0.531690}$ 
\\
3            &   $\num[math-rm=\mathbf]{0.913903}$  & $\num{0.843854}$ & $\num{0.848451}$   & $\num{0.383936}$
&&  $\num[math-rm=\mathbf]{0.852418}$  & $\num{0.818769}$ & $\num{0.800879}$   & $\num{0.406742}$
&&  $\num[math-rm=\mathbf]{0.787704}$  & $\num{0.601809}$ & $\num{0.549141}$   & $\num{0.590445}$ 
\\
4            &   $\num[math-rm=\mathbf]{0.992154}$  & $\num{0.983569}$ & $\num{0.947960}$   & $\num{0.327001}$  
&& $\num[math-rm=\mathbf]{0.873452}$  & $\num{0.827219}$ & $\num{0.756978}$   & $\num{0.409502}$  
&& $\num[math-rm=\mathbf]{0.800401}$  & $\num{0.697760}$ & $\num{0.695507}$   & $\num{0.473768}$
\\
5            &   $\num{1.160737}$  & $\num{1.503458}$ & $\num{1.348899}$   & $\num{2.284178}$  
&&  $\num[math-rm=\mathbf]{0.781347}$  & $\num{0.632111}$ & $\num{0.564522}$   & $\num{3.002037}$ 
&&  $\num{1.102007}$  & $\num{2.351196}$ & $\num{2.186977}$  &  $\num{3.823902}$  
\\
6            &  $\num{1.170188}$  & $\num{1.967256}$ & $\num{1.679097}$   & $\num{2.005065}$ 
&& $\num[math-rm=\mathbf]{0.694072}$  & $\num{0.518959}$ & $\num{0.516103}$   & $\num{2.849327}$
&& $\num{1.026344}$  & $\num{2.234939}$ & $\num{2.082018}$  & $\num{3.823448}$  
\\
7            &   $\num{1.116737}$  & $\num{3.066118}$ & $\num{2.547104}$   & $\num{1.626560}$  
&& $\num[math-rm=\mathbf]{0.807437}$  & $\num{0.608048}$ & $\num{0.556342}$   & $\num{1.970702}$ 
&& $\num{1.124555}$  & $\num{2.232059}$ & $\num{1.998179}$  & $\num{2.182250}$ 
\\
8            &   $\num{1.192011}$  & $\num{3.893396}$ & $\num{3.332523}$   & $\num{2.084858}$ 
&& $\num[math-rm=\mathbf]{0.811062}$  & $\num{0.707622}$ & $\num{0.620690}$   & $\num{2.818585}$ 
&& $\num{1.126199}$  & $\num{2.553890}$ & $\num{2.353221}$  & $\num{3.089752}$ 
\\
9            &  $\num{1.192056}$  & $\num{2.110374}$ & $\num{1.868557}$   & $\num{4.101009}$ 
&&  $\num[math-rm=\mathbf]{0.727050}$  & $\num{0.519994}$ & $\num{0.505922}$   & $\num{5.971537}$ 
&&  $\num{1.097499}$  & $\num{2.533257}$ & $\num{2.401743}$  & $\num{7.506477}$
\\
10            &  $\num[math-rm=\mathbf]{0.982171}$  & $\num{0.900149}$ & $\num{1.079000}$   & $\num{0.192385}$
&& $\num[math-rm=\mathbf]{0.993023}$  & $\num{0.954908}$ & $\num{1.094583}$   & $\num{0.189647}$
&& $\num[math-rm=\mathbf]{0.921341}$  & $\num{0.913100}$ & $\num{0.902282}$   & $\num{0.229356}$ 
\\
11            &  $\num{1.000030}$  & $\num{1.005919}$ & $\num{26.257707}$   & $\num{0.210796}$  
&& $\num[math-rm=\mathbf]{0.994315}$  & $\num{1.017491}$ & $\num{17.331039}$   & $\num{0.319371}$
&& $\num{1.011239}$  & $\num{1.019648}$ & $\num{11.243629}$  & $\num{0.467432}$ 
\\
12            &   $\num[math-rm=\mathbf]{0.986314}$  & $\num{0.963294}$ & $\num{1.034330}$   & $\num{0.152095}$ 
&&  $\num[math-rm=\mathbf]{0.994915}$  & $\num{0.975917}$ & $\num{1.012617}$   & $\num{0.155356}$ 
&&  $\num[math-rm=\mathbf]{0.971818}$  & $\num{0.913364}$ & $\num{1.013293}$   & $\num{0.153497}$
\\
13            &   $\num{1.005983}$  & $\num{1.018743}$ & $\num{1.284126}$   & $\num{0.158934}$   
&&  $\num[math-rm=\mathbf]{0.988944}$  & $\num{0.977869}$ & $\num{1.221540}$   & $\num{0.167077}$  
&&  $\num[math-rm=\mathbf]{0.970965}$  & $\num{1.055403}$ & $\num{1.338357}$   & $\num{0.175413}$ 
\\
14            &   $\num[math-rm=\mathbf]{0.982906}$  & $\num{0.993985}$ & $\num{29.936114}$   & $\num{0.202276}$
&&  $\num[math-rm=\mathbf]{0.995123}$  & $\num{0.978775}$ & $\num{26.130849}$   & $\num{0.231732}$ 
&&  $\num[math-rm=\mathbf]{0.991224}$  & $\num{1.002322}$ & $\num{18.662147}$   & $\num{0.325663}$
\\
\bottomrule
            \end{tabular}
            \begin{tablenotes}
            \item $\num{400}$ replications. $\num{50}$ bootstraps. Sample size is 1000.\\
            \end{tablenotes}
            \end{threeparttable}
            \end{table}

\begin{table}[htbp]
    \centering\caption{\label{tab:S1:IVQR-QR:comparetau} Finite-sample relative rRMSE of IVQR-QR averaging, aggressive estimators, and QR estimator in JTPA-based simulation model 1.}
            \sisetup{round-precision=2,round-mode=places}
    \begin{threeparttable}
    \begin{tabular}{ccccScccScccccScccc}
    \toprule
    &                     \multicolumn{4}{c}{$\tau=0.2$} &&               \multicolumn{4}{c}{$\tau=0.3$} &&          \multicolumn{4}{c}{$\tau=0.4$} \\
    \cmidrule{2-5} \cmidrule{7-10}  \cmidrule{12-15}  
    &                     \multicolumn{2}{c}{$\mathrm{IVQR.QR}$}&               && & \multicolumn{2}{c}{$\mathrm{IVQR.QR}$} &               &&  &
    \multicolumn{2}{c}{$\mathrm{IVQR.QR}$} &               
      \\
    \cmidrule{2-3} \cmidrule{7-8} \cmidrule{12-13}  
    $\mathrm{DGP }$   
    & $\mathrm{AVG}$ & $\mathrm{AGG}$ & $\mathrm{QR}$ & $\mathrm{CON}$ 
    && $\mathrm{AVG}$ & $\mathrm{AGG}$ & $\mathrm{QR}$ & $\mathrm{CON}$ 
    && $\mathrm{AVG}$ & $\mathrm{AGG}$ & $\mathrm{QR}$ & $\mathrm{CON}$ 
 \\
    \midrule
1            &   $\num[math-rm=\mathbf]{0.941798}$  & $\num{0.871945}$ & $\num{0.597444}$   & $\num{0.200699}$ 
&&  $\num[math-rm=\mathbf]{0.949044}$  & $\num{0.915647}$ & $\num{0.644577}$   & $\num{0.222795}$
&&  $\num[math-rm=\mathbf]{0.893952}$  & $\num{0.870897}$ & $\num{0.640454}$   & $\num{0.262160}$  
\\
2            &   $\num[math-rm=\mathbf]{0.909462}$  & $\num{0.873552}$ & $\num{0.801688}$   & $\num{0.238507}$ 
&& $\num[math-rm=\mathbf]{0.997286}$  & $\num{1.102791}$ & $\num{1.017636}$   & $\num{0.242750}$  
&&  $\num{1.001453}$  & $\num{1.002885}$ & $\num{1.077115}$   & $\num{0.275171}$ 
\\
3            &    $\num[math-rm=\mathbf]{0.993463}$  & $\num{1.260438}$ & $\num{1.187777}$   & $\num{0.268979}$
&&  $\num{1.052697}$  & $\num{1.348312}$ & $\num{1.471804}$   & $\num{0.289591}$  
&&  $\num{1.052231}$  & $\num{1.284234}$ & $\num{1.470588}$   & $\num{0.339472}$ 
\\
4            &   $\num[math-rm=\mathbf]{0.981629}$  & $\num{1.679892}$ & $\num{1.975982}$   & $\num{0.286971}$  
&&  $\num{1.050814}$  & $\num{1.798038}$ & $\num{2.512307}$   & $\num{0.280525}$ 
&&  $\num[math-rm=\mathbf]{0.973294}$  & $\num{2.005678}$ & $\num{2.517579}$   & $\num{0.323231}$ 
\\
5            &   $\num[math-rm=\mathbf]{0.927540}$  & $\num{0.828456}$ & $\num{0.608481}$   & $\num{0.262225}$
&&  $\num[math-rm=\mathbf]{0.884933}$  & $\num{0.809720}$ & $\num{0.584871}$   & $\num{0.435454}$ 
&&  $\num[math-rm=\mathbf]{0.868485}$  & $\num{0.745226}$ & $\num{0.568249}$   & $\num{0.845524}$  
\\
6            &  $\num[math-rm=\mathbf]{0.946775}$  & $\num{0.859464}$ & $\num{0.709994}$   & $\num{0.292543}$ 
&& $\num[math-rm=\mathbf]{0.969298}$  & $\num{0.904188}$ & $\num{0.806262}$   & $\num{0.471132}$ 
&& $\num[math-rm=\mathbf]{0.924969}$  & $\num{0.881627}$ & $\num{0.862383}$   & $\num{0.753200}$ 
\\
7            &  $\num{1.071536}$  & $\num{1.977770}$ & $\num{2.317601}$   & $\num{0.334383}$  
&& $\num{1.022212}$  & $\num{2.057648}$ & $\num{2.736132}$   & $\num{0.521376}$ 
&& $\num{1.053508}$  & $\num{2.748242}$ & $\num{3.140756}$   & $\num{0.773788}$  
\\
8            &    $\num{1.114988}$  & $\num{2.230924}$ & $\num{2.519431}$   & $\num{0.392785}$ 
&&   $\num{1.106857}$  & $\num{2.358828}$ & $\num{2.926117}$   & $\num{0.664187}$
&&   $\num{1.100109}$  & $\num{3.339572}$ & $\num{3.596283}$   & $\num{0.949278}$ 
\\
9            &     $\num[math-rm=\mathbf]{0.875496}$  & $\num{0.818780}$ & $\num{0.730609}$   & $\num{0.336555}$
&&  $\num[math-rm=\mathbf]{0.966692}$  & $\num{0.902254}$ & $\num{0.836768}$   & $\num{0.705640}$ 
&&  $\num{1.014700}$  & $\num{1.032971}$ & $\num{0.944072}$   & $\num{1.341854}$ 
\\
10            &   $\num[math-rm=\mathbf]{0.881206}$  & $\num{0.812385}$ & $\num{0.515840}$   & $\num{0.262022}$
&&  $\num[math-rm=\mathbf]{0.921218}$  & $\num{0.858971}$ & $\num{0.578228}$   & $\num{0.190439}$ 
&&  $\num[math-rm=\mathbf]{0.838552}$  & $\num{0.823072}$ & $\num{0.497556}$   & $\num{0.186220}$  
\\
11            &  $\num[math-rm=\mathbf]{0.929318}$  & $\num{0.864419}$ & $\num{0.642252}$   & $\num{0.421388}$ 
&&  $\num[math-rm=\mathbf]{0.926744}$  & $\num{0.840139}$ & $\num{0.555495}$   & $\num{0.276590}$   
&&  $\num[math-rm=\mathbf]{0.927355}$  & $\num{0.839659}$ & $\num{0.566140}$   & $\num{0.213634}$  
\\
12            &    $\num[math-rm=\mathbf]{0.992317}$  & $\num{1.473281}$ & $\num{1.680375}$   & $\num{0.296785}$  
&&  $\num[math-rm=\mathbf]{0.954414}$  & $\num{1.446733}$ & $\num{2.097819}$   & $\num{0.228425}$  
&&  $\num{1.013899}$  & $\num{1.721345}$ & $\num{2.542312}$   & $\num{0.181721}$  
\\
13            &   $\num[math-rm=\mathbf]{0.981839}$  & $\num{1.514421}$ & $\num{1.645674}$   & $\num{0.331279}$ 
&& $\num[math-rm=\mathbf]{0.977167}$  & $\num{1.598036}$ & $\num{1.918343}$   & $\num{0.254676}$ 
&&  $\num{1.050460}$  & $\num{1.875202}$ & $\num{2.386289}$   & $\num{0.194289}$ 
\\
14            &  $\num[math-rm=\mathbf]{0.850159}$  & $\num{1.156678}$ & $\num{1.171333}$   & $\num{0.793780}$  
&&  $\num[math-rm=\mathbf]{0.999237}$  & $\num{1.386553}$ & $\num{1.608324}$   & $\num{0.406347}$
&&  $\num[math-rm=\mathbf]{0.930318}$  & $\num{1.783418}$ & $\num{2.158330}$   & $\num{0.253328}$ 
\\
\midrule
    &                     \multicolumn{4}{c}{$\tau=0.6$} &&               \multicolumn{4}{c}{$\tau=0.7$} &&          \multicolumn{4}{c}{$\tau=0.8$} \\
    \cmidrule{2-5} \cmidrule{7-10}  \cmidrule{12-15}  
    &                     \multicolumn{2}{c}{$\mathrm{IVQR.QR}$}&               && & \multicolumn{2}{c}{$\mathrm{IVQR.QR}$} &               &&  &
    \multicolumn{2}{c}{$\mathrm{IVQR.QR}$} &               
      \\
    \cmidrule{2-3} \cmidrule{7-8} \cmidrule{12-13}  
    $\mathrm{DGP }$   
    & $\mathrm{AVG}$ & $\mathrm{AGG}$ & $\mathrm{QR}$ & $\mathrm{CON}$ 
    && $\mathrm{AVG}$ & $\mathrm{AGG}$ & $\mathrm{QR}$ & $\mathrm{CON}$ 
    && $\mathrm{AVG}$ & $\mathrm{AGG}$ & $\mathrm{QR}$ & $\mathrm{CON}$ 
 \\
    \midrule
1            &   $\num[math-rm=\mathbf]{0.911259}$  & $\num{0.883685}$ & $\num{0.559458}$   & $\num{0.374756}$  
&& $\num[math-rm=\mathbf]{0.886677}$  & $\num{0.818061}$ & $\num{0.544955}$   & $\num{0.441311}$
&& $\num[math-rm=\mathbf]{0.951114}$  &  $\num{0.881630}$ & $\num{0.615108}$ & $\num{0.586038}$ 
\\
2            &    $\num{1.041292}$  & $\num{1.113628}$ & $\num{0.984880}$   & $\num{0.386376}$ 
&&  $\num{1.035867}$  & $\num{1.157675}$ & $\num{1.084132}$   & $\num{0.397837}$
&&  $\num[math-rm=\mathbf]{0.961038}$  &  $\num{1.009065}$ & $\num{0.969146}$ & $\num{0.531690}$ 
\\
3            &   $\num{1.046038}$  & $\num{1.423761}$ & $\num{1.728843}$   & $\num{0.383936}$ 
&&  $\num{1.047463}$  & $\num{1.469343}$ & $\num{1.753731}$   & $\num{0.406742}$   
&&  $\num{1.006336}$  & $\num{1.217522}$ & $\num{1.377888}$   & $\num{0.590445}$ 
\\
4            &   $\num{1.062945}$  & $\num{2.097041}$ & $\num{3.055599}$   & $\num{0.327001}$  
&&  $\num{1.014734}$  & $\num{1.588099}$ & $\num{2.620650}$   & $\num{0.409502}$   
&&  $\num{1.045229}$  & $\num{1.728266}$ & $\num{2.511686}$   & $\num{0.473768}$ 
\\
5            &   $\num[math-rm=\mathbf]{0.816930}$  & $\num{0.711175}$ & $\num{0.577266}$   & $\num{2.284178}$    
&&  $\num[math-rm=\mathbf]{0.867164}$  & $\num{0.787746}$ & $\num{0.591234}$   & $\num{3.002037}$ 
&&  $\num[math-rm=\mathbf]{0.873461}$  & $\num{0.775237}$ & $\num{0.636861}$   & $\num{3.823902}$ 
\\
6            &  $\num[math-rm=\mathbf]{0.894223}$  & $\num{0.907162}$ & $\num{0.815450}$   & $\num{2.005065}$ 
&&  $\num[math-rm=\mathbf]{0.957873}$  & $\num{0.943390}$ & $\num{0.812931}$   & $\num{2.849327}$ 
&&  $\num[math-rm=\mathbf]{0.891685}$  & $\num{0.878466}$ & $\num{0.681488}$   & $\num{3.823448}$ 
\\
7            &   $\num{1.013931}$  & $\num{2.403023}$ & $\num{2.863077}$   & $\num{1.626560}$  
&&  $\num{1.012977}$  & $\num{2.376025}$ & $\num{2.791606}$   & $\num{1.970702}$ 
&&  $\num{1.019857}$  & $\num{2.084913}$ & $\num{2.598556}$   & $\num{2.182250}$
\\
8            &    $\num{1.102689}$  & $\num{3.338470}$ & $\num{3.367218}$   & $\num{2.084858}$  
&&  $\num{1.048446}$  & $\num{2.759596}$ & $\num{3.002576}$   & $\num{2.818585}$ 
&&  $\num{1.000867}$  & $\num{2.353455}$ & $\num{2.721317}$   & $\num{3.089752}$
\\
9            &  $\num[math-rm=\mathbf]{0.973719}$  & $\num{0.888872}$ & $\num{0.837853}$   & $\num{4.101009}$ 
&&  $\num[math-rm=\mathbf]{0.919984}$  & $\num{0.870516}$ & $\num{0.797345}$   & $\num{5.971537}$ 
&&  $\num[math-rm=\mathbf]{0.889460}$  & $\num{0.797108}$ & $\num{0.808708}$   & $\num{7.506477}$ 
\\
10            &  $\num[math-rm=\mathbf]{0.906869}$  & $\num{0.787995}$ & $\num{0.515212}$   & $\num{0.192385}$
&&  $\num[math-rm=\mathbf]{0.935285}$  & $\num{0.906043}$ & $\num{0.586668}$   & $\num{0.189647}$ 
&&  $\num[math-rm=\mathbf]{0.932533}$  & $\num{0.834646}$ & $\num{0.599918}$   & $\num{0.229356}$  
\\
11            &  $\num[math-rm=\mathbf]{0.911134}$  & $\num{0.797678}$ & $\num{0.552700}$   & $\num{0.210796}$
&&  $\num[math-rm=\mathbf]{0.922679}$  & $\num{0.793259}$ & $\num{0.495643}$   & $\num{0.319371}$ 
&&  $\num[math-rm=\mathbf]{0.781523}$  & $\num{0.623178}$ & $\num{0.744593}$  & $\num{0.467432}$ 
\\
12            &  $\num{1.027274}$  & $\num{1.933319}$ & $\num{2.812847}$   & $\num{0.152095}$ 
&&  $\num[math-rm=\mathbf]{0.996047}$  & $\num{1.338984}$ & $\num{2.692791}$   & $\num{0.155356}$
&&  $\num{1.030904}$  &  $\num{2.033685}$ & $\num{2.779391}$ & $\num{0.153497}$
\\
13            &  $\num[math-rm=\mathbf]{0.986569}$  & $\num{2.105844}$ & $\num{2.797674}$   & $\num{0.158934}$    
&&  $\num[math-rm=\mathbf]{0.977001}$  & $\num{1.534452}$ & $\num{2.741115}$   & $\num{0.167077}$
&&  $\num{1.046166}$  &  $\num{2.012718}$ & $\num{2.865149}$ & $\num{0.175413}$ 
\\
14            &   $\num{1.012414}$  & $\num{1.977426}$ & $\num{2.876150}$   & $\num{0.202276}$
&&  $\num{1.026333}$  & $\num{2.106319}$ & $\num{3.042188}$   & $\num{0.231732}$ 
&&  $\num{1.050328}$  & $\num{2.367321}$ & $\num{3.166435}$   & $\num{0.325663}$ 
\\
\bottomrule
            \end{tabular}
            \begin{tablenotes}
            \item $\num{400}$ replications. $\num{50}$ bootstraps. Sample size is 1000.\\
            \end{tablenotes}
            \end{threeparttable}
            \end{table}

\pagebreak

\subsection{Simulation model 2 results at different quantile levels}

\subsubsection{Gaussian error case}
\label{appdx:M2:Gaussian}

\begin{table}[htbp]
    \centering\caption{\label{tab:M2:S1S22:tau=0.2} Relative rRMSE in fixed-theta simulation model 2.1 at tau=0.2.}
            \sisetup{round-precision=2,round-mode=places}
    \begin{threeparttable}
    \begin{tabular}{cSScccScccSccccc}
    \toprule
    &                    && \multicolumn{2}{c}{$\mathrm{IVQR.2SLS}$}&               && \multicolumn{2}{c}{$\mathrm{IVQR.QR} $} &               &&  &  &   \\
    \cmidrule{4-5} \cmidrule{8-9}
    $\mathrm{DGP }$  &  $\mathrm{Endog}$  &&  $\mathrm{AVG} $ & $\mathrm{AGG}$       & $\mathrm{2SLS}$ && $\mathrm{AVG}$ & $\mathrm{AGG}$         & $\mathrm{QR}$ && $\mathrm{BS} $    & $\mathrm{SEE}$& \\
    \midrule
1          &    0  &&   $\num[math-rm=\mathbf]{0.912554}$  & $\num{0.904838}$ & $\num{0.695853}$  && $\num[math-rm=\mathbf]{0.632489}$  &  $\num{0.472075}$ & $\num{0.460235}$ && $\num[math-rm=\mathbf]{0.563689}$ &  $\num{0.694144}$ & \\
2          & 0.05  &&   $\num[math-rm=\mathbf]{0.899496}$  & $\num{0.881015}$ & $\num{0.657251}$  && $\num[math-rm=\mathbf]{0.661141}$  &  $\num{0.527252}$ & $\num{0.519990}$ && $\num[math-rm=\mathbf]{0.574081}$ &  $\num{0.656300}$ & \\
3          &  0.1  &&   $\num[math-rm=\mathbf]{0.948484}$  & $\num{0.948501}$ & $\num{0.732717}$  && $\num[math-rm=\mathbf]{0.784556}$  &  $\num{0.806223}$ & $\num{0.795545}$ && $\num[math-rm=\mathbf]{0.717411}$ &  $\num{0.746167}$ & \\
4          & 0.15  &&   $\num[math-rm=\mathbf]{0.923230}$  & $\num{0.944788}$ & $\num{0.680327}$  && $\num[math-rm=\mathbf]{0.832448}$  &  $\num{1.017197}$ & $\num{0.994308}$ && $\num[math-rm=\mathbf]{0.735321}$ &  $\num{0.671508}$ & \\
5          &  0.2  &&   $\num[math-rm=\mathbf]{0.949964}$  & $\num{0.926758}$ & $\num{0.666606}$  && $\num[math-rm=\mathbf]{0.887786}$  &  $\num{1.282756}$ & $\num{1.267169}$ && $\num[math-rm=\mathbf]{0.765192}$ &  $\num{0.671165}$ & \\
6          & 0.25  &&   $\num[math-rm=\mathbf]{0.942025}$  & $\num{0.910500}$ & $\num{0.719735}$  && $\num[math-rm=\mathbf]{0.913747}$  &  $\num{1.610236}$ & $\num{1.579899}$ && $\num[math-rm=\mathbf]{0.770856}$ &  $\num{0.725537}$ & \\
7          &  0.3  &&   $\num[math-rm=\mathbf]{0.934619}$  & $\num{0.972465}$ & $\num{0.681975}$  && $\num[math-rm=\mathbf]{0.946881}$  &  $\num{1.832664}$ & $\num{1.792756}$ && $\num[math-rm=\mathbf]{0.773705}$ &  $\num{0.692967}$ & \\
8          & 0.35  &&   $\num[math-rm=\mathbf]{0.947837}$  & $\num{0.881156}$ & $\num{0.668714}$  && $\num[math-rm=\mathbf]{0.958236}$  &  $\num{2.157011}$ & $\num{2.132194}$ && $\num[math-rm=\mathbf]{0.771860}$ &  $\num{0.674264}$ & \\
9          &  0.4  &&   $\num[math-rm=\mathbf]{0.906140}$  & $\num{0.932177}$ & $\num{0.655680}$  && $\num[math-rm=\mathbf]{0.962512}$  &  $\num{2.557154}$ & $\num{2.477202}$ && $\num[math-rm=\mathbf]{0.741929}$ &  $\num{0.662283}$ & \\
\bottomrule
            \end{tabular}
            \begin{tablenotes}
            \item $\num{200}$ replications. $\num{50}$ bootstraps. Sample size is 1000.\\
            \end{tablenotes}
            \end{threeparttable}
            \end{table}

\begin{table}[htbp!]
    \centering\caption{\label{tab:M2:S1S22:tau=0.3} Relative rRMSE in fixed-theta simulation model 2.1 at tau=0.3.}
            \sisetup{round-precision=2,round-mode=places}
    \begin{threeparttable}
    \begin{tabular}{cSScccScccSccccc}
    \toprule
    &                    && \multicolumn{2}{c}{$\mathrm{IVQR.2SLS}$}&               && \multicolumn{2}{c}{$\mathrm{IVQR.QR} $} &               &&  &  &   \\
    \cmidrule{4-5} \cmidrule{8-9}
    $\mathrm{DGP }$  &  $\mathrm{Endog}$  &&  $\mathrm{AVG} $ & $\mathrm{AGG}$       & $\mathrm{2SLS}$ && $\mathrm{AVG}$ & $\mathrm{AGG}$         & $\mathrm{QR}$ && $\mathrm{BS} $    & $\mathrm{SEE}$& \\
    \midrule
1          &    0  &&   $\num[math-rm=\mathbf]{0.930014}$  & $\num{0.930419}$ & $\num{0.745751}$  && $\num[math-rm=\mathbf]{0.628842}$  &  $\num{0.454545}$ & $\num{0.454717}$ && $\num[math-rm=\mathbf]{0.593561}$ &  $\num{0.748013}$ & \\
2          & 0.05  &&   $\num[math-rm=\mathbf]{0.986889}$  & $\num{0.955293}$ & $\num{0.783030}$  && $\num[math-rm=\mathbf]{0.705617}$  &  $\num{0.596166}$ & $\num{0.581637}$ && $\num[math-rm=\mathbf]{0.656178}$ &  $\num{0.787999}$ & \\
3          &  0.1  &&   $\num[math-rm=\mathbf]{0.918447}$  & $\num{0.916201}$ & $\num{0.777544}$  && $\num[math-rm=\mathbf]{0.788658}$  &  $\num{0.823297}$ & $\num{0.807444}$ && $\num[math-rm=\mathbf]{0.744570}$ &  $\num{0.779326}$ & \\
4          & 0.15  &&   $\num[math-rm=\mathbf]{0.929851}$  & $\num{0.931470}$ & $\num{0.663141}$  && $\num[math-rm=\mathbf]{0.839917}$  &  $\num{1.027878}$ & $\num{1.009847}$ && $\num[math-rm=\mathbf]{0.734071}$ &  $\num{0.671870}$ & \\
5          &  0.2  &&   $\num[math-rm=\mathbf]{0.952740}$  & $\num{0.972782}$ & $\num{0.800605}$  && $\num[math-rm=\mathbf]{0.945531}$  &  $\num{1.469600}$ & $\num{1.454345}$ && $\num[math-rm=\mathbf]{0.866432}$ &  $\num{0.799774}$ & \\
6          & 0.25  &&   $\num[math-rm=\mathbf]{0.911091}$  & $\num{0.849382}$ & $\num{0.698091}$  && $\num[math-rm=\mathbf]{0.946628}$  &  $\num{1.594224}$ & $\num{1.568749}$ && $\num[math-rm=\mathbf]{0.818973}$ &  $\num{0.705457}$ & \\
7          &  0.3  &&   $\num[math-rm=\mathbf]{0.945599}$  & $\num{0.953980}$ & $\num{0.718174}$  && $\num[math-rm=\mathbf]{0.960896}$  &  $\num{2.036718}$ & $\num{2.003893}$ && $\num[math-rm=\mathbf]{0.796740}$ &  $\num{0.719578}$ & \\
8          & 0.35  &&   $\num[math-rm=\mathbf]{0.944680}$  & $\num{0.960956}$ & $\num{0.779274}$  && $\num[math-rm=\mathbf]{0.982549}$  &  $\num{2.476480}$ & $\num{2.443201}$ && $\num[math-rm=\mathbf]{0.860426}$ &  $\num{0.786235}$ & \\
9          &  0.4  &&   $\num[math-rm=\mathbf]{0.925776}$  & $\num{0.900187}$ & $\num{0.711287}$  && $\num[math-rm=\mathbf]{0.966771}$  &  $\num{2.594347}$ & $\num{2.530612}$ && $\num[math-rm=\mathbf]{0.795334}$ &  $\num{0.719044}$ & \\
\bottomrule
            \end{tabular}
            \begin{tablenotes}
            \item $\num{200}$ replications. $\num{50}$ bootstraps. Sample size is 1000.\\
            \end{tablenotes}
            \end{threeparttable}
            \end{table}

\begin{table}[htbp]
    \centering\caption{\label{tab:M2:S1S22:tau=0.4} Relative rRMSE in fixed-theta simulation model 2.1 at tau=0.4.}
            \sisetup{round-precision=2,round-mode=places}
    \begin{threeparttable}
    \begin{tabular}{cSScccScccSccccc}
    \toprule
    &                    && \multicolumn{2}{c}{$\mathrm{IVQR.2SLS}$}&               && \multicolumn{2}{c}{$\mathrm{IVQR.QR} $} &               &&  &  &   \\
    \cmidrule{4-5} \cmidrule{8-9}
    $\mathrm{DGP }$  &  $\mathrm{Endog}$  &&  $\mathrm{AVG} $ & $\mathrm{AGG}$       & $\mathrm{2SLS}$ && $\mathrm{AVG}$ & $\mathrm{AGG}$         & $\mathrm{QR}$ && $\mathrm{BS} $    & $\mathrm{SEE}$ & \\
    \midrule
1          &    0  &&   $\num[math-rm=\mathbf]{0.941300}$  & $\num{0.941004}$ & $\num{0.768011}$  && $\num[math-rm=\mathbf]{0.642046}$  &  $\num{0.478941}$ & $\num{0.466754}$ && $\num[math-rm=\mathbf]{0.582446}$ &  $\num{0.852524}$ & \\
2          & 0.05  &&   $\num[math-rm=\mathbf]{0.914320}$  & $\num{0.886529}$ & $\num{0.756304}$  && $\num[math-rm=\mathbf]{0.671020}$  &  $\num{0.574163}$ & $\num{0.564033}$ && $\num[math-rm=\mathbf]{0.621477}$ &  $\num{0.841034}$ & \\
3          &  0.1  &&   $\num[math-rm=\mathbf]{0.888317}$  & $\num{0.921065}$ & $\num{0.766891}$  && $\num[math-rm=\mathbf]{0.758524}$  &  $\num{0.802277}$ & $\num{0.795053}$ && $\num[math-rm=\mathbf]{0.676812}$ &  $\num{0.830918}$ & \\
4          & 0.15  &&   $\num[math-rm=\mathbf]{0.922819}$  & $\num{0.955135}$ & $\num{0.822389}$  && $\num[math-rm=\mathbf]{0.860452}$  &  $\num{1.235585}$ & $\num{1.210753}$ && $\num[math-rm=\mathbf]{0.812073}$ &  $\num{0.891056}$ & \\
5          &  0.2  &&   $\num[math-rm=\mathbf]{0.931372}$  & $\num{0.926990}$ & $\num{0.761819}$  && $\num[math-rm=\mathbf]{0.883927}$  &  $\num{1.437655}$ & $\num{1.435180}$ && $\num[math-rm=\mathbf]{0.814931}$ &  $\num{0.839641}$ & \\
6          & 0.25  &&   $\num[math-rm=\mathbf]{0.893593}$  & $\num{0.894784}$ & $\num{0.738380}$  && $\num[math-rm=\mathbf]{0.942634}$  &  $\num{1.665645}$ & $\num{1.636619}$ && $\num[math-rm=\mathbf]{0.819394}$ &  $\num{0.815148}$ & \\
7          &  0.3  &&   $\num[math-rm=\mathbf]{0.938526}$  & $\num{0.908385}$ & $\num{0.824919}$  && $\num[math-rm=\mathbf]{0.960790}$  &  $\num{2.191715}$ & $\num{2.156783}$ && $\num[math-rm=\mathbf]{0.890250}$ &  $\num{0.877262}$ & \\
8          & 0.35  &&   $\num[math-rm=\mathbf]{0.922473}$  & $\num{0.927142}$ & $\num{0.736169}$  && $\num[math-rm=\mathbf]{0.968999}$  &  $\num{2.399767}$ & $\num{2.358589}$ && $\num[math-rm=\mathbf]{0.822829}$ &  $\num{0.814069}$ & \\
9          &  0.4  &&   $\num[math-rm=\mathbf]{0.911974}$  & $\num{0.920334}$ & $\num{0.804927}$  && $\num{1.005281}$  &  $\num{2.964080}$ & $\num{2.912014}$ && $\num[math-rm=\mathbf]{0.889765}$ &  $\num{0.864508}$ & \\
\bottomrule
            \end{tabular}
            \begin{tablenotes}
            \item $\num{200}$ replications. $\num{50}$ bootstraps. Sample size is 1000.\\
            \end{tablenotes}
            \end{threeparttable}
            \end{table}

\begin{table}[htbp]
    \centering\caption{\label{tab:M2:S1S22:tau=0.6} Relative rRMSE in fixed-theta simulation model 2.1 at tau=0.6.}
            \sisetup{round-precision=2,round-mode=places}
    \begin{threeparttable}
    \begin{tabular}{cSScccScccSccccc}
    \toprule
    &                    && \multicolumn{2}{c}{$\mathrm{IVQR.2SLS}$}&               && \multicolumn{2}{c}{$\mathrm{IVQR.QR} $} &               &&  &  &   \\
    \cmidrule{4-5} \cmidrule{8-9}
    $\mathrm{DGP }$  &  $\mathrm{Endog}$  &&  $\mathrm{AVG} $ & $\mathrm{AGG}$       & $\mathrm{2SLS}$ && $\mathrm{AVG}$ & $\mathrm{AGG}$         & $\mathrm{QR}$ && $\mathrm{BS} $    & $\mathrm{SEE}$ & \\
    \midrule
1          &    0  &&   $\num[math-rm=\mathbf]{0.932295}$  & $\num{0.972621}$ & $\num{0.785067}$  && $\num[math-rm=\mathbf]{0.627050}$  &  $\num{0.428602}$ & $\num{0.421007}$ && $\num[math-rm=\mathbf]{0.568392}$ &  $\num{0.821139}$ & \\
2          & 0.05  &&   $\num[math-rm=\mathbf]{0.907214}$  & $\num{0.903169}$ & $\num{0.750760}$  && $\num[math-rm=\mathbf]{0.665684}$  &  $\num{0.566199}$ & $\num{0.537802}$ && $\num[math-rm=\mathbf]{0.626208}$ &  $\num{0.827579}$ & \\
3          &  0.1  &&   $\num[math-rm=\mathbf]{0.924905}$  & $\num{0.943333}$ & $\num{0.740025}$  && $\num[math-rm=\mathbf]{0.757536}$  &  $\num{0.801435}$ & $\num{0.793176}$ && $\num[math-rm=\mathbf]{0.672323}$ &  $\num{0.818483}$ & \\
4          & 0.15  &&   $\num[math-rm=\mathbf]{0.910625}$  & $\num{0.942470}$ & $\num{0.784000}$  && $\num[math-rm=\mathbf]{0.862959}$  &  $\num{1.170549}$ & $\num{1.158481}$ && $\num[math-rm=\mathbf]{0.795096}$ &  $\num{0.840943}$ & \\
5          &  0.2  &&   $\num[math-rm=\mathbf]{0.934916}$  & $\num{0.929970}$ & $\num{0.741476}$  && $\num[math-rm=\mathbf]{0.929153}$  &  $\num{1.438181}$ & $\num{1.430027}$ && $\num[math-rm=\mathbf]{0.821462}$ &  $\num{0.869129}$ & \\
6          & 0.25  &&   $\num[math-rm=\mathbf]{0.902626}$  & $\num{0.891302}$ & $\num{0.775917}$  && $\num[math-rm=\mathbf]{0.955691}$  &  $\num{1.781201}$ & $\num{1.765892}$ && $\num[math-rm=\mathbf]{0.843190}$ &  $\num{0.860525}$ & \\
7          &  0.3  &&   $\num[math-rm=\mathbf]{0.956031}$  & $\num{0.976060}$ & $\num{0.811011}$  && $\num[math-rm=\mathbf]{0.990251}$  &  $\num{2.182120}$ & $\num{2.138888}$ && $\num[math-rm=\mathbf]{0.909753}$ &  $\num{0.899843}$ & \\
8          & 0.35  &&   $\num[math-rm=\mathbf]{0.919713}$  & $\num{0.919871}$ & $\num{0.751771}$  && $\num[math-rm=\mathbf]{0.974633}$  &  $\num{2.447289}$ & $\num{2.396635}$ && $\num[math-rm=\mathbf]{0.829006}$ &  $\num{0.840090}$ & \\
9          &  0.4  &&   $\num[math-rm=\mathbf]{0.936625}$  & $\num{0.931761}$ & $\num{0.732867}$  && $\num[math-rm=\mathbf]{0.976737}$  &  $\num{2.852819}$ & $\num{2.799648}$ && $\num[math-rm=\mathbf]{0.810845}$ &  $\num{0.804491}$ & \\
\bottomrule
            \end{tabular}
            \begin{tablenotes}
            \item $\num{200}$ replications. $\num{50}$ bootstraps. Sample size is 1000.\\
            \end{tablenotes}
            \end{threeparttable}
            \end{table}

\begin{table}[htbp]
    \centering\caption{\label{tab:M2:S1S22:tau=0.7} Relative rRMSE in fixed-theta simulation model 2.1 at tau=0.7.}
            \sisetup{round-precision=2,round-mode=places}
    \begin{threeparttable}
    \begin{tabular}{cSScccScccSccccc}
    \toprule
    &                    && \multicolumn{2}{c}{$\mathrm{IVQR.2SLS}$}&               && \multicolumn{2}{c}{$\mathrm{IVQR.QR} $} &               &&  &  &   \\
    \cmidrule{4-5} \cmidrule{8-9}
    $\mathrm{DGP }$  &  $\mathrm{Endog}$  &&  $\mathrm{AVG} $ & $\mathrm{AGG}$       & $\mathrm{2SLS}$ && $\mathrm{AVG}$ & $\mathrm{AGG}$         & $\mathrm{QR}$ && $\mathrm{BS} $    & $\mathrm{SEE}$ & \\
    \midrule
1          &    0  &&   $\num[math-rm=\mathbf]{0.938639}$  & $\num{0.929722}$ & $\num{0.732812}$  && $\num[math-rm=\mathbf]{0.665792}$  &  $\num{0.446796}$ & $\num{0.427769}$ && $\num[math-rm=\mathbf]{0.640949}$ &  $\num{0.738509}$ & \\
2          & 0.05  &&   $\num[math-rm=\mathbf]{0.932310}$  & $\num{0.930205}$ & $\num{0.701639}$  && $\num[math-rm=\mathbf]{0.683981}$  &  $\num{0.555895}$ & $\num{0.546676}$ && $\num[math-rm=\mathbf]{0.628711}$ &  $\num{0.698239}$ & \\
3          &  0.1  &&   $\num[math-rm=\mathbf]{0.937098}$  & $\num{0.936386}$ & $\num{0.700403}$  && $\num[math-rm=\mathbf]{0.759962}$  &  $\num{0.794360}$ & $\num{0.790327}$ && $\num[math-rm=\mathbf]{0.717073}$ &  $\num{0.694108}$ & \\
4          & 0.15  &&   $\num[math-rm=\mathbf]{0.913315}$  & $\num{0.931323}$ & $\num{0.760025}$  && $\num[math-rm=\mathbf]{0.886380}$  &  $\num{1.142344}$ & $\num{1.119383}$ && $\num[math-rm=\mathbf]{0.769081}$ &  $\num{0.756690}$ & \\
5          &  0.2  &&   $\num[math-rm=\mathbf]{0.926945}$  & $\num{0.922710}$ & $\num{0.779164}$  && $\num[math-rm=\mathbf]{0.949700}$  &  $\num{1.493225}$ & $\num{1.482712}$ && $\num[math-rm=\mathbf]{0.844859}$ &  $\num{0.773334}$ & \\
6          & 0.25  &&   $\num[math-rm=\mathbf]{0.956132}$  & $\num{0.992639}$ & $\num{0.740241}$  && $\num[math-rm=\mathbf]{0.939966}$  &  $\num{1.672360}$ & $\num{1.668143}$ && $\num[math-rm=\mathbf]{0.832532}$ &  $\num{0.746445}$ & \\
7          &  0.3  &&   $\num[math-rm=\mathbf]{0.905105}$  & $\num{0.911932}$ & $\num{0.761440}$  && $\num[math-rm=\mathbf]{0.961689}$  &  $\num{2.156402}$ & $\num{2.127319}$ && $\num[math-rm=\mathbf]{0.823798}$ &  $\num{0.762869}$ & \\
8          & 0.35  &&   $\num[math-rm=\mathbf]{0.956419}$  & $\num{0.976421}$ & $\num{0.767141}$  && $\num{1.014680}$  &  $\num{2.592047}$ & $\num{2.557790}$ && $\num[math-rm=\mathbf]{0.851999}$ &  $\num{0.763230}$ & \\
9          &  0.4  &&   $\num[math-rm=\mathbf]{0.908085}$  & $\num{0.890684}$ & $\num{0.721491}$  && $\num[math-rm=\mathbf]{0.976115}$  &  $\num{2.651701}$ & $\num{2.584554}$ && $\num[math-rm=\mathbf]{0.803980}$ &  $\num{0.727223}$ & \\
\bottomrule
            \end{tabular}
            \begin{tablenotes}
            \item $\num{200}$ replications. $\num{50}$ bootstraps. Sample size is 1000.\\
            \end{tablenotes}
            \end{threeparttable}
            \end{table}

\begin{table}[htbp]
    \centering\caption{\label{tab:M2:S1S22:tau=0.8} Relative rRMSE in fixed-theta simulation model 2.1 at tau=0.8.}
            \sisetup{round-precision=2,round-mode=places}
    \begin{threeparttable}
    \begin{tabular}{cSScccScccSccccc}
    \toprule
    &                    && \multicolumn{2}{c}{$\mathrm{IVQR.2SLS}$}&               && \multicolumn{2}{c}{$\mathrm{IVQR.QR} $} &               &&  &  &   \\
    \cmidrule{4-5} \cmidrule{8-9}
    $\mathrm{DGP }$  &  $\mathrm{Endog}$  &&  $\mathrm{AVG} $ & $\mathrm{AGG}$       & $\mathrm{2SLS}$ && $\mathrm{AVG}$ & $\mathrm{AGG}$         & $\mathrm{QR}$ && $\mathrm{BS} $    & $\mathrm{SEE}$ & \\
    \midrule
1          &    0  &&   $\num[math-rm=\mathbf]{0.938936}$  & $\num{0.903381}$ & $\num{0.700645}$  && $\num[math-rm=\mathbf]{0.644842}$  &  $\num{0.453057}$ & $\num{0.427055}$ && $\num[math-rm=\mathbf]{0.570289}$ &  $\num{0.692649}$ & \\
2          & 0.05  &&   $\num[math-rm=\mathbf]{0.989837}$  & $\num{1.003194}$ & $\num{0.745132}$  && $\num[math-rm=\mathbf]{0.706387}$  &  $\num{0.593060}$ & $\num{0.578406}$ && $\num[math-rm=\mathbf]{0.656623}$ &  $\num{0.743466}$ & \\
3          &  0.1  &&   $\num[math-rm=\mathbf]{0.939553}$  & $\num{0.918079}$ & $\num{0.695216}$  && $\num[math-rm=\mathbf]{0.751819}$  &  $\num{0.741166}$ & $\num{0.727343}$ && $\num[math-rm=\mathbf]{0.667569}$ &  $\num{0.705668}$ & \\
4          & 0.15  &&   $\num[math-rm=\mathbf]{0.926631}$  & $\num{0.941766}$ & $\num{0.661339}$  && $\num[math-rm=\mathbf]{0.822981}$  &  $\num{1.007703}$ & $\num{0.993716}$ && $\num[math-rm=\mathbf]{0.753312}$ &  $\num{0.666738}$ & \\
5          &  0.2  &&   $\num[math-rm=\mathbf]{0.927442}$  & $\num{0.912672}$ & $\num{0.732559}$  && $\num[math-rm=\mathbf]{0.912402}$  &  $\num{1.306970}$ & $\num{1.287492}$ && $\num[math-rm=\mathbf]{0.788009}$ &  $\num{0.736738}$ & \\
6          & 0.25  &&   $\num[math-rm=\mathbf]{0.926608}$  & $\num{0.957610}$ & $\num{0.712921}$  && $\num[math-rm=\mathbf]{0.930087}$  &  $\num{1.577030}$ & $\num{1.561710}$ && $\num[math-rm=\mathbf]{0.819698}$ &  $\num{0.709635}$ & \\
7          &  0.3  &&   $\num[math-rm=\mathbf]{0.910973}$  & $\num{0.943295}$ & $\num{0.690537}$  && $\num[math-rm=\mathbf]{0.950223}$  &  $\num{1.888220}$ & $\num{1.894471}$ && $\num[math-rm=\mathbf]{0.780809}$ &  $\num{0.696275}$ & \\
8          & 0.35  &&   $\num[math-rm=\mathbf]{0.899239}$  & $\num{0.888192}$ & $\num{0.722863}$  && $\num[math-rm=\mathbf]{0.944672}$  &  $\num{2.197099}$ & $\num{2.135770}$ && $\num[math-rm=\mathbf]{0.784284}$ &  $\num{0.728300}$ & \\
9          &  0.4  &&   $\num[math-rm=\mathbf]{0.951325}$  & $\num{0.970825}$ & $\num{0.749307}$  && $\num[math-rm=\mathbf]{0.995542}$  &  $\num{2.727588}$ & $\num{2.668856}$ && $\num[math-rm=\mathbf]{0.828182}$ &  $\num{0.738477}$ & \\
\bottomrule
            \end{tabular}
            \begin{tablenotes}
            \item $\num{200}$ replications. $\num{50}$ bootstraps. Sample size is 1000.\\
            \end{tablenotes}
            \end{threeparttable}
            \end{table}

\begin{figure}[htbp]
\includegraphics[width=0.45\textwidth, height=0.3\textheight, trim=35 20 20 70]{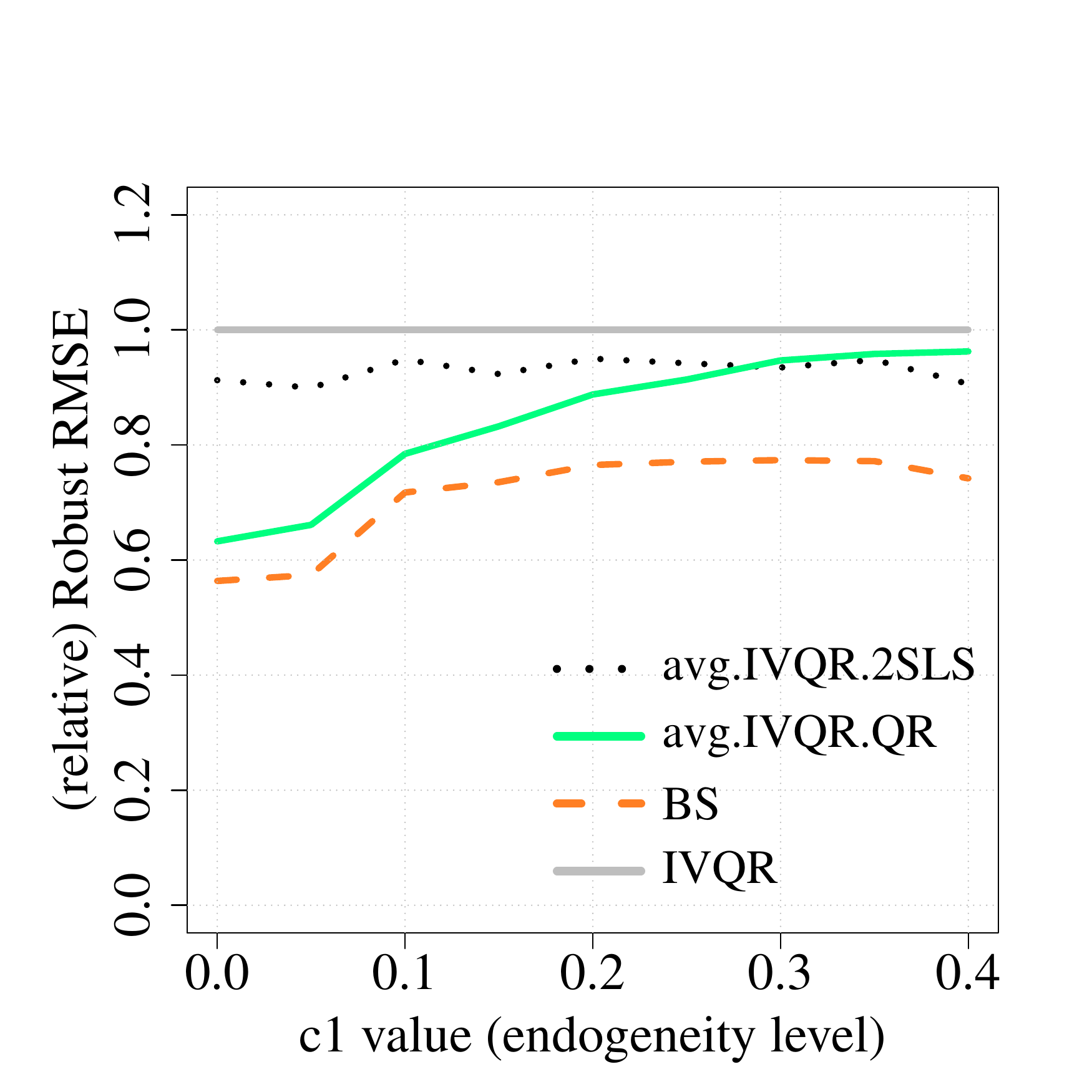}
\hfill
\includegraphics[width=0.45\textwidth, height=0.3\textheight, trim=35 20 20 70]{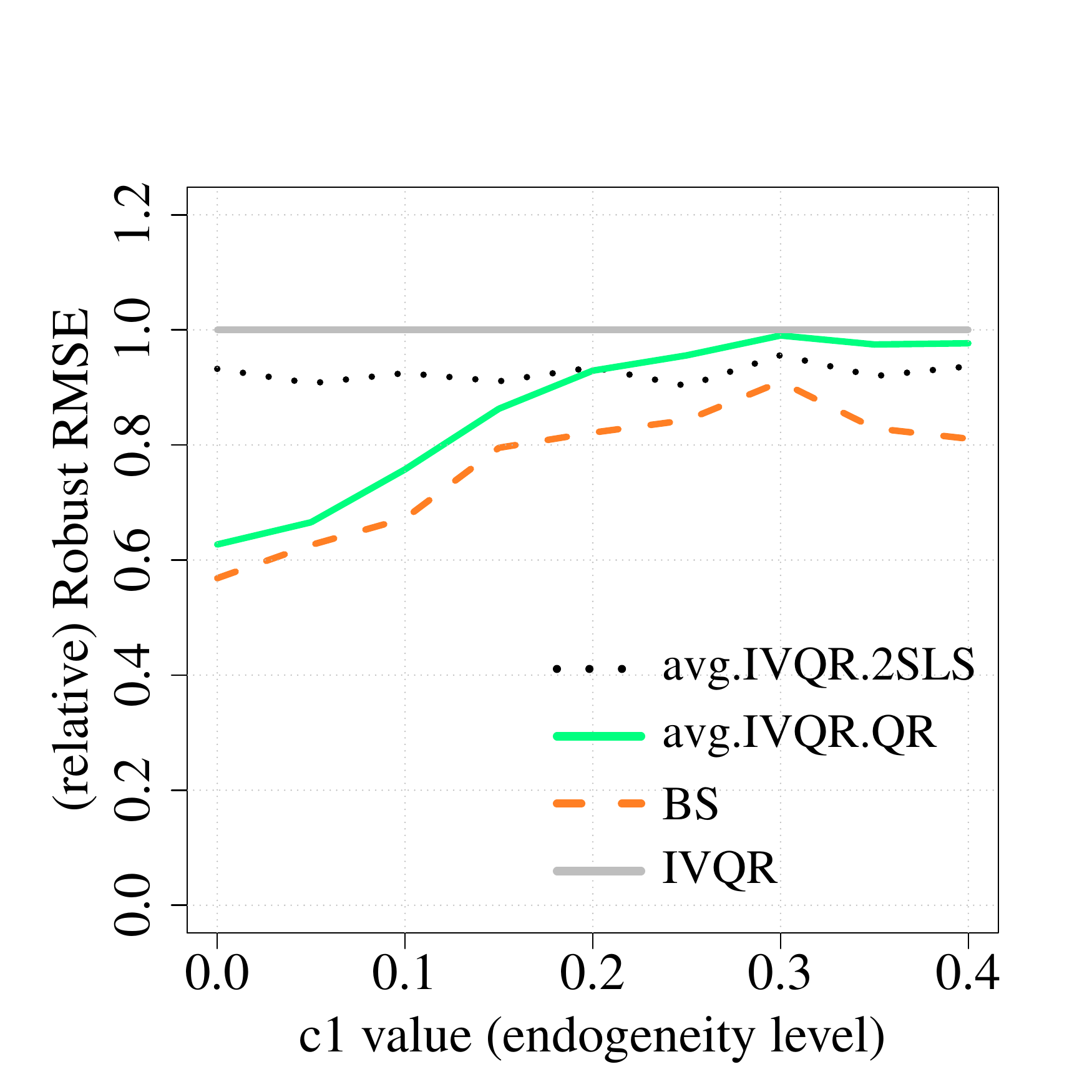}
\hfill\null
\includegraphics[width=0.45\textwidth, height=0.3\textheight, trim=35 20 20 70]{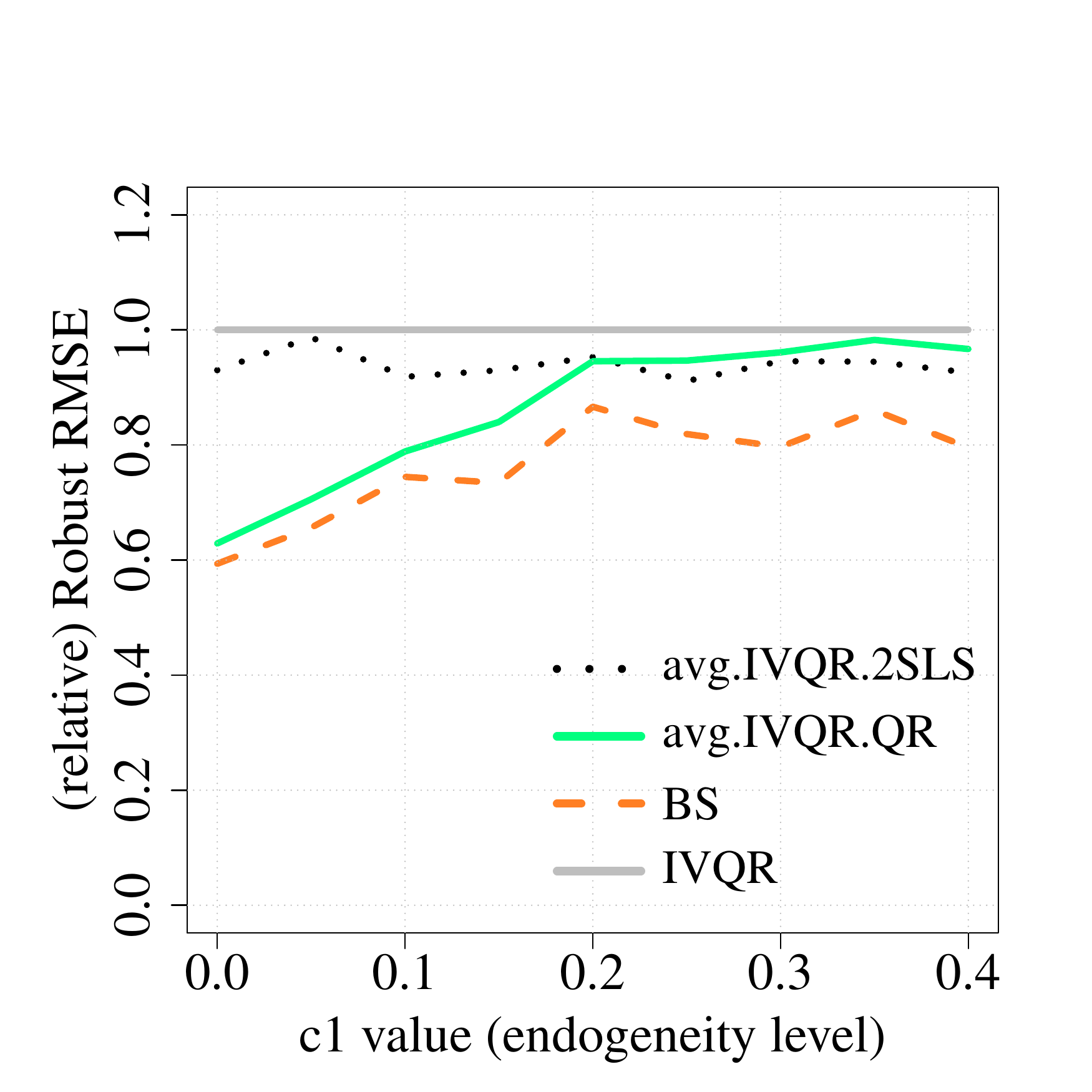}
\hfill
\includegraphics[width=0.45\textwidth, height=0.3\textheight, trim=35 20 20 70 ]{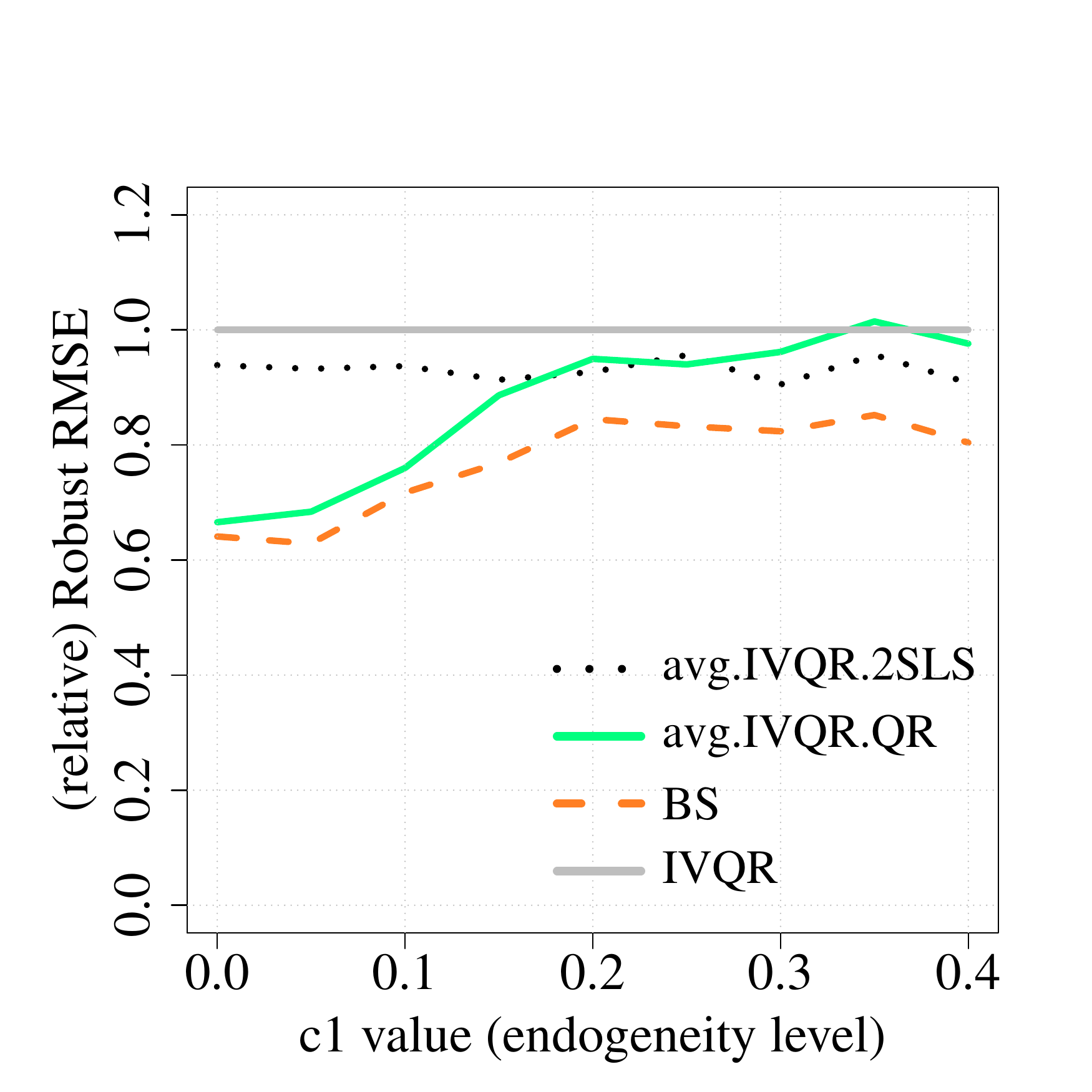}
\hfill\null
\includegraphics[width=0.45\textwidth, height=0.3\textheight, trim=35 20 20 70 ]{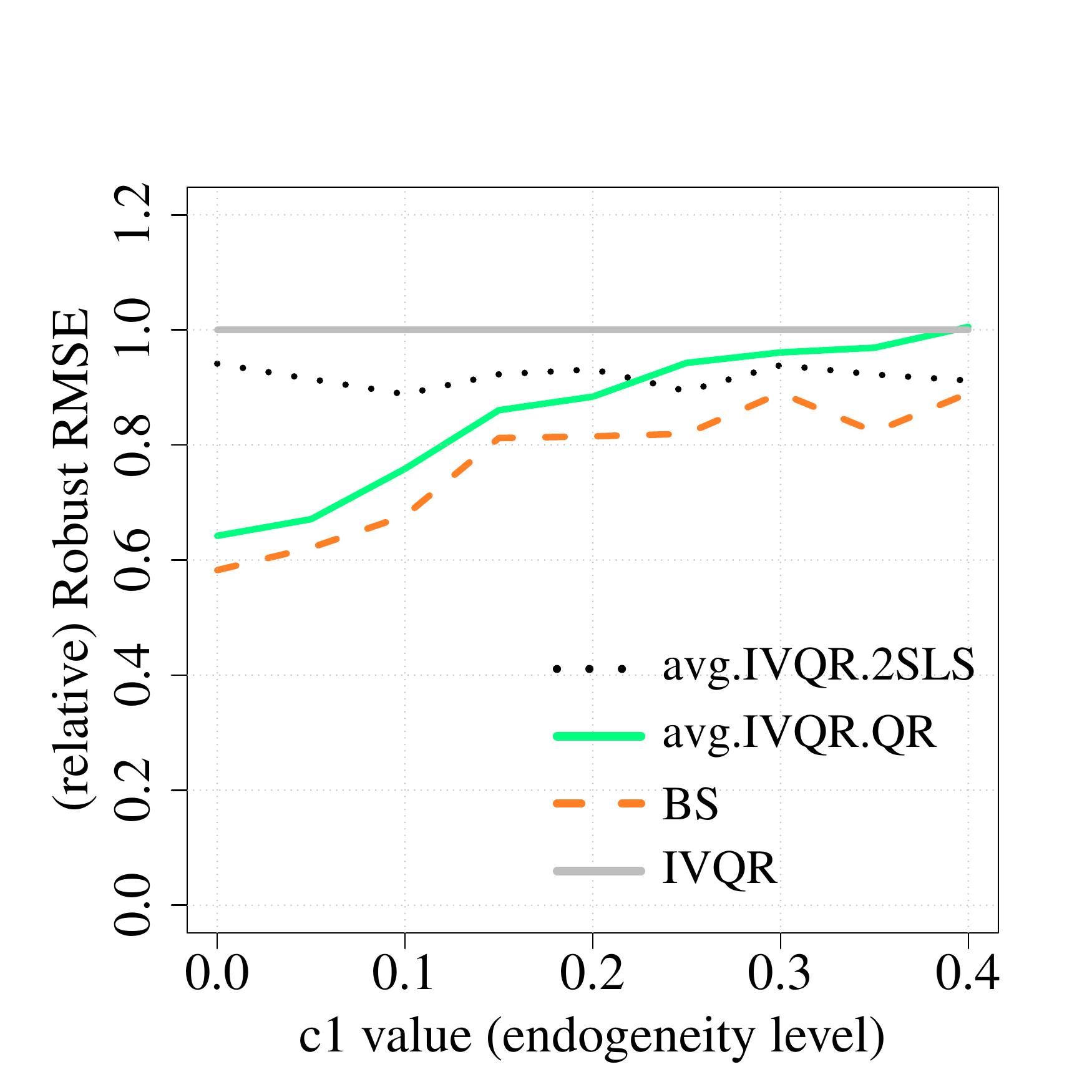}
\hfill
\includegraphics[width=0.45\textwidth, height=0.3\textheight, trim=35 20 20 70]{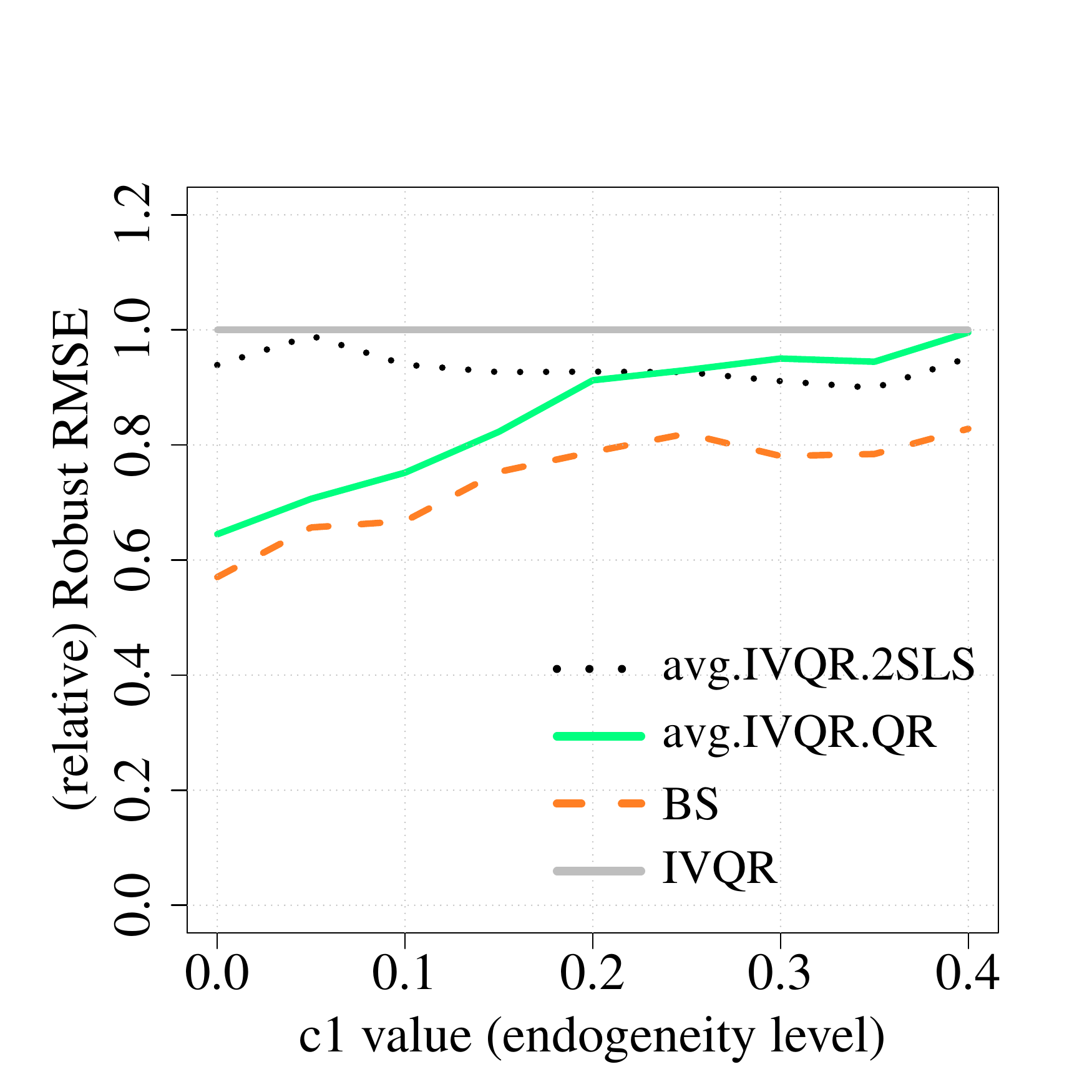}
\caption{\label{fig:M2:S1S22:compare_tau}%
Relative rRMSE in simulation model 2.1 (Gaussian error) at $\tau=0.2$ (left up), $\tau=0.3$ (left middle), $\tau=0.4$ (up bottom), $\tau=0.6$ (right up), $\tau=0.7$ (right middle),$\tau=0.8$ (right bottom), based on $\num{200}$ replications and $\num{50}$ bootstraps. Sample size n=1000.} 
\end{figure}

\pagebreak

\subsubsection{Non-Gaussian error case}
\label{appdx:M2:non-Gaussian}

\begin{table}[htbp]
    \centering\caption{\label{tab:M2:S1S21:tau=0.5} Relative rRMSE in fixed-theta simulation model 2.2 at tau=0.5 (median).}
            \sisetup{round-precision=2,round-mode=places}
    \begin{threeparttable}
    \begin{tabular}{cSScccScccSccccc}
    \toprule
    &                    && \multicolumn{2}{c}{$\mathrm{IVQR.2SLS}$}&               && \multicolumn{2}{c}{$\mathrm{IVQR.QR} $} &               &&  &  &   \\
    \cmidrule{4-5} \cmidrule{8-9}
    $\mathrm{DGP }$  &  $\mathrm{Endog}$  &&  $\mathrm{AVG} $ & $\mathrm{AGG}$       & $\mathrm{2SLS}$ && $\mathrm{AVG}$ & $\mathrm{AGG}$         & $\mathrm{QR}$ && $\mathrm{BS} $    & $\mathrm{SEE}$ & \\
    \midrule
1          &    0  &&   $\num[math-rm=\mathbf]{0.933166}$  & $\num{0.977039}$ & $\num{0.829148}$  && $\num[math-rm=\mathbf]{0.649330}$  &  $\num{0.433104}$ & $\num{0.415311}$ && $\num[math-rm=\mathbf]{0.581875}$ &  $\num{0.871879}$ & \\
2          & 0.05  &&   $\num[math-rm=\mathbf]{0.950633}$  & $\num{0.977121}$ & $\num{0.835832}$  && $\num[math-rm=\mathbf]{0.678719}$  &  $\num{0.548316}$ & $\num{0.548175}$ && $\num[math-rm=\mathbf]{0.608213}$ &  $\num{0.870095}$ & \\
3          &  0.10  &&   $\num[math-rm=\mathbf]{0.942768}$  & $\num{0.994438}$ & $\num{0.880619}$  && $\num[math-rm=\mathbf]{0.787032}$  &  $\num{0.827284}$ & $\num{0.816406}$ && $\num[math-rm=\mathbf]{0.733478}$ &  $\num{0.866133}$ & \\
4          & 0.15  &&   $\num[math-rm=\mathbf]{0.909213}$  & $\num{0.952970}$ & $\num{0.847706}$  && $\num[math-rm=\mathbf]{0.846951}$  &  $\num{1.108374}$ & $\num{1.089305}$ && $\num[math-rm=\mathbf]{0.781893}$ &  $\num{0.849499}$ & \\
5          &  0.20  &&   $\num[math-rm=\mathbf]{0.946654}$  & $\num{1.022409}$ & $\num{0.886523}$  && $\num[math-rm=\mathbf]{0.910879}$  &  $\num{1.530068}$ & $\num{1.524478}$ && $\num[math-rm=\mathbf]{0.855920}$ &  $\num{0.901316}$ & \\
6          & 0.25  &&   $\num[math-rm=\mathbf]{0.963794}$  & $\num{1.012445}$ & $\num{0.904009}$  && $\num[math-rm=\mathbf]{0.934694}$  &  $\num{1.739299}$ & $\num{1.733614}$ && $\num[math-rm=\mathbf]{0.880215}$ &  $\num{0.915744}$ & \\
7          &  0.30  &&   $\num[math-rm=\mathbf]{0.936276}$  & $\num{0.976490}$ & $\num{0.882975}$  && $\num[math-rm=\mathbf]{0.981626}$  &  $\num{2.090170}$ & $\num{2.066818}$ && $\num[math-rm=\mathbf]{0.891600}$ &  $\num{0.868277}$ & \\
8          & 0.35  &&   $\num[math-rm=\mathbf]{0.967486}$  & $\num{1.024155}$ & $\num{0.958527}$  && $\num[math-rm=\mathbf]{0.965456}$  &  $\num{2.480717}$ & $\num{2.427608}$ && $\num[math-rm=\mathbf]{0.911479}$ &  $\num{0.887186}$ & \\
9          &  0.40  &&   $\num[math-rm=\mathbf]{0.956044}$  & $\num{0.974692}$ & $\num{0.874485}$  && $\num[math-rm=\mathbf]{0.943147}$  &  $\num{2.788741}$ & $\num{2.706202}$ && $\num[math-rm=\mathbf]{0.887236}$ &  $\num{0.892896}$ & \\
\bottomrule
            \end{tabular}
            \begin{tablenotes}
            \item $\num{200}$ replications. $\num{50}$ bootstraps. Sample size is 1000.\\
            \end{tablenotes}
            \end{threeparttable}
            \end{table}

\begin{figure}[htbp!]
 \centering
\includegraphics[width=0.5\textwidth, height=3.5in]{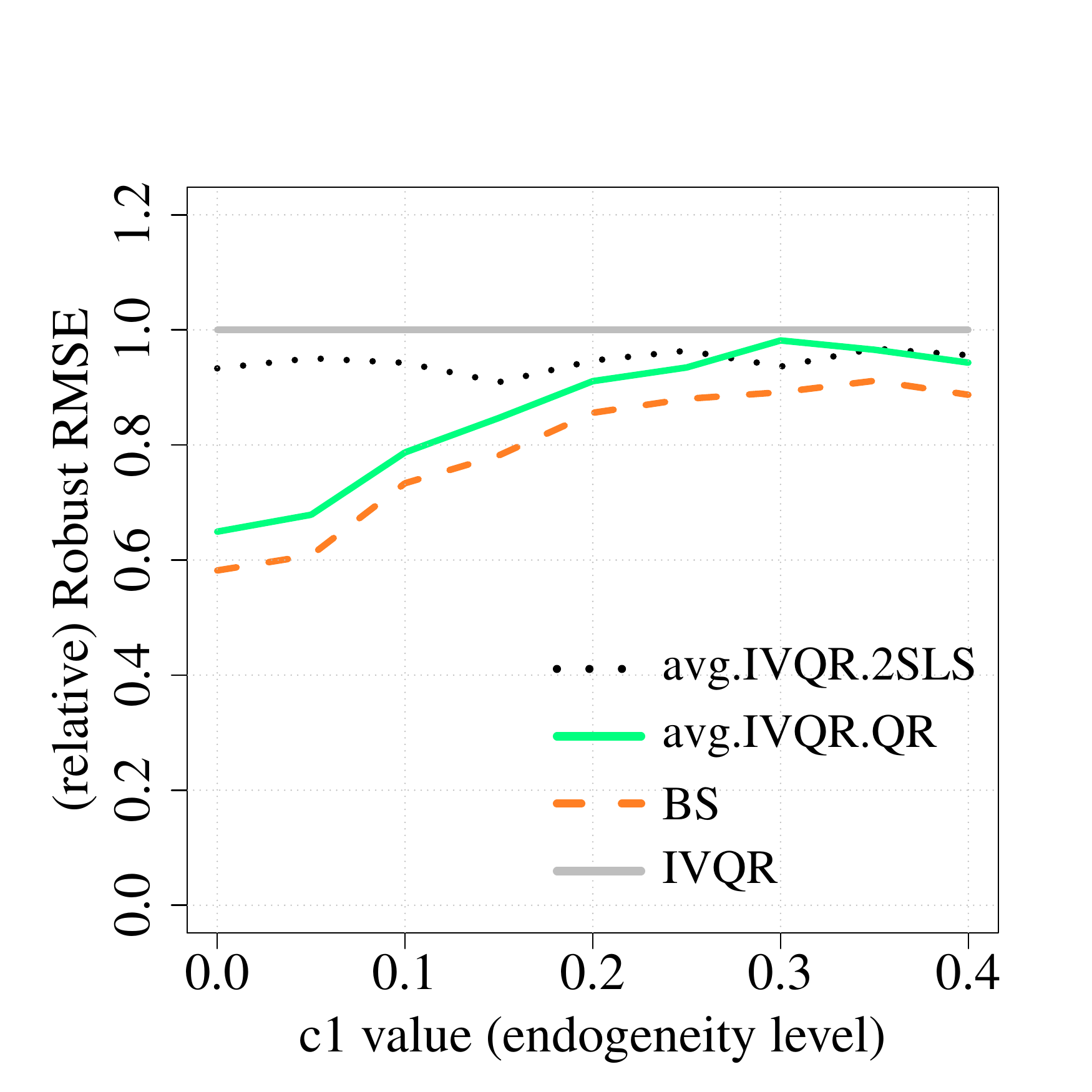}
\caption{\label{fig:sim-M2:CLS-FH-S1S21-tau0.5}%
Relative rRMSE in simulation model 2.2 (non-Gaussian error) at $\tau=0.5$ based on $\num{200}$ replications and $\num{50}$ bootstraps. Sample size n=1000.} 
\end{figure}

\begin{table}[htbp]
    \centering\caption{\label{tab:M2:S1S21:tau=0.2} Relative rRMSE in fixed-theta simulation model 2.2 at tau=0.2.}
            \sisetup{round-precision=2,round-mode=places}
    \begin{threeparttable}
    \begin{tabular}{ccScccScccSccccc}
    \toprule
    &                    && \multicolumn{2}{c}{$\mathrm{IVQR.2SLS}$}&               && \multicolumn{2}{c}{$\mathrm{IVQR.QR} $} &               &&  &  &   \\
    \cmidrule{4-5} \cmidrule{8-9}
    $\mathrm{DGP }$  &  $\mathrm{Endog}$  &&  $\mathrm{AVG} $ & $\mathrm{AGG}$       & $\mathrm{2SLS}$ && $\mathrm{AVG}$ & $\mathrm{AGG}$         & $\mathrm{QR}$ && $\mathrm{BS} $    & $\mathrm{SEE}$ & \\
    \midrule
1          & $   0$  &&   $\num[math-rm=\mathbf]{0.985965}$  & $\num{0.986148}$ & $\num{1.164131}$  && $\num[math-rm=\mathbf]{0.605033}$  &  $\num{0.461143}$ & $\num{0.443345}$ && $\num[math-rm=\mathbf]{0.635738}$ &  $\num{1.037887}$ & \\
2          & $0.05$  &&   $\num{1.032437}$  & $\num{1.055472}$ & $\num{1.366091}$  && $\num[math-rm=\mathbf]{0.695882}$  &  $\num{0.622655}$ & $\num{0.590196}$ && $\num[math-rm=\mathbf]{0.714712}$ &  $\num{1.195117}$ & \\
3          & $ 0.10$  &&   $\num[math-rm=\mathbf]{0.992888}$  & $\num{1.024126}$ & $\num{1.280595}$  && $\num[math-rm=\mathbf]{0.718035}$  &  $\num{0.791446}$ & $\num{0.756270}$ && $\num[math-rm=\mathbf]{0.786995}$ &  $\num{1.088795}$ & \\
4          & $0.15$  &&   $\num[math-rm=\mathbf]{0.984362}$  & $\num{1.003759}$ & $\num{1.211431}$  && $\num[math-rm=\mathbf]{0.823548}$  &  $\num{1.016474}$ & $\num{0.992095}$ && $\num[math-rm=\mathbf]{0.851530}$ &  $\num{1.092944}$ & \\
5          & $ 0.20$  &&   $\num[math-rm=\mathbf]{0.983763}$  & $\num{1.027150}$ & $\num{1.214644}$  && $\num[math-rm=\mathbf]{0.876411}$  &  $\num{1.221520}$ & $\num{1.183604}$ && $\num[math-rm=\mathbf]{0.928679}$ &  $\num{1.061601}$ & \\
6          & $0.25$  &&   $\num{1.019296}$  & $\num{1.034378}$ & $\num{1.302392}$  && $\num[math-rm=\mathbf]{0.958077}$  &  $\num{1.667269}$ & $\num{1.634471}$ && $\num[math-rm=\mathbf]{0.942775}$ &  $\num{1.155674}$ & \\
7          & $ 0.30$  &&   $\num{1.005048}$  & $\num{1.037079}$ & $\num{1.268411}$  && $\num[math-rm=\mathbf]{0.952193}$  &  $\num{1.949369}$ & $\num{1.908922}$ && $\num[math-rm=\mathbf]{0.938604}$ &  $\num{1.095906}$ & \\
8          & $0.35$  &&   $\num[math-rm=\mathbf]{0.977357}$  & $\num{1.004641}$ & $\num{1.235724}$  && $\num[math-rm=\mathbf]{0.959564}$  &  $\num{2.345582}$ & $\num{2.253235}$ && $\num[math-rm=\mathbf]{0.974943}$ &  $\num{1.089091}$ & \\
9          & $ 0.40$  &&   $\num[math-rm=\mathbf]{0.984137}$  & $\num{0.986364}$ & $\num{1.330826}$  && $\num[math-rm=\mathbf]{0.975818}$  &  $\num{2.718775}$ & $\num{2.604009}$ && $\num[math-rm=\mathbf]{0.985004}$ &  $\num{1.123513}$ & \\
\bottomrule
            \end{tabular}
            \begin{tablenotes}
            \item $\num{200}$ replications. $\num{50}$ bootstraps. Sample size is 1000.\\
            \end{tablenotes}
            \end{threeparttable}
            \end{table}

\begin{table}[htbp]
    \centering\caption{\label{tab:M2:S1S21:tau=0.3} Relative rRMSE in fixed-theta simulation model 2.2 at tau=0.3.}
            \sisetup{round-precision=2,round-mode=places}
    \begin{threeparttable}
    \begin{tabular}{ccScccScccSccccc}
    \toprule
    &                    && \multicolumn{2}{c}{$\mathrm{IVQR.2SLS}$}&               && \multicolumn{2}{c}{$\mathrm{IVQR.QR} $} &               &&  &  &   \\
    \cmidrule{4-5} \cmidrule{8-9}
    $\mathrm{DGP }$  &  $\mathrm{Endog}$  &&  $\mathrm{AVG} $ & $\mathrm{AGG}$       & $\mathrm{2SLS}$ && $\mathrm{AVG}$ & $\mathrm{AGG}$         & $\mathrm{QR}$ && $\mathrm{BS} $    & $\mathrm{SEE}$ & \\
    \midrule
1          & $   0$  &&   $\num[math-rm=\mathbf]{0.971663}$  & $\num{1.048237}$ & $\num{1.125891}$  && $\num[math-rm=\mathbf]{0.634262}$  &  $\num{0.456787}$ & $\num{0.442481}$ && $\num[math-rm=\mathbf]{0.631906}$ &  $\num{0.823319}$ & \\
2          & $0.05$  &&   $\num{1.003301}$  & $\num{1.053934}$ & $\num{1.109154}$  && $\num[math-rm=\mathbf]{0.686522}$  &  $\num{0.589211}$ & $\num{0.565426}$ && $\num[math-rm=\mathbf]{0.715400}$ &  $\num{0.908110}$ & \\
3          & $ 0.10$  &&   $\num[math-rm=\mathbf]{0.985405}$  & $\num{1.034346}$ & $\num{1.084259}$  && $\num[math-rm=\mathbf]{0.742147}$  &  $\num{0.730333}$ & $\num{0.722185}$ && $\num[math-rm=\mathbf]{0.753168}$ &  $\num{0.848187}$ & \\
4          & $0.15$  &&   $\num[math-rm=\mathbf]{0.984451}$  & $\num{1.041436}$ & $\num{1.114378}$  && $\num[math-rm=\mathbf]{0.850875}$  &  $\num{1.088469}$ & $\num{1.083617}$ && $\num[math-rm=\mathbf]{0.845325}$ &  $\num{0.848937}$ & \\
5          & $ 0.20$  &&   $\num{1.006103}$  & $\num{1.047927}$ & $\num{1.163181}$  && $\num[math-rm=\mathbf]{0.929411}$  &  $\num{1.488540}$ & $\num{1.487896}$ && $\num[math-rm=\mathbf]{0.937252}$ &  $\num{0.926147}$ & \\
6          & $0.25$  &&   $\num[math-rm=\mathbf]{0.998943}$  & $\num{1.062765}$ & $\num{1.113493}$  && $\num[math-rm=\mathbf]{0.936973}$  &  $\num{1.662090}$ & $\num{1.647114}$ && $\num[math-rm=\mathbf]{0.908682}$ &  $\num{0.840307}$ & \\
7          & $ 0.30$  &&   $\num[math-rm=\mathbf]{0.990134}$  & $\num{1.057446}$ & $\num{1.060950}$  && $\num[math-rm=\mathbf]{0.955772}$  &  $\num{1.996429}$ & $\num{1.952347}$ && $\num[math-rm=\mathbf]{0.960992}$ &  $\num{0.857752}$ & \\
8          & $0.35$  &&   $\num[math-rm=\mathbf]{0.988894}$  & $\num{1.071067}$ & $\num{1.201547}$  && $\num[math-rm=\mathbf]{0.974663}$  &  $\num{2.445913}$ & $\num{2.401229}$ && $\num{1.003178}$ &  $\num{0.893437}$ & \\
9          & $ 0.40$  &&   $\num[math-rm=\mathbf]{0.986867}$  & $\num{1.033834}$ & $\num{1.087412}$  && $\num[math-rm=\mathbf]{0.951769}$  &  $\num{2.698196}$ & $\num{2.610765}$ && $\num[math-rm=\mathbf]{0.978063}$ &  $\num{0.865174}$ & \\
\bottomrule
            \end{tabular}
            \begin{tablenotes}
            \item $\num{200}$ replications. $\num{50}$ bootstraps. Sample size is 1000.\\
            \end{tablenotes}
            \end{threeparttable}
            \end{table}

\begin{table}[htbp]
    \centering\caption{\label{tab:M2:S1S21:tau=0.4} Relative rRMSE in fixed-theta simulation model 2.2 at tau=0.4.}
            \sisetup{round-precision=2,round-mode=places}
    \begin{threeparttable}
    \begin{tabular}{ccScccScccSccccc}
    \toprule
    &                    && \multicolumn{2}{c}{$\mathrm{IVQR.2SLS}$}&               && \multicolumn{2}{c}{$\mathrm{IVQR.QR} $} &               &&  &  &   \\
    \cmidrule{4-5} \cmidrule{8-9}
    $\mathrm{DGP }$  &  $\mathrm{Endog}$  &&  $\mathrm{AVG} $ & $\mathrm{AGG}$       & $\mathrm{2SLS}$ && $\mathrm{AVG}$ & $\mathrm{AGG}$         & $\mathrm{QR}$ && $\mathrm{BS} $    & $\mathrm{SEE}$ & \\
    \midrule
1          & $   0$  &&   $\num[math-rm=\mathbf]{0.942432}$  & $\num{0.986965}$ & $\num{1.001899}$  && $\num[math-rm=\mathbf]{0.640468}$  &  $\num{0.469533}$ & $\num{0.457496}$ && $\num[math-rm=\mathbf]{0.613198}$ &  $\num{0.886612}$ & \\
2          & $0.05$  &&   $\num[math-rm=\mathbf]{0.975302}$  & $\num{1.056243}$ & $\num{1.012504}$  && $\num[math-rm=\mathbf]{0.685716}$  &  $\num{0.558249}$ & $\num{0.545467}$ && $\num[math-rm=\mathbf]{0.666865}$ &  $\num{0.870498}$ & \\
3          & $ 0.10$  &&   $\num[math-rm=\mathbf]{0.977430}$  & $\num{1.009010}$ & $\num{0.944167}$  && $\num[math-rm=\mathbf]{0.774485}$  &  $\num{0.805500}$ & $\num{0.784705}$ && $\num[math-rm=\mathbf]{0.749540}$ &  $\num{0.862592}$ & \\
4          & $0.15$  &&   $\num[math-rm=\mathbf]{0.983445}$  & $\num{1.102423}$ & $\num{1.035851}$  && $\num[math-rm=\mathbf]{0.855253}$  &  $\num{1.144311}$ & $\num{1.136162}$ && $\num[math-rm=\mathbf]{0.846393}$ &  $\num{0.822742}$ & \\
5          & $ 0.20$  &&   $\num[math-rm=\mathbf]{0.959831}$  & $\num{1.066556}$ & $\num{1.045529}$  && $\num[math-rm=\mathbf]{0.933007}$  &  $\num{1.473200}$ & $\num{1.458715}$ && $\num[math-rm=\mathbf]{0.908898}$ &  $\num{0.864985}$ & \\
6          & $0.25$  &&   $\num[math-rm=\mathbf]{0.983220}$  & $\num{1.022076}$ & $\num{0.959585}$  && $\num[math-rm=\mathbf]{0.979510}$  &  $\num{1.643607}$ & $\num{1.631658}$ && $\num[math-rm=\mathbf]{0.945834}$ &  $\num{0.843035}$ & \\
7          & $ 0.30$  &&   $\num[math-rm=\mathbf]{0.969065}$  & $\num{1.055938}$ & $\num{1.005546}$  && $\num[math-rm=\mathbf]{0.943129}$  &  $\num{2.039433}$ & $\num{2.007738}$ && $\num[math-rm=\mathbf]{0.930300}$ &  $\num{0.858890}$ & \\
8          & $0.35$  &&   $\num[math-rm=\mathbf]{0.992367}$  & $\num{1.043289}$ & $\num{1.030924}$  && $\num[math-rm=\mathbf]{0.954163}$  &  $\num{2.329011}$ & $\num{2.284743}$ && $\num[math-rm=\mathbf]{0.934160}$ &  $\num{0.840382}$ & \\
9          & $ 0.40$  &&   $\num[math-rm=\mathbf]{0.932130}$  & $\num{1.030064}$ & $\num{0.982342}$  && $\num[math-rm=\mathbf]{0.973219}$  &  $\num{2.809253}$ & $\num{2.718179}$ && $\num[math-rm=\mathbf]{0.921131}$ &  $\num{0.829017}$ & \\
\bottomrule
            \end{tabular}
            \begin{tablenotes}
            \item $\num{200}$ replications. $\num{50}$ bootstraps. Sample size is 1000.\\
            \end{tablenotes}
            \end{threeparttable}
            \end{table}

\begin{table}[htbp]
    \centering\caption{\label{tab:M2:S1S21:tau=0.6} Relative rRMSE in fixed-theta simulation model 2.2 at tau=0.6.}
            \sisetup{round-precision=2,round-mode=places}
    \begin{threeparttable}
    \begin{tabular}{ccScccScccSccccc}
    \toprule
    &                    && \multicolumn{2}{c}{$\mathrm{IVQR.2SLS}$}&               && \multicolumn{2}{c}{$\mathrm{IVQR.QR} $} &               &&  &  &   \\
    \cmidrule{4-5} \cmidrule{8-9}
    $\mathrm{DGP }$  &  $\mathrm{Endog}$  &&  $\mathrm{AVG} $ & $\mathrm{AGG}$       & $\mathrm{2SLS}$ && $\mathrm{AVG}$ & $\mathrm{AGG}$         & $\mathrm{QR}$ && $\mathrm{BS} $    & $\mathrm{SEE}$ & \\
    \midrule
1          & $   0$  &&   $\num[math-rm=\mathbf]{0.901977}$  & $\num{0.928575}$ & $\num{0.770761}$  && $\num[math-rm=\mathbf]{0.645696}$  &  $\num{0.477261}$ & $\num{0.463431}$ && $\num[math-rm=\mathbf]{0.587616}$ &  $\num{0.902581}$ & \\
2          & $0.05$  &&   $\num[math-rm=\mathbf]{0.848807}$  & $\num{0.847127}$ & $\num{0.691109}$  && $\num[math-rm=\mathbf]{0.641421}$  &  $\num{0.537181}$ & $\num{0.494978}$ && $\num[math-rm=\mathbf]{0.594821}$ &  $\num{0.787453}$ & \\
3          & $ 0.10$  &&   $\num[math-rm=\mathbf]{0.891112}$  & $\num{0.899207}$ & $\num{0.741558}$  && $\num[math-rm=\mathbf]{0.771001}$  &  $\num{0.837923}$ & $\num{0.842538}$ && $\num[math-rm=\mathbf]{0.704470}$ &  $\num{0.844147}$ & \\
4          & $0.15$  &&   $\num[math-rm=\mathbf]{0.899067}$  & $\num{0.905653}$ & $\num{0.797338}$  && $\num[math-rm=\mathbf]{0.846570}$  &  $\num{1.208408}$ & $\num{1.187941}$ && $\num[math-rm=\mathbf]{0.795798}$ &  $\num{0.899594}$ & \\
5          & $ 0.20$  &&   $\num[math-rm=\mathbf]{0.888571}$  & $\num{0.861603}$ & $\num{0.766086}$  && $\num[math-rm=\mathbf]{0.903913}$  &  $\num{1.404869}$ & $\num{1.389336}$ && $\num[math-rm=\mathbf]{0.839882}$ &  $\num{0.861214}$ & \\
6          & $0.25$  &&   $\num[math-rm=\mathbf]{0.920350}$  & $\num{0.946757}$ & $\num{0.783695}$  && $\num[math-rm=\mathbf]{0.954796}$  &  $\num{1.760076}$ & $\num{1.739580}$ && $\num[math-rm=\mathbf]{0.886149}$ &  $\num{0.861654}$ & \\
7          & $ 0.30$  &&   $\num[math-rm=\mathbf]{0.920773}$  & $\num{0.922389}$ & $\num{0.751853}$  && $\num[math-rm=\mathbf]{0.969283}$  &  $\num{2.079273}$ & $\num{2.040869}$ && $\num[math-rm=\mathbf]{0.878936}$ &  $\num{0.910682}$ & \\
8          & $0.35$  &&   $\num[math-rm=\mathbf]{0.857746}$  & $\num{0.819161}$ & $\num{0.741572}$  && $\num[math-rm=\mathbf]{0.939499}$  &  $\num{2.267172}$ & $\num{2.228388}$ && $\num[math-rm=\mathbf]{0.810939}$ &  $\num{0.807294}$ & \\
9          & $ 0.40$  &&   $\num[math-rm=\mathbf]{0.865908}$  & $\num{0.847995}$ & $\num{0.777066}$  && $\num{1.003632}$  &  $\num{2.692668}$ & $\num{2.630881}$ && $\num[math-rm=\mathbf]{0.869565}$ &  $\num{0.883217}$ & \\
\bottomrule
            \end{tabular}
            \begin{tablenotes}
            \item $\num{200}$ replications. $\num{50}$ bootstraps. Sample size is 1000.\\
            \end{tablenotes}
            \end{threeparttable}
            \end{table}

\begin{table}[htbp]
    \centering\caption{\label{tab:M2:S1S21:tau=0.7} Relative rRMSE in fixed-theta simulation model 2.2 at tau=0.7.}
            \sisetup{round-precision=2,round-mode=places}
    \begin{threeparttable}
    \begin{tabular}{ccScccScccSccccc}
    \toprule
    &                    && \multicolumn{2}{c}{$\mathrm{IVQR.2SLS}$}&               && \multicolumn{2}{c}{$\mathrm{IVQR.QR} $} &               &&  &  &   \\
    \cmidrule{4-5} \cmidrule{8-9}
    $\mathrm{DGP }$  &  $\mathrm{Endog}$  &&  $\mathrm{AVG} $ & $\mathrm{AGG}$       & $\mathrm{2SLS}$ && $\mathrm{AVG}$ & $\mathrm{AGG}$         & $\mathrm{QR}$ && $\mathrm{BS} $    & $\mathrm{SEE}$ & \\
    \midrule
1          & $   0$  &&   $\num[math-rm=\mathbf]{0.818808}$  & $\num{0.774343}$ & $\num{0.609321}$  && $\num[math-rm=\mathbf]{0.625038}$  &  $\num{0.450636}$ & $\num{0.431199}$ && $\num[math-rm=\mathbf]{0.564420}$ &  $\num{0.852658}$ & \\
2          & $0.05$  &&   $\num[math-rm=\mathbf]{0.880588}$  & $\num{0.845390}$ & $\num{0.690303}$  && $\num[math-rm=\mathbf]{0.694359}$  &  $\num{0.564080}$ & $\num{0.565038}$ && $\num[math-rm=\mathbf]{0.629387}$ &  $\num{0.928183}$ & \\
3          & $ 0.10$  &&   $\num[math-rm=\mathbf]{0.816081}$  & $\num{0.748445}$ & $\num{0.633686}$  && $\num[math-rm=\mathbf]{0.773877}$  &  $\num{0.761339}$ & $\num{0.757124}$ && $\num[math-rm=\mathbf]{0.694770}$ &  $\num{0.886899}$ & \\
4          & $0.15$  &&   $\num[math-rm=\mathbf]{0.803565}$  & $\num{0.772054}$ & $\num{0.628178}$  && $\num[math-rm=\mathbf]{0.808699}$  &  $\num{0.995587}$ & $\num{0.984468}$ && $\num[math-rm=\mathbf]{0.695331}$ &  $\num{0.865434}$ & \\
5          & $ 0.20$  &&   $\num[math-rm=\mathbf]{0.818873}$  & $\num{0.815652}$ & $\num{0.637650}$  && $\num[math-rm=\mathbf]{0.897051}$  &  $\num{1.382942}$ & $\num{1.363840}$ && $\num[math-rm=\mathbf]{0.779155}$ &  $\num{0.888535}$ & \\
6          & $0.25$  &&   $\num[math-rm=\mathbf]{0.846026}$  & $\num{0.825192}$ & $\num{0.683328}$  && $\num[math-rm=\mathbf]{0.962380}$  &  $\num{1.769683}$ & $\num{1.737131}$ && $\num[math-rm=\mathbf]{0.778568}$ &  $\num{0.860926}$ & \\
7          & $ 0.30$  &&   $\num[math-rm=\mathbf]{0.844457}$  & $\num{0.843036}$ & $\num{0.656264}$  && $\num[math-rm=\mathbf]{0.991065}$  &  $\num{2.119637}$ & $\num{2.068324}$ && $\num[math-rm=\mathbf]{0.853562}$ &  $\num{0.919079}$ & \\
8          & $0.35$  &&   $\num[math-rm=\mathbf]{0.790992}$  & $\num{0.744770}$ & $\num{0.604874}$  && $\num[math-rm=\mathbf]{0.937999}$  &  $\num{2.031738}$ & $\num{1.987206}$ && $\num[math-rm=\mathbf]{0.752715}$ &  $\num{0.868969}$ & \\
9          & $ 0.40$  &&   $\num[math-rm=\mathbf]{0.840503}$  & $\num{0.806064}$ & $\num{0.646829}$  && $\num[math-rm=\mathbf]{0.963600}$  &  $\num{2.498278}$ & $\num{2.419822}$ && $\num[math-rm=\mathbf]{0.774546}$ &  $\num{0.900748}$ & \\
\bottomrule
            \end{tabular}
            \begin{tablenotes}
            \item $\num{200}$ replications. $\num{50}$ bootstraps. Sample size is 1000.\\
            \end{tablenotes}
            \end{threeparttable}
            \end{table}

\begin{table}[htbp]
    \centering\caption{\label{tab:M2:S1S21:tau=0.8} Relative rRMSE in fixed-theta simulation model 2.2 at tau=0.8.}
            \sisetup{round-precision=2,round-mode=places}
    \begin{threeparttable}
    \begin{tabular}{ccScccScccSccccc}
    \toprule
    &                    && \multicolumn{2}{c}{$\mathrm{IVQR.2SLS}$}&               && \multicolumn{2}{c}{$\mathrm{IVQR.QR} $} &               &&  &  &   \\
    \cmidrule{4-5} \cmidrule{8-9}
    $\mathrm{DGP }$  &  $\mathrm{Endog}$  &&  $\mathrm{AVG} $ & $\mathrm{AGG}$       & $\mathrm{2SLS}$ && $\mathrm{AVG}$ & $\mathrm{AGG}$         & $\mathrm{QR}$ && $\mathrm{BS} $    & $\mathrm{SEE}$ & \\
    \midrule
1          & $   0$  &&   $\num[math-rm=\mathbf]{0.821933}$  & $\num{0.733994}$ & $\num{0.509435}$  && $\num[math-rm=\mathbf]{0.643116}$  &  $\num{0.472628}$ & $\num{0.436924}$ && $\num[math-rm=\mathbf]{0.519343}$ &  $\num{0.563281}$ & \\
2          & $0.05$  &&   $\num[math-rm=\mathbf]{0.802569}$  & $\num{0.719385}$ & $\num{0.525693}$  && $\num[math-rm=\mathbf]{0.665014}$  &  $\num{0.530699}$ & $\num{0.520795}$ && $\num[math-rm=\mathbf]{0.559490}$ &  $\num{0.581545}$ & \\
3          & $ 0.10$  &&   $\num[math-rm=\mathbf]{0.833646}$  & $\num{0.768216}$ & $\num{0.526505}$  && $\num[math-rm=\mathbf]{0.768155}$  &  $\num{0.802893}$ & $\num{0.768267}$ && $\num[math-rm=\mathbf]{0.642832}$ &  $\num{0.591115}$ & \\
4          & $0.15$  &&   $\num[math-rm=\mathbf]{0.832923}$  & $\num{0.756548}$ & $\num{0.529786}$  && $\num[math-rm=\mathbf]{0.838635}$  &  $\num{1.011841}$ & $\num{1.002279}$ && $\num[math-rm=\mathbf]{0.676322}$ &  $\num{0.584173}$ & \\
5          & $ 0.20$  &&   $\num[math-rm=\mathbf]{0.813924}$  & $\num{0.725903}$ & $\num{0.509276}$  && $\num[math-rm=\mathbf]{0.866982}$  &  $\num{1.284544}$ & $\num{1.251119}$ && $\num[math-rm=\mathbf]{0.654457}$ & $\num{0.571413}$ & \\
6          & $0.25$  &&   $\num[math-rm=\mathbf]{0.840842}$  & $\num{0.764893}$ & $\num{0.508802}$  && $\num[math-rm=\mathbf]{0.924171}$  &  $\num{1.464504}$ & $\num{1.427234}$ && $\num[math-rm=\mathbf]{0.720800}$ &  $\num{0.560036}$ & \\
7          & $ 0.30$  &&   $\num[math-rm=\mathbf]{0.801013}$  & $\num{0.727230}$ & $\num{0.486811}$  && $\num[math-rm=\mathbf]{0.937894}$  &  $\num{1.698238}$ & $\num{1.669410}$ && $\num[math-rm=\mathbf]{0.680512}$ &  $\num{0.548156}$ & \\
8          & $0.35$  &&   $\num[math-rm=\mathbf]{0.835048}$  & $\num{0.757762}$ & $\num{0.530311}$  && $\num{1.001051}$  &  $\num{2.084615}$ & $\num{2.027790}$ && $\num[math-rm=\mathbf]{0.728396}$ &  $\num{0.587661}$ & \\
9          & $ 0.40$  &&   $\num[math-rm=\mathbf]{0.839215}$  & $\num{0.736296}$ & $\num{0.541051}$  && $\num[math-rm=\mathbf]{0.985016}$  &  $\num{2.264286}$ & $\num{2.168639}$ && $\num[math-rm=\mathbf]{0.689692}$ &  $\num{0.594530}$ & \\
\bottomrule
            \end{tabular}
            \begin{tablenotes}
            \item $\num{200}$ replications. $\num{50}$ bootstraps. Sample size is 1000.\\
            \end{tablenotes}
            \end{threeparttable}
            \end{table}

\begin{figure}[htbp]
\includegraphics[width=0.45\textwidth, height=0.3\textheight, trim=35 20 20 70]{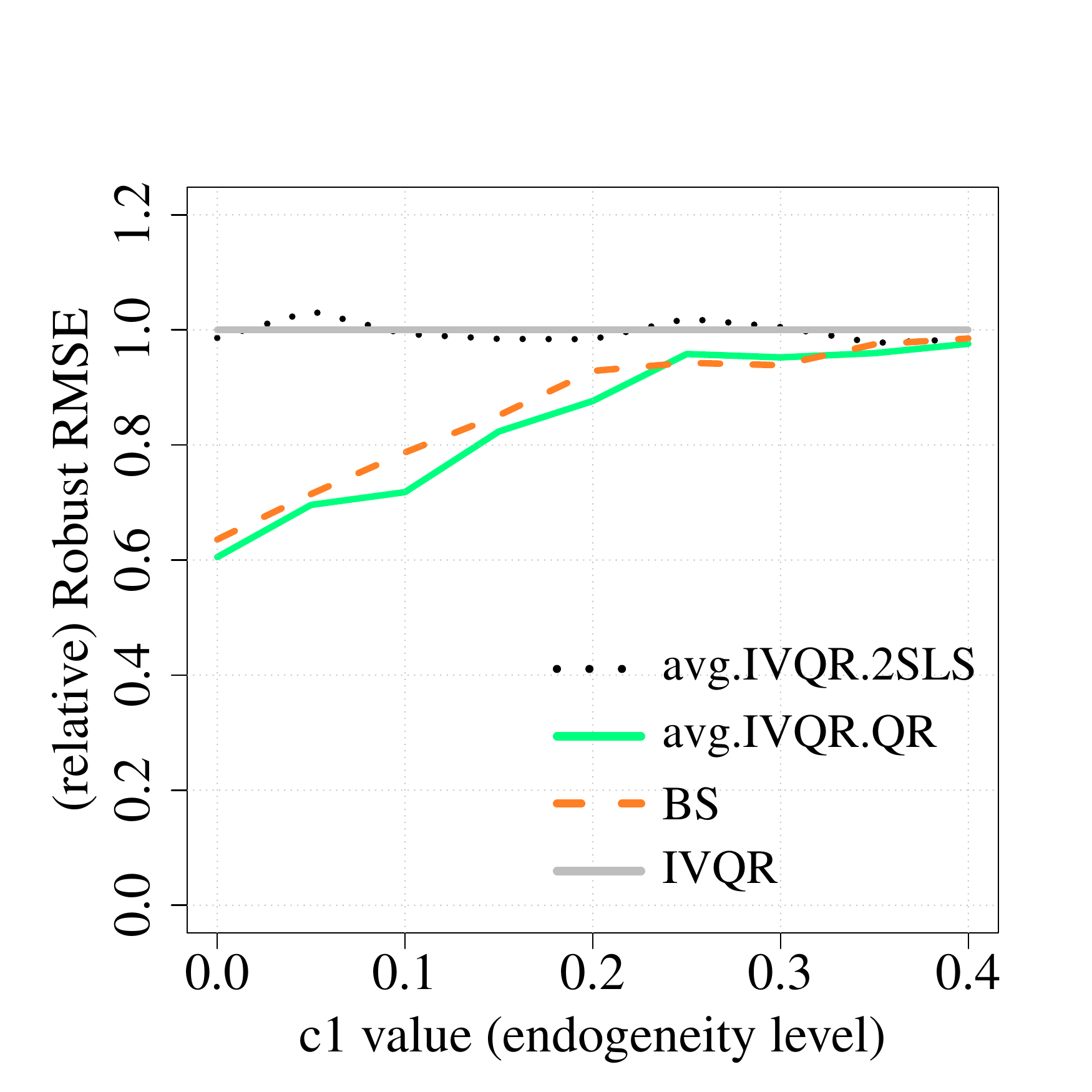}
\hfill
\includegraphics[width=0.45\textwidth, height=0.3\textheight, trim=35 20 20 70]{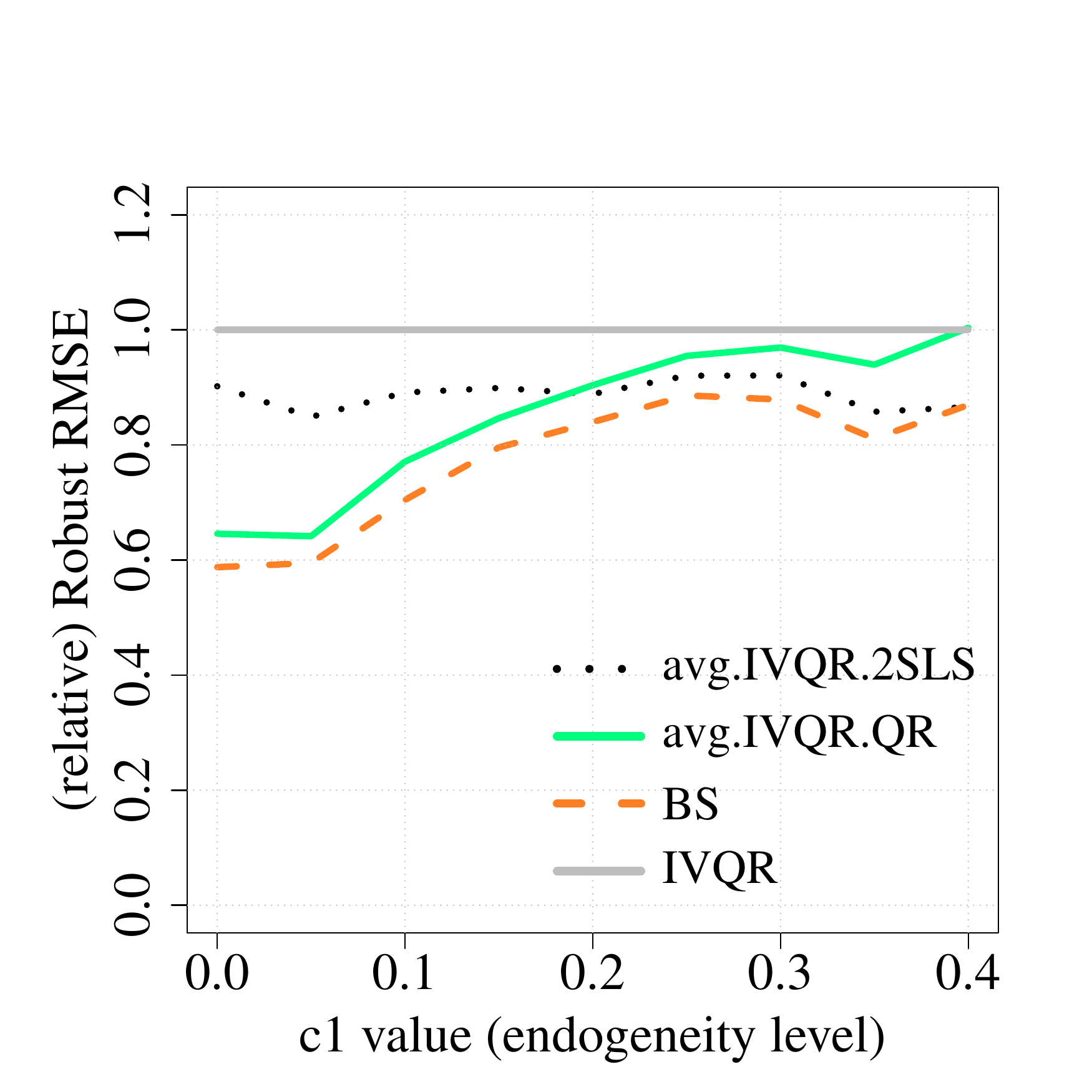}
\hfill\null
\includegraphics[width=0.45\textwidth, height=0.3\textheight, trim=35 20 20 70]{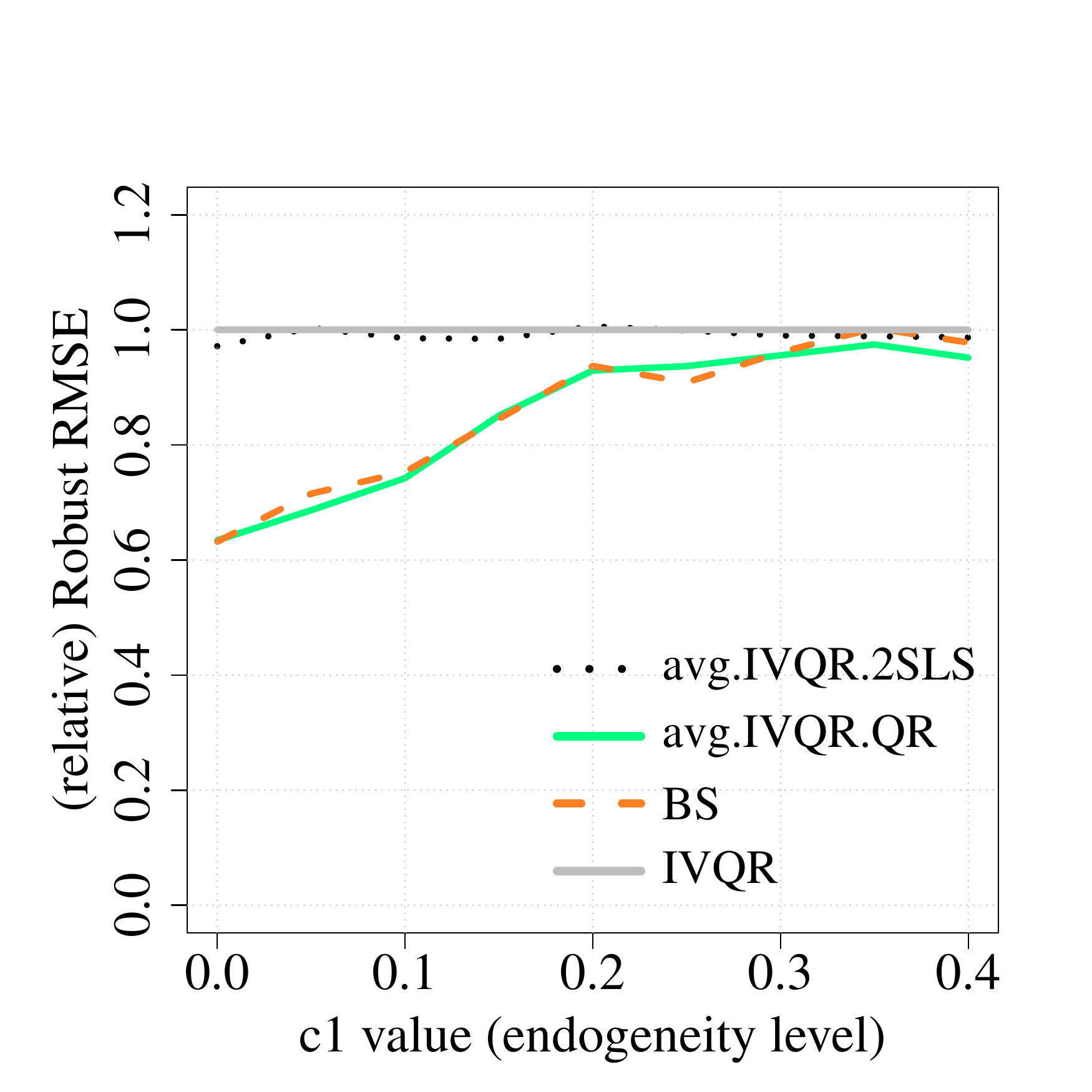}
\hfill
\includegraphics[width=0.45\textwidth, height=0.3\textheight, trim=35 20 20 70 ]{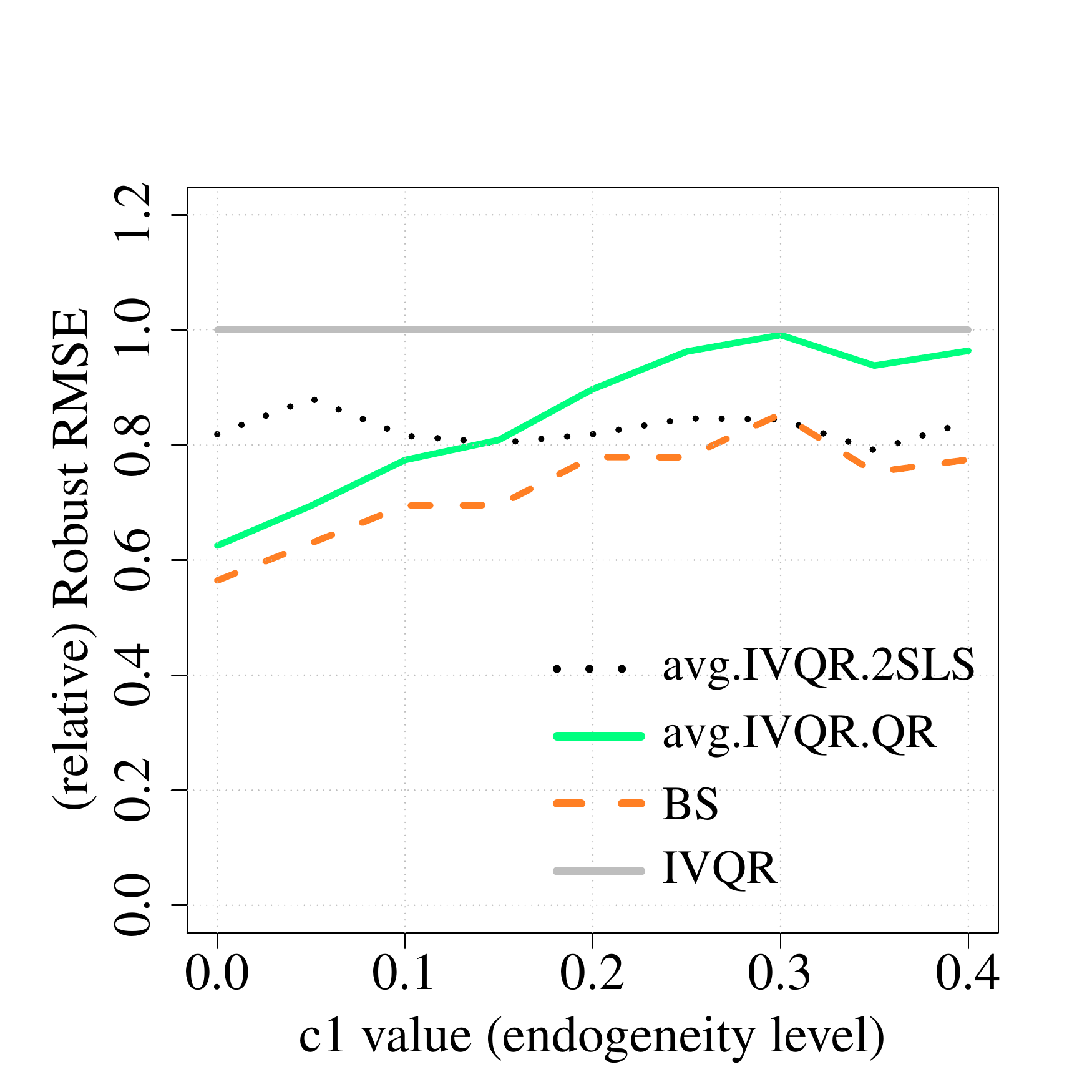}
\hfill\null
\includegraphics[width=0.45\textwidth, height=0.3\textheight, trim=35 20 20 70 ]{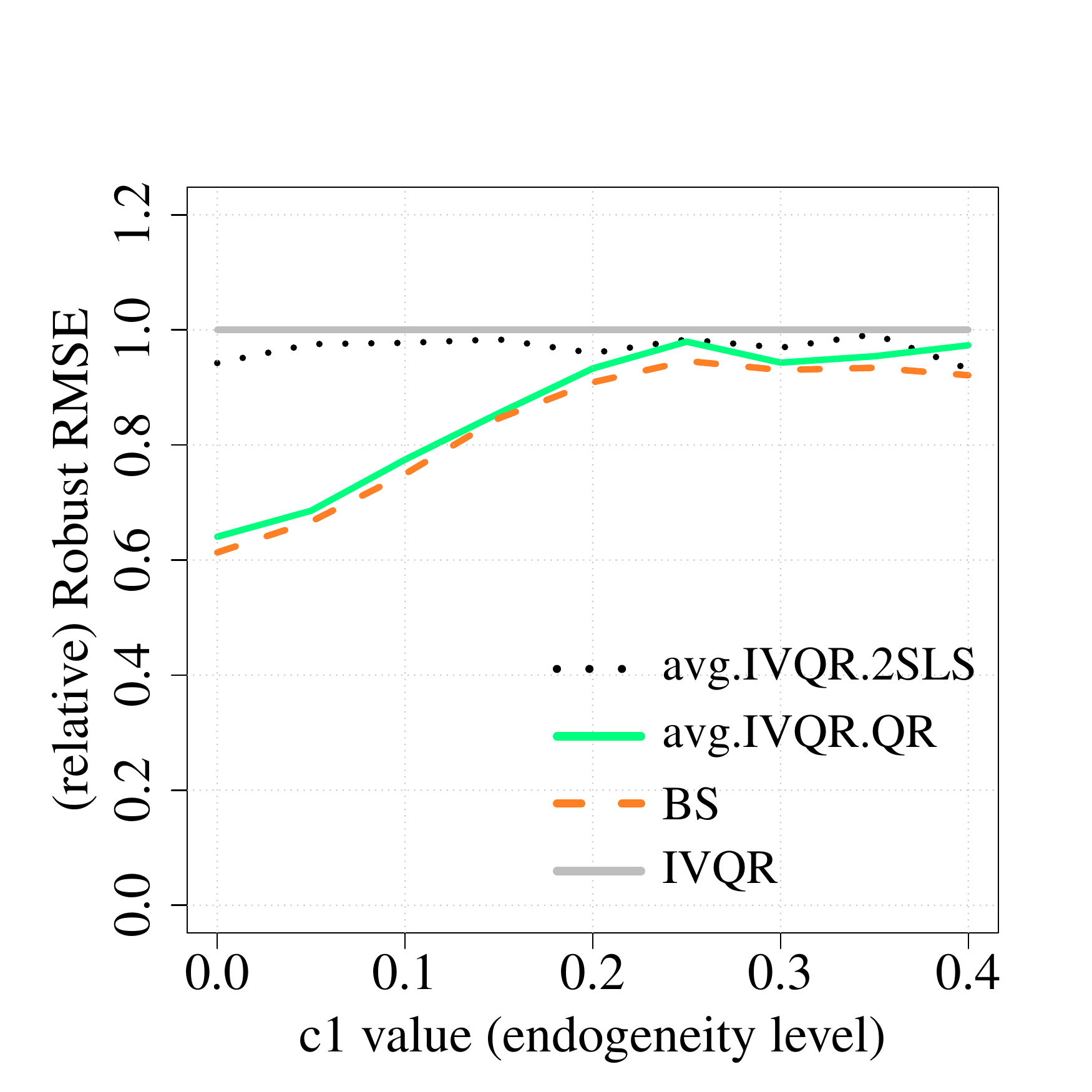}
\hfill
\includegraphics[width=0.45\textwidth, height=0.3\textheight, trim=35 20 20 70]{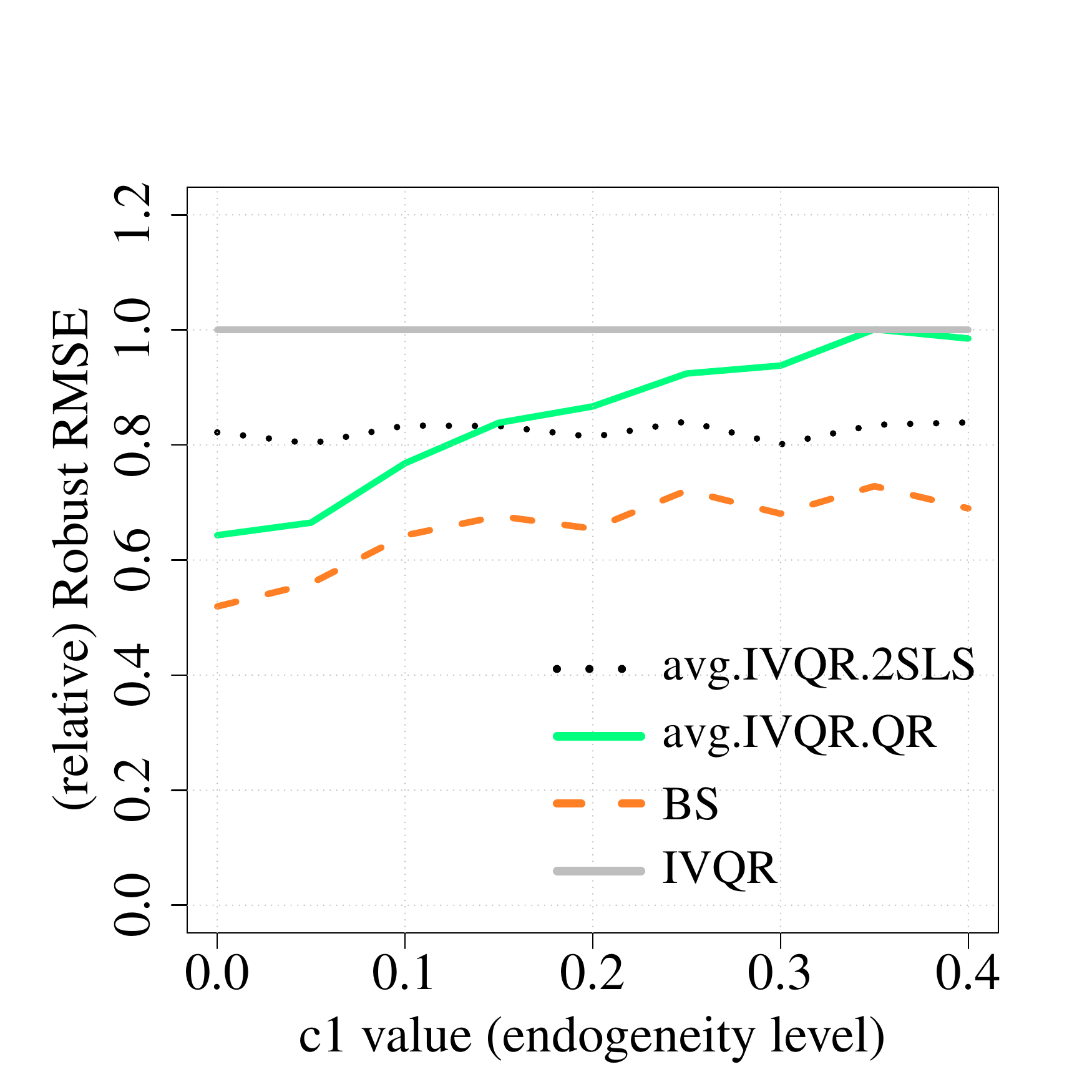}\\
\caption{\label{fig:M2:S1S21:compare_tau}%
Relative rRMSE in simulation model 2.2 (non-Gaussian error) at $\tau=0.2$ (left up), $\tau=0.3$ (left middle), $\tau=0.4$ (left bottom), $\tau=0.6$ (right up), $\tau=0.7$ (right middle),$\tau=0.8$ (right bottom), based on $\num{200}$ replications and $\num{50}$ bootstraps. Sample size n=1000.} 
\end{figure}



\subsection{Simulation model 3 results at different quantile levels}
\label{appdx:M3}

Additional tables corresponding to the figures in \cref{appdx:M3} are available in the supplemental appendix. 

\begin{table}[htbp!]
    \centering\caption{\label{tab:bound:S3:vary hetero} The lower and upper bounds of relative rRMSE in simulation model 3 at three fixed endogeneity and varying heterogeneity.}
            \sisetup{round-precision=3,round-mode=places}
    \begin{threeparttable}
    \begin{tabular}{cccScccScccSccc}
    \toprule
    &&          && \multicolumn{8}{c}{$\mathrm{Fixed \ Endog }$}               \\
    \cmidrule{5-12}
    &&          && \multicolumn{2}{c}{$\mathrm{No}\ (c=0)$}              && \multicolumn{2}{c}{$\mathrm{Some}\ (c=0.2) $}     && \multicolumn{2}{c}{$\mathrm{Much}\ (c=0.4) $}    \\
     \cmidrule{5-6} \cmidrule{8-9} \cmidrule{11-12}
    &&          && $\mathrm{Lower}$   & $\mathrm{Upper}$   &&   $\mathrm{Lower}$  & $\mathrm{Upper}$    && $\mathrm{Lower}$  & $\mathrm{Upper}$            &&    \\
    \cmidrule{3-12}
      && $\hat{\theta}_{\mathrm{AVG.2SLS}}$ &&  $\num{0.969577}$  & $\num{1.021010}$  && $\num{0.942747}$  &  $\num{1.044357}$ && $\num{0.958864}$ & $\num{1.032918}$  \\
$\tau=0.2$ && $\hat{\theta}_{\mathrm{AVG.QR}}$ &&  $\num{0.821836}$  & $\num{0.931055}$  && $\num{0.993061}$  &  $\num{1.018057}$ && $\num{0.988507}$ & $\num{1.001809}$  \\
      && $\hat{\theta}_{\mathrm{BS}}$ &&  $\num{0.795371}$  & $\num{1.065354}$  && $\num{0.789777}$  &  $\num{1.074997}$ && $\num{0.814863}$ & $\num{1.072128}$  \\
    \cmidrule{3-12}
      && $\hat{\theta}_{\mathrm{AVG.2SLS}}$ &&  $\num{0.967765}$  & $\num{1.021182}$  && $\num{0.935481}$  &  $\num{1.007504}$ && $\num{0.955472}$ & $\num{1.025268}$  \\
$\tau=0.3$ && $\hat{\theta}_{\mathrm{AVG.QR}}$ &&  $\num{0.823643}$  & $\num{0.889685}$  && $\num{0.997770}$  &  $\num{1.006209}$ && $\num{0.990556}$ & $\num{1.000673}$  \\
      && $\hat{\theta}_{\mathrm{BS}}$ &&  $\num{0.822331}$  & $\num{1.019420}$  && $\num{0.786762}$  &  $\num{1.037749}$ && $\num{0.816520}$ & $\num{1.062801}$  \\
    \cmidrule{3-12}
      && $\hat{\theta}_{\mathrm{AVG.2SLS}}$ &&  $\num{0.986476}$  & $\num{1.011116}$  && $\num{0.949473}$  &  $\num{1.020236}$ && $\num{0.960573}$ & $\num{1.011081}$  \\
$\tau=0.4$ && $\hat{\theta}_{\mathrm{AVG.QR}}$ &&  $\num{0.811193}$  & $\num{0.892960}$  && $\num{0.996253}$  &  $\num{1.004808}$ && $\num{0.993039}$ & $\num{1.001408}$  \\
      && $\hat{\theta}_{\mathrm{BS}}$ &&  $\num{0.795891}$  & $\num{0.989374}$  && $\num{0.814171}$  &  $\num{1.053954}$ && $\num{0.879285}$ & $\num{1.051338}$  \\
    \cmidrule{3-12}
      && $\hat{\theta}_{\mathrm{AVG.2SLS}}$ &&  $\num{0.921558}$  & $\num{1.006412}$  && $\num{0.960113}$  &  $\num{1.024411}$ && $\num{0.958913}$ & $\num{0.996904}$  \\
$\tau=0.5$ && $\hat{\theta}_{\mathrm{AVG.QR}}$ &&  $\num{0.830799}$  & $\num{0.919396}$  && $\num{0.996671}$  &  $\num{1.002288}$ && $\num{0.998377}$ & $\num{1.001428}$  \\
      && $\hat{\theta}_{\mathrm{BS}}$ &&  $\num{0.769622}$  & $\num{0.929397}$  && $\num{0.850083}$  &  $\num{1.014216}$ && $\num{0.859666}$ & $\num{1.006431}$  \\
    \cmidrule{3-12}
      && $\hat{\theta}_{\mathrm{AVG.2SLS}}$ &&  $\num{0.941824}$  & $\num{0.999079}$  && $\num{0.925064}$  &  $\num{0.994143}$ && $\num{0.954774}$ & $\num{1.015155}$  \\
$\tau=0.6$ && $\hat{\theta}_{\mathrm{AVG.QR}}$ &&  $\num{0.842839}$  & $\num{0.928874}$  && $\num{0.992314}$  &  $\num{1.000468}$ && $\num{0.995785}$ & $\num{1.000706}$  \\
      && $\hat{\theta}_{\mathrm{BS}}$ &&  $\num{0.792263}$  & $\num{0.879152}$  && $\num{0.797588}$  &  $\num{0.972436}$ && $\num{0.826679}$ & $\num{0.977792}$  \\
    \cmidrule{3-12}
      && $\hat{\theta}_{\mathrm{AVG.2SLS}}$ &&  $\num{0.903325}$  & $\num{0.984811}$  && $\num{0.913862}$  &  $\num{0.973105}$ && $\num{0.900183}$ & $\num{1.001527}$  \\
$\tau=0.7$ && $\hat{\theta}_{\mathrm{AVG.QR}}$ &&  $\num{0.862907}$  & $\num{0.948820}$  && $\num{0.988170}$  &  $\num{0.995181}$ && $\num{0.997861}$ & $\num{1.000398}$  \\
      && $\hat{\theta}_{\mathrm{BS}}$ &&  $\num{0.794074}$  & $\num{0.898076}$  && $\num{0.665717}$  &  $\num{0.841015}$ && $\num{0.743293}$ & $\num{0.842545}$  \\
    \cmidrule{3-12}
      && $\hat{\theta}_{\mathrm{AVG.2SLS}}$ &&  $\num{0.927831}$  & $\num{0.976097}$  && $\num{0.919111}$  &  $\num{0.982062}$ && $\num{0.920880}$ & $\num{0.983158}$  \\
$\tau=0.8$ && $\hat{\theta}_{\mathrm{AVG.QR}}$ &&  $\num{0.877183}$  & $\num{0.977167}$  && $\num{0.989498}$  &  $\num{0.995521}$ && $\num{0.997245}$ & $\num{0.999976}$  \\
      && $\hat{\theta}_{\mathrm{BS}}$ &&  $\num{0.751498}$  & $\num{0.942940}$  && $\num{0.594733}$  &  $\num{0.822040}$ && $\num{0.659582}$ & $\num{0.797359}$  \\

      \bottomrule
            \end{tabular}
            \begin{tablenotes}
            \item $\num{200}$ replications. $\num{50}$ bootstraps. Sample size is 1000.\\
            \end{tablenotes}
            \end{threeparttable}
            \end{table}

\begin{table}[htbp]
      \centering\caption{\label{tab:bound:S3:vary endo} The lower and upper bounds of relative rRMSE in simulation model 3 at three fixed heterogeneity and varying endogeneity.}
              \sisetup{round-precision=3,round-mode=places}
      \begin{threeparttable}
      \begin{tabular}{cccScccScccSccc}
      \toprule
      &&          && \multicolumn{8}{c}{$\mathrm{Fixed \ Hetero }$}               \\
      \cmidrule{5-12}
      &&          && \multicolumn{2}{c}{$\mathrm{No}\ (\mathrm{hetero}=0)$}              && \multicolumn{2}{c}{$\mathrm{Some}\ (\mathrm{hetero}=0.5) $}     && \multicolumn{2}{c}{$\mathrm{Much}\ (\mathrm{hetero}=1) $}    \\
      \cmidrule{5-6} \cmidrule{8-9} \cmidrule{11-12}
      &&          && $\mathrm{Lower}$   & $\mathrm{Upper}$   &&   $\mathrm{Lower}$  & $\mathrm{Upper}$    && $\mathrm{Lower}$  & $\mathrm{Upper}$            &&    \\
      \cmidrule{3-12}
      && $\hat{\theta}_{\mathrm{AVG.2SLS}}$ &&  $\num{0.940338}$  & $\num{0.989084}$  && $\num{0.969949}$  &  $\num{1.023346}$ && $\num{0.997005}$ & $\num{1.034830}$  \\
$\tau=0.2$ && $\hat{\theta}_{\mathrm{AVG.QR}}$ &&  $\num{0.860337}$  & $\num{1.002612}$  && $\num{0.841141}$  &  $\num{1.012564}$ && $\num{0.866528}$ & $\num{1.013626}$  \\
      && $\hat{\theta}_{\mathrm{BS}}$ &&  $\num{0.776494}$  & $\num{0.824634}$  && $\num{1.008582}$  &  $\num{1.075296}$ && $\num{1.012550}$ & $\num{1.069652}$  \\
      \cmidrule{3-12}
      && $\hat{\theta}_{\mathrm{AVG.2SLS}}$ &&  $\num{0.930176}$  & $\num{1.001961}$  && $\num{0.978666}$  &  $\num{1.006226}$ && $\num{0.981683}$ & $\num{1.026054}$  \\
$\tau=0.3$ && $\hat{\theta}_{\mathrm{AVG.QR}}$ &&  $\num{0.849978}$  & $\num{0.999664}$  && $\num{0.834134}$  &  $\num{1.004380}$ && $\num{0.834125}$ & $\num{1.002484}$  \\
      && $\hat{\theta}_{\mathrm{BS}}$ &&  $\num{0.784113}$  & $\num{0.858217}$  && $\num{0.998901}$  &  $\num{1.054327}$ && $\num{0.970440}$ & $\num{1.031315}$  \\
      \cmidrule{3-12}
      && $\hat{\theta}_{\mathrm{AVG.2SLS}}$ &&  $\num{0.919255}$  & $\num{1.011570}$  && $\num{0.976497}$  &  $\num{1.020466}$ && $\num{0.985229}$ & $\num{1.017733}$  \\
$\tau=0.4$ && $\hat{\theta}_{\mathrm{AVG.QR}}$ &&  $\num{0.841303}$  & $\num{1.009771}$  && $\num{0.845490}$  &  $\num{1.003508}$ && $\num{0.880731}$ & $\num{1.002516}$  \\
      && $\hat{\theta}_{\mathrm{BS}}$ &&  $\num{0.804544}$  & $\num{0.940342}$  && $\num{0.963760}$  &  $\num{1.032182}$ && $\num{0.978236}$ & $\num{1.033021}$  \\
      \cmidrule{3-12}
      && $\hat{\theta}_{\mathrm{AVG.2SLS}}$ &&  $\num{0.906251}$  & $\num{0.991062}$  && $\num{0.973068}$  &  $\num{1.021872}$ && $\num{0.981040}$ & $\num{1.015973}$  \\
$\tau=0.5$ && $\hat{\theta}_{\mathrm{AVG.QR}}$ &&  $\num{0.857462}$  & $\num{1.000541}$  && $\num{0.875519}$  &  $\num{1.001435}$ && $\num{0.888192}$ & $\num{1.001399}$  \\
      && $\hat{\theta}_{\mathrm{BS}}$ &&  $\num{0.802951}$  & $\num{0.883027}$  && $\num{0.877296}$  &  $\num{1.003506}$ && $\num{0.886404}$ & $\num{1.018404}$  \\
      \cmidrule{3-12}
      && $\hat{\theta}_{\mathrm{AVG.2SLS}}$ &&  $\num{0.889529}$  & $\num{0.996740}$  && $\num{0.965629}$  &  $\num{1.003244}$ && $\num{0.962025}$ & $\num{1.015973}$  \\
$\tau=0.6$ && $\hat{\theta}_{\mathrm{AVG.QR}}$ &&  $\num{0.842899}$  & $\num{0.999995}$  && $\num{0.907748}$  &  $\num{1.000420}$ && $\num{0.895110}$ & $\num{1.000699}$  \\
      && $\hat{\theta}_{\mathrm{BS}}$ &&  $\num{0.790830}$  & $\num{0.894048}$  && $\num{0.845967}$  &  $\num{0.975165}$ && $\num{0.751630}$ & $\num{0.954729}$  \\
      \cmidrule{3-12}
      && $\hat{\theta}_{\mathrm{AVG.2SLS}}$ &&  $\num{0.937520}$  & $\num{0.992026}$  && $\num{0.895892}$  &  $\num{0.959120}$ && $\num{0.961816}$ & $\num{1.004971}$  \\
$\tau=0.7$ && $\hat{\theta}_{\mathrm{AVG.QR}}$ &&  $\num{0.843241}$  & $\num{0.999671}$  && $\num{0.912638}$  &  $\num{0.999054}$ && $\num{0.929948}$ & $\num{0.999062}$  \\
      && $\hat{\theta}_{\mathrm{BS}}$ &&  $\num{0.801435}$  & $\num{0.882384}$  && $\num{0.682959}$  &  $\num{0.851977}$ && $\num{0.647930}$ & $\num{0.879719}$  \\
      \cmidrule{3-12}
      && $\hat{\theta}_{\mathrm{AVG.2SLS}}$ &&  $\num{0.948435}$  & $\num{1.005422}$  && $\num{0.930249}$  &  $\num{0.954654}$ && $\num{0.955824}$ & $\num{0.989458}$  \\
$\tau=0.8$ && $\hat{\theta}_{\mathrm{AVG.QR}}$ &&  $\num{0.868048}$  & $\num{1.000797}$  && $\num{0.969205}$  &  $\num{1.000430}$ && $\num{0.976276}$ & $\num{0.999664}$  \\
      && $\hat{\theta}_{\mathrm{BS}}$ &&  $\num{0.766116}$  & $\num{0.830854}$  && $\num{0.662789}$  &  $\num{0.907786}$ && $\num{0.573559}$ & $\num{0.943787}$  \\

\bottomrule
              \end{tabular}
              \begin{tablenotes}
              \item $\num{200}$ replications. $\num{50}$ bootstraps. Sample size is 1000.\\
              \end{tablenotes}
              \end{threeparttable}
              \end{table}

\begin{figure}[htbp]
\includegraphics[width=0.45\textwidth, height=0.3\textheight, trim=35 20 20 70]{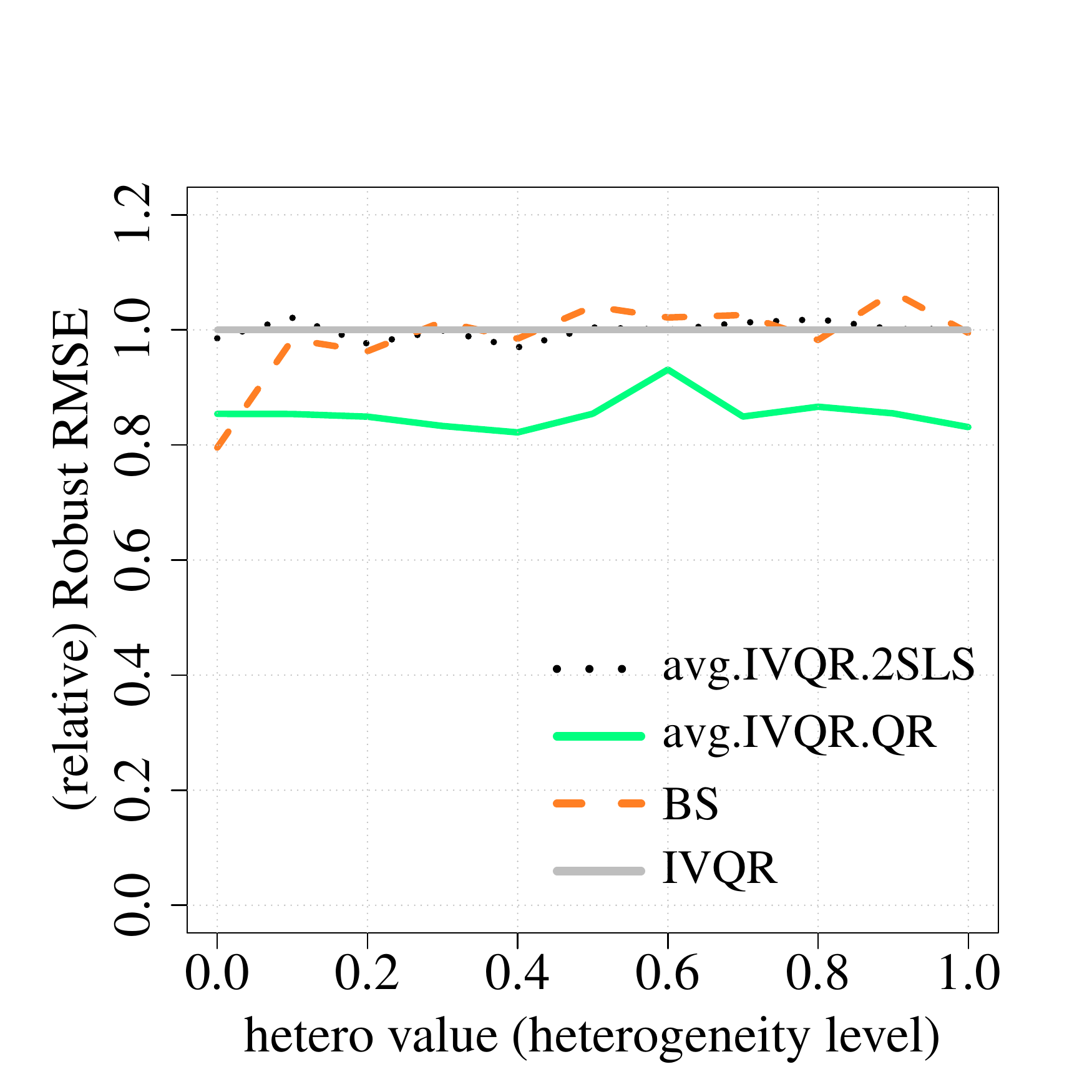}
\hfill
\includegraphics[width=0.45\textwidth, height=0.3\textheight, trim=35 20 20 70]{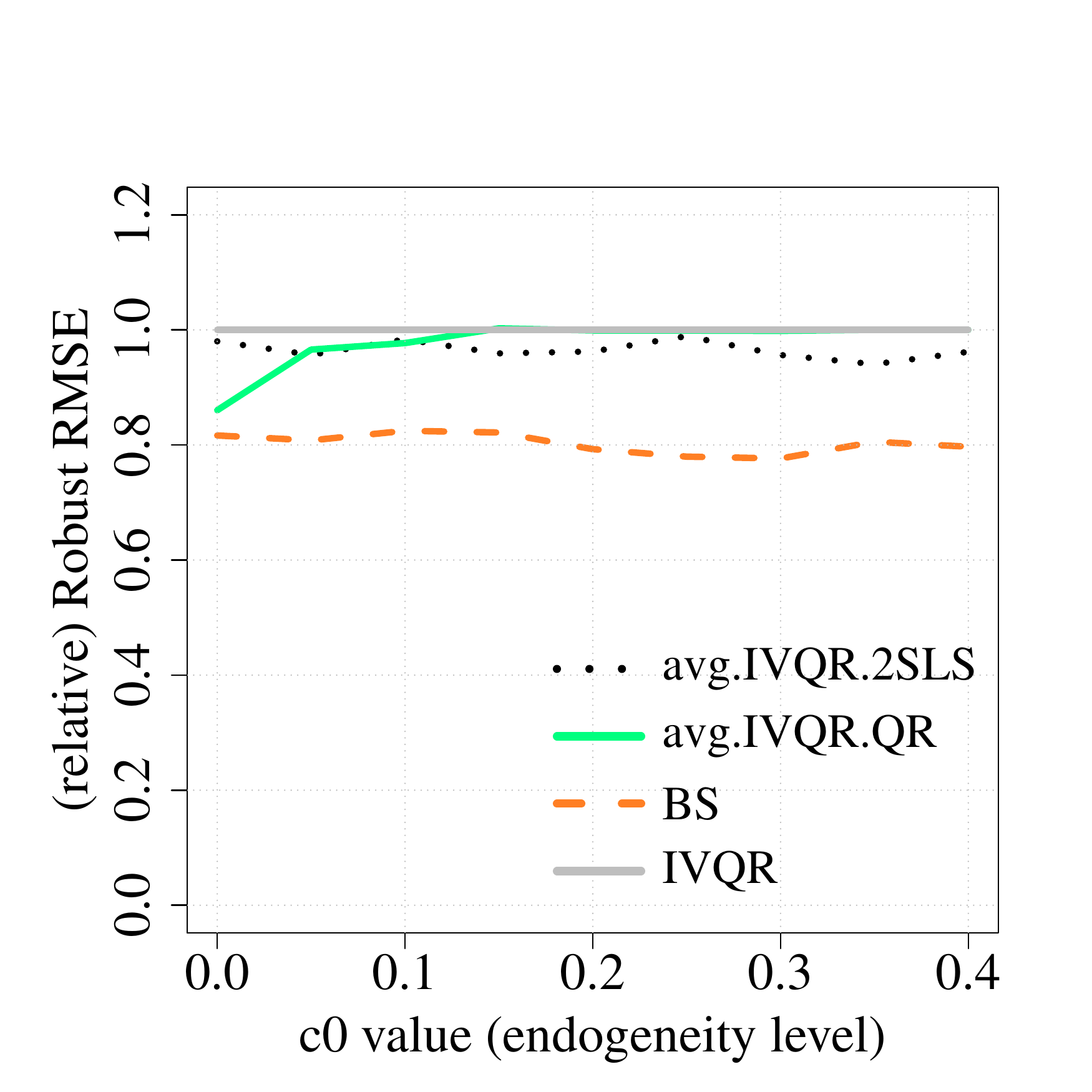}
\hfill\null
\includegraphics[width=0.45\textwidth, height=0.3\textheight, trim=35 20 20 70]{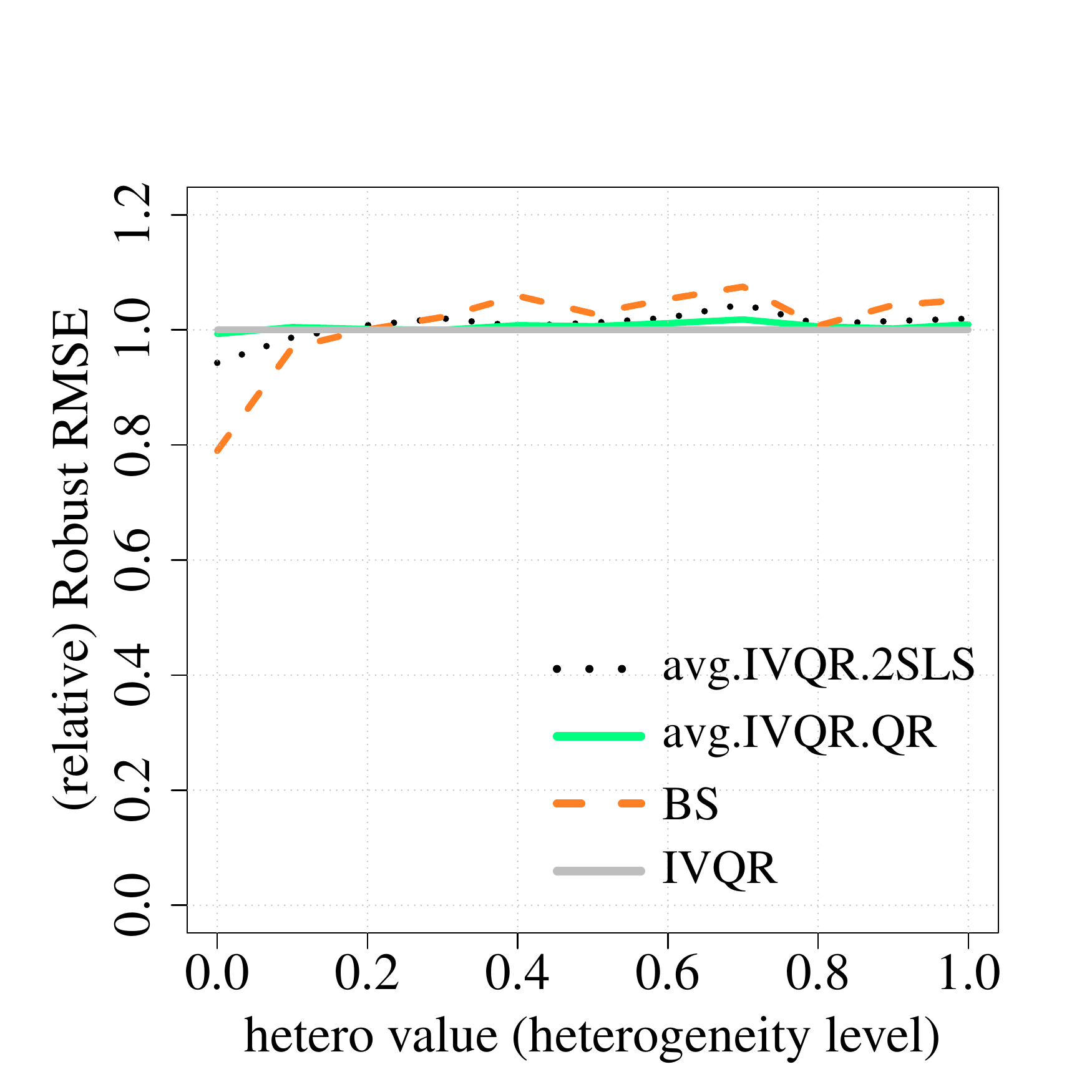}
\hfill
\includegraphics[width=0.45\textwidth, height=0.3\textheight, trim=35 20 20 70 ]{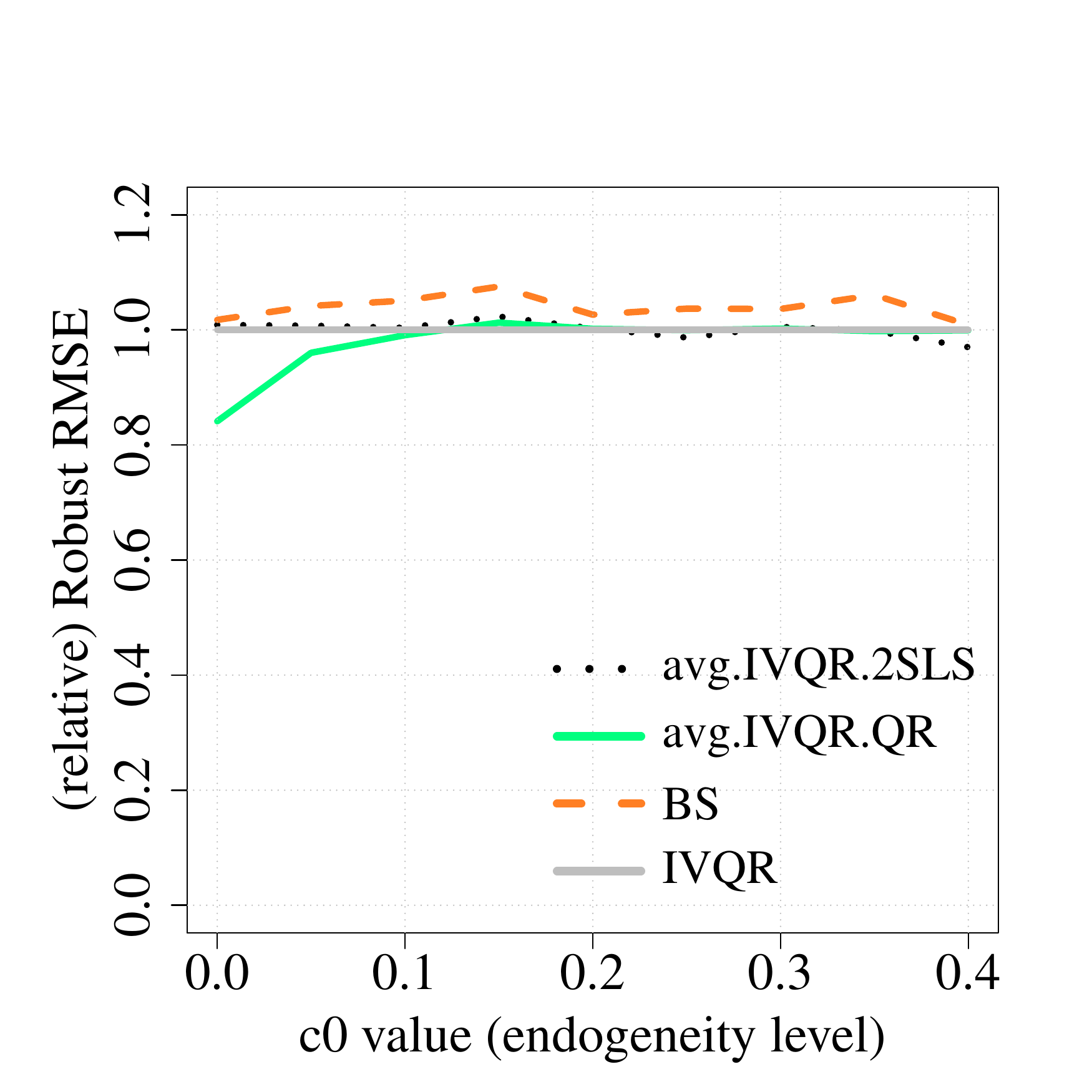}
\hfill\null
\includegraphics[width=0.45\textwidth, height=0.3\textheight, trim=35 20 20 70 ]{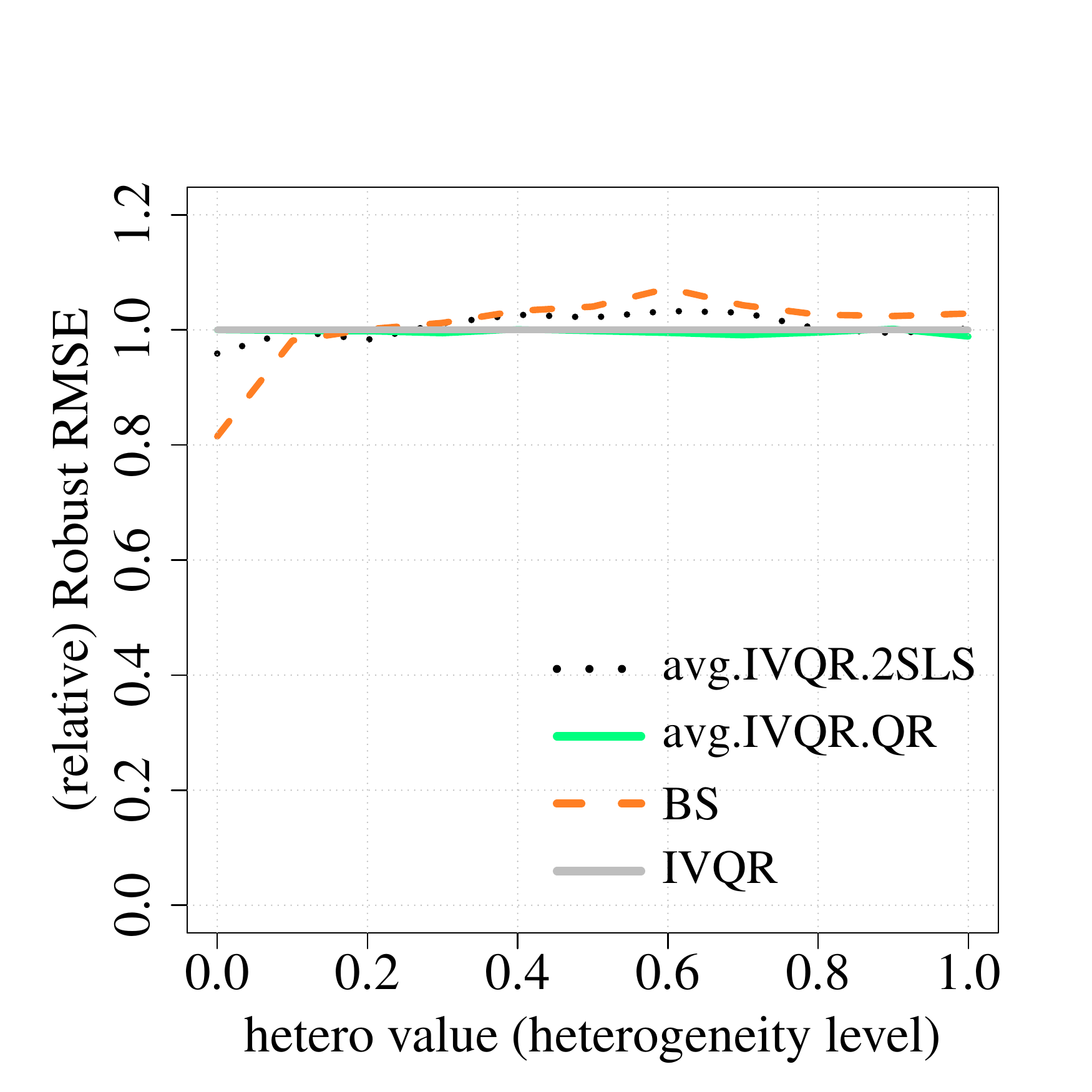}
\hfill
\includegraphics[width=0.45\textwidth, height=0.3\textheight, trim=35 20 20 70]{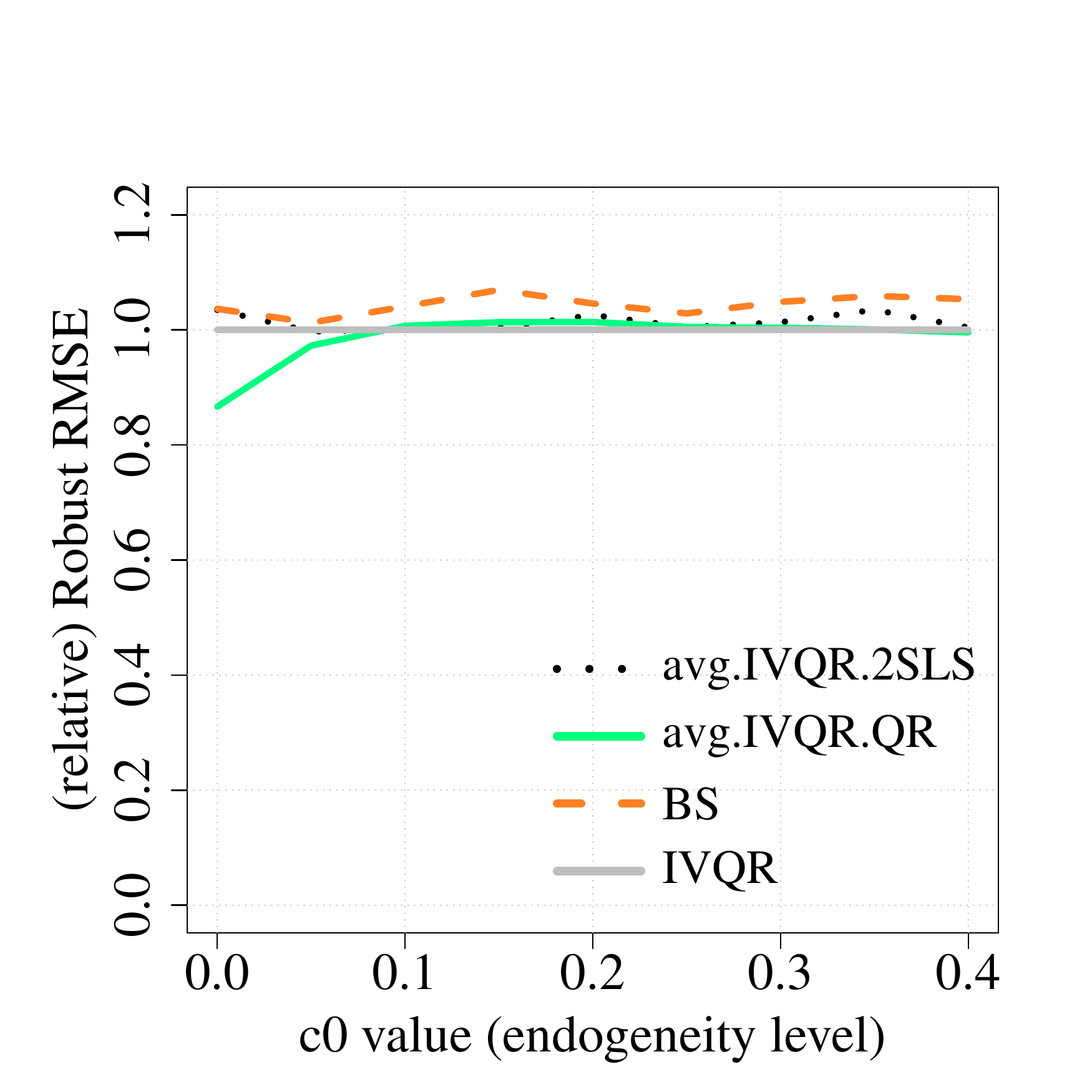}
\caption{\label{fig:M3:S1S22:tau0.2}%
Relative rRMSE in simulation model 3 at $\tau=0.2$ quantile level in 6 cases: fixed endogeneity level $c_0=0$  (left up), $c_0=0.2$ (left middle), $c_0=0.4$ (left bottom) and varying heterogeneity; and fixed heterogeneity level $hetero=0$ (right up), $hetero=0.5$ (right middle), $hetero=1$ (right bottom) and varying endogeneity,
based on $\num{200}$ replications and $\num{50}$ bootstraps. Sample size n=1000.} 
\end{figure}

\begin{figure}[htbp]
\includegraphics[width=0.45\textwidth, height=0.3\textheight, trim=35 20 20 70]{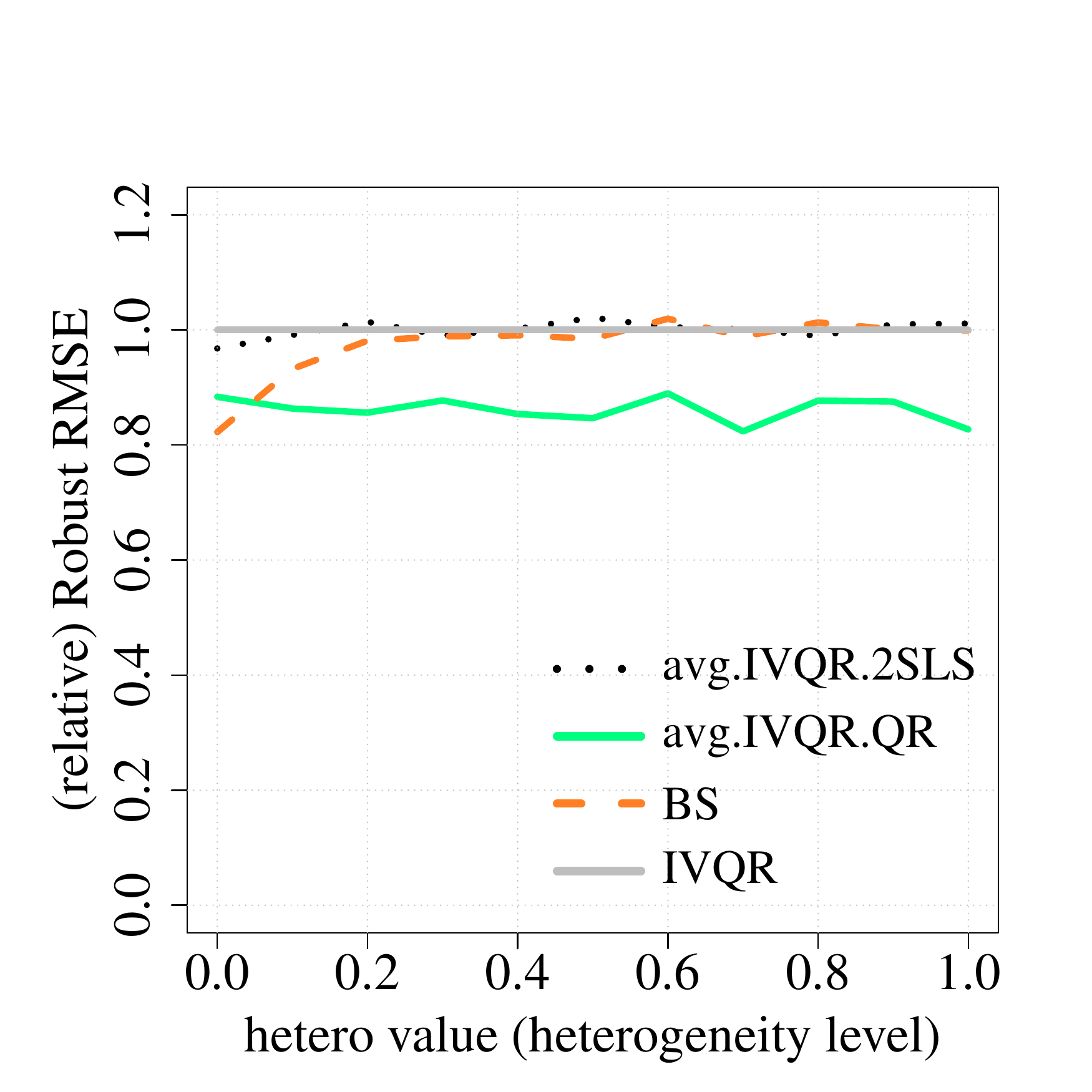}
\hfill
\includegraphics[width=0.45\textwidth, height=0.3\textheight, trim=35 20 20 70]{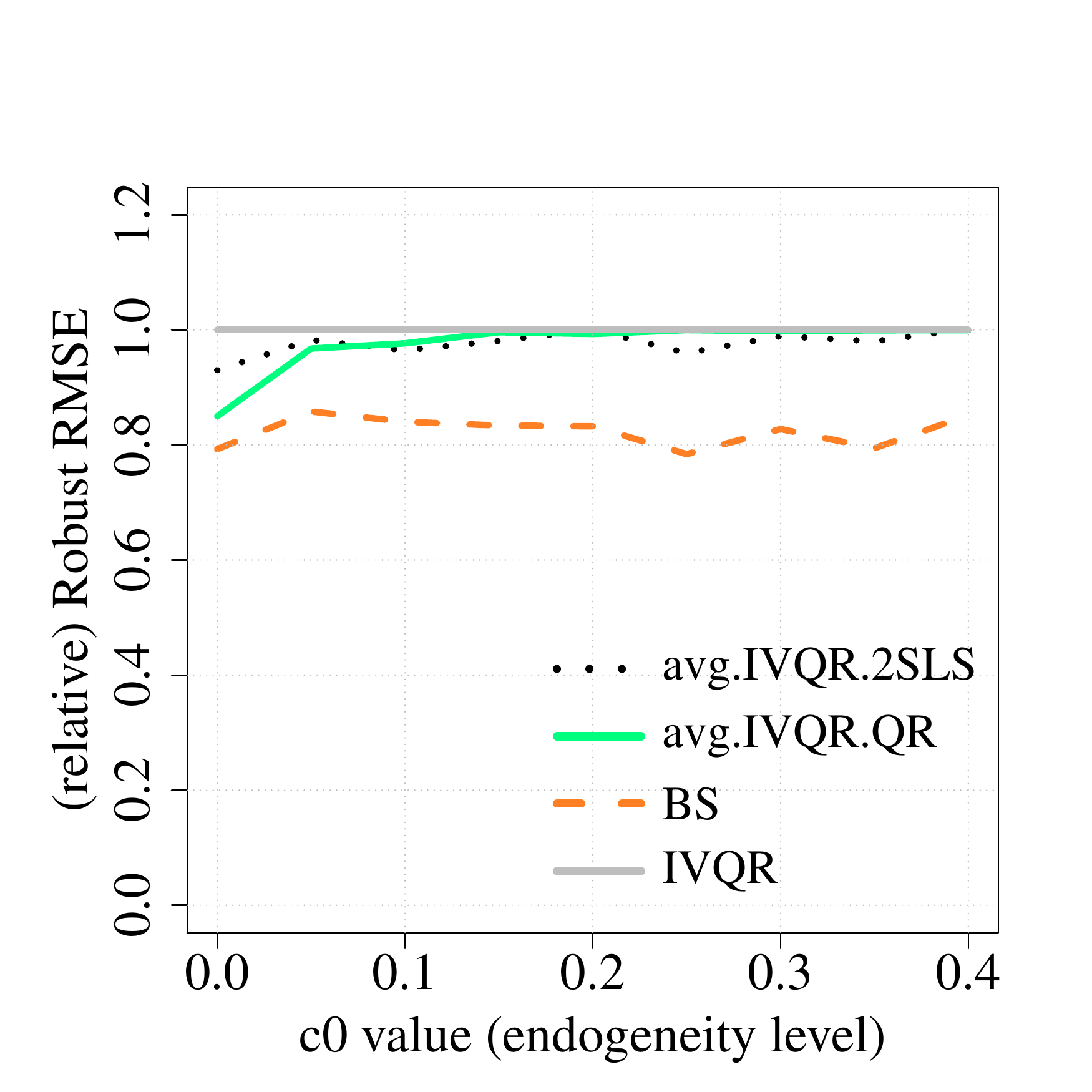}
\hfill\null
\includegraphics[width=0.45\textwidth, height=0.3\textheight, trim=35 20 20 70]{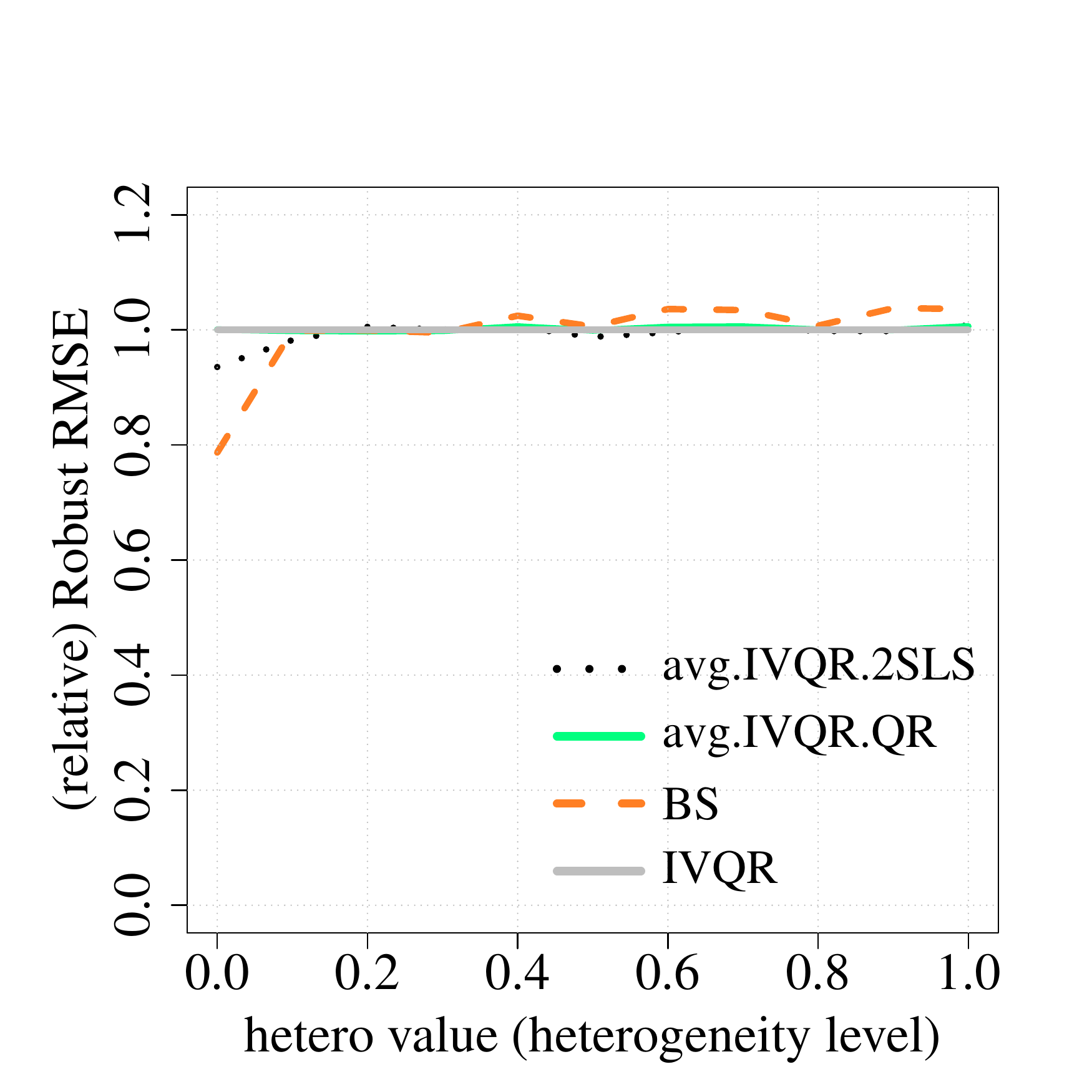}
\hfill
\includegraphics[width=0.45\textwidth, height=0.3\textheight, trim=35 20 20 70 ]{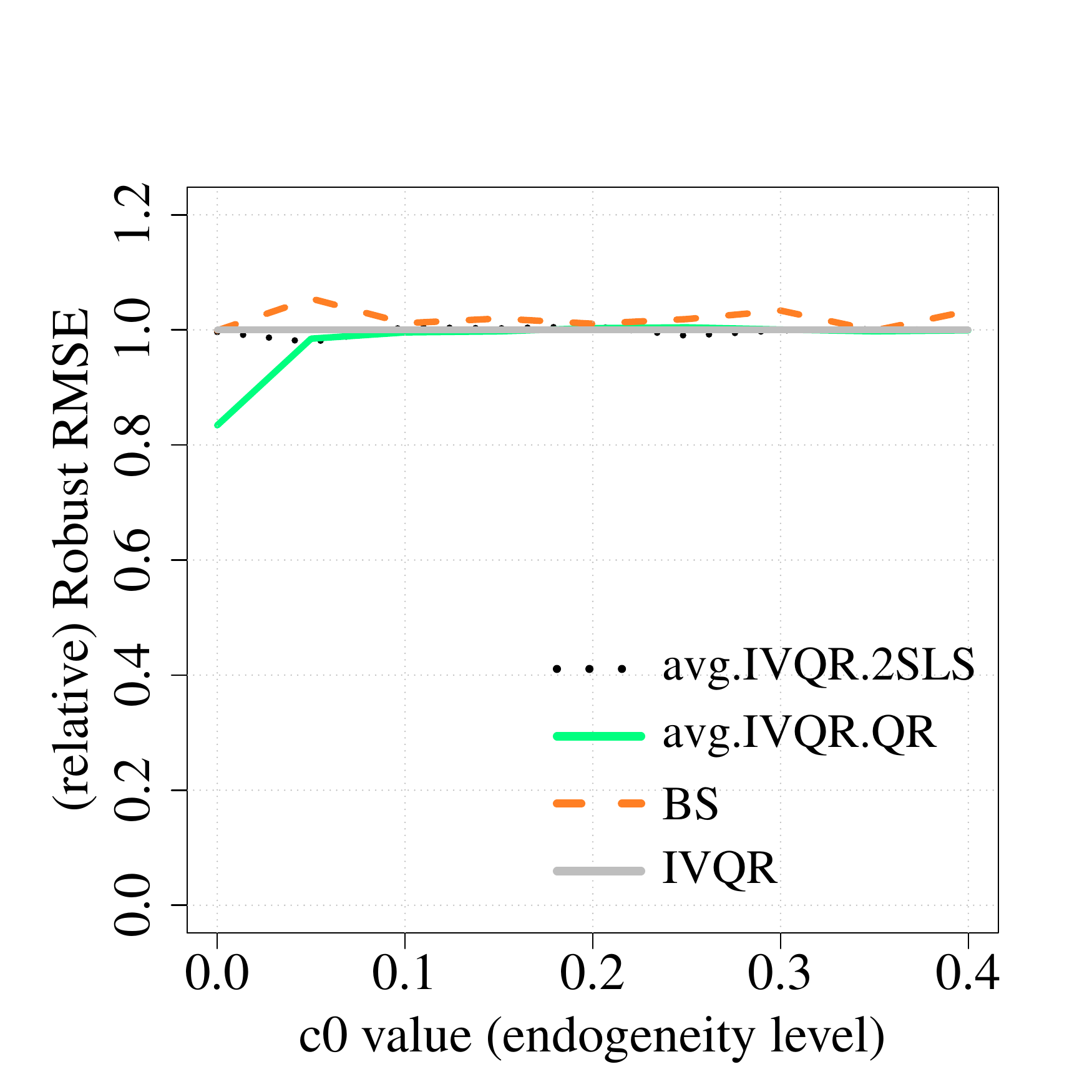}
\hfill\null
\includegraphics[width=0.45\textwidth, height=0.3\textheight, trim=35 20 20 70 ]{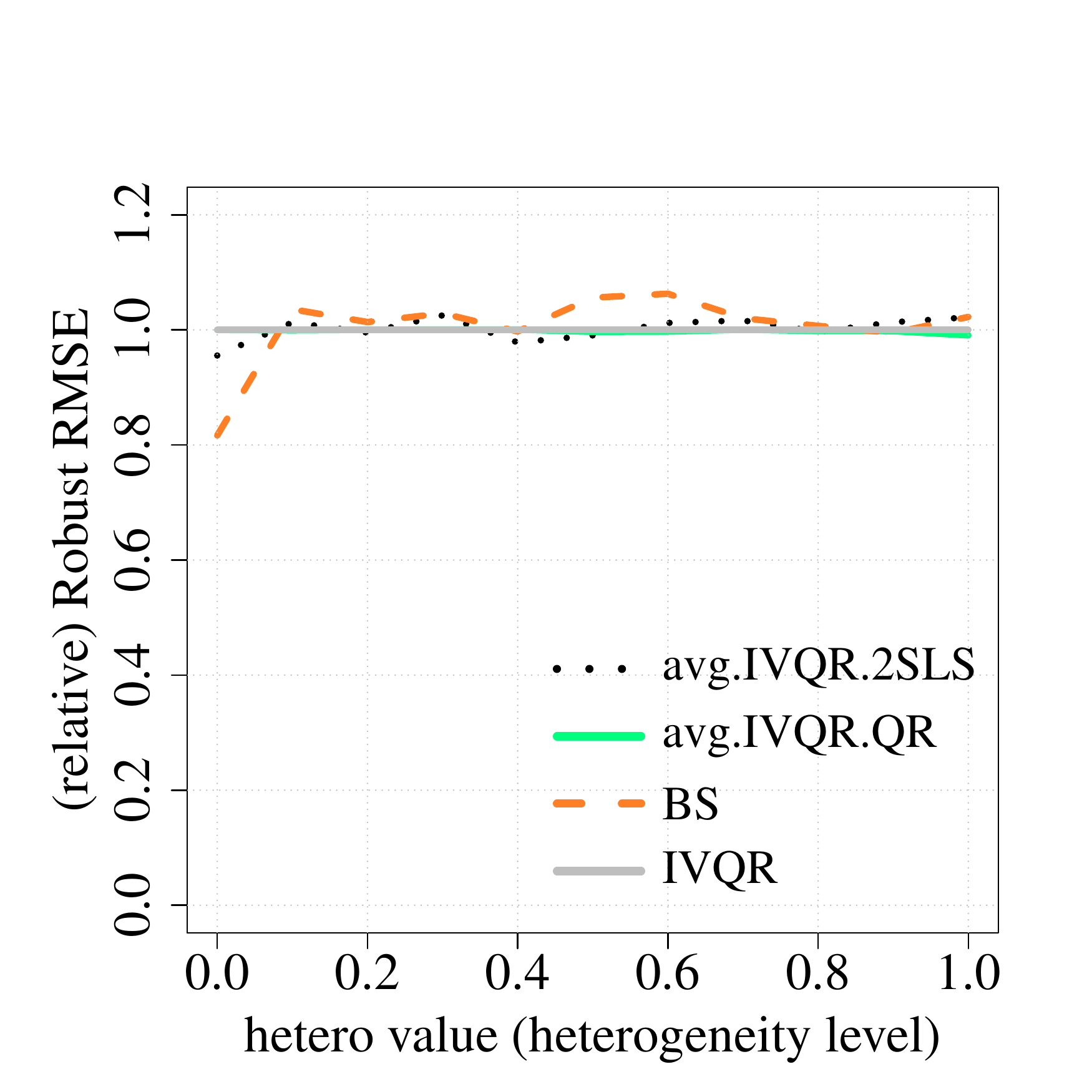}
\hfill
\includegraphics[width=0.45\textwidth, height=0.3\textheight, trim=35 20 20 70]{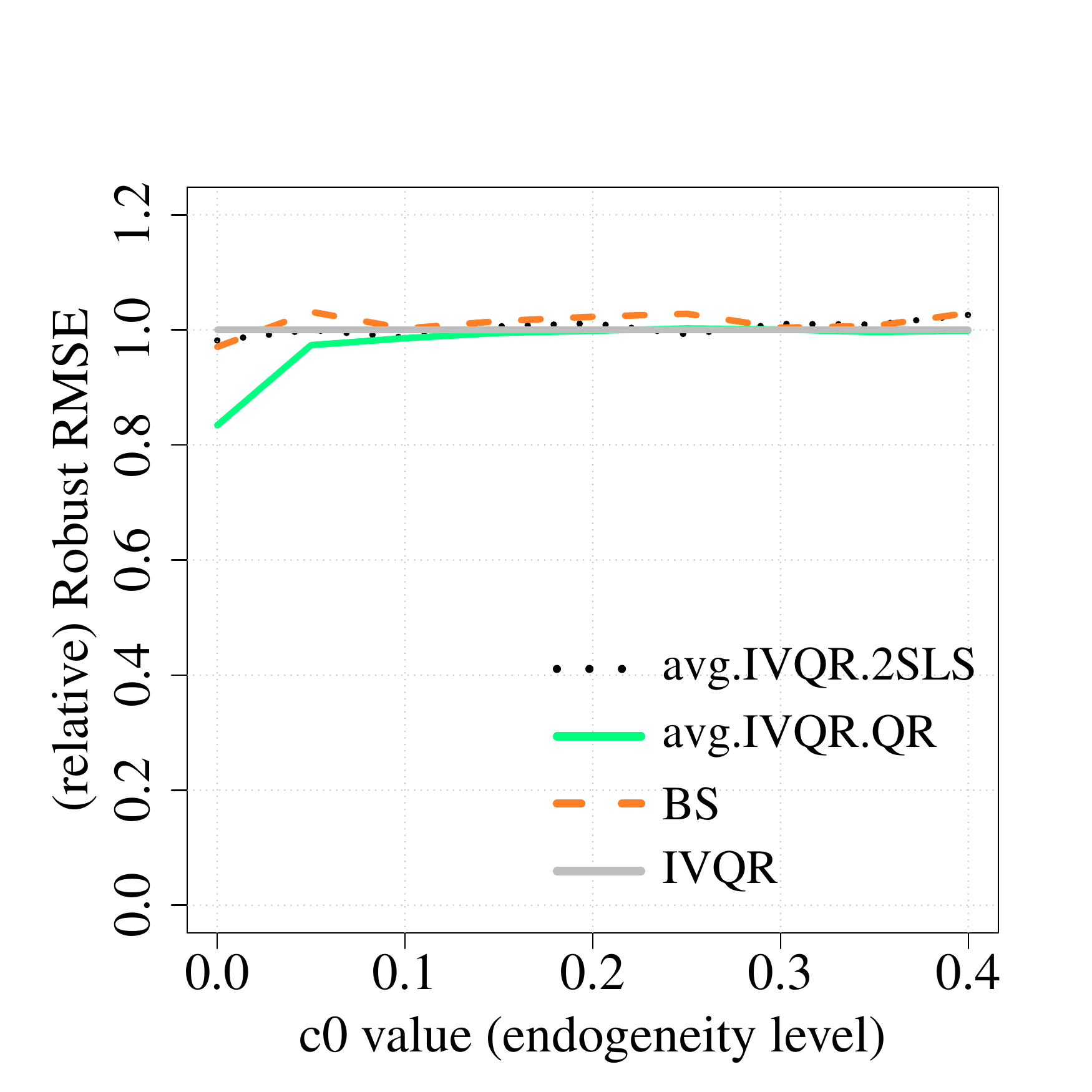}
\caption{\label{fig:M3:S1S22:tau0.3}%
Relative rRMSE in simulation model 3 at $\tau=0.3$ quantile level in 6 cases: fixed endogeneity level $c_0=0$  (left up), $c_0=0.2$ (left middle), $c_0=0.4$ (left bottom) and varying heterogeneity; and fixed heterogeneity level $hetero=0$ (right up), $hetero=0.5$ (right middle), $hetero=1$ (right bottom) and varying endogeneity,
based on $\num{200}$ replications and $\num{50}$ bootstraps. Sample size n=1000.} 
\end{figure}

\begin{figure}[htbp]
\includegraphics[width=0.45\textwidth, height=0.3\textheight, trim=35 20 20 70]{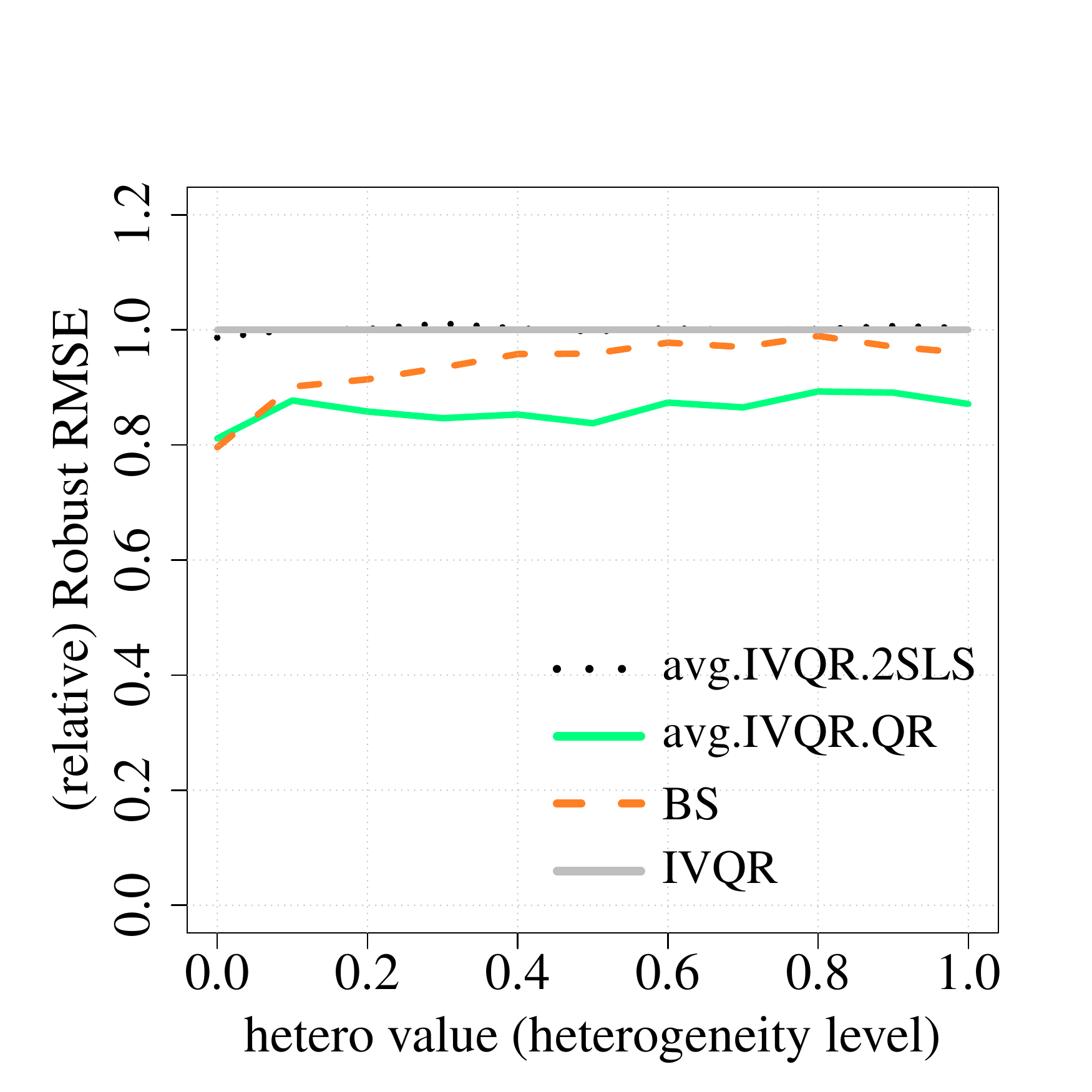}
\hfill
\includegraphics[width=0.45\textwidth, height=0.3\textheight, trim=35 20 20 70]{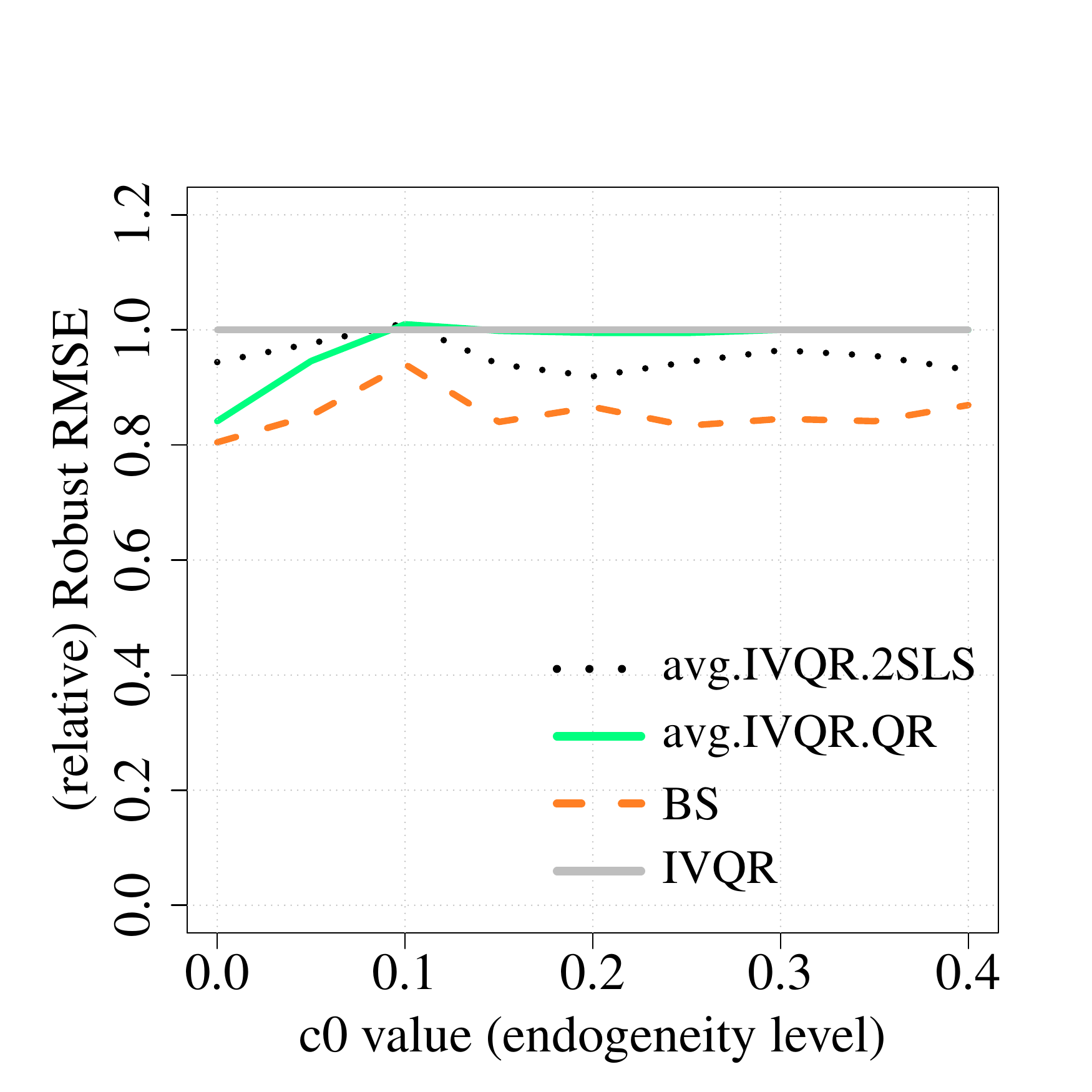}
\hfill\null
\includegraphics[width=0.45\textwidth, height=0.3\textheight, trim=35 20 20 70]{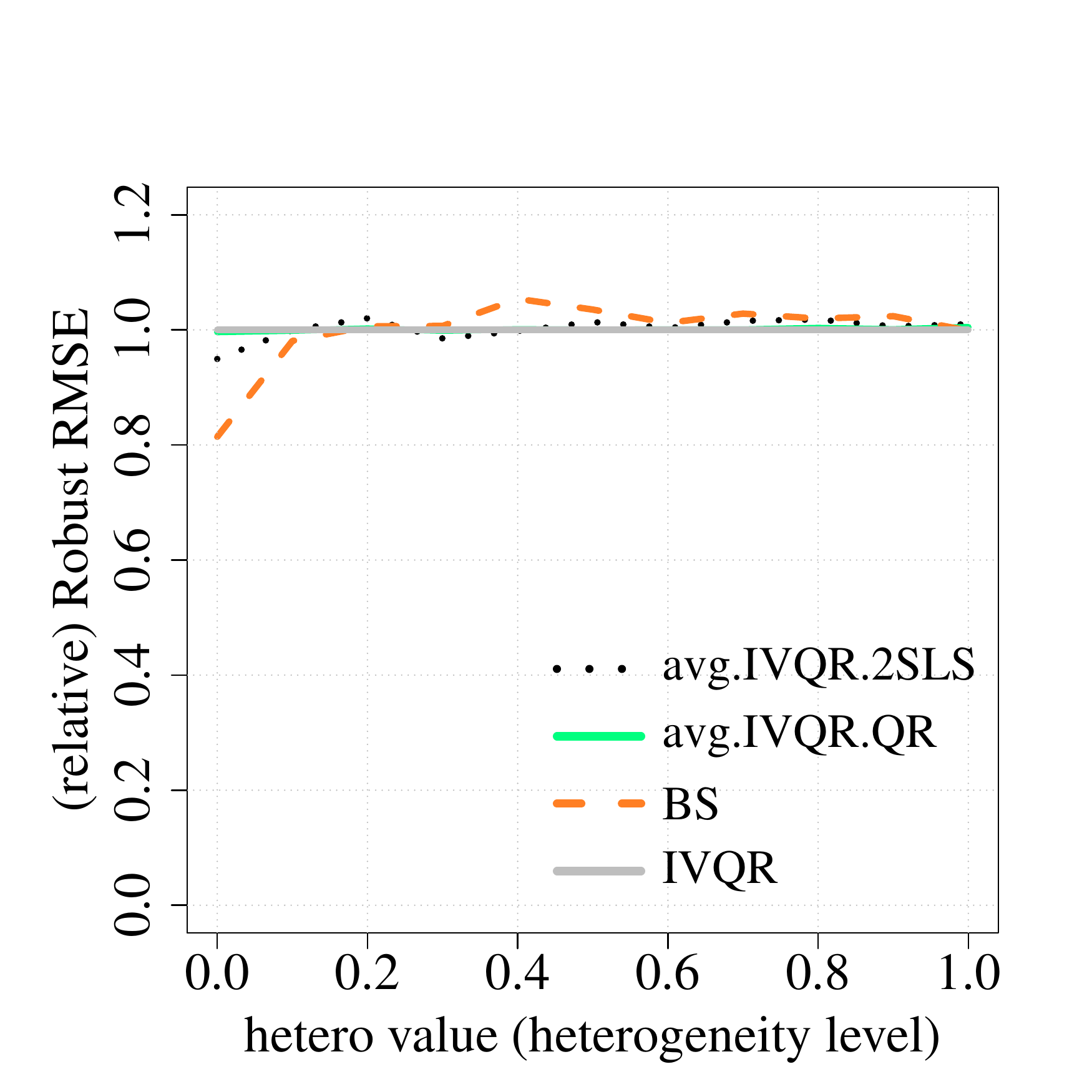}
\hfill
\includegraphics[width=0.45\textwidth, height=0.3\textheight, trim=35 20 20 70 ]{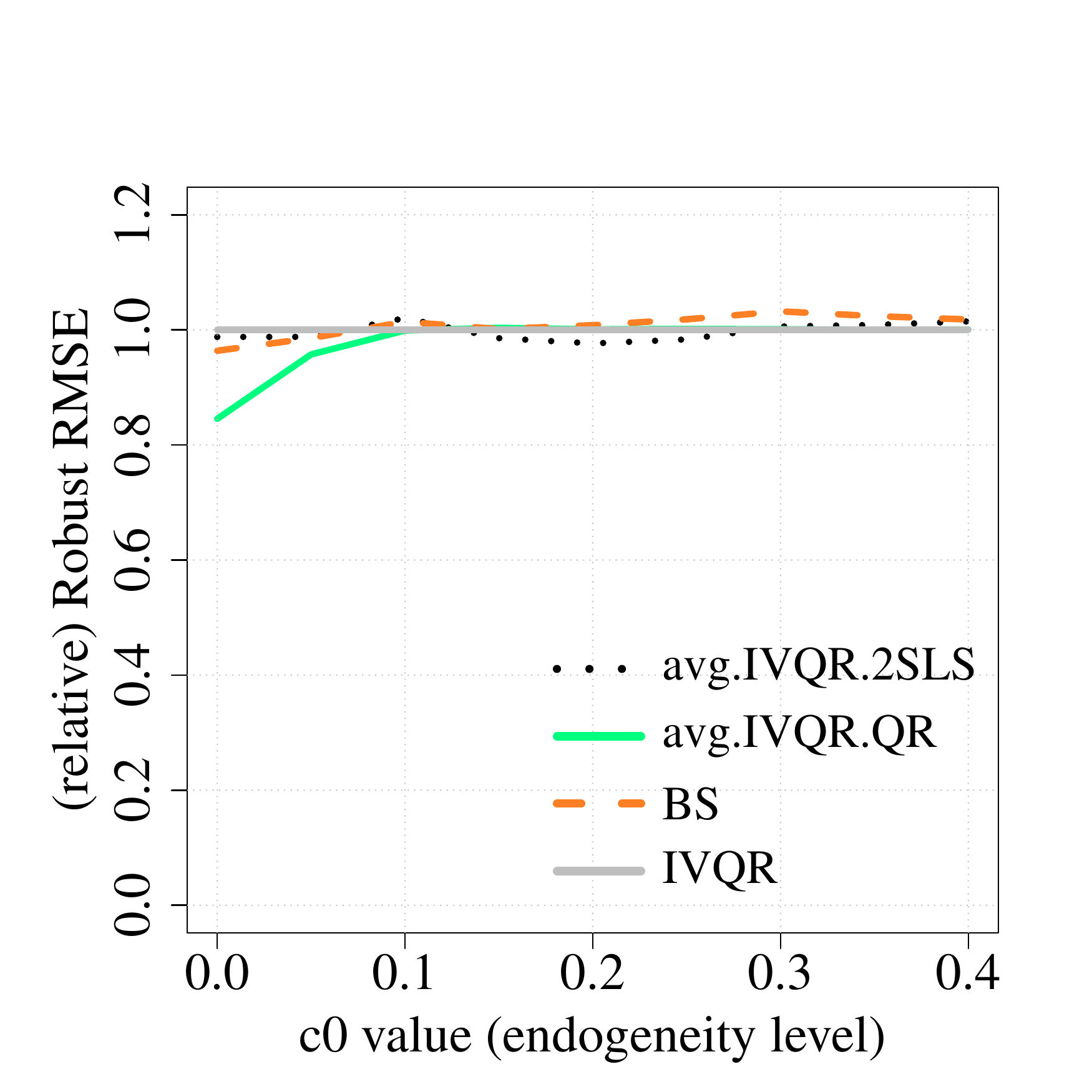}
\hfill\null
\includegraphics[width=0.45\textwidth, height=0.3\textheight, trim=35 20 20 70 ]{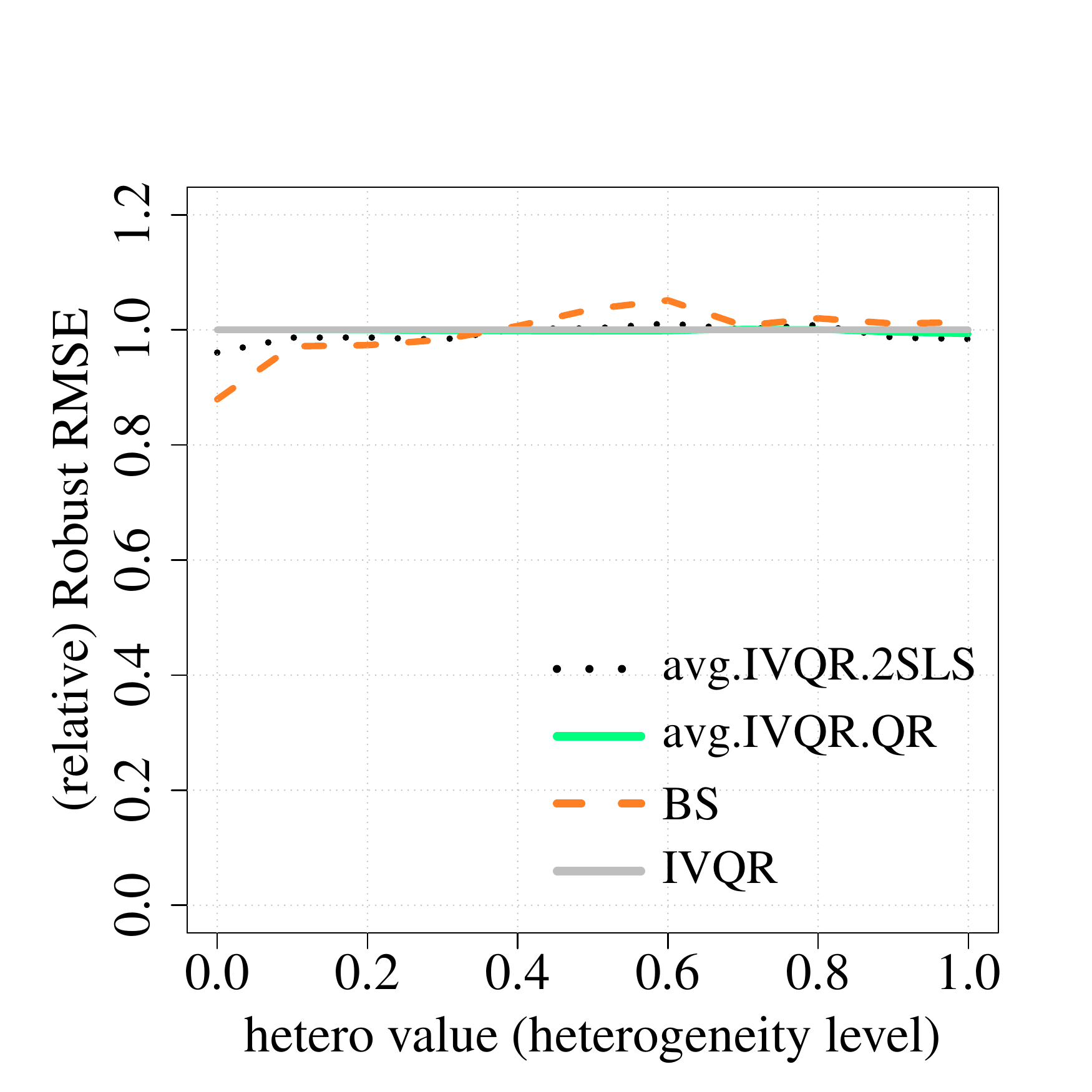}
\hfill
\includegraphics[width=0.45\textwidth, height=0.3\textheight, trim=35 20 20 70]{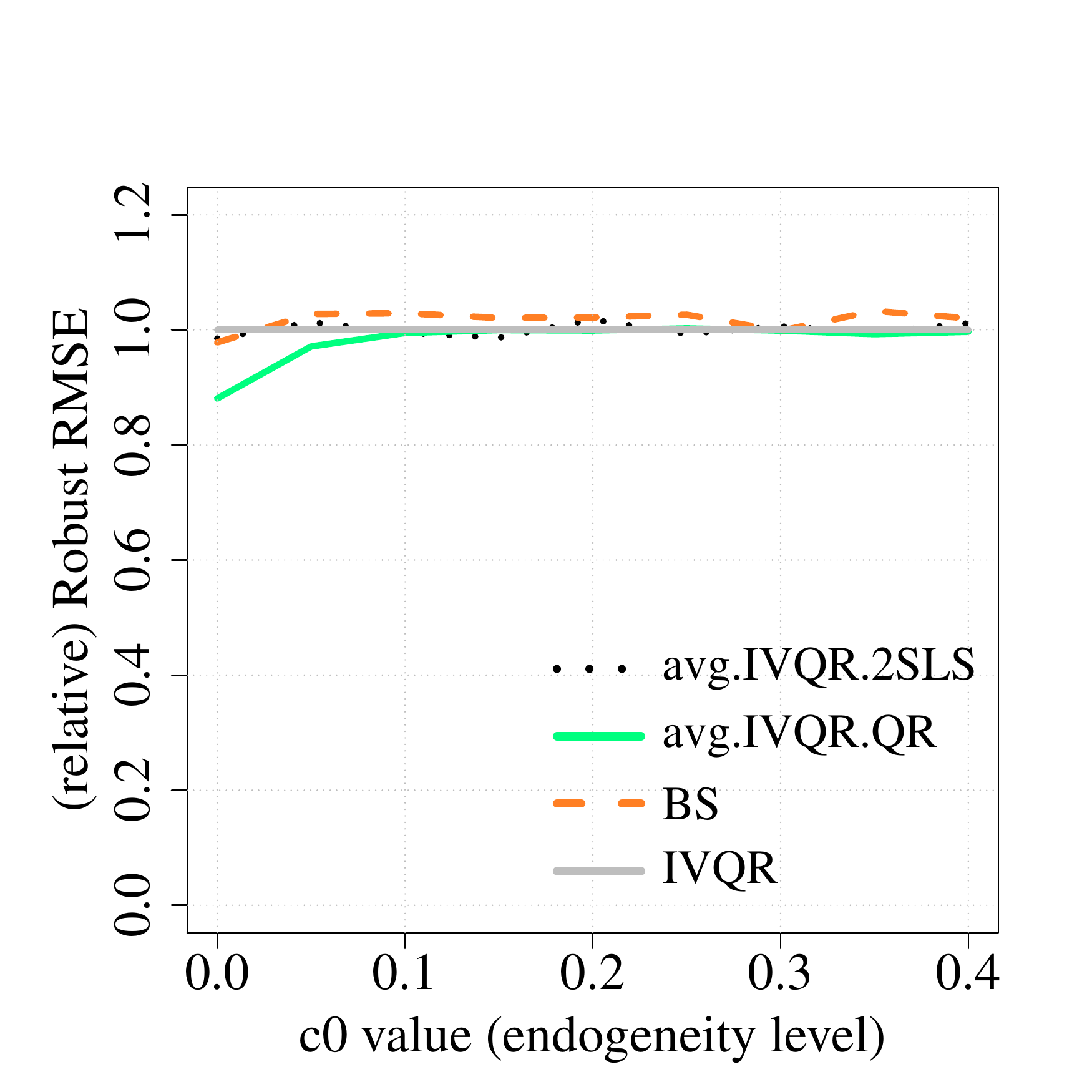}
\caption{\label{fig:M3:S1S22:tau0.4}%
Relative rRMSE in simulation model 3 at $\tau=0.4$ quantile level in 6 cases: fixed endogeneity level $c_0=0$  (left up), $c_0=0.2$ (left middle), $c_0=0.4$ (left bottom) and varying heterogeneity; and fixed heterogeneity level $hetero=0$ (right up), $hetero=0.5$ (right middle), $hetero=1$ (right bottom) and varying endogeneity,
based on $\num{200}$ replications and $\num{50}$ bootstraps. Sample size n=1000.} 
\end{figure}

\begin{figure}[htbp]
\includegraphics[width=0.45\textwidth, height=0.3\textheight, trim=35 20 20 70]{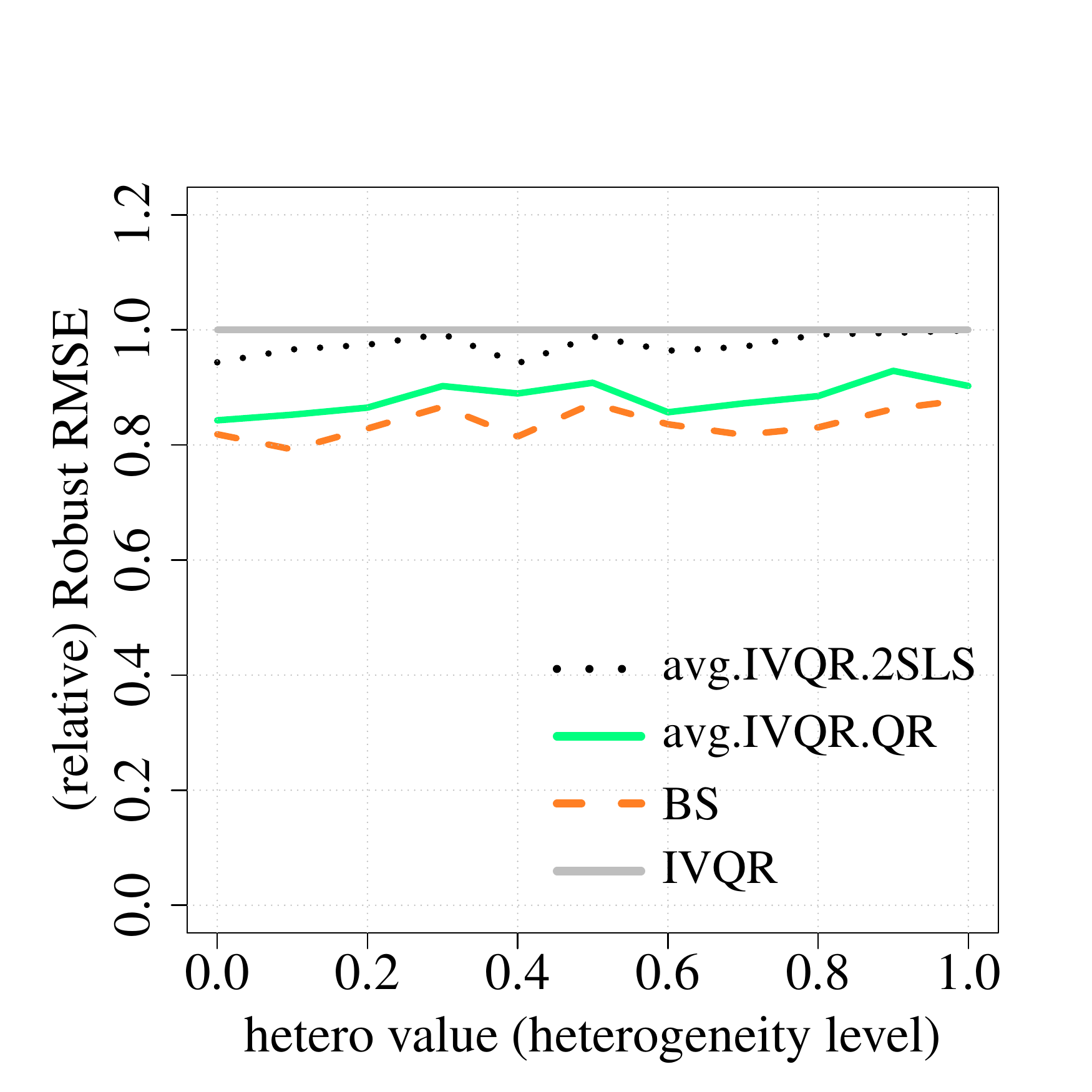}
\hfill
\includegraphics[width=0.45\textwidth, height=0.3\textheight, trim=35 20 20 70]{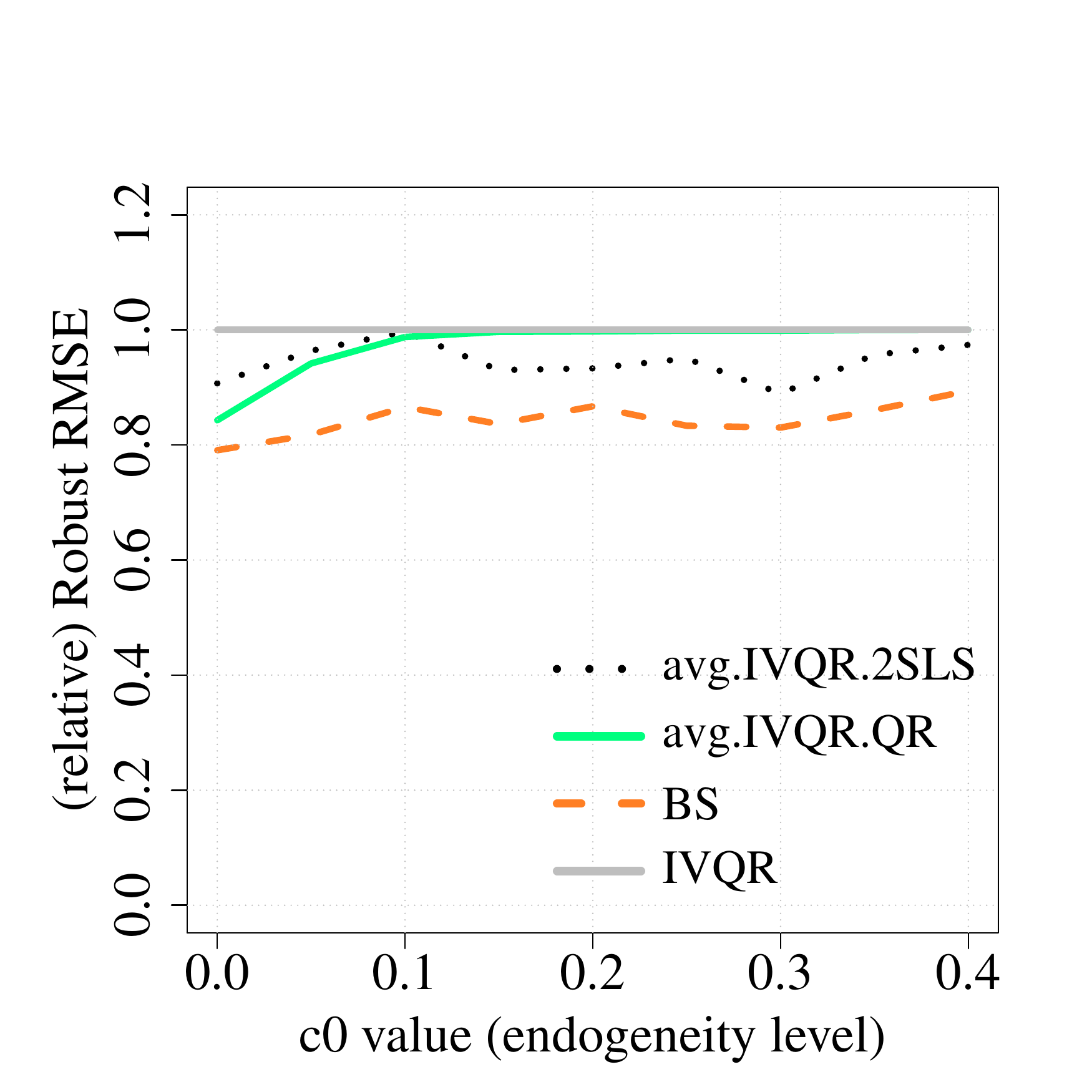}
\hfill\null
\includegraphics[width=0.45\textwidth, height=0.3\textheight, trim=35 20 20 70]{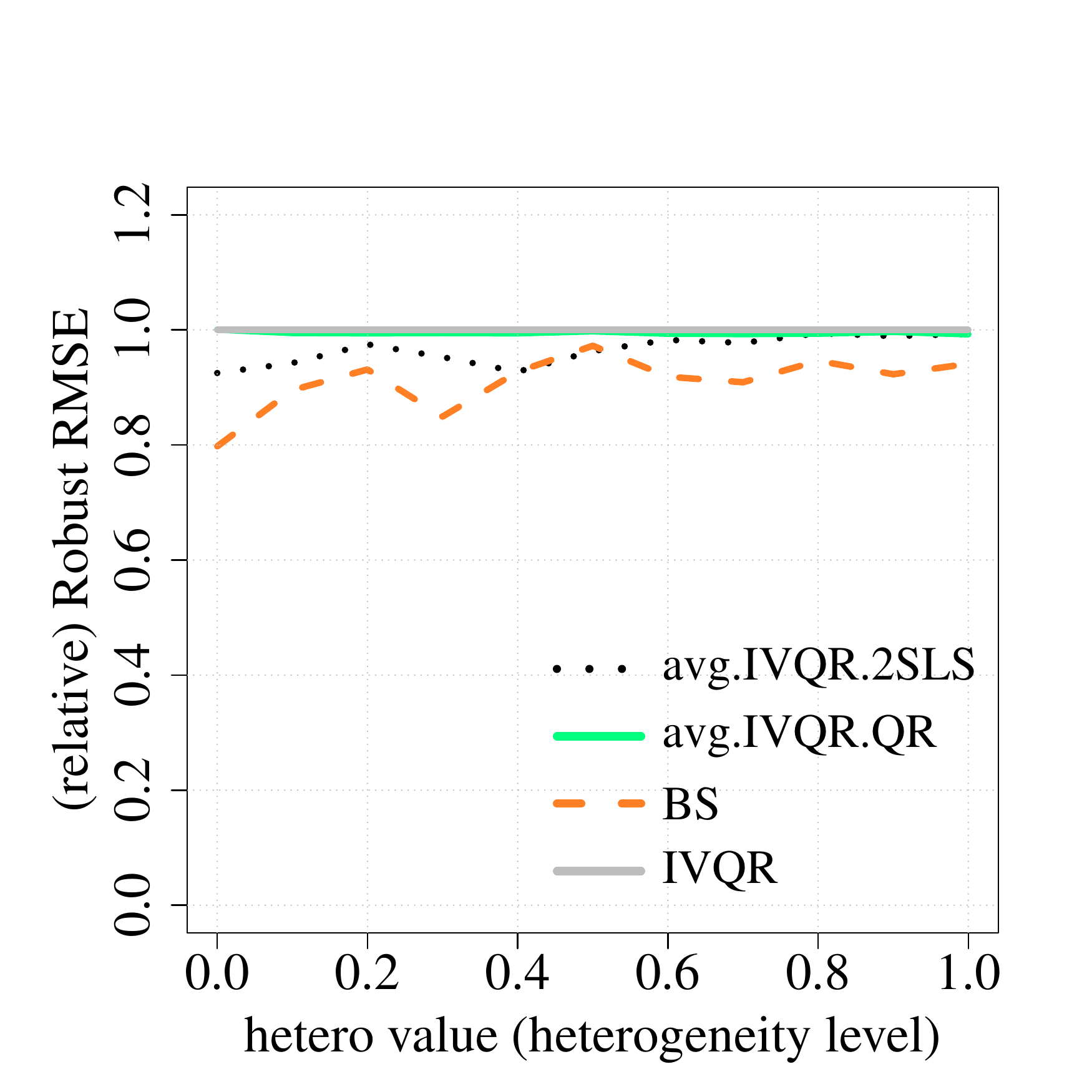}
\hfill
\includegraphics[width=0.45\textwidth, height=0.3\textheight, trim=35 20 20 70 ]{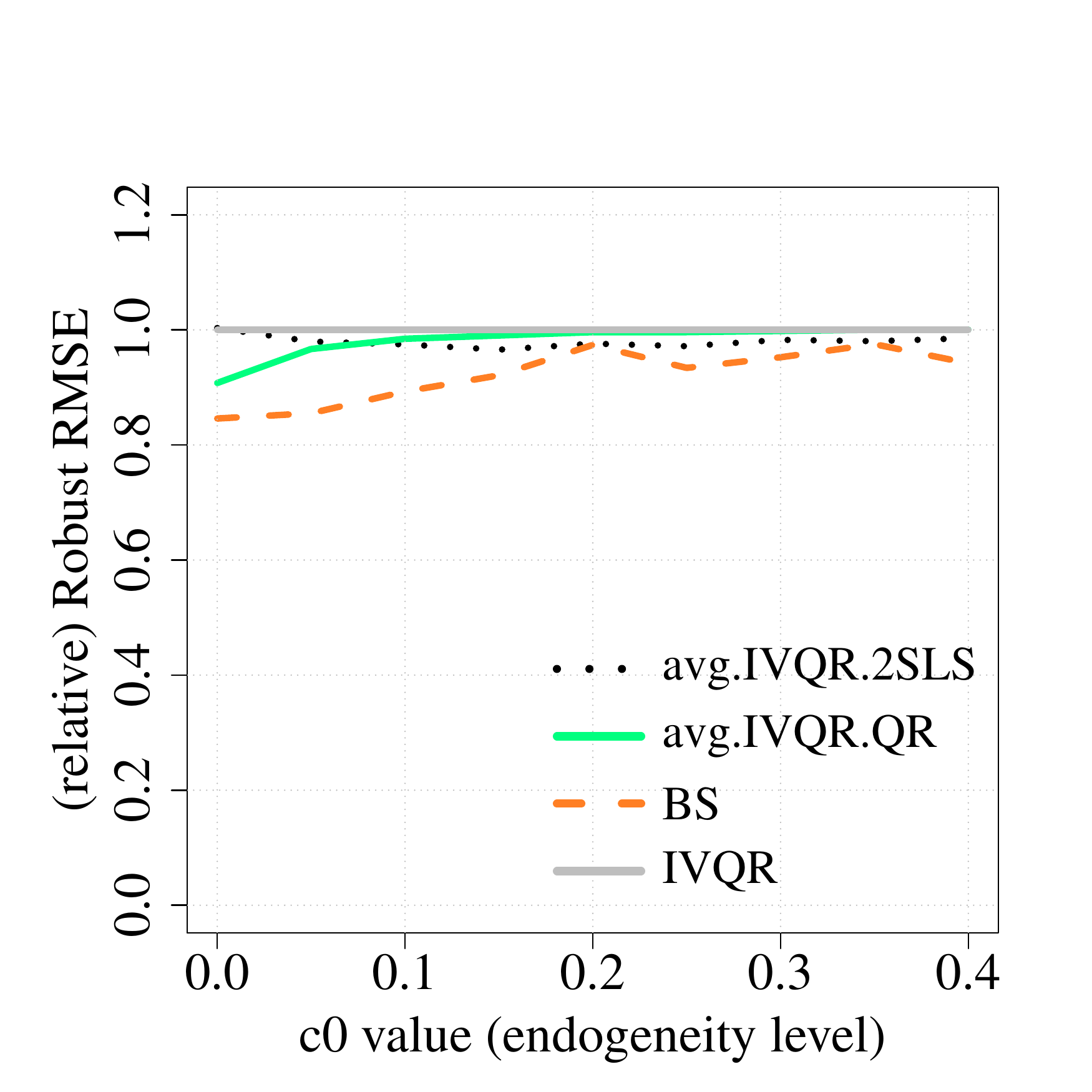}
\hfill\null
\includegraphics[width=0.45\textwidth, height=0.3\textheight, trim=35 20 20 70 ]{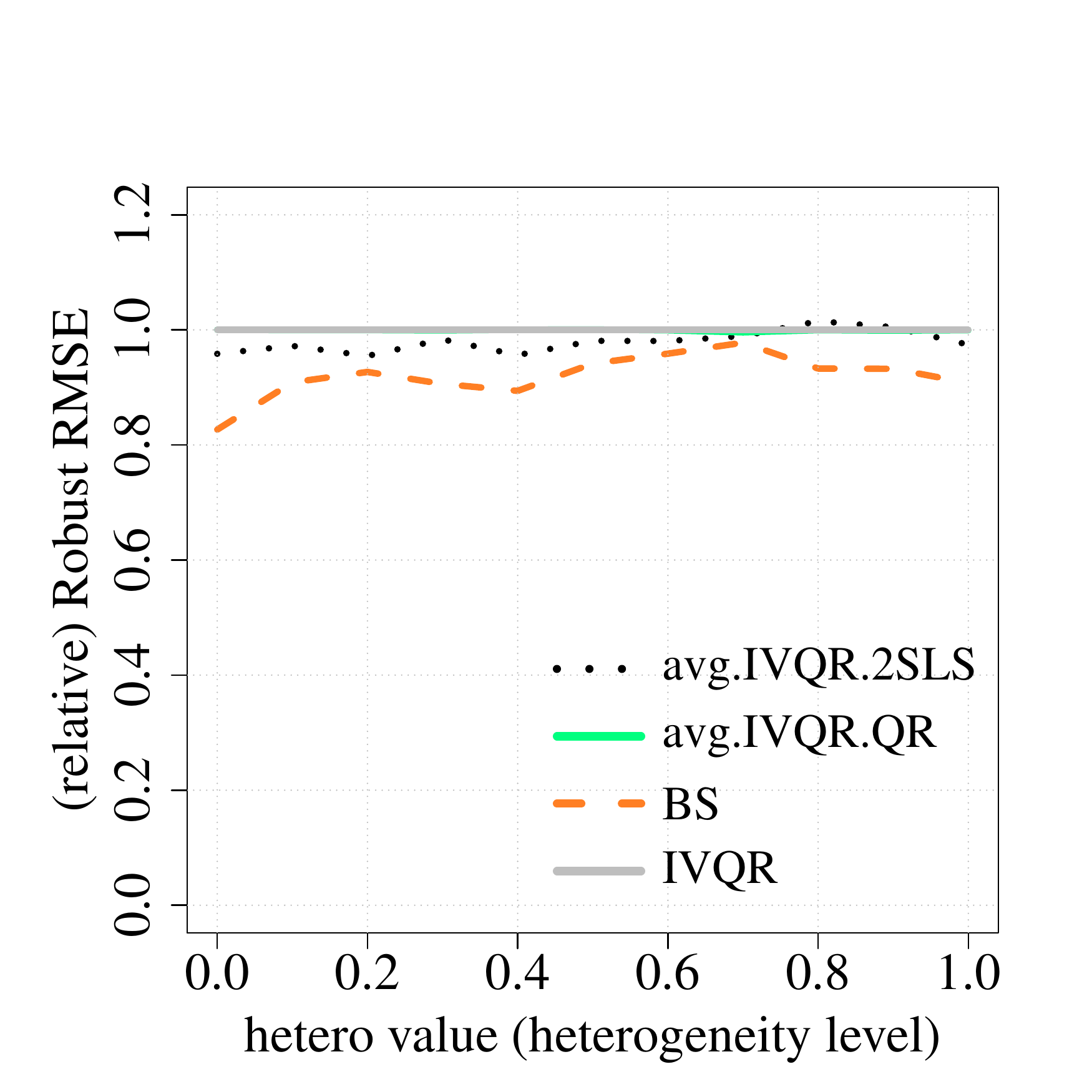}
\hfill
\includegraphics[width=0.45\textwidth, height=0.3\textheight, trim=35 20 20 70]{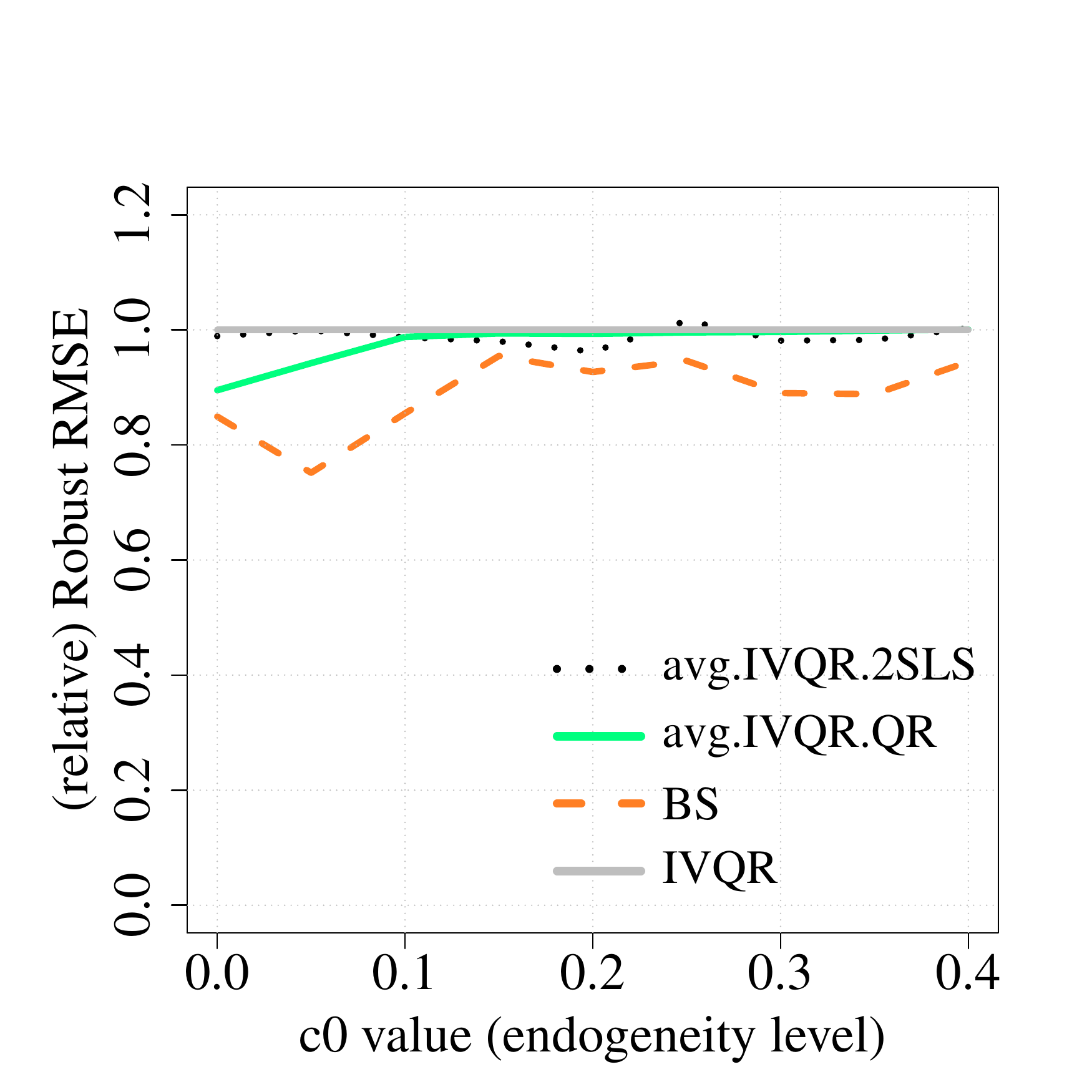}
\caption{\label{fig:M3:S1S22:tau0.6}%
Relative rRMSE in simulation model 3 at $\tau=0.6$ quantile level in 6 cases: fixed endogeneity level $c_0=0$  (left up), $c_0=0.2$ (left middle), $c_0=0.4$ (left bottom) and varying heterogeneity; and fixed heterogeneity level $hetero=0$ (right up), $hetero=0.5$ (right middle), $hetero=1$ (right bottom) and varying endogeneity,
based on $\num{200}$ replications and $\num{50}$ bootstraps. Sample size n=1000.} 
\end{figure}

\begin{figure}[htbp]
\includegraphics[width=0.45\textwidth, height=0.3\textheight, trim=35 20 20 70]{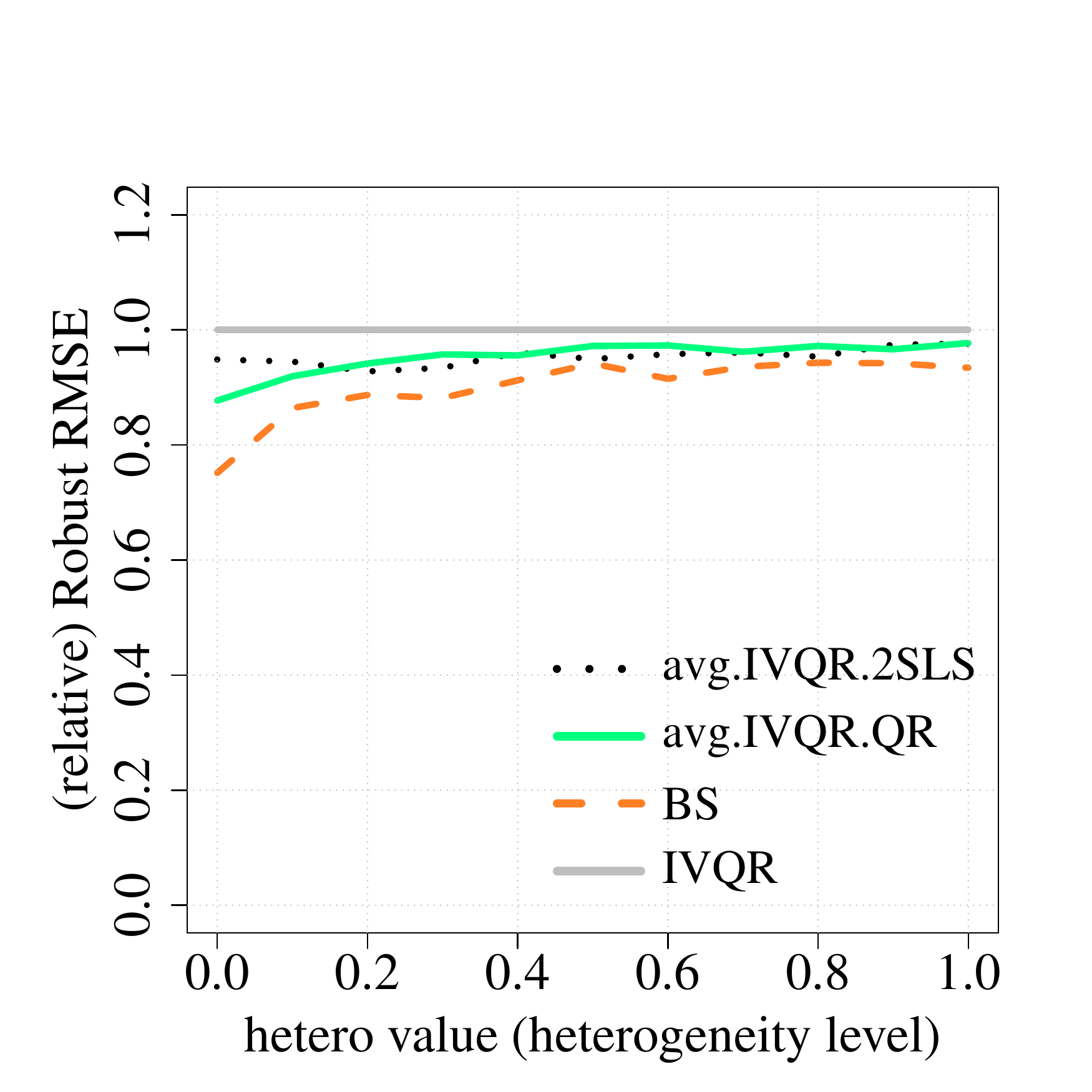}
\hfill
\includegraphics[width=0.45\textwidth, height=0.3\textheight, trim=35 20 20 70]{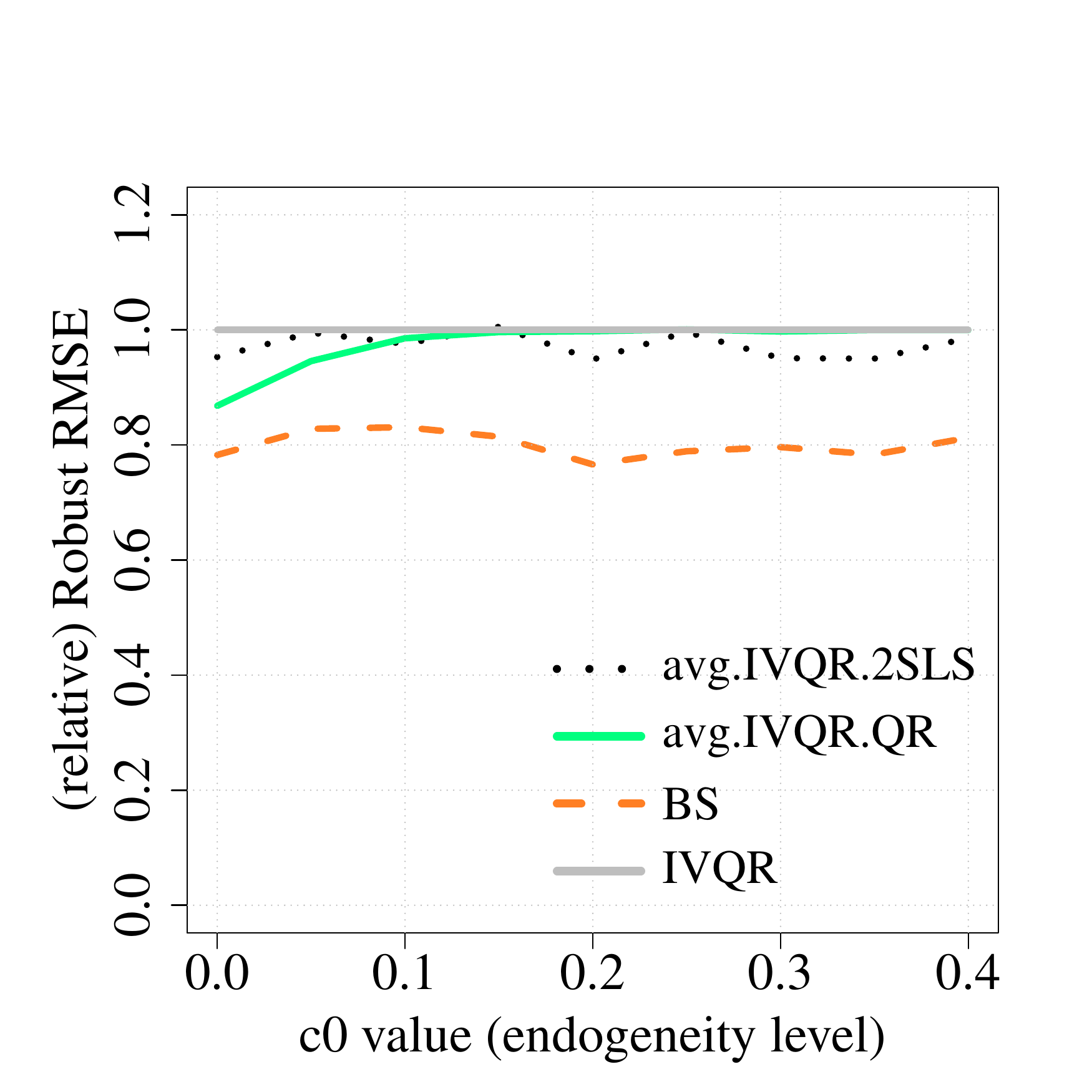}
\hfill\null
\includegraphics[width=0.45\textwidth, height=0.3\textheight, trim=35 20 20 70]{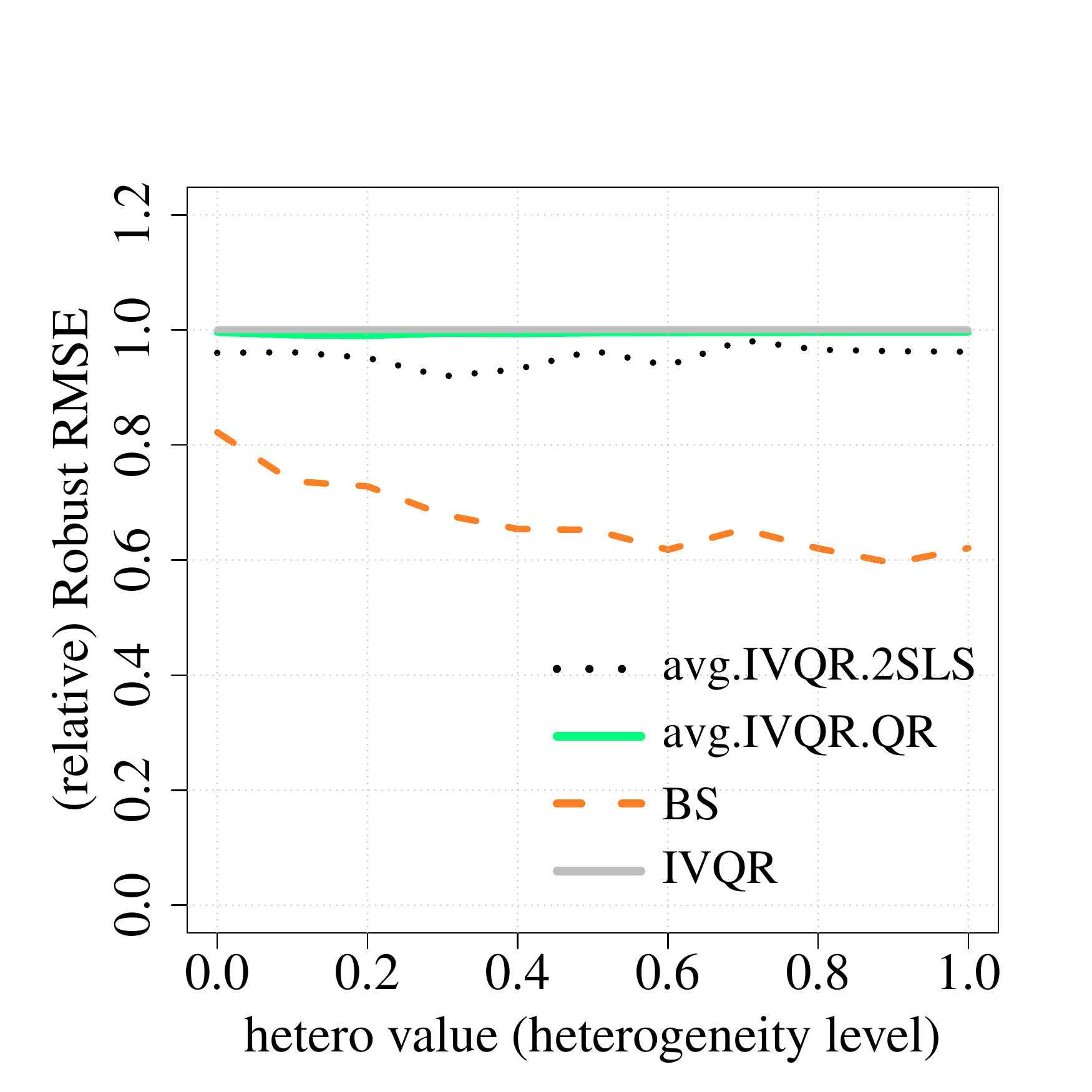}
\hfill
\includegraphics[width=0.45\textwidth, height=0.3\textheight, trim=35 20 20 70 ]{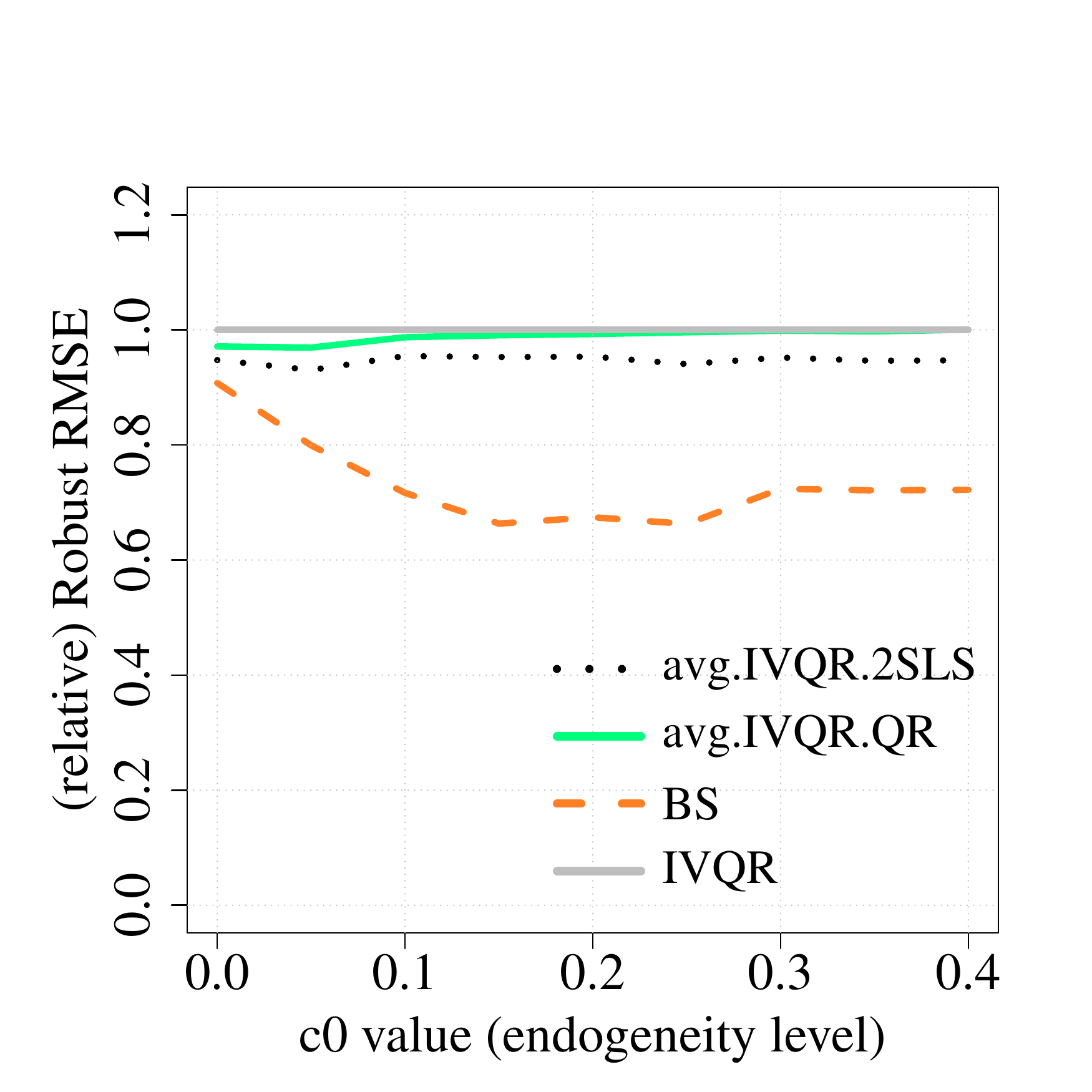}
\hfill\null
\includegraphics[width=0.45\textwidth, height=0.3\textheight, trim=35 20 20 70 ]{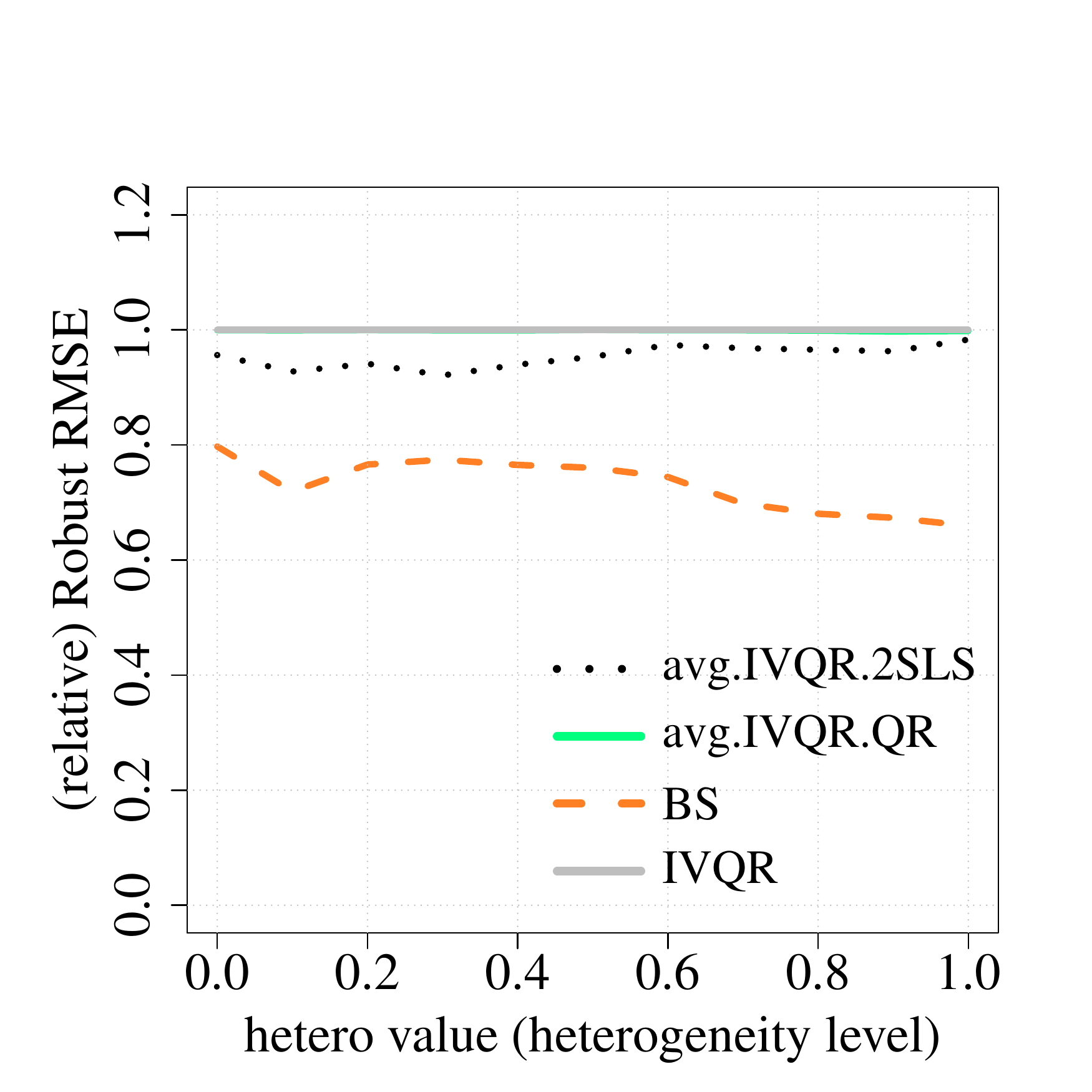}
\hfill
\includegraphics[width=0.45\textwidth, height=0.3\textheight, trim=35 20 20 70]{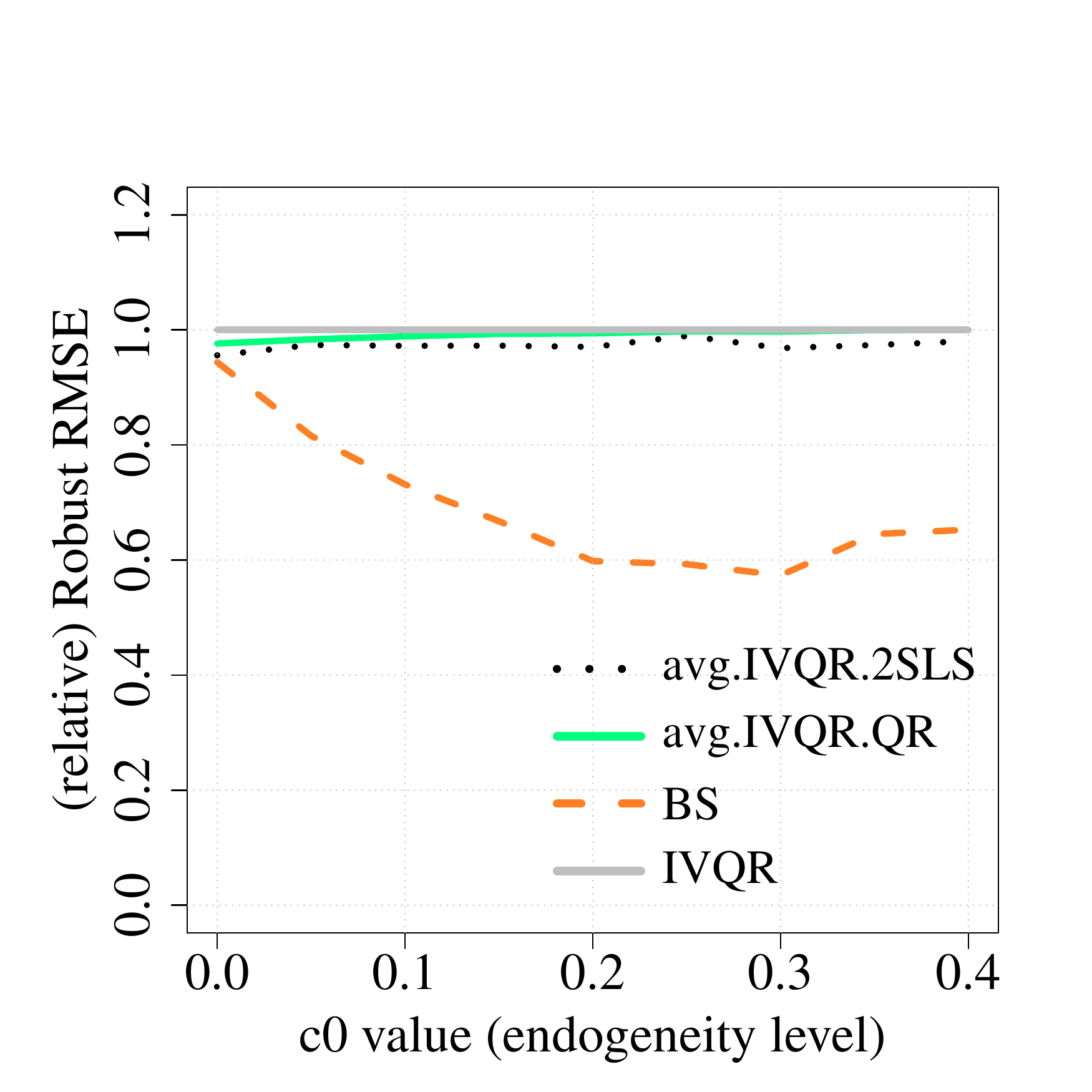}
\caption{\label{fig:M3:S1S22:tau0.8}%
Relative rRMSE in simulation model 3 at $\tau=0.8$ quantile level in 6 cases: fixed endogeneity level $c_0=0$  (left up), $c_0=0.2$ (left middle), $c_0=0.4$ (left bottom) and varying heterogeneity; and fixed heterogeneity level $hetero=0$ (right up), $hetero=0.5$ (right middle), $hetero=1$ (right bottom) and varying endogeneity,
based on $\num{200}$ replications and $\num{50}$ bootstraps. Sample size n=1000.} 
\end{figure}

\end{document}